\documentclass[a4paper,12pt]{book}
\usepackage{ucs}
\usepackage[utf8x]{inputenc} 

\usepackage{multicol} 
\usepackage{mathrsfs} 
\usepackage{amsmath, amsthm, amssymb,bm}

\RequirePackage[T1]{fontenc}
\usepackage{amsfonts}

\usepackage[
  breaklinks,
  naturalnames,
  ]{hyperref}

\usepackage{amsrefs}

\makeatletter
\renewcommand\PrintNames@a[4]{%
    \PrintSeries{\name}
        {#1}
        {}{ y \set@othername}
        {,}{ \set@othername}
        {,}{ y \set@othername}
        {#2}{#4}{#3}%
}

\BibSpec{article}{%
    +{}  {\PrintAuthors}                {author}
    +{,} { \textit}                     {title}
    +{.} { }                            {part}
    +{:} { \textit}                     {subtitle}
    +{,} { \PrintContributions}         {contribution}
    +{.} { \PrintPartials}              {partial}
    +{,} { }                            {journal}
    +{}  { \textbf}                     {volume}
    +{}  { \PrintDatePV}                {date}
    +{,} { \issuetext}                  {number}
    +{,} { \eprintpages}                {pages}
    +{,} { }                            {status}
    +{,} { \PrintDOI}                   {doi}
    +{,} { available at \eprint}        {eprint}
    +{}  { \parenthesize}               {language}
    +{}  { \PrintTranslation}           {translation}
    +{;} { \PrintReprint}               {reprint}
    +{.} { }                            {note}
    +{.} {}                             {transition}
    +{}  {\SentenceSpace \PrintReviews} {review}
}
\makeatother

\usepackage[all]{xy}

\usepackage{braket} 
\usepackage{enumitem}

\newtheorem{thm}{Theorem}[section]
\newtheorem*{thm*}{Theorem}
\newtheorem{cor}[thm]{Corollary}
\newtheorem{lem}[thm]{Lemma}
\newtheorem{propo}[thm]{Proposition}
\theoremstyle{plain}

\newtheorem{defi}[thm]{Definition}
\theoremstyle{definition}

\newtheorem{example}[thm]{Example}
\newtheorem{rem}[thm]{Remark}
\theoremstyle{remark}

\usepackage[pdftex]{graphicx}
\usepackage{capt-of}

\usepackage[usenames,dvipsnames,svgnames,table]{xcolor}

\usepackage{a4,color,palatino,fancyhdr}
 
\usepackage{longtable}
\usepackage[acronym,xindy,toc]{glossaries}
\newacronym{qft}{QFT}{quantum field theory}
\newacronym{pqft}{pQFT}{perturbative quantum field theory}
\newacronym{tvs}{TVS}{topological vector space}
\newacronym{lcs}{LCS}{locally convex space}
\newglossaryentry{angelsperarea}{
  name = $a$ ,
  description = The number of angels per unit area,
}
\newglossaryentry{infdiff}{name=space of infinitely differentiable functions,symbol={\[\mathscr{E}\left( \Omega\right)  \]},
description=blah}
\makeglossaries

\usepackage{hyperref}
\hypersetup{
    colorlinks=true,
    linkcolor=blue,
    filecolor=magenta,      
    urlcolor=blue,
    citecolor=violet,
}
\usepackage{makeidx}
\makeindex

\begin{document}

\thispagestyle{empty}
\fancyhead[LE]{\small\bfseries\leftmark}
\fancyhead[RO]{\small\bfseries\rightmark}
\fancyfoot[C]{\small\bfseries\thepage}
\begin {center}

\textbf{\large Quantum field theory and renormalization à la Stückelberg-Petermann-Epstein-Glaser}
\vspace{.2 cm}

\vspace{1.5cm}
\textbf{Virginia Gali}

\vspace{4 cm}

\includegraphics[scale=.4]{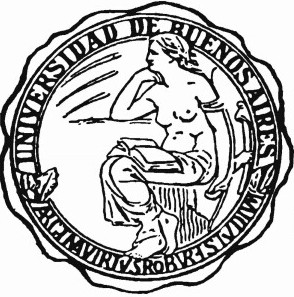}

\vspace{3cm}

\textbf{\large Tesis de Licenciatura en Ciencias Físicas}

\textbf{\large Facultad de Ciencias Exactas y Naturales}

\textbf{\large Universidad de Buenos Aires}
\vspace{1 cm}

Diciembre 2015

\end {center}

\newpage

\thispagestyle{empty}
\noindent
\textit{O God, I could be bounded in a
\\nutshell and count myself a King of in-
\\finite space.}
\\

\noindent
Hamlet, II, 2. 
\newpage
\pagenumbering{roman}
\tableofcontents
\newpage
\chapter*{Aknowledgements\markboth{AKNOWLEDGEMENTS}{AKNOWLEDGEMENTS}}
\addcontentsline{toc}{chapter}{Aknowledgements}
First of all I would like to thank my parents for their priceless support and continued encouragement to follow my studies all along these years.
For his excellent guidance, I am most grateful to my advisor, Professor Estanislao Herscovich.
It is incredible the amount of knowledge a person can possess and I find it necessary to express my gratitude in his sharing it with me all along the past year.
It was a pleasure to be his student.

I want to thank the examiners of my master thesis, Professors Rafael Ferraro, Mauricio Leston and Mariano Suarez Álvarez, whose corrections helped to improve my work considerably.
Finally, I would like to thank Professors Victoria Bekeris and Mirta Gil por their infinite patience with me...

\newpage
\chapter*{Notations\markboth{NOTATIONS}{NOTATIONS}}
\addcontentsline{toc}{chapter}{Notations}
\paragraph{Sets of numbers:}
The sets of natural numbers and the set of integers  will be denoted with the symbols $\mathbb{N}$ and $\mathbb{Z}$, respectively.
The symbols $\mathbb{Q}$, $\mathbb{R}$ and $\mathbb{C}$ will denote the fields of rational, real and complex numbers, respectively. 

\paragraph{Multi-index notations:}
$k=(k_1,\dots,k_n)$ in $\mathbb{N}^n$ is called an \textit{n-integer}\index{n@$n$"-integer} or \emph{multi-index}\index{multi"-index}, and the \textit{order of k} \index{order of a multi"-index} is defined as $|k|=k_1+\dots+k_n$. 
We will also use the following notations: 
	\begin{equation*}
		\begin{split}
			k!=(k_1!)\dots(k_n!)&\mbox{,}\quad x^k=x_1^{k_1}\dots x_n^{k_n}\ \mbox{ if }\ 	x=(x_i)\in \mathbb{C}^n\\ \mbox{ and }\ \  &\partial^k_x=\frac{\partial^{|k|}}{\partial x^k}=\frac{\partial^{|k|}}{\partial x^{k_1}\dots\partial x^{k_n}}. 
		\end{split}
	\end{equation*}
If $l$ and $k$ are two $n$-integers, we denote by $k+l$ the $n$-integer defined as $k+l=(k_1+l_1,\dots,k_n+l_n)$. 
Ordering $\mathbb{N}^n$ by the relation 
	\begin{equation*}
		k\leq l\iff k_i\leq l_i \quad \forall i=1\dots n\mbox{,}
	\end{equation*} 
we can define $k-l=(k_1-l_1,\dots ,k_n-l_n)$ if $k\geq l $ and we write 
	\begin{equation*}
		\binom{l}{k}=\frac{l!}{(l-k)!k!}\mbox{,}\quad \mbox{ if } l\geq k\mbox{;} \mbox{ otherwise } \binom{l}{k}=0.  
	\end{equation*}
\paragraph{Pullbacks and pushforwards:}\label{pulpush}
We denote by
	\begin{itemize}
		\item $\psi^\ast(f)$ the pullback of a function\index{pullback of a function} $f$ given by
			\begin{equation*}
				\psi^\ast(f)=f\circ \psi.
			\end{equation*}
		\item  $\psi_\ast(f)$ the pushforward of a function\index{pushforward of a function} $f$ given by
			\begin{equation*}
				\begin{split}
					\psi_\ast(f)&=f\circ \psi^{-1}\\&=\left( \psi^{-1}\right) ^\ast(f)\mbox{,}
				\end{split}
			\end{equation*}
		where the last equation makes sense only if  $\psi$ is invertible.
		\item $\psi^\ast(t)$ the pullback of a distribution\index{pullback of a distribution} $t$ given by 
			\begin{equation*}
				\big\langle\psi^\ast(t),f\big\rangle=\big\langle t,f\circ \psi^{-1} \big\rangle\mbox{,}
			\end{equation*}
		where the last equation makes sense only if  $\psi$ is invertible.
		\item  $\psi_\ast(t)$ the pushforward of a distribution \index{pushforward of a distribution} $t$ given by
			\begin{equation*}
				\big\langle\psi_\ast(t),f\big\rangle=\big\langle t,f\circ \psi \big\rangle.
			\end{equation*}
	\end{itemize}

\chapter*{Introduction}\label{introooo}
\markboth{INTRODUCTION}{INTRODUCTION}
\addcontentsline{toc}{chapter}{Introduction}
The objective of this thesis is to analyze certain results presented by Nguyen Viet Dang in his article on the  extension of distributions on Riemannian manifolds (\textit{cf.} \cite{vietdang}). 
Some of his proofs were thoroughly revised.
In addition, corrections or  more detailed descriptions and explanations were added to them.

In the present work, we study the renormalization of perturbative quantum field theory (pQFT) as a problem of extension of distributions originally defined on the complement of a closed set in a manifold.
Our  approach is simpler in the case of locally Euclidean quantum field theories (QFT).
Therefore, we choose to work with spacetimes which are $d$-dimensional Riemannian manifolds  (instead of pseudo-Riemannian manifolds). 
This has the advantage that we deal only with the Green functions and there are no time ordered products of fields.
The geometric view also favours the study of renormalization in coordinate space.
This is crucial for theories with curved spacetime backgrounds; in such scenarios, translation invariance is lost and the study of renormalization in momentum space is not possible.
In addition, we do not specify the theory to which the Green functions belong to, so our study separates the problem from  particular models of pQFT. 

Let us briefly explain the heuristic contents of pQFT. 
Consider a real scalar field $\varphi$ defined on a  $d$-dimensional Riemannian manifold $\mathcal{M}$.
To simplify the exposition we consider the classical field configuration space $\mathscr{E}\left(\mathcal{M}\right)$ (and not only the space of solutions of the field equations)\footnote{$\mathscr{E}\left(\mathcal{M}\right)$ denotes the space of $\mathcal{C}^\infty$ functions on $\mathcal{M}$.}.

We define the space $\mathcal{F}$ of \emph{observables}  as the set of all functionals
	\begin{equation*}
		V:\mathscr{E}\left(\mathcal{M}\right)\rightarrow\mathbb{C}
	\end{equation*}   
which are infinitely differentiable in the sense of Bastiani-Michal (see \cite{bastiani}, \cite{Mi} and \cite{Mi2}) and  whose functional derivatives
	\begin{equation}\label{derivfunc}
		\frac{\delta^nV}{\delta\varphi^n}
	\end{equation}
are distributions with compact support on $\mathcal{M}^n$, for every $n$ in $\mathbb{N}$.
There is an additional defining condition on the wavefront sets 
	\begin{equation}\label{wfsets}
		WF\left(\frac{\delta^nV}{\delta\varphi^n}\right) 
	\end{equation}
which is a microlocal version of translation invariance (see \cite{dutsch}).

The space of \emph{nonlocal functionals} $\mathcal{F}_{0}$ is the subspace of $\mathcal{F}$ defined by the stronger requirement that, for every $n$ in $\mathbb{N}$, the functional derivative \eqref{derivfunc} is a smooth function with compact support outside the big diagonal
		\begin{equation}\label{bigdiag}
			D_n=\Big\lbrace \left( x_i \right)_{i=1}^n:\exists i\neq j \mbox{ such that } x_i=x_j \Big\rbrace\subseteq\mathcal{M}^n.
		\end{equation}

The space of \emph{local functionals} $\mathcal{F}_{loc}$ is defined to be the subspace of $\mathcal{F}$ with the additional condition that, for every $n$ in $\mathbb{N}$,
	\begin{equation*}
		\frac{\delta^nV}{\delta\varphi^n}(x_1,\dots,x_n)=0 
	\end{equation*}
if $x_i\neq x_j$ for some pair $(i,j)$.

Let 
	\begin{equation}\label{greengreen}
		G:\mathcal{M}^2\setminus\left\lbrace (x,y)\in\mathcal{M}^2 : x=y \right\rbrace \rightarrow \mathbb{C}
	\end{equation}
be a Green function of the theory.
Every such function $G$ induces a product on the space of nonlocal functionals in the following way.
For every pair $V$ and $W$ in $\mathcal{F}_{0}$, we define the $\star_G$-\emph{product} by the formula
	\begin{equation}\label{starprod}
		\begin{split}
			\left(  V\star_G W\right) \left(\varphi \right) = \sum_{n=0}^\infty &\frac{\hbar^n}{n!}\int_{\mathcal{M}^{2n}}\mbox{d}x_1\dots \mbox{d}y_1\dots \frac{\delta^nV}{\delta\varphi(x_1)\dots \delta\varphi(x_n)}\\
			& G(x_1,y_1)\dots G(x_n,y_n)\frac{\delta^nW}{\delta\varphi(y_1)\dots \delta\varphi(y_n)}.
		\end{split}
	\end{equation}
The above integral clearly depends on the chosen Green function $G$. 
It exists for nonlocal functionals $V$ and $W$, and it is associative and distributive (see \cite{dutsch}).
Definition \eqref{starprod} can be extended to a product of  $n$-th order 
	\begin{equation}\label{starprod2}
		V_1\star_G\dots \star_G V_n\mbox{,}
	\end{equation}
for $V_i$ in $\mathcal{F}_{0}$ for every $i=1,\dots,n$.

In this thesis, we are going to focus  on the task of defining a  product like \eqref{starprod2} for generic functions whose support might intersect the big diagonal \eqref{bigdiag}. 
Observe that, in general, \eqref{starprod} or \eqref{starprod2} are ill defined in this case, due to the occurrence of UV singularities  whenever the arguements in the Green functions coincide.
Therefore, our task will be to extend the above definition of $\star_G$ to generic functions.
This is a simpler problem than the attempt to find such an  extension for general functionals.
The latter is a much more interesting  challenge.
It consists of  extending \eqref{starprod2} to functionals in a certain subspace of the tensor product
	\begin{equation*}
		\left( \mathcal{F}_0+\mathcal{F}_{loc}\right)^{\otimes n}.
	\end{equation*}
Using the conditions imposed on the wavefront sets \eqref{wfsets}, the already known formulae for functionals in $\mathcal{F}_0$ (eqs. \eqref{starprod}, \eqref{starprod2}) and distributing all the appearing  terms, the latter problem  reduces to finding an expression for \eqref{starprod2} when the arguments $V_i$ only belong to the space $\mathcal{F}_{loc}$.

As it stands, it is clear that \eqref{starprod2} is ill defined for $V_1,\dots ,V_n$ in $\mathcal{F}_{loc}$, because of the appearance of UV singularities  whenever the arguements in the Green functions coincide.
Therefore, the $\star_G$-product should be defined in an alternative axiomatic way, as a linear and totally symmetric map
	\begin{align}
	\star_G^n:	\mathcal{F}_{loc}^{\otimes n}&\rightarrow \mathcal{F}\\
		V_1\otimes \dots \otimes V_n&\mapsto 	V_1\star_G\dots \star_G V_n.\nonumber
	\end{align}
The defining axioms could then be given in terms of the generating functional, the S-matrix,
	\begin{equation*}
		\mathbf{S}:\mathcal{F}_{loc}\rightarrow \mathcal{F},
	\end{equation*} 
with
	\begin{equation}\label{smatrix}
		\mathbf{S}(V)=\sum_{n=0}^{\infty} \frac{\overbrace{V\star_G\dots\star_G V}^{n-times}}{n!}
	\end{equation}
for every $V$ in $\mathcal{F}_{loc}$.
Or vice versa, the product is obtained from $\mathbf{S}$ by
	\begin{equation*}
		V\star_G\dots \star_G V=\mathbf{S}^{\left(n\right) }\left(0 \right) \left(V^{\otimes n} \right)\mbox{,}
	\end{equation*} 	
where $\mathbf{S}^{\left(n\right) }\left(0 \right) $ denotes the $n$-th derivative of $\mathbf{S}$ at the origin.
We see that imposing physical axioms on $\mathbf{S}$ is equivalent to imposing corresponding  axioms on $\star_G^n$.

Causal perturbation theory requires that $\mathbf{S}$ satisfies the following conditions (see \cite{ zeidler}):\\

\begin{changemargin}{0.75cm}{0.75cm} 
\noindent\textit{Locality}: physical interactions are localized, and this condition can be realized according to the  Epstein–Glaser setting, by requiring that the S-matrix is an operator-valued tempered distribution.
In particular, if the support of the distribution $V$ in $\mathcal{F}_{loc}$ is
concentrated on a small region of the $d$-dimensional spacetime manifold, then the interaction is localized.\\

\noindent\textit{Unitarity}: to conserve the total probability  we require that $\mathbf{S}$ is a unitary operator, \textit{i.e.}
	\begin{equation*}
		\overline{\mathbf{S}\left(-V\right)}\star_G \mathbf{S}\left(\overline{V}\right)=1\mbox{,}
	\end{equation*}
where the bar means complex conjugation.\\

\noindent\textit{Causality}: whenever the supports of the functionals are disjoint
	\begin{equation*}
		\operatorname{Supp}\left(V\right)\cap\operatorname{Supp}\left(W\right)=\emptyset\mbox{,}
	\end{equation*}
we require that the S-matrix satisfies the factorization property
	\begin{equation*}
		\mathbf{S}\left(V+W\right)=\mathbf{S}\left(V\right)\star_G\mathbf{S}\left(W\right).
	\end{equation*}
\end{changemargin}
The previous axioms should be considered if one seeks to extend the definition of \eqref{starprod2} for $V_1,\dots ,V_n$ in $\mathcal{F}_{loc}$.

As mentioned earlier, we are going to focus in the less ambitious task of defining the $\star_G$ product for generic functions, instead of distributions.
This problem is much more simple because we do not need to set conditions on the wavefront sets \eqref{wfsets}, as we only deal with products of functions.
Our task reduces then to finding extensions of finite linear combinations of products of the type
	\begin{equation}\label{atland}
		f\prod_{1\leq i \leq j \leq n} G^{n_{ij}}\left( x_i,x_j\right)\mbox{,}
	\end{equation} 
where $f$ belongs to $\mathscr{E}\left( \mathcal{M}^n\right)$ and $n_{ij}$ to $\mathbb{N}$, such that they satisfy the required physical axioms.

Since the Green functions \eqref{greengreen} are regular functions on $\mathcal{M}^2\setminus\left\lbrace x=y\right\rbrace$, finite sums of products of the type \eqref{atland} are regular functions  defined on the subspace of all pairwise distinct arguments $\left(x_1,\dots ,x_{n} \right)$ in $\mathcal{M}^{n}$.
This means that they form an algebra, and we are going  to extend its elements to distributions over the whole space $\mathcal{M}^n$ by means of a \emph{system of renormalization maps}
	\begin{equation*}
		\left\lbrace \mathcal{R}_n\right\rbrace_{n\in\mathbb{N}}.
	\end{equation*}
Specifically, it is a family of linear extension operators $\mathcal{R}_n$, whose action on expressions of the type \eqref{atland}  guarantee the convergence of the integrals in \eqref{starprod} and those appearing in the general case \eqref{starprod2}.

In this thesis we shall construct inductively such a system of renormalization maps, satisfying certain natural axioms due to N. Nikolov \cite{nikolov}, but not until Chapter \ref{tengofiebree}.

In the first four chapters we develop the mathematical background needed to derive the results presented in this thesis.
In the other chapters we apply our analytic machinery to the study of pQFT on Riemannian manifolds, treating the problem of extension of distributions.

In  Chapter \ref{joann1} some preliminaries on functional analysis are given in order to understand the topology and basic properties of topological vector spaces and, in particular, of Fréchet spaces.
We mainly follow the exposition of W. Rudin, \cite{rudin}.
The results given in Lemma \ref{Mullova}; and propositions \ref{joann2}, \ref{ramiro} and \ref{gershwin} are of particular interest.

In Chapter \ref{asdf}, we recall the theory of spaces of continuously differentiable functions and  spaces of test functions on open subsets of the Euclidean space $\mathbb{R}^d$, making emphasis in the description of their topology.
This is done in order to give a precise definition of a distribution, as a continuous linear functional defined on the space of test functions. At the end of the chapter, in Proposition \ref{joannnn}, we give an equivalent definition of distribution, which is very useful for explicit calculations.

In Chapter \ref{joann3}, we introduce jet spaces, and in particular, the subspace of  differentiable functions in the sense of Whitney, following the text of B. Malgrange, \cite{malgrange}.
At the end of the chapter we present Whitney's Extension Theorem (see Theorem \ref{whitney}), taken from E. Bierstone's text, \cite{bierstone}, whose utility becomes clear in Chapter \ref{cucuu} (see below).

In Chapter \ref{macri} we present a generalization of  the concepts introduced in Chapter \ref{asdf}.
We define the spaces of continuously differentiable functions and the spaces of test functions on open subsets of some manifold $\mathcal{M}$, following Dieudonne's text, \cite{dieu}.
Later on, a definition of a distribution is given, as a continuous linear functional on the space of test functions.
This is done in an analogous way to the definition given in Chapter \ref{asdf}, with the obvious changes.
In Proposition \ref{rachmaninov2}, we give the equivalent definition of distribution in this general case, which is very useful in explicit calculations.
Theorem \ref{gluedistrib} is of particular interest, as it is one of the results on which the proof of the existence of a system of renormalization maps rests.
In the last section of this chapter, we introduce a natural condition of moderate growth for a distribution $t$ along a closed subset $X$ of a manifold $\mathcal{M}$, following  \cite{vietdang}.
This condition measures the
singular behaviour of $t$ near $X$.
The moderate growth condition given by the author in \cite{vietdang} is rather unclear so we define it in a much clearer way.
We also state and prove in Proposition \ref{viir} that the product of a distribution with moderate growth and a $\mathcal{C}^\infty$ function gives a distribution with moderate growth.
This result is mentioned at the beginning of \S 1.1 of \cite{vietdang}, but the author gives no demonstration.

The importance of the notion of a distribution with moderate growth becomes clear in Chapter \ref{cucuu}, as it is a necessary and sufficient condition for a distribution originally defined on $\mathcal{M}\setminus X$, to be extendible to $\mathcal{M}$, for $X\subseteq\mathcal{M}$ a closed subset. 
This last statement is the equivalence of conditions \textit{(i)} and \textit{(ii)} listed in Theorem \ref{elteo}, that we choose to transcribe here due to the central role it plays in this work.
\begin{thm*}
For a distribution $t$ defined on  $\mathcal{M}\setminus X$ the three following claims are equivalent:
	\begin{itemize}
		\item [(i)] $t$ has moderate growth along $X$.
		\item [(ii)] $t$ is extendible to $\mathcal{M}$.
		\item [(iii)] There is a family of $\mathcal{C}^\infty$ functions
			\begin{equation*}
				(\beta_\lambda)_{\lambda\in(0,1]}
			\end{equation*}
		defined on $\mathcal{M}$ such that
			\begin{enumerate}
				\item $\beta_\lambda=0\ $ in a neighborhood of $X$,
				\item $\lim\limits_{\lambda\to 0}\beta_\lambda(x)=1\mbox{,}\ $ for every $x$ in $\mathcal{M}\setminus X$, 
			\end{enumerate}
		and a family of distributions
			\begin{equation*}
				(c_\lambda)_{\lambda\in(0,1]}
			\end{equation*} 
		 on $\mathcal{M}$ supported on $X$ such that the following limit
			\begin{equation*}
				\lim\limits_{\lambda\to 0}t\beta_\lambda -c_\lambda
			\end{equation*}
		exists and defines a continuous extension of $t$ to the  manifold $\mathcal{M}$.
	\end{itemize}
\end{thm*}
\noindent The whole chapter is dedicated to the proof of the previous  extension theorem.
It  gives an explicit way to define a continuous extension of a given distribution with moderate growth along a closed subset of a manifold.
On the other hand, it is fairly long and rather technical, but we provide it completely.
To do so, we follow the proof given by Dang in \cite{vietdang}.
It has to be noted that the author makes some nontrivial omissions in his demonstration. 
For instance, he relies on the result in Theorem 1.2 of  \cite{vietdang}, where the author does not give a definition of splitting of a short exact sequence of locally convex spaces, nor does he specify the topology of the dual spaces that appear in it.
In this work we restate the result in Theorem 1.2 of \cite{vietdang} in a clearer way.
Its proof relies on Whitney's Extension Theorem, given in Chapter \ref{joann3}, and our version of it is complete and well-organized.

In Chapter \ref{joann6} the whole analytical machinery developed in the previous chapters  is used to extend the so called \emph{Feynman amplitudes} that usually appear in QFT and are given by
	\begin{equation*}
		\prod_{1\leq i<j\leq n } G^{n_{ij}}\left( x_i,x_j\right). 
	\end{equation*}
First, in \S\ref{secrenprod} we introduce the notion of tempered function along a closed set $X$.
It is a natural growth condition for a function $f$ defined on $\mathcal{M}\setminus X$, that implies the moderate growth condition for the distribution $T_f$, defined on a compactly supported function 
	\begin{equation*}
		\psi:\mathcal{M}\setminus X\rightarrow \mathbb{C}
	\end{equation*}
by
	\begin{equation*}
		\left\langle T_f, \psi  \right\rangle =\int_{\mathcal{M}\setminus X}f\psi\ \mbox{d}\mu_g\mbox{,}
	\end{equation*}
where $\mbox{d}\mu_g$ is the volume form associated to a metric $g$ defined on $\mathcal{M}$.
We then apply the extension techniques developed in Chapter \ref{cucuu} to
establish in Theorem \ref{maggotbrain}  that the product  of a tempered function and a distribution has a continuous  extension to the closed set $X$.
This means that the space of extendible distributions $\mathcal{T}_{\mathcal{M}\setminus X}(\mathcal{M})$ is a $\mathcal{T}(X,\mathcal{M})$-module, if we denote by $\mathcal{T}(X,\mathcal{M})$ the space of tempered functions.
In \S\ref{expliFA}, we show that the  Green functions $G(x,y)$ associated to the Riemannian structure are tempered along the diagonal $\left\lbrace x=y\right\rbrace$ contained in $\mathcal{M}^2$ (see Lemma \ref{silvia}).
Both Theorem \ref{maggotbrain} and Lemma \ref{silvia} appear in \cite{vietdang}, but we give  an exhaustive proof of them.
We use Lemma \ref{silvia} together with the results of the previous section  in the proof of  Theorem \ref{dinsky}, where we show that Feynman amplitudes are extendible.

As explained above, in QFT, the renormalization procedure is not only involved with the extension of Feynman amplitudes.
Renormalization is intended to be applicable over the algebra which they generate along with $\mathcal{C}^\infty$ functions, namely
	\begin{equation}\label{moustache}
		\mathcal{O}(D_I,\Omega):=\left\langle f\prod_{i<j\in I}\left.G^{n_{ij}}\left(x_i,x_j \right)\right|_{\Omega\setminus D_I}: n_{ij}\in\mathbb{N} \ \forall i<j\in I, \ f\in\mathscr{E}\left( \Omega\right)  \right\rangle _{\mathbb{C}}\mbox{,}
	\end{equation}
where $I$ is a finite subset of $\mathbb{N}$, $\Omega$ is an open subset of the product $\mathcal{M}^I$ of $\left| I\right|$ copies of the Riemannian manifold $\mathcal{M}$, and
	\begin{equation*}
		D_I=\Big\lbrace \left( x_i \right)_{i\in I}:\exists i,j\in I,  i\neq j, x_i=x_j \Big\rbrace.
	\end{equation*}
This is where the need of a system of renormalization maps as a collection of linear extension operators
	\begin{equation*}
		\mathcal{R}_{\Omega}^I: \mathcal{O}\left( D_I,\Omega\right) \rightarrow \mathscr{D}\left( \Omega\right)' 
	\end{equation*}
becomes apparent\footnote{$\mathscr{D}\left( \Omega\right)'$ is the symbol generally used to denote the space of distributions defined on an open set $\Omega$ (see Definitions \ref{defid1} and \ref{mariacallas}).}.

In the first section of Chapter \ref{tengofiebree}, we list the axioms the renormalization maps should satisfy in order to define a coherent renormalization procedure, following \cite{nikolov}.
The most important one is the \emph{factorization axiom}, which guarantees the compatibility with the fundamental requirement of locality.
The factorization axiom appears in \cite{vietdang}, but we chose to follow \cite{todorov} where it is defined in a clearer manner.
In Theorem \ref{exrenmap} we prove  the existence of a family of renormalization maps satisfying our adequate version of renormalization axioms (\textit{cf.} Theorem 4.2 of \cite{vietdang}):
\begin{thm*}
There exists a collection of renormalization maps 
	\begin{equation}\label{genoa}
		\mathscr{R}=\left\lbrace \mathcal{R}_{\Omega}^{I}: \mathcal{O}\left( D_I,\Omega\right) \rightarrow \mathscr{D}\left( \Omega\right)'  :I\subseteq \mathbb{N}\mbox{,} |I|<\infty\mbox{,} \ \Omega\subseteq \mathcal{M}^I\ open \right\rbrace\mbox{,}
	\end{equation}
that satisfies the renormalization axioms listed in \S \ref{axioms}.
\end{thm*}

\clearpage
\chapter {Preliminaries on functional analysis}\label{joann1}

\renewcommand{\headrulewidth}{.4pt}
\fancyhf{}
\fancyhead[LE]{\small\bfseries\leftmark}
\fancyhead[RO]{\small\bfseries\rightmark}
\fancyfoot[LE,RO]{\small\bfseries\thepage}
\pagenumbering{arabic}

In what follows, an introduction to some basic concepts and notations associated with vector spaces is given.
Vector spaces equipped with a topology compatible with the vector space operations are defined next, swiftly turning to the particular case of locally convex topological vector spaces, in particular those which are Fr\'echet. 
The reader familiarized with general aspects of topology and functional analysis can skip this chapter.
\section{Vector spaces}

Let $E$ be a vector space over $\mathbb{C}$. If $A$ and $B$ are subsets of $E$, $x$ belongs to $E$ and $\lambda$ is a complex number, the following notations will be used:
	\begin{equation*}
		\begin{split}
			x+A&=\left\lbrace x+a: a\in A \right\rbrace\mbox{,}\\ 
			x-A&=\left\lbrace x-a: a\in A \right\rbrace\mbox{,}\\
			A+B&=\left\lbrace a+b: a\in A\mbox{,}\ b\in B \right\rbrace\mbox{,}\\
			\lambda A&=\left\lbrace \lambda a: a\in A \right\rbrace.
		\end{split}
	\end{equation*}
Some types of subsets of $E$ are introduced in the following definition.
\begin{defi}[See \cite{rudin}, Ch. 1, \S 1.4]
Let $E$ be a $\mathbb{C}$-vector space.
A subset $B$ of $E$ is called \emph{balanced}\index{balanced set} if 
	\begin{equation*}
		B=\left\lbrace \lambda b: \lambda \in \mathbb{C}\mbox{, } |\lambda|\leq 1\mbox{, } b \in 	B\right\rbrace.
	\end{equation*}
In other words, $B$ is balanced if 
	\begin{equation*}
		\lambda B\subseteq B\mbox{,}
	\end{equation*}
for every complex number $\lambda$ with $|\lambda|\leq 1$.

A subset A of $E$ is called \emph{absorbing}\index{absorbing set} if for every  $x$ in $E$ there is a positive number $\varepsilon$ such that 
	\begin{equation*}
	tx \in A\ 
	\end{equation*}
for every $t$ such that $0 \leq t < \varepsilon$; or equivalently, if
	\begin{equation*}
		E=\bigcup_{n\in\mathbb{N}}nA.
	\end{equation*}
	
A subset $C$ of a $\mathbb{C}$-vector space $E$ is \emph{convex}\index{convex set} if 
	\begin{equation*}
		\lambda x   +  (1  -  \lambda)y \in C
	\end{equation*}
for all elements $x\mbox{, }y$ in   $C$, and every $\lambda$ in  $[0,1]$.
\end{defi}
\section{Topological vector spaces}

It is of special interest the class of vector spaces equipped with a certain topological structure.
\begin{defi}[See \cite{meisevogt}, Ch. 22 or \cite{rudin}, Ch. 1, \S 1.6]\label{trutru}
A \emph{topological vector space (TVS)}\index{topological vector space}  $E$ is a $\mathbb{C}$-vector space equipped with a topology for which points are closed sets, and addition $+:E\times E\rightarrow E$ and scalar multiplication $\cdot:\mathbb{C}\times E \rightarrow E$ are continuous. 
A topology $\tau$ on a $\mathbb{C}$-vector space $E$ is called a \emph{vector space topology}\index{vector space topology} if $\left( E,\tau\right) $ is a topological vector space.

A \emph{neighborhood}\index{neighborhood} of a point $x$ is an open subset of $E$ that contains $x$.
A collection $\mathscr{B}$ of neighborhoods of a point $x$ is a \emph{local base at} $x$ \index{local base at a point}if every neighborhood of $x$ contains an element of $\mathscr{B}$.

\end{defi}
\begin{rem}[Invariance of the topology and the metric under translations. See  \cite{rudin}, Ch. 1, \S 1.7]\label{invtop}
Let $E$ be a TVS. Associate to each $h$ in $E$ the \emph{translation operator}\index{translation operator} $T_{x}$, defined by
	\begin{align}\label{trasla}
		T_{h}: E&\rightarrow E\\
		x& \mapsto x+h. \nonumber
	\end{align}
To each $\lambda\neq 0$ in $\mathbb{C}$ associate the \emph{multiplication operator}\index{multiplication operator} $M_{\lambda}$, defined by
	\begin{align}\label{multi}
		M_{\lambda}: E&\rightarrow E\\
		x& \mapsto \lambda x. \nonumber
	\end{align}

The vector space axioms together with the fact that the vector space operations are continuous imply that $T_{h}$ and $M_{\lambda}$ are homeomorphisms.

The fact that $T_{h}$ is a homeomorphism imply that every vector space topology $\tau$ is \emph{translation-invariant}\index{translation"-invariant topology}: a subset $V$ of a TVS $E$ is open if and only if $T_{h}(V)=h+V$ is open for every $h$  in  $E$. 
Thus, $\tau$ is determined by a local base at any point we choose.
For instance, if $E$ is a TVS, the neighborhoods of some $x$ in $E$ may be expressed in the following form:
	\begin{equation*}
		x+V:=\left\lbrace x+v: v\in V \right\rbrace\mbox{,} \ \mbox{ V a zero neighborhood}.
	\end{equation*} 
In particular, a local base at any element $x$ is determined by a local base at zero.
In the vector  space context, the term \emph{local base}\index{local base} will allways mean a local base at zero.
\end{rem}
Throughout this thesis we are going to work with the category of topological vector spaces with continuous linear morphisms.
This is not an abelian category \index{abelian category}(see \cite{grothe}).
Nevertheless it is an exact category \index{exact category} (see \cite{quillen}) where the notion of exact sequence is the usual set theoretic definition. 
However, the notion of a splitting of a short exact sequence may be ambiguous in principle, so we will provide it to be clear.
\begin{defi} [See \cite{meisevogt}, Ch. 9] \label{galletitasdelimon}
Let $\left(A_i,f_i \right) _{i\in\mathbb{Z}}$ be a sequence of TVS  $A_i$ and linear continuous maps 
	\begin{equation*}
		f_i:A_i\rightarrow A_{i+1}.
	\end{equation*}
The sequence is \emph{exact at the $i$-th position}\index{exactness of a sequence at a position} if $\operatorname{Im}\left( f_{i-1}\right) =\operatorname{Ker}\left(f_i \right)$; this means that it satisfies this condition for the underlying vector spaces.
The sequence is \emph{exact} \index{exact sequence} if it is exact at each position.

The \emph{support}\index{support of a sequence} of the sequence $\left(A_i,f_i \right) _{i\in\mathbb{Z}}$ is the set of integers $i$ such that $A_i$ is not zero.

A \emph{short exact sequence} \index{short exact sequence} is an exact  sequence whose support is contained in a set of the form
	\begin{equation*}
		\left\lbrace i-1,i,i+1\right\rbrace
	\end{equation*}
for some $i$ in $\mathbb{Z}$.
We then write
	\begin{equation} \label{george12}
		0\rightarrow A \overset{f}{\rightarrow} B \overset{g}{\rightarrow} C \rightarrow 0. 
	\end{equation}
The sequence \eqref{george12}  \emph{splits} \index{split} if there exists an isomorphism $\Omega: A\oplus C\rightarrow B$ of TVS such that the following diagram commutes:
	\begin{equation}\label{diag}
		\xymatrix{A
		\ar@{=}[d]_{} \ar[r]^{f}& B\ar[r]^{g}& C\ar@{=}[d]_{}\\
		A\ar@{^{(}->}[r]^{i{\hspace*{0.5 cm}}}& A\oplus C\ar[u]^{\Omega} \ar@{->>}[r]^{\pi {\hspace*{-0.2 cm}}}& B
 		}
	\end{equation} 
where $i:A\hookrightarrow A\oplus C$ is the canonical inclusion and $\pi: A\oplus C \twoheadrightarrow C $ is the canonical projection.
\end{defi}
\section{Locally convex spaces and the Hahn-Banach Theorem}

An important class of TVS is given by the following family.
\begin{defi} [See \cite{rudin}, Ch. 1, \S 1.8]\label{triojgh}
A \emph{locally convex space}\index{locally convex space} $($LCS$)$ is a TVS $E$ which has a zero  neighborhood basis whose members are  convex sets. Such a basis will be called \emph{convex basis}\index{convex basis}, for short.

A \emph{locally convex topology}\index{locally convex topology} on a $\mathbb{C}$-vector space $E$ is a topology $\tau$, for which $\left(E, \tau \right) $ is a LCS.   
\end{defi}
\begin{propo}[See \cite{rudin}, Ch. 1, Thm. 1.14]\label{kingcrimson}
If $E$ is a TVS, then the following hold:
	\begin{itemize}
		\item  [(i)] Every zero neighborhood contains a balanced zero neighborhood.
		\item  [(ii)] Every convex zero neighborhood contains a convex balanced zero neighborhood.
	\end{itemize}
In particular, every TVS has a \emph{balanced local base}\index{balanced local base} (\textit{i.e.} a zero neighborhood basis consisting of balanced sets); and every LCS has a \emph{convex balanced local base}\index{convex balanced local base} (\textit{i.e.} a zero neighborhood basis consisting of convex balanced  sets).
\end{propo}
\begin{proof}
We begin by proving \textit{(i)}.
Consider a zero neighborhood $U$.
Since 
	\begin{equation*}
		\cdot:\mathbb{C}\times E \rightarrow E
	\end{equation*}
is continuous and $U$ is a neighborhood of $0=\cdot\left( 0,0\right) $, there exist a positive number $\varepsilon$ and a zero neighborhood $W$ such that
	\begin{equation*}
		\delta W\subseteq U  
	\end{equation*} 
for every $\delta$ such that $|\delta|<\varepsilon$.
Define $V$ by 
	\begin{equation*}
		V=\bigcup_{|\delta|<\varepsilon}\delta W.
	\end{equation*}
Clearly, $V$ is a zero neighborhood,  it is contained in $U$ and $\lambda V$ is contained in $V$ for every $\lambda$ such that $|\lambda|\leq 1$.

To prove \textit{(ii)}, suppose $U$ is a convex neighborhood of zero. 
Define $A$ by
	\begin{equation*}
		A=\bigcap_{|\alpha|=1} \alpha U.
	\end{equation*}
By \textit{(i)} there exists a balanced zero neighborhood $V$ contained in $U$.
Since $V$ is balanced,
	\begin{equation*}
		\begin{split}
			& V\subseteq\alpha^{-1} \left( \alpha V\right) \subseteq\alpha^{-1}V \subseteq V \ \mbox{for every } \alpha \mbox{ such that }|\alpha|=1.\\
			\implies &\alpha^{-1}V=V \subseteq U\ \mbox{for every } \alpha \mbox{ such that }|\alpha|=1.\\
			\implies& V \subseteq \alpha U \ \mbox{for every } \alpha \mbox{ such that }|\alpha|=1.
		\end{split}
	\end{equation*}
Thus $V$ is contained in $A$ which implies that the interior of $A$ is a neighborhood of zero.
Clearly $A^{\circ}$ is contained in $U$. 
Being an intersection of convex sets, $A$ is convex. 
It follows that $A^{\circ}$ is convex. 
Indeed, since $A^\circ$ is contained in $ A$ and $A$ is convex, we have
	\begin{equation*}
		t\, A^\circ+(1-t)A^\circ\subseteq A
	\end{equation*}
if $0<t<1$. The two sets on the left are open; hence so is their sum. 
Then, since every open subset of $A$ is a subset of $A^\circ$, $A^\circ$ is convex.

To prove that $A^{\circ}$ is a neighborhood with the desired properties it remains to see that $A^{\circ}$ is balanced. 
We begin by showing that $A$ is balanced: choose $r$ and $\beta$ such that $0\leq r \leq 1$, $|\beta|=1$. 
Then 
	\begin{equation*}
		r\beta A= \bigcap_{|\alpha|=1} r\beta \alpha U = \bigcap_{|\alpha|=1} r\alpha U .
	\end{equation*}
Since $\alpha U$ is a convex set that contains zero, we have 
	\begin{equation*}
		r \alpha U\subseteq \alpha U.
	\end{equation*}
Thus, 
	\begin{equation*}
		r\beta A\subseteq A\mbox{,}
	\end{equation*}
and $A$ is balanced.
Next, let $0<|\alpha|\leq 1$. 
Then, 
	\begin{equation*}
		\alpha A^\circ=\left( \alpha A \right)^\circ\mbox{,}
	\end{equation*}
since the map $M_\alpha$ (defined in \eqref{multi})  is an homeomorphism for every $\alpha$ such that  $\alpha\neq 0$.
Hence,
	\begin{equation*}
		\alpha A^\circ \subseteq \alpha A \subseteq A\mbox{,}
	\end{equation*}
since $A$ is balanced.
But  $\alpha A^\circ$ is open, so $\alpha A^\circ$ is contained in $ A^\circ $, and this holds for every $\alpha$ such that  $0<|\alpha|\leq 1$.

If $\alpha=0$, then
	\begin{equation*}
		\alpha A^\circ=\left\lbrace 0 \right\rbrace \subseteq A^\circ\mbox{,}
	\end{equation*}
as $A^\circ$ is a zero neighborhood.
Then, $\alpha A^\circ$ is contained in $ A^\circ$ for every $\alpha$ such that $0\leq |\alpha| \leq 1$, which means that $A^\circ$ is  balanced.   
\end{proof}
The following result is of fundamental importance in functional analisys and in all that follows. 
\begin{thm}[Hahn-Banach. See \cite{rudin}, Thm. 3.6]\label{hahnbanach} \index{Hahn"-Banach Theorem}
If $\lambda$ is a continuous linear functional on a subspace $M$ of a LCS $E$, then there exists a continuous linear functional $\Lambda$ such that $\Lambda|_M = \lambda$.
\end{thm}
\section{Metrization}

A topology $\tau$ on a set $E$ is \emph{metrizable}\index{metrizable topology} if there is a metric $d$ on $E$ which is compatible with $\tau$.
In that case, the collection of balls with radius $1/n$ centered at $x$ is a local base at $x$.
This gives a necessary condition for a topological space $E$ to be metrizable. In Theorem \ref{liszt} we  show that if $E$ is a TVS, it turns out to be also sufficient.
First we give the following definition.
\begin{defi}
A metric $d$ defined on a vector space $E$ is  \emph{invariant} \index{invariant metric} if
	\begin{equation*}
		d(x+z,y+z)=d(x,y)
	\end{equation*}
for all elements $x\mbox{, }y\mbox{, }z$ in $E$.
\end{defi}
\begin{thm}[See \cite{rudin}, Ch. 1, Thm. 1.24] \label{liszt}
If $E$ is a TVS with a countable local base at zero $( $or by invariance, at any point $x$; see Remark \ref{invtop}$)$, then there is a metric $d$ on $E$ such that
	\begin{itemize}
		\item[(i)] $d$ is compatible with the topology of $E$,
		\item[(ii)] the open balls centered at zero are balanced, and 
		\item[(iii)] d is invariant.
	\end{itemize}  
If, in addition, $E$ is a LCS, then $d$ can be chosen to be so as to satisfy \textit{(i)}, \textit{(ii)}, \textit{(iii)} and also 
	\begin{itemize}
		\item[(iv)] all open balls are convex.
	\end{itemize}
\end{thm}
\begin{proof}
By Theorem  \ref{kingcrimson}, $E$ has a balanced local base $\left\lbrace V_n \right\rbrace_{n\in\mathbb{N}}$  at zero such that 
	\begin{equation}\label{arbol}
		V_{n+1}+V_{n+1}+V_{n+1}+V_{n+1}\subseteq V_n
	\end{equation} 
for every $ n$ in $\mathbb{N}$.
When $E$ is a LCS, this local base can be chosen so that each $V_n$ is also convex.

Let $D$ be the set of all rational numbers $r$ of the form 
	\begin{equation*}
		r=\sum_{n=1}^\infty c_n(r)2^{-n}\mbox{,}
	\end{equation*}
where each $c_n(r)$ is $0$ or $1$ and only finitely many are $1$.
Thus, each $r$ in $D$ satisfies the inequalities $0\leq r< 1$.
Put $A(r)=E$ if $r\geq 1$. 
For every $r$ in $D$, define 
 	\begin{equation}\label{xaleo}
		A(r)=c_1(r)V_1+c_2(r)V_2+c_3(r)V_3+\dotsb.
	\end{equation}
Note that each of these sums is actually finite. 
Define 
	\begin{equation*}
		f(x)=\inf\left\lbrace r: x\in A(r)\right\rbrace 
	\end{equation*}
for $x$ in $E$, and 
	\begin{equation*}
		d(x,y)=f(x-y) 
	\end{equation*}
for $x\mbox{, }y$ in $E$.
To show that $d$ is a metric, we will use the following property, to be proven later:
	\begin{equation}\label{property}
		A(r)+A(s)\subseteq A(r+s) 
	\end{equation}
for all elements $r\mbox{, }s$ in $D$.
Since every $A(r)$ contains $0$,  \eqref{property} implies 
	\begin{equation*}
		A(r)\subseteq A(r)+A(t-r)\subseteq A(t)
	\end{equation*}
if $r<t$.
Thus, $\left\lbrace A(r)\right\rbrace$ is totally ordered by set inclusion. 

We claim that
	\begin{equation}\label{sombrero}
		f(x+y)\leq f(x)+f(y)
	\end{equation}
for all elements $x\mbox{, }y$ in $E$.
As $f\leq 1$, in the proof of \eqref{sombrero} we may assume that the right side is less than unity.
Fix $\varepsilon>0$.
There exist $r$ and $s$ in $D$ such that 
	\begin{equation*}
		f(x)<r\mbox{,} \quad f(y)<s
		\mbox{,}\quad 
		r+s<f(x)+f(y)+\varepsilon.
	\end{equation*}
Thus, by definition of $f$, $x$ belongs to $A(r)$ and  $y$ belongs to $A(s)$; which in turn implies by  \eqref{property} that $x+y$ belongs to $A(r+s)$.     
Now \eqref{sombrero} follows, because 
	\begin{equation*}
		f(x+y)\leq  r+s\leq f(x)+f(y)+\varepsilon\mbox{,}
	\end{equation*}
and $\varepsilon$ was arbitrary.

Since each $A(r)$ is balanced, $f(x)=f(-x)$.
It is obvious that $f(0)=0$.
If $x\neq 0$, then $x$ is not in $V_n=A(1/2^n)$, for some $n$ in $\mathbb{N}$, and so $f(x)\geq 1/2^n>0$.

These properties of $f$ show that $d$ is an invariant metric on $E$.
The open balls centered at zero are the open sets 
	\begin{equation*}
		B_{\delta}(0)=\left\lbrace x: d(x,0)=f(x)<\delta \right\rbrace=\bigcup_{r<\delta} A(r). 
	\end{equation*}
If $\delta<1/2^n$, then $B_{\delta}(0)$ is contained in $V_n$.
Hence $\left\lbrace B_{\delta}(0)\right\rbrace $ is a local base for the topology of $E$.

This proves \textit{(i)}.
Since each $A(r)$ is balanced, so is each  $B_{\delta}(0)$.
If each $V_n$ is convex, so is each $A(r)$,  and therefore the same is true for each $B_{\delta}(0)$.

We turn to the proof of \eqref{property}.
If $r+s\geq 1$, then $A(r+s)=E$ and \eqref{property} is obvious. 
We may therefore assume that $r+s<1$, and we will use the following simple proposition about addition in the binary system of notation:
\begin{changemargin}{0.75cm}{0.75cm}
\textit{If r, s and r+s are in D and $c_n(r)+c_n(s)\neq c_n(r+s)$ for some n, then at the smallest n where this happens we have $c_n(r)=c_n(s)=0$, $c_n(r+s)=1$. }   
\end{changemargin}
Put $\alpha_n=c_n(r)$, $\beta_n=c_n(s)$, $\gamma_n=c_n(r+s)$.
If $\alpha_n+\beta_n=\gamma_n$ for every $n$ in $\mathbb{N}$ then \eqref{xaleo} shows that $A(r)+A(s)=A(r+s).$
In the other case, let $N$ be the smallest integer for which  $\alpha_N+\beta_N\neq\gamma_N$.
Then, $ \gamma_N=1, \ \alpha_N=\beta_N=0$, as mentioned above.
Hence, 
	\begin{equation*}
		\begin{split}
			A(r)&\subseteq \alpha_1V_1+\dotsb+ \alpha_{N-1}V_{N-1}+V_{N+1}+V_{N+2}+\dotsb\\
			&\subseteq \alpha_1V_1+\dotsb+ \alpha_{N-1}V_{N-1}+V_{N+1}+V_{N+1}.
		\end{split}
	\end{equation*}      
Likewise, 
	\begin{equation*}
		A(s)\subseteq \beta_1V_1+\dotsb + \beta_{N-1}V_{N-1}+V_{N+1}+V_{N+1}.
	\end{equation*}
Since $\alpha_n+\beta_n=\gamma_n$ for every $n<N$, \eqref{arbol} now leads to 
	\begin{equation*}
		A(r)+A(s)\subseteq \gamma_1V_1+\dotsb + \gamma_{N-1}V_{N-1}+V_{N}\subseteq A(r+s)\mbox{,}
	\end{equation*}
because $\gamma_N=1$.
\end{proof}
\begin{thm}[See \cite{rudin}, Ch. 1, Thm. 1.28]\label{sarasa} 
$\quad$
\begin{itemize}
		\item[(i)] If d is an invariant metric on a vector space $E$ then
			\begin{equation*}
				d(nx,0)\leq nd(x,0)\mbox{,}
			\end{equation*}
		for every $x$ in   $E$ and every $n$  in  $\mathbb{N}$. 
		\item[(ii)] If $\left(  x_n\right)_{n\in\mathbb{N}} $ is a sequence in a metrizable TVS $E$ and if $x_n\rightarrow 0$ as $n\rightarrow \infty$, then there are positive scalars $\gamma_n$ such that $\gamma_n\rightarrow \infty$ as $n\rightarrow\infty$ and $\gamma_nx_n\rightarrow 0$.
	\end{itemize}
\end{thm}
\begin{proof}
Statement \textit{(i)} follows from 
	\begin{equation*}
		d(nx,0)\leq  \sum_{k=1}^{n}d\left( kx,(k-1)x\right) =nd(x,0).
	\end{equation*}
To prove \textit{(ii)}, let $d$ be as in \textit{(i)}, compatible with the topology of $E$.
Since $d(x_n,0)\rightarrow 0$, there is an increasing sequence of positive integers $n_k$ such that $d(x_n,0)<k^{-2}$ if $n>n_k$.
Put 
	\begin{equation*}
	\gamma_n=\left\{ 
		\begin{array}{c c l}
			1\mbox{,}  &\mbox{ if } & n  <  n_1\mbox{,} \\
			k\mbox{,}   &\mbox{ if } & n_k  \leq n<  n_{k+1}.
		\end{array}
	\right.
	\end{equation*}
For every $n$ in $\mathbb{N}$, 
	\begin{equation*}
		d(\gamma_n x_n,0)=d(k x_n,0)\leq kd(x_n,0)<\frac{1}{k}.
	\end{equation*}
Hence $\gamma_n x_n\rightarrow 0$ as $n\rightarrow\infty$.
\end{proof}
\section{Boundedness and continuity of linear mappings}
 
Within the context of Banach spaces, boundedness and continuity of linear mappings are two strongly related concepts.
With the objective of generalizing these ideas a notion of what is meant by a bounded subset of a TVS is now introduced, followed by the definition of bounded linear mappings between TVS and the relation between both notions.
\begin{defi}[See \cite{rudin}, Ch. 1, \S 1.6]
A subset $B$ of a TVS $E$ is \emph{bounded}\index{bounded set} if for every zero neighborhood  $U$ in $E$ there is a positive  number $\varepsilon$ such that 
	\begin{equation*}
		\lambda B \subseteq   U\mbox{,}
	\end{equation*}
for every $\lambda <\varepsilon$.  
\end{defi}
\begin{rem}[See \cite{rudin}, Ch. 1, \S 1.29]
If $E$ is a TVS with a compatible metric $d$, bounded sets in $E$ do not generally coincide with the sets which have bounded diameter. 
 
When $E$ is a Banach space and $d$ is the metric induced by the norm, then bounded sets in $E$, as defined above, coincide with the sets that have bounded diameter.  
\end{rem}
\begin{defi}[See \cite{rudin}, Ch. 1, \S 1.8]
Let $E$ ba a TVS with topology $\tau$.
$E$ is \emph{locally bounded }\index{locally bounded topological vector space}if $0$ has a bounded neighborhood.
\end{defi}
\begin{defi}[See \cite{rudin}, Ch. 1, \S 1.31]
Let $\Lambda:E\rightarrow F$ be a linear map  between TVS. The function $\Lambda$ is \emph{bounded}\index{bounded function} if it maps bounded sets into bounded sets, \textit{i.e.} if $\Lambda(B)$ is a bounded subset of $F$  for every bounded subset $B$ of $E$.
\end{defi}
 
The following result relates the notions presented above, in the special case of linear functionals on a TVS.
\begin{thm} [See \cite{rudin}, Ch. 1, Thm. 1.18] \label{pepitos}
Let  $\Lambda:E\rightarrow \mathbb{C}$ be a  nonzero linear functional whose domain is a TVS. 
Then, the following properties are equivalent:
	\begin{itemize}
		\item[(i)] $\Lambda$ is continuous.
		\item[(ii)] $\operatorname{Ker}\left( \Lambda\right) $ is closed.
		\item[(iii)] $\operatorname{Ker}\left( \Lambda\right) $ is not dense in $E$.
		\item[(iv)] $\Lambda$ is bounded in some neighborhood $V$ of zero (\textit{i.e.} there is some neighborhood $V$ of zero such that the subset $\Lambda\left(V \right)$ of $\mathbb{C}$ is bounded).
	\end{itemize}
\end{thm}
\begin{proof}
Since $\operatorname{Ker}\left( \Lambda\right) =\Lambda^{-1}\left( \left\lbrace 0\right\rbrace \right) $ and $\left\lbrace 0\right\rbrace$ is a closed subset of $\mathbb{C}$, \textit{(i)} implies \textit{(ii)}.
By hypothesis, $\operatorname{Ker}\left( \Lambda\right) \neq E$. 
Hence \textit{(ii)} implies \textit{(iii)}.

Assume \textit{(iii)} holds, so that the complement of $\operatorname{Ker}\left( \Lambda\right) $ has nonempty interior. 
By Theorem \ref{kingcrimson},
	\begin{equation}\label{babysnakes}
		\left( x+V\right) \cap \operatorname{Ker}\left( \Lambda\right)=\emptyset\mbox{,}
	\end{equation}
for some $x$ in $E$ and some balanced neighborhood $V$ of zero.
As $V$ is balanced and $\Lambda$ is linear, $\Lambda \left( V\right) $ is a balanced subset of $\mathbb{C}$. 
Thus, either $\Lambda\left(  V\right) $ is bounded, in which case \textit{(iv)} holds, or $\Lambda \left( V\right) =\mathbb{C}$.
In the latter case, there exists $y$ in $V$ such that $\Lambda \left( y\right) =- \Lambda \left( x\right) $, and so $x+y$ belongs to $\operatorname{Ker}\left( \Lambda\right) $, in contradiction with \eqref{babysnakes}. 
Thus, \textit{(iii)} implies \textit{(iv)}.
 
Finally, if \textit{(iv)} holds, there exists a positive constant $M$ such that, for every $x$ in $V$, $\left| \Lambda \left( x\right)  \right|< M $.
If $r>0$ and if $W=(r/M)V$, then $\left| \Lambda \left( x \right) \right|< r $ for every $x$ in $W$.
Hence $\Lambda$ is continuous at the origin. 
As the topology on $E$ is translation invariant, $\Lambda$ is continuous at every point of $E$.  
\end{proof}
\begin{thm}[See \cite{rudin}, Ch. 1, Thm. 1.30]\label{mafi12}
The following two properties of a set $B$ in a TVS $E$ are equivalent:
	\begin{itemize}
		\item[(i)] $B$ is bounded. 
		\item[(ii)] If $\left(  x_n \right)_{n\in\mathbb{N}}$  is a sequence in $B$  and $\left(  \alpha_n \right)_{n\in\mathbb{N}}$ is a sequence of scalars such that $\alpha_n\rightarrow 0$ as $n\rightarrow \infty$, then $\alpha_n x_n\rightarrow 0$ as $n\rightarrow \infty$.
	\end{itemize}
\end{thm}
\begin{proof}
Suppose $B$ is bounded.
Then, for every balanced neighborhood of zero $V$ of $E$, $tB$ is contained in $V $ for some $t$.
If
	\begin{equation*}
		\left\lbrace x_n : n\in\mathbb{N}\right\rbrace \subseteq B
	\end{equation*}
and $\alpha_n\rightarrow 0$, there exists $N$ such that $|\alpha_n|<t$ if $n> N$.
Since $tB$ is contained in $ V$ and $V$ is balanced, $\alpha_n x_n$ belongs to $V$ for every $n>N$.
Thus $\alpha_n x_n\rightarrow 0$. 

Conversely, if $B$ is not bounded, there is a neighborhood $V$ of zero such that for every $n$ in  $\mathbb{N}$ there exists $r_n<1/n$ and $x_n$ in $B$, such that $r_nx_n$ is not in $V$.
Then, $\left(  r_n\right)_{n\in\mathbb{N}}  $ is a sequence of scalars with $r_n\rightarrow 0$ but the sequence $\left(  r_nx_n \right)_{n\in\mathbb{N}}  $ does not converge to zero.
\end{proof}
\begin{thm} [See  \cite{rudin}, Ch. 1, \S 1.32]\label{mielitas}
Suppose $E$ and $F$ are TVS and $\Lambda:E\rightarrow F$ is linear.
Among the following four properties of $\Lambda$, the implications
	\begin{equation*}
		(i)\implies (ii) \implies (iii)
	\end{equation*}
hold. 
If $E$ is metrizable, then also
	\begin{equation*}
		(iii)\implies (iv) \implies (i)\mbox{,}
	\end{equation*}
so that all four properties are equivalent:
	\begin{itemize}
		\item[(i)] $\Lambda$ is continuous.
		\item[(ii)] $\Lambda$ is bounded.
		\item[(iii)] If $x_n\rightarrow 0$ then  $\left\lbrace\Lambda x_n : n\in\mathbb{N}\right\rbrace $ is bounded.
		\item[(iv)]If $x_n\rightarrow 0$ then $\Lambda x_n\rightarrow 0$. 
	\end{itemize}
\end{thm}
\begin{proof}
Assume \textit{(i)}, let $B$ be a bounded set in $E$, and let $W$ be a neighborhood of $0$ in $F$.
Since $\Lambda $ is continuous, there is a neighborhood $V$ of $0$ in $E$ such that $\Lambda (V)$ is contained in $W$.
Since $B$ is bounded, $B$ is contained in $ tV$ for all large values of $t$.
Then,  
	\begin{equation*}
		\Lambda (B)\subseteq \Lambda (tV) =t\Lambda(V) \subseteq t W.
	\end{equation*}
This shows that $\Lambda(B)$ is a bounded set in $F$. Thus \textit{(i)} implies \textit{(ii)}.
Since convergent sequences are bounded, \textit{(ii)} implies \textit{(iii)}.

Assume now that $E$ is metrizable, that $\Lambda$ satisfies \textit{(iii)} and that $x_n\rightarrow 0$.
By Theorem \ref{sarasa}, there are positive scalars $\gamma_n\rightarrow \infty$ such that $\gamma_n x_n \rightarrow 0$.
Hence, 
	\begin{equation*}
		\left\lbrace \Lambda (\gamma_n x_n): n\in\mathbb{N}\right\rbrace 
	\end{equation*}
is a bounded set in $F$, and now Theorem \ref{mafi12} implies that
	\begin{equation*}
		\Lambda\left(  x_n\right) = \gamma_n^{-1}\Lambda(\gamma_n x_n)\rightarrow 0\mbox{,} \qquad \mbox{ as } n\rightarrow\infty.
	\end{equation*}
Finally, assume that \textit{(i)} fails.
Then there is a neighborhood $W$ of zero in $F$ such that $\Lambda^{-1}(W) $ contains no neighborhood of zero in $E$.
As $E$ has a countable local base, there is a sequence $\left(  x_n \right) _{n\in\mathbb{N}} $ in $E$ such that $x_n\rightarrow 0$ but $\Lambda x_n$ is not in $W$.
Thus \textit{(iv)} fails.
\end{proof}
\section{Seminorms and local convexity}

\begin{defi}[See \cite{rudin}, Ch. 1, Def. 1.33]\label{defisn}
A \emph{seminorm}\index{seminorm} on  a vector space $E$ is a real valued function $p$ on $E$ such that
	\begin{itemize}
		\item[1.] $p(x+y)\leq p(x)+p(y)\mbox{,}\  $ and
		\item[2.] $p(\alpha x) =|\alpha| p(x)$
	\end{itemize} 
for all elements $x$ and $y$ in $E$ and every scalar $\alpha$ in $\mathbb{C}$.
$p$ is called a \emph{norm}\index{norm} if it also  satisfies
	\begin{itemize}
		\item[3.] $p(x)\neq 0$ if $x\neq 0$.
	\end{itemize}
A family $\mathscr{P}$ of seminorms is  \emph{separating} \index{separating family of seminorms}if to each $x\neq 0$ corresponds at least one $p$ in $\mathscr{P}$ with $p(x)\neq 0$. 
\end{defi}
Seminorms are closely related to local convexity in two ways: in every LCS there exists a separating family of continuous seminorms. 
Conversely, if  $\mathscr{P}$ is a separating family of seminorms on a vector space $E$, then $\mathscr{P}$ can be used to define a locally convex topology on $E$ with the property that every $p$ in $\mathscr{P}$ is continuous. 
\begin{thm}[See \cite{rudin}, Ch. 1, Thm. 1.36] \label{gaugin} Let $E$ be a TVS and suppose $\mathscr{B}$ is a zero neighborhood basis whose members are convex balanced sets. 
Associate to every $V$ in $\mathscr{B}$ its \emph{Minkowski functional}\index{Minkowski functional} $\mu_V$ defined by:
	\begin{equation} \label{minki}
		\mu_V(x)=\inf\left\lbrace t>0:t^{-1}x\in V \right\rbrace\mbox{,}
	\end{equation} 
for every $x$ in $E$.
Then
	\begin{itemize}
		\item[(i)] $V=\left\lbrace x\in E : \mu_V(x)<1 \right\rbrace$\mbox{,} and  
		\item[(ii)] $\left\lbrace \mu_V:V\in\mathscr{B}\right\rbrace$ is a separating family of continuous seminorms on $E$.
	\end{itemize}
\end{thm}
\begin{proof}
If $x$ belongs to $V$, then $x/t$ is in $V$ for some $t<1$, because $V$ is open; hence $\mu_V(x)<1$.
If $x$ is not in $V$, then the fact that $x/t$ is in $V$ implies $t\geq 1$, because $V$ is balanced; hence $\mu_V(x)\geq 1$.
This proves \textit{(i)}.

Next, we show that each $\mu_V$ is a seminorm (which is a consequence of the convexity of $V$ and the fact that $V$ is balanced and absorbing).
We begin by proving the triangular inequality.

Given $x$, $y$ in $E$ and $\varepsilon>0$, there exists $t>0$ such that  $t^{-1}x$ belongs to $ V$ and $t<\mu_V(x)+\varepsilon$.
Likewise,  there exists $s>0$ such that $s^{-1}y$ belongs to $V$ and $s<\mu_V(y)+\varepsilon$.
As $V$ is convex,
	\begin{equation*}
		\frac{x+y}{t+s}=\frac{t}{t+s}\frac{x}{t}+\frac{s}{t+s}\frac{y}{s}\in V\mbox{,}
	\end{equation*}
so that $\mu_V(x+y)\leq t+s <\mu_V(x)+\mu_V(y)+2\varepsilon$.
As $\varepsilon>0$ is arbitrary, we get $\mu_V(x+y)\leq \mu_V(x)+\mu_V(y)$.

It remains to show that for any $\lambda$ in $\mathbb{C}$ and $x$ in $E$, $\mu_V(\lambda x)=|\lambda|\mu_V(x)$.
It is clear from the definition \eqref{minki}  of $\mu_V$ that $\mu_V(tx)=t\mu_V(x)$ for every $x$ in $E$ and $t\geq 0$. 
Now, if $\lambda\neq 0$,
	\begin{equation*}
		\begin{split}
			\mu_V(\lambda x)&=\inf \left\lbrace t>0: t^{-1}\lambda x\in V\right\rbrace=\inf \left\lbrace t>0:  \left( \frac{t}{\lambda} \right) ^{-1} x\in  V\right\rbrace \\ 
			& =\inf \left\lbrace t>0:  \left( \frac{t}{|\lambda|} \right) ^{-1} x\in  V\right\rbrace =|\lambda|\inf \left\lbrace \frac{t}{|\lambda|}>0:  \left( \frac{t}{|\lambda|} \right) ^{-1} x\in  V\right\rbrace \\
			&=|\lambda|\inf \left\lbrace s>0:  s^{-1} x\in  V\right\rbrace=|\lambda|\mu_V(x)\mbox{,}
		\end{split}
	\end{equation*}
where in the third equality we have used the fact that $V$ is balanced.
Thus, $\mu_V$ is a seminorm.

The fact that $\mu_V$ is continuous follows from the triangular inequality and \textit{(i)}: given $\varepsilon>0$, whenever $x-y$ belongs to $\varepsilon V$, we have
	\begin{equation*}
		\left| \mu_V(x)-\mu_V(x)\right| \leq \mu_V(x-y)=\varepsilon\mu_V\left( {\varepsilon}^{-1}(x-y)\right) <\varepsilon.
	\end{equation*}
Finally, if $x$ belongs to $E$ and $x\neq 0$, then $x$ does not belong to $V$ for some $V$ in $\mathscr{B}.$
For this $V$, $\mu_V(x)\geq 1$.
Thus, the set $\left\lbrace \mu_V:V\in\mathscr{B}\right\rbrace $ is separating.
\end{proof}
\begin{thm}[See \cite{rudin}, Ch. 1, Thm. 1.37] \label{degas}
Suppose $\mathscr{P}$ is a separating family of seminorms on a vector space $E$.
Associate to each $p$ in $\mathscr{P}$ and to each $n$ in  $\mathbb{N}$ the set 
	\begin{equation*}\label{vpn}
		V(p,n)=\left\lbrace x: p(x)<\frac{1}{n}\right\rbrace .
	\end{equation*}
Let $\mathscr{B}$ be the collection of all finite intersections of sets of the form $V(p,n)$.
Then, $\mathscr{B}$ is a convex balanced local base for a topology $\tau$ on $E$, which turns $E$ into a LCS such that 
	\begin{itemize}
		\item[(i)] every $p$ in $\mathscr{P}$ is continuous, and 
		\item[(ii)] A subset $B$ of $E$ is bounded  if and only if every $p$ in $\mathscr{P}$ is bounded on $B$.
	\end{itemize}
\end{thm}
\begin{proof}
Declare a subset $V$ of $E$ to be open if and only if $V$ is a (possibly empty) union of translates of members of $\mathscr{B}$.
This clearly defines a translation-invariant topology $\tau$ on $E$. 
It is also clear that each set  $V(p,n)$, and therefore each member of $\mathscr{B}$, is convex and balanced. 
Thus, $\mathscr{B}$ is a convex balanced local base for $\tau$.

Suppose $x$ is an element of $E$, such that  $x\neq 0$.
As $\mathscr{P}$ is a separating family of seminorms, $p(x)>0$ for some $p$ in $\mathscr{P}$.
Then, for some $n$ in $\mathbb{N}$, $p(nx)>1$ and therefore $x$ is not in $V(p,n)$.
Then, $0$ is not in the neighborhood $x-V(p,n)$ of $x$, so that $x$ is not in the closure of $\left\lbrace 0\right\rbrace $.
Thus, the set $\left\lbrace 0\right\rbrace $ is closed and, since $\tau$ is translation-invariant, every point of $E$ is a closed set.

Next, we show that addition and scalar multiplication are continuous.
Let $U$ be a neighborhood of $0$ in $E$.
Then,
	\begin{equation}\label{i}
		\bigcap_{i=1}^m V(p_i,n_i)\subseteq U
	\end{equation}
for some $p_1\mbox{,}\dots\mbox{, }p_m$ in $\mathscr{P}$ and some positive integers $n_1\mbox{,}\dots\mbox{, }n_m$.
Put
	\begin{equation}\label{v}
		V=\bigcap_{i=1}^m V(p_i,2n_i).
	\end{equation}
Since every $p$ in $\mathscr{P}$ is subadditive, $V+V\subseteq U$.
This proves that addition is continuous.

Suppose now that $x$ is an element of $E$, $\alpha$ is a scalar, and $U$ and $V$ are as above.
As any open set is absorbing,  $x$ belongs to $sV$ for some $s>0$.
Put $t=s/(1+|\alpha|s)$.
If $y$ is in $x+tV$ and $|\beta-\alpha|<1/s$, then
	\begin{equation*}
		\beta y-\alpha x=\beta (y-x)+(\beta-\alpha)x\mbox{,}
	\end{equation*}   
which lies in 
	\begin{equation*}
		|\beta|t V+|\beta - \alpha|sV\subseteq V+V\subseteq U\mbox{,}
	\end{equation*}
since
	\begin{equation*}
		|\beta|t=|\beta|\frac{s}{1+|\alpha|s}\leq \frac{\left( |\beta-\alpha|+|\alpha|\right)s}{1+|\alpha|s}< \frac{\left( 1/s+|\alpha|\right)s}{1+|\alpha|s}=1\mbox{,}
	\end{equation*}
and $V$ is balanced.
This proves that scalar multiplication is continuous.

Thus, $E$ is a LCS. 
The definition of $V(p,n)$ shows that every $p$ in $\mathscr{P}$ is continuous at zero.                                           

Finally, suppose a subset  $B$ of $E$ is bounded.
Fix $p$ in $\mathscr{P}$. 
Since $V(p,1)$ is a neighborhood of $0$, 
	\begin{equation*}
		B\subseteq kV(p,1)
	\end{equation*}
for some $k<\infty$.
Hence, $p(x)<k$ for every $x$ in $B$. 
It follows that every $p$ in $\mathscr{P}$ is bounded on $B$.

Conversely, suppose $B$ satisfies this condition, $U$ is a neighborhood of $0$, and \eqref{i} holds.
There are numbers $M_i<\infty$ such that $p_i<M_i$ on $B$ for every $i$ such that $1\leq i \leq m$.
If $n>M_in_i$ for $1\leq i \leq m$, it follows that $B$ is contained in $nU$, so that $B$ is bounded.
\end{proof}
\begin{propo}[See \cite{rudin}, Ch. 1, Rmk. 1.38]\label{aylan}
If $\mathscr{B}$ is a convex balanced local base for the topology $\tau$ of a LCS $E$, then $\mathscr{B}$ generates a separating family,
	\begin{equation*}
		\mathscr{P}=\left\lbrace \mu_V:V\in\mathscr{B}\right\rbrace
	\end{equation*}
of continuous seminorms on $E$, as in Theorem \ref{gaugin}.  
This family $\mathscr{P}$ in turn induces a topology $\tau_1$ on $E$, by the process described in Theorem \ref{degas}.
It turns out that $\tau=\tau_1$.
\end{propo}
\begin{proof}
Every $p$ in $\mathscr{P}$ is $\tau$ continuous, so that the sets $V(p,n)$ of Theorem \ref{degas} are in $\tau$.
Hence, $\tau_1\subseteq \tau$.

Conversely, if $W$ belongs to $\mathscr{B}$ and $p=\mu_W$, then
	\begin{equation*}
		W=\left\lbrace x\in E : \mu_W(x)<1 \right\rbrace=V(p,1).
	\end{equation*} 
Thus, $W$ belongs to $\tau_1$ for every $W$ in $\mathscr{B}$. 
This implies that $\tau\subseteq \tau_1$.
\end{proof}
\begin{rem}[See \cite{rudin}, Ch. 1, Rmk. 1.38]\label{pinocho}
If $\mathscr{P}=\left\lbrace p_i \right\rbrace_{i\in\mathbb{N}}$ is a countable separating family of seminorms on a vector space $E$ then, by Theorem  \ref{degas},  $\mathscr{P}$ induces a topology $\tau$ on $E$ with a countable convex balanced local base $\mathscr{B}$, turning $E$ into a LCS.
As $\mathscr{B}$ is countable, by Theorem \ref{liszt} we have that there exists an invariant metric $d$ on $E$ which induces the same topology $\tau$ on $E$, and the open balls centered at zero are convex and balanced.
  
Another compatible invariant metric $d'$ can be defined in terms of  $\mathscr{P}$ by the formula
	\begin{equation}\label{metric}
		d'(x,y)=\sum_{i=1}^{\infty} \frac{1}{2^i} \frac{p_i(x-y)}{\left[ 1+p_i(x-y)\right] }.
	\end{equation}
It has to be noted that the balls defined by the metric in equation \eqref{metric} need not be convex, while the balls defined by the metric $d$ given by Theorem \ref{liszt}  satisfy this condition.   
Another invariant metric $d''$ can be defined in order to correct this, namely
	\begin{equation}
		d''(x,y)=\max_{i}  \frac{c_i p_i(x-y)}{\left[ 1+p_i(x-y)\right] }\mbox{,}
	\end{equation} 
where $\left(  c_i\right)_{i\in\mathbb{N}}$ is some fixed sequence of positive numbers which converges to zero as $i\rightarrow \infty$.
The balls defined by $d''$ form a convex balanced local base. 
\end{rem}
\begin{defi}
A LCS $E$ is a \emph{Fr\'echet space}\index{Fr\'echet space} if its topology is induced by a complete  invariant  metric.
\end{defi}
\section{Open Mapping Theorem}
 
Let $\Lambda: E\rightarrow F$ be a map between TVS.
We say that $\Lambda$ is \emph{open at a point} $p$ in $E$ if $\Lambda(V)$ contains a  neighborhood of $\Lambda(p)$ whenever $V$ is a neighborhood of $p$.
We say that $\Lambda$ is \emph{open}\index{open map} if $\Lambda(U)$ is open in $F$ whenever $U$ is open in $E$.
 
It is clear that $\Lambda$ is open if and only if $\Lambda$ is open at every point of $E$.
Because of the invariance of vector topologies, it follows that a linear mapping of one TVS into another is open if and only if it is open at the origin.
Let us also note that a one-to-one continuous mapping $\Lambda$ of $E$ onto $F$ is a homeomorphism precisely when $\Lambda$ is open.
\begin{thm}[Open Mapping Theorem. See  \cite{rudin}, \S 2.12, Corollaries (a) and (b)]\index{Open Mapping Theorem}
Let  $\Lambda: E\rightarrow F$ be a continuous surjective linear mapping between Fr\'echet spaces. 
The following statements are true:
	\begin{itemize}
		\item[(i)] The map $\Lambda$ is open.
		\item[(ii)] If the map $\Lambda$ satisfies \textit{(i)} and is one-to-one, then $\Lambda^{-1}:E\rightarrow F$ is continuous.
	\end{itemize}
\end{thm}
\section{Quotient spaces}

\begin{defi}[See \cite{rudin}, Ch. 1, Def. 1.40]\label{espcoc1} Let $F$ be a subspace of a vector space $E$.
For each $x$ in $E$ let $\pi(x)$ be the coset of $F$ that contains $x$:
	\begin{equation*}
		\pi(x)=x+F.
	\end{equation*}
These cosets are elements of a vector space $E/F$ called the \emph{quotient space of $E$ modulo $F$}\index{quotient space}, in which addition and scalar multiplication is such that
	\begin{equation}\label{quotientoperations}
		\pi(x)+\pi(y)=\pi(x+y) \mbox{ and } \lambda\pi(x)=\pi(\lambda x).
	\end{equation}
Since $F$ is a vector space, the operations \eqref{quotientoperations} are well-defined.
By \eqref{quotientoperations}
	\begin{equation}
		\pi: E \rightarrow E/F
	\end{equation}
is a linear mapping with $F$ as its null space, which we call the \emph{quotient map of $E$ onto $E/F$}\index{quotient map}.

If $\tau$ is a vector topology on $E$, and $F$ is a closed subspace of $E$, let $\tau_F$ be the collection of all subsets $V$  of  $E/F$ for which $\pi^{-1}\left(V \right) $ belongs to $\tau$.
Then $\tau_F$ turns out to be a topology on $E/F$, called the \emph{quotient topology}\index{quotient topology}.
Some of its properties are listed in the following theorem.
\end{defi}
\begin{thm}[See \cite{rudin}, Ch. 1, Thm. 1.41]\label{espcoc2}
Let $F$ be a closed subspace of a TVS E.
Let $\tau$ be a topology on $E$ and define $\tau_F$ as above.
	\begin{itemize}
		\item[(i)]$\tau_F$ is a vector space topology on $E/F$; the quotient map $\pi: E \rightarrow E/F$ is linear continuous and open.
		\item[(ii)]If $\mathscr{B}$ is a local base for $\tau$,  then the collection of all sets $\pi\left( V \right)$ with $V$ in $\mathscr{B}$ is a base for $\tau_F$.
		\item[(iii)]Each of the following properties of $E$ is inherited by $E/F$: local convexity, local boundedness, metrizability.
		\item[(iv)]If $E$ is a Fréchet space or Banach space, so is $E/F$.
	\end{itemize}
\end{thm}
\begin{proof}
Since
	\begin{equation*}
		\pi^{-1}\left( V\cap W\right)=\pi^{-1}\left( V\right) \cap \pi^{-1}\left( W\right)
	\end{equation*}
and
	\begin{equation*}
		\pi^{-1}\left(\bigcup_{i\in I}V_i\right) =\bigcup_{i\in I}\pi^{-1}\left( V_i\right)\mbox{,}
	\end{equation*}
$\tau_F$ is a topology.
A subset $H$ of $E/F$ is $\tau_F$-closed if and only if $\pi^{-1}\left(H \right) $ is $\tau$-closed.
In particular, every point of $E/F$ is closed since
	\begin{equation*}
		\pi^{-1}\left(\pi(x) \right)=x+ F 
	\end{equation*}
and $F$ is closed.

The continuity of $\pi$ follows directly from the definition of $\tau_F$.
Next, suppose $V$ belongs to $\tau$.
Since
	\begin{equation*}
		\pi^{-1}\left(\pi\left( V\right)  \right)=V+ F
	\end{equation*}
and $V+F$ belongs to $\tau$, it follows that $\pi(V)$ belongs to $\tau_F$.
Thus $\pi$ is an open mapping.

If now $W$ is a neighborhood of $0$ in $E/F$, there is a neighborhood $V$ of $0$ in $E$ such that
	\begin{equation*}
		V+V\subseteq \pi^{-1} \left( W\right). 
	\end{equation*} 
Hence, 
	\begin{equation*}
		\pi\left( V\right) + 	\pi\left( V\right)\subseteq W. 
	\end{equation*}
Since $\pi$ is open, $\pi(V)$ is a neighborhood of $0$ in $E/F$.
Addition is therefore continuous in $E/F$.

The continuity of scalar multiplication in $E/F$ is proved in the same manner.
This proves \textit{(i)}.
It is clear that \textit{(i)} implies \textit{(ii)}.
With the aid of Theorem  \ref{liszt}, it follows that \textit{(ii)} implies \textit{(iii)}.

Suppose next that $d$ is an invariant metric on $E$, compatible with $\tau$.
Define $\rho$ by
	\begin{equation*}
		\rho\left(\pi(x),\pi(y) \right) =\inf\left\lbrace d(x-y,z): z\in F \right\rbrace. 
	\end{equation*}
This may be interpreted as the distance from $x-y$ to $F$.
It can be verified that $\rho$ is well-defined and that it is an invariant metric on $E/F$.
Since
	\begin{equation*}
		\pi\left( \left\lbrace x: d(x,0)<r \right\rbrace \right) =\left\lbrace u: \rho(u,0)<r \right\rbrace\mbox{,} 
	\end{equation*}
it follows from \textit{(ii)} that $\rho$ is compatible with $\tau_F$.

If $E$ is normed, this definition of $\rho$ specialices to shield what is usually called the \emph{quotient norm}\index{quotient norm} of $E/F$:
	\begin{equation*}
		\left\| \pi(x)\right\|=\inf \left\lbrace \left\|x-z \right\| : z \in F \right\rbrace .
	\end{equation*}

To prove \textit{(iv)} we have to show that $\rho$ is a complete metric whenever $d$ is complete.

Suppose $\left(u_n \right)_{n\in\mathbb{N}} $ is a Cauchy sequence in $E/F$, relative to $\rho$.
There is a subsequence $\left(u_{n_i} \right)_{i\in\mathbb{N}} $ with 
	\begin{equation*}
		\rho\left( u_{n_i},u_{n_{i+1}}\right)<\frac{1}{2^i}. 
	\end{equation*}
One can then inductively choose $x_i$ in $E$ such that  $\pi(x_i)=u_{n_i}$ and
	\begin{equation*}
		d\left( x_{i},x_{i+1}\right)<\frac{1}{2^i}. 
	\end{equation*} 
If $d$ is complete, the Cauchy sequence $\left( x_i\right)_{i\in\mathbb{N}}$ converges to some $x$ in $E$.
The continuity of $\pi$ implies that $u_{n_{i}}\rightarrow \pi(x)$ as $i\rightarrow \infty$.
But if a Cauchy sequence has a convergent subsequence then the full sequence must converge.
Hence $\rho$ is complete and this concludes the proof of Theorem \ref{espcoc2}.
\end{proof}
\section{Topologies on the dual}\label{kutiii}

Let $E$ be a topological vector space over $\mathbb{C}$, and $E'$ its \emph{continuous dual space}\index{continuous dual of a topological vector space}; that is to say, the vector space of all continuous linear forms on $E$. If $x'$ is in $E'$, we shall denote by $\braket{x',x}$ its value at the point $x$ in $E$ (see \cite{treves}, Ch. 19).
\subsubsection{Weak star topology}\label{wstdef}
The \emph{weak star topology}\index{weak star topology} on $E'$, denoted by $\sigma(E',E)$ or $\sigma^\ast$, is the topology which has as a basis of neighborhoods of zero the family of the sets of the form
	\begin{equation*}
		W_\varepsilon\left( x_1,\dots,x_n\right)=\left\lbrace x'\in E : \left| \braket{x',x_j}\right|\leq \varepsilon\mbox{, }j=1 \dots n\right\rbrace. 
	\end{equation*}
Here, $\left\lbrace x_1, \dots ,x_n\right\rbrace$ is a finite subset of $E$ and $\varepsilon >0$.

The weak star topology on $E'$ is the topology of pointwise convergence in $E$.
Thus, a net $\left( x'_\lambda\right)_{\lambda\in\Lambda}$ of continuous linear functionals  converges to zero if at each point $x$ of $E$, their values $\braket{x'_\lambda,x}$ converge to zero in the complex plane (see \cite{treves}, Ch. 19, Example I). 
\subsubsection{Strong topology}\label{stdef}
The \emph{strong topology}\index{strong topology} on $E'$, denoted by $b^\ast$, is the topology which has as a basis of neighborhoods of zero  the family of the sets of the form
	\begin{equation*}
		W_{\varepsilon}(B)=\left\lbrace x'\in E' : \sup\limits_{x\in B}\left|\braket{x',x} \right| \leq \varepsilon\right\rbrace\mbox{,}
	\end{equation*}
with $B$ a bounded subset of $E$ and $\varepsilon >0$.

The strong topology on $E'$ is the topology of uniform convergence on bounded subsets of $E$.
Thus, a net $\left( x'_\lambda\right)_{\lambda\in\Lambda}$ of continuous linear functionals  converges to zero if for every bounded set $B$, $x'_{\lambda}$ converges uniformly to zero over $B$ (see \cite{treves}, Ch. 19, Example IV).\\

Let $\Lambda:E\rightarrow F$ be a continuous linear mapping between two TVS.  Given $\varphi$ in $F'$, it is clear that the composition $\varphi\circ \Lambda$ is a continuous linear functional on $E$.
This defines a map
	\begin{align}\label{trdef}
	   \Lambda': F' &\rightarrow E'
	   \\
	   \varphi &\mapsto \varphi \circ \Lambda\nonumber
	\end{align}
known as the \emph{transpose}\index{transpose of a map} of $\Lambda$, which is obviously linear.

There is an important result concerning the continuity of the transpose of a linear map $\Lambda$.
\begin{propo}[See \cite{treves}, Corollary of Prop. 19.5]\label{siqueiros}
Let $E,F$ be two TVS and $\Lambda$ a continuous linear map from $E$ to $F$. 
Then, 
	\begin{equation*}
		\Lambda': F' \rightarrow E' 
	\end{equation*}
is continuous when the duals $E'$ and $F'$ carry the weak star topology 
(resp. the strong topology).
\end{propo}
\begin{proof} For a sketch of the proof, see the explanation on page 199 of \cite{treves}, just before Proposition 19.5. 
The proof relies on many concepts and partial results, discussed previously  in Tr{\`e}ves' book.   
\end{proof}
\section{Orthogonality relations} 

In working with dual spaces of TVS some relations of orthogonality are going to be necessary
(see \cite{langlang93}, Ch. XV, \S 2).
We first give the following definitions and notations.
\begin{defi}\label{orthdef}
Let $E$ be a TVS. 
If $M$ is a subspace of $E$, we write
	\begin{equation*}
		M^{\perp}=\left\lbrace \varphi\in E' : \braket{\varphi,x}=0\quad\forall x \in M \right\rbrace .
	\end{equation*}
If $N$  is a subspace of $ E'$, we write
	\begin{equation*}
		^{\perp}N=\left\lbrace x\in E : \braket{\varphi,x}=0\quad\forall \varphi \in N \right\rbrace .
	\end{equation*}
We say that $M^{\perp}$ $($resp. $^{\perp}N$$)$ is the \emph{orthogonal} \index{orthogonal of a subspace}of $M$ $($resp. $N$$)$.
\end{defi}
Note that $^{\perp}N$ is a closed subspace of $E$, being the intersection of the kernels of a family of continuous linear functionals, namely $\left\lbrace\varphi:\varphi\in N  \right\rbrace $.
In addition, when $E'$ is given the weak star topology, $M^{\perp}$ is a closed subspace of $E'$.
Indeed, let $\left(  x'_{\lambda} \right) _{\lambda\in\Lambda}$ be a net in $M^{\perp}$ which converges to some element $x'$, \textit{i.e.} $\braket{x'_\lambda,x}$ converges to $\braket{x',x}$ for every $x$ in $E$. 
By definition of $M^{\perp}$, $\braket{x'_\lambda,x}=0$ for every $x$ in $M$. 
Then, for every $x$ in $M$, $\braket{x',x}=0$, thus $x'$ belongs to $M^{\perp}$. 
When $E'$ is given the strong topology, $M^{\perp}$ is also a closed subspace of $E'$.
To see this, let $\left( x'_{\lambda} \right)_{\lambda\in\Lambda}$ be a net in $M^{\perp}$ which converges to some element $x'$, that is, $\braket{x'_\lambda,\cdot}$ converges uniformly to $\braket{x',\cdot}$ on bounded sets. 
As $\left\lbrace x \right\rbrace $  is bounded for every $x$ in $E$, $0=\braket{x'_\lambda,x}$ converges to $\braket{x',x}$ for every $x$ in $M$.
Thus $x'$ belongs to $M^{\perp}$. 

An alternative way to show that $M^{\perp}$ is a closed subspace of $E'$ is contained in the following argument.
When the duals $M'$ and $E'$ carry the weak star topology or the strong topology, by Proposition \ref{siqueiros}, 
the inclusion
	\begin{equation*}
		i:M\hookrightarrow E
	\end{equation*}
induces a continuous linear map by restriction
	\begin{equation*}
		i':E'\hookrightarrow M'
	\end{equation*}
whose kernel is precisely $M^{\perp}$, and therefore $M^{\perp}$ is closed in $E'$.
\begin{defi}[See \cite{conway}, Ch. 5, Def. 1.6]\label{poldef}
Let $E$ be a LCS. 
If $M\subseteq E$, the  \emph{polar} of $M$, denoted by $M^{\circ}$, is the subset of $E'$ defined by 
	\begin{equation*}
		M^{\circ}=\left\lbrace\varphi \in E':\left|  \braket{\varphi,x}\right| \leq 1\ \forall x\in M \right\rbrace\mbox{.}
	\end{equation*}
If $N\subseteq E'$, the \emph{polar} of $N$, denoted by $^{\circ}N$, is the subset of $E$ defined by 
	\begin{equation*}
		^{\circ}N=\left\lbrace x \in E:\left|  \braket{\varphi,x}\right| \leq 1\ \forall \varphi\in N \right\rbrace.
	\end{equation*}
\end{defi}
The next result, the Bipolar Theorem, is fundamental in functional analysis.
Corollary \ref{egmont}, derived from it, is  very useful in what follows.
\begin{thm}[Bipolar Theorem. See \cite{conway}, Ch. 5, Thm. 1.8]\index{Bipolar Theorem}
Let $E$ be a LCS and $F\subseteq E$.
Then, $^\circ \left( F^\circ\right)$ is the intersection of all closed convex balanced subsets of $E$ that contain $F$; \textit{i.e.} the set $^\circ \left( F^\circ\right)$ is the \emph{closed convex balanced hull} of $F$.
\end{thm}
\begin{cor}[See \cite{conway}, Ch. 5, Corollary 1.9]\label{egmont}
If $E$ is a LCS and $F\subseteq \left( E',\sigma^\ast\right)$, then  $\left(  ^\circ F\right) ^\circ$ is the closed convex balanced hull of $F$ with respect to the weak star topology in $E'$.
\end{cor}
Now we can state and prove the following useful result.
\begin{propo}\label{bipo}
Let $E$ be a LCS.
Then
	\begin{itemize}
		\item[(i)] $\bar{F}=\; ^\perp \left( F^\perp\right)$ if $F$ is a subspace of $E$.
		\item[(ii)]$\bar{F}=\;\left(  ^\perp F\right) ^\perp$ if $F$ is a subspace of $\left( E',\sigma^\ast\right) $.
	\end{itemize}
\end{propo}
\begin{proof}\label{polarsdef}
Suppose $F$ is a subspace of $E$.
Then, given $\varphi$ in $F^\circ$, we know that $\left|  \braket{\varphi,x}\right| \leq 1$ for every $x$ in $F$.
As $F$ is a subspace of $E$, we have that $\left|  \braket{\varphi,nx}\right| \leq 1$ for every $x$ in $F$ and for every $n$ in $\mathbb{N}$, which in turn implies that $\left|  \braket{\varphi,x}\right| \leq n^{-1}$ for every $x$ in 
 $F$ and for every $n$ in $\mathbb{N}$. 
Therefore,  $\braket{\varphi,x}=0$ for every $x$ in $F$ and we conclude that $F^\circ= F^\perp$. 
In particular $F^\circ$ is convex.
In a similar way, if $F$ is a subspace of $E'$ we have that $^\circ F=^\perp F$, which is also convex.
By the Bipolar Theorem and Corollary \ref{egmont} we conclude the assertion of the Proposition.
\end{proof}
\begin{propo}\label{snoopy}
\textit{(ii)} in Proposition \ref{bipo} is not true when $E'$ is given the strong topology $b^\ast$, not even if $E$ is a Banach space.
\end{propo}
\begin{proof}
It suffices to see that there exist a Banach space $E$ and a closed subspace $S$ of $\left(E',b^\ast \right)$ such that $S$ is not closed in the space $\left(E',\sigma^\ast\right)$.
For this, it suffices to find a Banach space $E$ which is not reflexive when considering the strong topology on the sucessive duals,
\textit{i.e.} a Banach space $E$ such that the canonical inclusion 
	\begin{equation*}
		E\hookrightarrow \left(\left( E',b^\ast\right)', b^\ast \right) 
	\end{equation*}
is not surjective.
For if such a space $E$ were found, and as it is always true that a LCS is reflexive when considering the weak star topology on the sucessive duals, \textit{i.e.}
	\begin{equation*}
		E = \left(\left( E',\sigma^\ast\right)', \sigma^\ast \right)\mbox{,}
	\end{equation*}
then we could assert that there exists a functional $L:E'\rightarrow \mathbb{C}$ which is continuous when $E'$ has the strong topology but not when it is given the weak star topology (in this case $L$ cannot be given  by the evaluation on some fixed $v$ in $E$, which is always true for the maps of $\left(\left( E',\sigma^\ast\right)', \sigma^\ast \right)$).

Taking $S$ to be the kernel of such a functional $L$, we then have that $S$ is closed  in the space  $\left(\left( E',b^\ast\right)', b^\ast \right)$ but not in the space $\left(\left( E',\sigma^\ast\right)', \sigma^\ast \right)$.

It only remains to note that such a Banach space $E$, which is not reflexive, exists, namely $E=\ell_1$, where  
	\begin{equation}\label{elle12}
		\ell_1:=\left\lbrace\left(a_i \right)_{i\in\mathbb{N}} : a_i\in \mathbb{C}\  \forall i\in\mathbb{N}\mbox{,}\mbox{ and }\sum_{i=1}^\infty \left|a_i \right|<\infty \right\rbrace.
	\end{equation}
\end{proof}
For every continuous linear map between Fr\'echet spaces there is a duality associated, given in the following result.
\begin{propo}\label{evgeny}
Let $\Lambda:E\rightarrow F$ be a continuous linear map between Fr\'echet spaces.
Then 
	\begin{itemize}
			\item[(i)]$\operatorname{Ker}\left( \Lambda'\right) =\left(\operatorname{Im}\left( \Lambda\right)  \right)^{\perp}$ and
			\item[(ii)] if the image of $\Lambda$ is closed, then so is the image of $\Lambda'$ and  $\ \operatorname{Im}\left( \Lambda'\right) =\left( \operatorname{Ker} \left( \Lambda\right)  \right)^{\perp}$.
	\end{itemize} 
\end{propo}  
\begin{proof}
The proof of \textit{(i)} is straightforward:
	\begin{equation*}
		\begin{split}
			\varphi\in \operatorname{Ker}\left( \Lambda'\right)  & \iff \Lambda'(\varphi)=0 \iff \varphi\circ \Lambda=0 \iff \braket{\varphi,\Lambda(x)}=0 \ \forall x\in E \\
			& \iff  \braket{\varphi,y}=0 \ \forall y\in \operatorname{Im}\left( \Lambda\right)  \iff \varphi \in \left(\operatorname{Im}\left( \Lambda\right)  \right)^{\perp}. 
		\end{split}
	\end{equation*}

Now we prove \textit{(ii)}. 
Let $\phi$ be an element of $F'$ and $x$ in $\operatorname{Ker}\left( \Lambda\right) $.
Then,
	\begin{equation*}
		\braket{\Lambda'(\phi) ,x}=\braket{\phi,\Lambda(x)}=0.
	\end{equation*}
Hence, 
	\begin{equation*}
		\operatorname{Im}\left( \Lambda'\right) \subseteq\left( \operatorname{Ker} \left( \Lambda\right)  \right)^{\perp}.
	\end{equation*}
Conversely, let $\varphi$ be in $E'$ such that $\braket{\varphi, x}=0$ for every $x$ in $\operatorname{Ker}\left( \Lambda\right) $.
By the First Isomorphism Theorem there exists a unique continuous linear mapping
	\begin{equation*}
		\sigma:E/\operatorname{Ker}\left( \Lambda\right) \rightarrow \operatorname{Im}\left( \Lambda\right)\mbox{,}
	\end{equation*}
which is bijective. 
By the Open Mapping Theorem, $\sigma$ is an isomorphism.
We define $\overline{\varphi}$ on $E/\operatorname{Ker}\left( \Lambda\right) $, by   $\braket{\overline{\varphi},x+\operatorname{Ker}\left( \Lambda\right) }=\braket{\varphi,x}$.
$\overline{\varphi}$ is well-defined because  $\braket{\varphi, x}=0$ for every $x$ in $\operatorname{Ker}\left( \Lambda\right) $.
Then,  $\overline{\varphi}\circ\sigma^{-1}$ is a functional on $\operatorname{Im}\left( \Lambda\right) $, which can be extended to a functional $\psi$ on $F$ by the Hahn-Banach Theorem.
	\begin{equation*}
		\xymatrix{
		E\ar[d]_\pi \ar@{->}[r]^\Lambda& \operatorname{Im}\left( \Lambda\right)  \ar@{->}[r]^{i} & F \ar@{-->}[ddll]^\psi\\
		E/\operatorname{Ker}\left( \Lambda\right)  \ar@{->}[d]_{\overline{\varphi}}\ar@{-->}[ur]_{\sigma}&\\
		\mathbb{C}&}
	\end{equation*}
Then, it is clear that $\Lambda'(\psi)=\varphi$, because given $x$ in $E$
	\begin{equation*}
		\begin{split}
				\braket{\Lambda'(\psi),x}&=\braket{\psi,\Lambda(x)}= \left\langle\overline{\varphi}\circ\sigma^{-1}, \Lambda(x)\right\rangle= \left\langle\overline{\varphi},\left( \sigma^{-1}\circ \Lambda\right) (x) \right\rangle\\
				&=\braket{\overline{\varphi},x+\operatorname{Ker}\left( \Lambda\right) }=\braket{\varphi,x}.
		\end{split}
	\end{equation*} 
Thus $\varphi$ is in $\operatorname{Im}\left( \Lambda'\right) $.
This proves that $\operatorname{Im}\left( \Lambda'\right) =\left( \operatorname{Ker} \left( \Lambda\right)  \right)^{\perp}$, and in particular proves that $ \operatorname{Im}\left( \Lambda'\right) $ is closed.
\end{proof}
Notice that the previous results can be expressed as follows.
\begin{propo}[See \cite{meisevogt}, Prop. 26.4]\label{joann2}
A sequence of Fréchet spaces and continuous linear maps
	\begin{equation} \label{george1}
 		0\rightarrow A \overset{f}{\rightarrow} B \overset{g}{\rightarrow} C \rightarrow 0 
	\end{equation}
is  exact if and only if the dual sequence of TVS and continuous linear maps  
	\begin{equation} \label{george2}
 		0\rightarrow C' \overset{g'}{\rightarrow} B' \overset{f'}{\rightarrow} A' \rightarrow 0 
	\end{equation}
is exact, when its vector spaces are endowed  with either the weak star topology or the strong topology. 
\end{propo}
\begin{propo}\label{ramiro}
Let
	\begin{equation} \label{george3}
		0\rightarrow A \overset{f}{\rightarrow} B \overset{g}{\rightarrow} C \rightarrow 0 
	\end{equation}
be a short exact sequence of Fr\'echet spaces.
The following statements are equivalent:
	\begin{itemize}
		\item[(i)]there exists a continuous linear section $s:C\rightarrow B$ of $g:B\rightarrow C$.
		\item[(ii)]there exists a continuous linear section $s:C\rightarrow B$ of $g:B\rightarrow C$ such that $\operatorname{Im}(s)$ is closed.
		\item[(iii)]there exists a continuous linear retraction $r:B\rightarrow A$ of $f:A\rightarrow B$.
		\item[(iv)]there exists a continuous linear retraction $r:B\rightarrow A$ of $f:A\rightarrow B$ such that $\operatorname{Im}(r)$ is closed.
		\item[(v)] There exists an isomorphism $\Omega: A\oplus C\rightarrow B$ such that the following diagram commutes
			\begin{equation}\label{dia}
				\xymatrix{A
				\ar@{=}[d]_{} \ar[r]^{f}& B\ar[r]^{g}& C\ar@{=}[d]_{}\\
				A\ar@{^{(}->}[r]^{i {\hspace*{0.5 cm}}}& A\oplus C\ar[u]^{\Omega} \ar@{->>}[r]^{\pi {\hspace*{-0.2 cm}}}& C}
			\end{equation} 
		where $i:A\hookrightarrow A\oplus C$ is the canonical inclusion and $\pi: A\oplus C \twoheadrightarrow C $ is the canonical projection.
	\end{itemize}
\end{propo}
\begin{proof}
We first observe that item \textit{(iv)}, does not give further information as $r$, being a retraction, is surjective.
However, to emphasize the symmetry of the result given by the Proposition, we chose to include the item in its statement.

Two of the implications are trivial, namely that \textit{(ii)} implies \textit{(i)} and that \textit{(iv)} implies \textit{(iii)}.
Now we prove that \textit{(v)} implies \textit{(ii)}. 
Let $\Omega: A\oplus C\rightarrow B$ be an isomorphism such that  the diagram \eqref{dia} commutes and define $s:C\rightarrow B$ by $s=\Omega\circ i_C$ where $i_C:C\rightarrow A\oplus C$ is the inclusion given by $i_C(c)=(0,c)$ for every $c$ in $C$.
$s$ is clearly linear, continuous and satisfies $g\circ s= \operatorname{Id}_C$.
Therefore $s$ is a section.
To show that $\operatorname{Im}(s)$ is closed, notice that as $g$ is surjective, $\operatorname{Im}(s)=\operatorname{Im}(s\circ g)$, so we may as well prove that $\operatorname{Im}(s\circ g)$ is closed. 
On the other hand, as  $g\circ s= \operatorname{Id}_C$, we have $\operatorname{Im}(s\circ g)=\operatorname{Ker}\left( \operatorname{Id}_{B}-s\circ g\right) $ which is closed as $\operatorname{Id}_{B}-s\circ g$ is a continuous linear function.

To show that \textit{(v)} implies \textit{(iv)} let $\Omega: A\oplus C\rightarrow B$ be an isomorphism such that  the diagram \eqref{dia} commutes and define $r:B\rightarrow A$ by $r=\pi_A\circ\Omega^{-1}$, where $\pi_A:A\oplus C\rightarrow A$ is the projection given by $\pi_A(a,c)=a$ for every $(a,c)$ in $A\oplus C$.
$r$ is clearly linear, continuous and satisfies $r\circ f= \operatorname{Id}_A$.
Therefore, $r$ is a retraction.
Moreover, by the fact that $r\circ f= \operatorname{Id}_A$, $\operatorname{Im}(r)=A$, which  is closed.

We turn to the proof of \textit{(iii)} implies \textit{(v)}.
Let $r:B\rightarrow A$ be a continuous linear retraction of $f:A\rightarrow B$.
Define  $\Theta: B \rightarrow A\oplus C$ by $\Theta(b)=\left( r(b),g(b)\right) $.
It is clear that $\Theta$ is linear, continuous and it makes the diagram in \eqref{dia} commutative.
To check injectivity, suppose that $\Theta(b)=0$ for some $b$ in $B$.
In particular $g(b)=0$, thus $b$ belongs to $\operatorname{Im}(f)$. 
Let $a\in A$ be such that $f(a)=b$. 
Then, as $\Theta(b)=0$ also implies that $r(b)=0$ we have that $0=\left( r\circ f\right) (a)=a$.
We conclude that $b=0$ and therefore, $\Theta $ is injective.   
To see that it is surjective take an element $(a,c)$ in $A\oplus C$.
By surjectivity of $g$ there exists an element $b'$ in $B$ such that $g(b')=c$.
Then, by exactness of \eqref{george3}, applying $g$ to any element of the form $b=b'+f(a')$ with $a'$ in $A$ also gives $c$ as a result.
It remains to find an adequate element $a'$ in $A$ such that $r(b)=a$.
A simple calculation shows that $a'=a-r(b')$ satisfies this condition.
Therefore taking $b=b'+f\left(a-r(b') \right)$ we have:
	\begin{equation*}
		\Theta(b)= \left( r\left( b'\right) +r\circ f\left( a-r(b') \right) ,g\left( b'\right) +g\circ f\left( a' \right) \right) =\left( a,c\right)\mbox{,}
	\end{equation*} 
where we have used that $r\circ f=\operatorname{Id}_A$ and that $g\circ f= 0$.
Thus, $\Theta$ is surjective.
By the Open Mapping Theorem, we conlude thet $\Theta$ is an isomorphism.

Finally we prove \textit{(i)} implies \textit{(v)}.
Let  $s:C\rightarrow B$ be a continuous linear section of $g:B\rightarrow C$.
Define $\Omega: A\oplus C \rightarrow B$ by $\Omega(a,c)=f(a)+s(c)$.
It is clear that $\Omega$ is linear, continuous  and it makes the diagram in \eqref{dia} commutative.
To see that it is surjective take an element $b$ in $B$ and write
	\begin{equation*}	
		b=b-(s\circ g)(b)+ (s\circ g)(b)\mbox{,}
	\end{equation*}
where clearly $b_2:=(s\circ g)(b)$ belongs to $\operatorname{Im}(s)$. 
The element $b_1:=b-(s\circ g)(b)$ satisfies:
	\begin{equation*}
		g(b_1)=g(b-(s\circ g)(b))=g(b)-g(b)=0.
	\end{equation*}
Thus, $b_1$ belongs to $\operatorname{Ker}(g)=\operatorname{Im}(f)$.
Therefore, $\Omega$ is surjective.

To see that $\Omega$ is injective suppose that $\Omega(a,c)=f(a)+s(c)=0$.
We must show that $a=0$ and $c=0$.
But $f(a)+s(c)=0$ implies $s(c)=f(-a)$.
Applying $g$, we obtain $c=\left( g\circ s\right) (c)=\left( g\circ f\right)(-a)=0$ by exactness of \eqref{george3}.
Then, $c=0$ and so, $f(a)=s(-c)=0$, which in turn  implies, by injectivity of $f$,  that $a$ is also equal to zero.
   
Thus, $\Omega$ is bijective and by the Open Mapping Theorem we conclude that it is an isomorphism.
\end{proof}
\begin{rem}
Proposition \ref{ramiro} shows that a short exact sequence of Fr\'echet spaces splits if any, and therefore all, of the properties listed in the proposition holds.
  
In addition, in the proof of \textit{(v)} implies \textit{(ii)} and of \textit{(v)} implies \textit{(iv)} the fact that the spaces are complete has not been used.
Therefore, this implications also hold for TVS in general:

If the short exact sequence of TVS
	\begin{equation} 
		0\rightarrow A \overset{f}{\rightarrow} B \overset{g}{\rightarrow} C \rightarrow 0 
	\end{equation}
splits, then 
	\begin{itemize}
		\item[(i)]there exists a continuous linear section $s:C\rightarrow B$ of $g:B\rightarrow C$ such that $\operatorname{Im}(s)$ is closed.
		\item[(ii)]there exists a continuous linear retraction $r:B\rightarrow A$ of $f:A\rightarrow B$ such that $\operatorname{Im}(r)$ is closed.
	\end{itemize}
\end{rem}
\begin{lem} \label{Mullova}The following statements are equivalent:
	\begin{itemize}
		\item[(i)] The short exact sequence of Fr\'echet spaces splits:
			\begin{equation} \label{george}
 				0\rightarrow A \overset{f}{\rightarrow} B \overset{g}{\rightarrow} C \rightarrow 0. 
			\end{equation}
		\item[(ii)] The dual short exact sequence of TVS $($provided with the weak star topology$)$ splits:
			\begin{equation} \label{harrison}
 				0\rightarrow C' \overset{g'}{\rightarrow} B' \overset{f'}{\rightarrow} A' \rightarrow 0. 
			\end{equation}
	\end{itemize}
\end{lem}
\begin{proof}

To prove that \textit{(i)} implies \textit{(ii)} we will show that there is an isomorphism  $\Omega: B' \rightarrow C'\oplus A' $ such that the following diagram commutes
	\begin{equation}\label{ivan}
		\xymatrix{
		C' \ar@{=}[d]_{} \ar[r]^{g'}& B'\ar[d]^{\Omega}\ar[r]^{f'}& A'\ar@{=}[d]_{}\\
		C'\ar@{^{(}->}[r]^{i{\hspace*{0.5 cm}}}& C'\oplus A' \ar@{->>}[r]^{\pi{\hspace*{-0.2 cm}}}& A'}
	\end{equation}
where $i:A\hookrightarrow A\oplus C$ is the canonical inclusion and $\pi: A\oplus C \twoheadrightarrow C $ is the canonical projection.
Indeed, let $T:C\rightarrow B$ be the continuous linear section of $g:B\rightarrow C$ in  \eqref{george} which satisfies $g\circ T= Id_C$. 
We define
	\begin{align}\label{PO}
	   \Omega: B'&\rightarrow C'\oplus A'
	   \\
	   \Psi &\mapsto \left(T'\left( \Psi\right) ,f'\left( \Psi \right) \right) \nonumber.
	\end{align}
It is immediate to see that $\Omega$ is linear. 
As $f$ is continuous and $A'$  and $B'$ have the weak star topology, then $f'$ is continuous. 
As $T$ is continuous and $B'$  and $C'$ have the weak star topology, then $T'$ is continuous (see Proposition \ref{siqueiros}).
Thus, $\Omega$ is continuous. 
Moreover, it makes \eqref{ivan} a commutative diagram. 
Indeed, if $\varphi$ is an element of $C'$ we write
	\begin{equation*}
		\begin{split}
			\left( \Omega\circ g' \right) (\varphi)&=\left( \left( T'\circ g' \right) (\varphi), \left( f'\circ g' \right)(\varphi) \right)=\left(\left(  g\circ T \right)'  (\varphi),0 \right)\\ &=  \left( \left( \operatorname{Id}_C\right) '(\varphi),0 \right) =  \left( \operatorname{Id}_{C'} (\varphi),0 \right)=\left(\varphi,0 \right)=i(\varphi)\mbox{,}
		\end{split} 
	\end{equation*}
where we have used the fact that $ f'\circ g'=0$ by exactness of \eqref{harrison} and the fact that $ g\circ T=\operatorname{Id}_C$.
The commutativity of the right hand side square in \eqref{ivan} follows immediately from the definition of $\Omega$.
Therefore, $\Omega$ defined as in \eqref{PO} makes \eqref{ivan} a commutative diagram.
 
To check injectivity, suppose $\Omega\left( \Psi\right)=0$. 
Then, 
$T'\left( \Psi\right)=\Psi\circ T=0$ and $f'\left( \Psi\right)=\Psi\circ f=0$. 
As \eqref{george} splits, for every $b$ in $B$ there exist unique elements $a\in A$ and $c\in C$ such that $b=f(a)+T(c)$  
(see proof of \textit{(i)} implies \textit{(v)} in Proposition  \ref{ramiro}).
Then for every $b$ in $B$ we have that 
	\begin{equation*}
		\Psi(b)=\Psi\left( f(a)+T(c)\right) =f'(\Psi)(a)+T'(\Psi)(c)=0\mbox{,}
	\end{equation*} 
so that $\Psi\equiv 0 $.

To check surjectivity, let $(\phi,\varphi)\in C'\oplus A'$. 
Define $\Psi$ in $B'$ as $\Psi(b)=\phi(c)+\varphi(a)$ where $a$ in  $A$ and $c$ in $C$ are the unique elements such that $b$ is written as $b=f(a)+T(c)$. 
Then,
	\begin{equation*}
		\Omega(\Psi)(c,a)=\left( \Psi\circ T(c),\Psi\circ f(a)\right)=\left( \phi(c),\varphi(a)\right) \ \forall (c,a)\in C\oplus A\mbox{,}
	\end{equation*}
so that $\Omega(\Psi)=\left( \phi,\varphi\right)$.
	
To prove that \textit{(ii)} implies \textit{(i)} we will show that $B=f(A)\oplus \tilde{C}$ where $ \tilde{C}\simeq C$. 
Denote by $R:A'\rightarrow B'$ the continuous linear section of $f': B'\rightarrow A'$ in  \eqref{harrison} such that  $f'\circ R= Id_{A'}$. 
Define 
	\begin{equation*}
		\tilde{C}:=^\perp R(A')\mbox{,}
	\end{equation*}
recalling that
	\begin{equation*}
		^\perp R(A')=\left\lbrace b\in B: \mu(b)=0\  \forall \mu\in R(A') \right\rbrace=\bigcap_{\mu\in R(A')}\operatorname{Ker}\left(  \mu\right).
	\end{equation*}
From the above definition it is immediate to see that $\tilde{C}$ is a closed subspace of $B$.

Now we turn to show that $f(A)\cap \tilde{C}=\left\lbrace  0\right\rbrace  $ and to do so let us suppose there is a nonzero $b$ in $f(A)\cap \tilde{C}$. 
As $b$ belongs to $f(A)$ and $f$ is injective there exists a unique element $a$ in $A$ such that $b=f(a)$. 
On the other hand, as $b$ belongs to $\tilde{C}$  we have that $R(\lambda)(b)=0$  for every $\lambda$ in $A'$. 
Then
	\begin{equation} \label{bach}
		0=R(\lambda)(b)=R(\lambda)\left( f(a)\right)=f'\left( R(\lambda)\right)(a)=\lambda(a) \ \forall \lambda \in A'\mbox{,}
	\end{equation}
where we have used the fact that $f'\circ R=Id_{A'}$ in the last equality. 
As $b\neq 0$ we have that $a\neq 0$,  then by the Hahn-Banach Theorem (see Theorem \ref{hahnbanach}) there exists $\lambda$ in $A'$ such that $\lambda (a)\neq 0$, in contradiction with \eqref{bach}. 
This proves what we wanted.

In order to show that $B$ is contained in $ f(A)\oplus \tilde{C}$ we first define:
	\begin{equation*}
		\tilde{A}:=^\perp g'(C')\mbox{,}
	\end{equation*}
recalling that
	\begin{equation*}
		^\perp g'(C') =\left\lbrace  b\in B: \mu(b)=0\ \forall \mu \in g'(C') \right\rbrace.
	\end{equation*}
We assert that $\tilde{A}=f(A)$. 
We begin by proving the inclusion to the right: take $b$ in $\tilde{A}$, then $\mu(b)=0$ for every $\mu$ in $g'(C')$. 
Suppose that $g(b)\neq 0$; then, there is an element $\phi$ of $ C'$, with $\phi\neq 0$, such that $\phi(g(b))=g'(\phi)(b)\neq 0$ and  $g'(\phi)$ belongs to $g'(C')$. 
Contradiction. 
So it must be $g(b)=0$ which implies that  $b$ belongs to $\operatorname{Ker} \left( g\right) =\operatorname{Im}\left( f\right) $. 
To prove the inclusion to the left, let $f(a)$ be in $f(A)$ and take $\mu$ in $g'(C')$. 
As $\mu=g'(\phi)$ for some $\phi$ in $C'$, we can write: 
	\begin{equation*}
		\mu(f(a))=g'(\phi)(f(a))=\left( \phi\circ g \right) (f(a))=\phi\left( \left(g \circ f\right)(a)\right) =0\mbox{,}
	\end{equation*}
where we have used the fact that $g \circ f=0$ in the  last equality. 
From this we conclude that $f(a)$ belongs to $\tilde{A}$. 

Suppose there exists some element $b$ in $B\setminus f(A)\oplus \tilde{C}=B\setminus \tilde{A}\oplus \tilde{C}$. 
Define the linear functional $\omega:\braket{b}\oplus \tilde{A}\oplus \tilde{C}\rightarrow \mathbb{C}$ by the following requirements:
	\begin{equation}
		\begin{split}
			\omega|_{\tilde{A}\oplus \tilde{C}}&= 0\mbox{,}\\
			\omega(b)&= 1.
		\end{split}
	\end{equation}
Note that $\omega$ is continuous by the fact that its kernel is closed.
Applying the Hahn-Banach Theorem, it follows that there exists some $\Omega$ in $B'$ such that
	\begin{equation*}
		\Omega|_{\braket{b}\oplus \tilde{A}\oplus \tilde{C}}=\omega.
	\end{equation*} 
But then $\Omega\neq 0$ and $\Omega|_ {\tilde{A}\oplus \tilde{C}}=0$,  so that 
	\begin{equation*}
		0\neq \Omega\in\left( \tilde{A}\oplus \tilde{C}\right)^{\perp}= \tilde{A}^\perp\cap \tilde{C}^\perp=\left[^\perp g'(C')\right] ^{\perp}\cap \left[^\perp R(A')\right] ^{\perp}=\overline{g'(C')}\cap \overline{R(A')}\mbox{,}
	\end{equation*}
where we have used \textit{(ii)} of Proposition \ref{bipo} in the last equality.
In addition, both $R(A')$ and $g'(C')$ are closed in the weak star topology.
This is due to the fact that $g'(C')=\operatorname{Ker}\left(  f'\right) $, with $f'$ continuous; and by the fact that $R(A')$ is isomorphic to the closed space $A'$.
Therefore $\overline{g'(C')}\cap \overline{R(A')}=g'(C')\cap R(A')$, and we conclude:
	\begin{equation*}
		0\neq \Omega\in g'(C')\cap R(A').
	\end{equation*}
But $g'(C')\cap R(A')=\left\lbrace 0\right\rbrace $, so we have reached a contradiction which came from assuming the existence of some element $b$ in $B\setminus \tilde{A}\oplus \tilde{C}$.
Then we conclude $B=\tilde{A}\oplus \tilde{C}$.    

Finally, let us show that $\tilde{C}\simeq C$. 
For this purpose, consider the restriction $g|_{\tilde{C}}:\tilde{C}\rightarrow C$. 
On the one side, $\operatorname{Ker} \left( g|_{\tilde{C}}\right) = \operatorname{Ker}\left(  g\right) \cap \tilde{C}=f(A)\cap \tilde{C}=\left\lbrace 0\right\rbrace $  (the last equality having been proved above), so that $g|_{\tilde{C}}$ is injective. 
Now, let $c$ be in  $C$. 
As $g$ is surjective there exists an element $b$ of $B$ such that $g(b)=c$. 
It has been shown that $B=f(A)\oplus \tilde{C}$ so that there exist unique elements $a$ in $A$ and $\tilde{c}$ in $\tilde{C}$ such that $b=f(a)+\tilde{c}$. 
Then,
	\begin{equation*}
		g|_{\tilde{C}}(\tilde{c})=g(\tilde{c})=g(f(a)+\tilde{c})=g(b)=c\mbox{,}
	\end{equation*}
where the fact that $g\circ f=0$ has been used in the second equality. 
It follows from this,  that $g|_{\tilde{C}}$ is surjective. 
To complete the proof we note that by the Open Mapping Theorem, $g|_{\tilde{C}}:\tilde{C}\rightarrow C$ is an isomorphism.
\end{proof}
\begin{rem}\label{seqspa23}
Note that in proving that \textit{(ii)} implies \textit{(i)} in Lemma \ref{Mullova} we have  used \textit{(ii)} of Proposition \ref{bipo}, which in turn holds only when duals are provided with the weak star topology (see Proposition \ref{snoopy}). 
If we wanted to prove the same result in the case in which the duals have the strong topology we would not be able to use \textit{(ii)} of Proposition \ref{bipo}.
In fact, if the duals are considered with the strong topology, the fact that \textit{(ii)} implies \textit{(i)} in Lemma \ref{Mullova} is false. 
It does not even hold for Banach spaces, as can be seen in the following example.
\end{rem}
\begin{example}\label{seqspa}
Let \label{linfc0}$\ell_{\infty}$ be the space of bounded complex valued sequences and let $c_0$ be the space of complex valued sequences whose limit is zero.
Then, the short exact sequence
	\begin{equation}\label{gaga}
		0\rightarrow c_0 \overset{i}{\rightarrow} \ell_\infty \overset{\pi}{\rightarrow} \ell_\infty / c_0 \rightarrow 0
	\end{equation}
does not split.  
However, the dual short exact sequence
	\begin{equation}\label{dual}
		0\rightarrow \left(\ell_{\infty}/c_0 \right)' \overset{\pi'}{\rightarrow}\ell_{\infty}'\overset{i'}{\rightarrow}c_0' \rightarrow 0
	\end{equation}
splits if the dual spaces are given the strong topology.

We omit the very technical proof of the fact that \eqref{gaga} does not split.
The reader is referred to \cite{conway}, Ch.  III, \S 13. 
By  Lemma \ref{Mullova} we know that \eqref{dual} does not split if the dual spaces are given the weak star topology.
Contrary to this, we are going to show that  \eqref{dual} splits if the duals have the strong topology instead. 

We begin by noting that $c_0'$ is isometrically isomorphic to $\ell_1$ defined by
	\begin{equation}\label{elle1}
		\ell_1:=\left\lbrace\left(a_i \right)_{i\in\mathbb{N}} : a_i\in \mathbb{C}\  \forall i\in\mathbb{N}\mbox{,}\mbox{ and }\sum_{i=1}^\infty \left|a_i \right|<\infty \right\rbrace.
	\end{equation}
Indeed, we can define $\psi:\ell_1 \rightarrow c_0' $ which on each element $a= \left(a_i \right)_{i\in\mathbb{N}}$ of $\ell_1$ is defined by
	\begin{align*}
		\psi(a):c_0&\rightarrow \mathbb{C}\\
		b &\mapsto  \sum_{i=1}^{\infty} a_i b_i  \nonumber\mbox{,}
	\end{align*}
where $b$ denotes a sequence $\left(b_i \right)_{i\in\mathbb{N}}$ in $c_0 $.
For each $a$ it is clear that $\psi(a)$ is linear.
It is continuous by the following bound:
	\begin{equation}\label{boundi1}
		\left| \psi(a)(b) \right| =  \left| \sum_{i=1}^{\infty}a_i b_i \right| \leq\left\| b\right\|_\infty \sum_{i=1}^\infty \left| a_i \right|\leq \left\| b\right\|_\infty \left\| a\right\|_1.  
	\end{equation}

It is clear that $\psi$ is linear.
Restricting equation \eqref{boundi1} to elements $b$ of unitary norm and taking supremum over all such elements we have
	\begin{equation*}
		\left\| \psi(a) \right\| \leq \left\| a\right\|_1\mbox{,}
	\end{equation*}
so that $\psi$ is continuous.
Moreover, if $\psi(a)=0$, evaluating on the sequence $e^{\left( k\right)}$ in $c_0$ defined by $\left( e^{\left( k\right)}\right) _i=\delta_{ki}$ for each $k$ in $\mathbb{N}$, we arrive at:
	\begin{equation*}
		0=\psi(a)\left( e^{\left( k\right)}\right) = \sum_{i=1}^{\infty} a_i e^{\left( k\right)}_i= \sum_{i=1}^{\infty} a_i \delta_{ki}=a_k. 
	\end{equation*}
Thus, $a=0$ and we conclude $\psi$ is injective.

Given $\mu$ in $c_0'$ let $a$ be the sequence defined by $a_i=\mu\left( e^{\left( i\right)}\right)  $ where $e^{\left( i\right)}$ was defined before.
We assert that
	\begin{itemize}
		\item[1.]$a$ belongs to $\ell_1$  and that 
		\item[2.]$\psi(a)=\mu$, thus $\psi$ is surjective.
	\end{itemize}
To prove the first assertion, we define the numbers 
	\begin{equation}
		\alpha_i= 
		\left\{ 						\begin{array}{lcc}
				\frac{\left|\mu\left( e^{\left( i\right)}\right) \right| }{\mu\left( e^{\left( i\right)}\right)}\mbox{,} & \mbox{if} & \mu\left( e^{\left( i\right)}\right)\neq 0\mbox{,} \\
				0\mbox{,} &  \mbox{if} & \mu\left( e^{\left( i\right)}\right)=0.
			\end{array}
		\right.
	\end{equation}
And now consider for each positive integer $k$ the sequence $b^k=\sum_{i=1}^{k}\alpha_i e^i$.
Then, we have that $b_k$ belongs to $c_0$ and that $\left\| b^k \right\|_{\infty} $ is either 1 or 0  for every $k$.
Then we can write for each positive integer $k$:
	\begin{equation*}
		\sum_{i=1}^{k} \left| \mu\left(  e^{\left( i\right)}  \right) \right| = \sum_{i=1}^{k} \alpha_i  \mu\left(e^{\left( i\right)}\right) = \left| \mu\left( b^k\right) \right|\leq\left\| \mu \right\|  \left\|   b^k\right\|_\infty\leq  \left\| \mu \right\|.
	\end{equation*} 
Then, as $\mu$ is continuous its norm is bounded and taking the limit of $k$ tending to infinity, we obtain
	\begin{equation*}
		\left\| a \right\|_1=\sum_{i=1}^{\infty} \left| \mu\left(  e^{\left( i\right)}\right) \right|\leq  \left\| \mu \right\|
	\end{equation*}
from where it follows that $a$ belongs to $\ell_1$.

The second assertion follows from
	\begin{equation*}
		\begin{split}
			\psi(a)(b)&=\sum_{i=1}^{\infty} a_i b_i= \sum_{i=1}^{\infty} \mu\left( e^{\left( i\right)}\right)  b_i = \sum_{i=1}^{\infty} \mu\left( b_i e^{\left( i\right)}\right) =\lim\limits_{N\rightarrow \infty}  \sum_{i=1}^{N} \mu\left( b_i e^{\left( i\right)}\right) \\ &=\mu\left( \lim\limits_{N\rightarrow \infty}  \sum_{i=1}^{N} b_i e^{\left( i\right)}\right)\mbox{,}
		\end{split}
	\end{equation*}
and the fact that in $c_0$ the following limit holds
	\begin{equation*}
		\lim\limits_{N\rightarrow \infty}  \sum_{i=1}^{N} b_i e^{\left( i\right)}=b.
	\end{equation*}
This concludes the proof of the fact that  $c_0'$ is isometrically isomorphic to $\ell_1$.

Turning to the space $\ell_\infty'$, we assert that it is isometrically isomorphic to the space $ba$ defined by:
	\begin{equation}\label{badef}
		ba=\left\lbrace \nu: \mathcal{P}(\mathbb{N})\rightarrow \mathbb{R} : \nu \mbox{ is finitely additive } \mbox{ and } \forall E\subseteq\mathbb{N}\mbox{,}\ \sup_{\pi}\sum_{A\in\pi}\left|\nu(A) \right|<\infty \right\rbrace\mbox{,}
	\end{equation}
where $\pi$ in the definition above runs through all partitions of $E$.
Note that the supremum increases with respect to set inclusion and taking $E=\mathbb{N}$ we obtain a uniform bound for every $E$ contained in $\mathbb{N}$.
Also, $ba$ is a vector space with the usual definitions of addition and scalar multiplication.
Moreover, $ba$ is a Banach space defining the following norm
	\begin{equation}\label{ooooooooo}
		\left\|  \nu\right\|= \sup_{\pi}\sum_{A\in\pi}\left|\nu(A) \right|<\infty\mbox{,}
	\end{equation}
where $\pi$ in the definition above runs through all partitions of $\mathbb{N}$.
The norm \eqref{ooooooooo} is known as the \emph{total variation} of $\nu$\index{total variation of a measure}.

Considering the measurable space $\left(\mathbb{N},\mathcal{P}(\mathbb{N})\right) $, each element $\nu$ in $ba$ defines a continuous functional of $\ell_\infty$ in the following way:
	\begin{align}\label{jjjjjjj}
		\ell_\infty&\rightarrow \mathbb{R}\\
		a &\mapsto  \int_{\mathbb{N}}a\ \mbox{d}\nu \nonumber\mbox{,}
	\end{align}
and this identification defines an isomorphism between $\ell_\infty'$ and  $ba$ (for a description of the consrtuction of the integral in \eqref{jjjjjjj} when $\nu$ is finitely additive, see \cite{diestel}, p. 77). 

Note that if $\nu$ is $\sigma$-additive, the integral in the definition above coincides with the sum:
	\begin{equation*}
		\int_{\mathbb{N}}a\ \mbox{d}\nu=\sum_{i\in\mathbb{N}}a_i\nu\left(\left\lbrace i\right\rbrace  \right)\mbox{,}
	\end{equation*}
where $a=\left(a_i \right)_{i\in\mathbb{N}}$.

These isomorphisms, namely $\ell_\infty'\sim ba$ and $c_0'\sim \ell_1$ give the following short exact sequence which is isometrically isomorphic to \eqref{dual}:
	\begin{equation}\label{dual2}
		0\rightarrow \left(\ell_{\infty}/c_0 \right)' \rightarrow  ba  \overset{\phi}{\rightarrow}  \ell_1 \rightarrow  0
	\end{equation} 
where the morphism $\phi$ is given by
	\begin{align*}
		\phi:  ba &\rightarrow \ell_1\\
		\nu&\mapsto \left( \nu\left(\left\lbrace i\right\rbrace  \right) \right)_{i\in \mathbb{N}}
	\end{align*} 
and where we also have the isomorphism $\operatorname{ker}\left( \phi\right) \sim  \left(\ell_{\infty}/c_0 \right)'$, which is a closed subspace of $ba$.

We seek to show that \eqref{dual2} splits (which is equivalent to the splitting of \eqref{dual}).
To see this, it suffices to show that there exist some closed subspace of $ba$ which complements $\operatorname{ker}\left(  \phi\right) $.

Consider the subspace $ca$ of $ba$ of signed, countably additive and bounded measures
	\begin{equation}\label{cadef}
		ca= \left\lbrace \nu\in ba: \nu \mbox{ is } \sigma\mbox{-additive} \right\rbrace .
	\end{equation}
The space $ca$ is closed in $ba$, because it is complete with respect to the norm \eqref{ooooooooo} (see \cite{langlang93}, p.199).
In fact, $\phi$ induces an isomorphism between $ca$ and $\ell_1$. 
We briefly prove this last statement.
The morphism $\phi$ restricted to $ca$ is inyective because every countably additive measure $\nu$ is completely determined by the sequence of values $\nu(\{n\})$ with $n$ in $\mathbb{N}$.
On the other hand, given any real sequence $a = \left( a_{n}\right)_{n\in\mathbb{N}}$ in $\ell_{1}$, the arrow $\nu_{a} : \mathcal{P}(\mathbb{N}) \rightarrow \mathbb{R}$ defined by $\nu_{a}(X) = \sum_{n \in X} a_{n}$ is well defined, is $\sigma$-additive by definition and it is bounded.
Therefore, $\nu_{a}$ belongs to $ca$. 
Moreover, $\phi(\nu_{a}) = a$ by construction.
Thus, the restriction of $\phi$ to $ca$ is surjective.
To complete the proof we finally note that
\begin{equation*}
\left\|  \nu\right\|=\left\|  \phi(\nu)\right\|_1
\end{equation*}
for every $\nu$ in $ca$.
The above equality follows from the next argument.
Take $\nu$ in $ca$.
On the one hand, for every $N$ in $\mathbb{N}$ we have
	\begin{equation*}
		\sum_{i=1}^N \left| \nu\left(\left\lbrace i \right\rbrace  \right) \right|\leq \sum_{i=1}^N \left| \nu\left(\left\lbrace i \right\rbrace  \right) \right|+\left| \nu\left(\mathbb{N}\setminus \left\lbrace 1,\dots , N\right\rbrace  \right) \right|\mbox{,}
	\end{equation*}
and the right hand side of the above inequality is a sum over the partition
	\begin{equation*}
		\pi=\left\lbrace \left\lbrace 1\right\rbrace,\dots,\left\lbrace N\right\rbrace ,\mathbb{N}\setminus \left\lbrace 1,\dots , N\right\rbrace  \right\rbrace
	\end{equation*}
of $\mathbb{N}$.
Thus, by definition of the norm in $ba$ \eqref{ooooooooo}, we have that for every $N$ in $\mathbb{N}$ 
	\begin{equation*}
		\sum_{i=1}^N \left| \nu\left(\left\lbrace i \right\rbrace  \right) \right|\leq \left\|\nu \right\| \mbox{,}
	\end{equation*}
and therefore, taking the limit $N\rightarrow \infty$, we have that $\left\| \phi\left(\nu \right) \right\|_1\leq \left\|\nu \right\|  $.

On the other hand, for any partition $\pi$ of $\mathbb{N}$, we have that
	\begin{equation*}
		\sum_{A\in\pi}\left|\nu\left( A \right)  \right|\leq\sum_{A\in\pi}\sum_{i\in A}\left|\nu\left( \left\lbrace i\right\rbrace \right)  \right|=\sum_{i\in \mathbb{N}}\left|\nu\left( \left\lbrace i\right\rbrace \right)  \right|=\left\|\phi\left(\nu \right)  \right\|_1.
	\end{equation*}
Taking supremum over all such partitions $\pi$ on the left hand side of the previous inequality we arrive at $\left\|\nu  \right\|\leq\left\|\phi\left(\nu \right)  \right\|_1$.

As $ca\cap \operatorname{ker} \left( \phi\right) =\left\lbrace 0\right\rbrace$, \eqref{dual2} splits in the category of Banach spaces, and therefore so does \eqref{dual}.
\end{example}
\begin{propo}\label{baca}
Note that when considering  $ba$ with the weak star topology, $ca$ is not a closed subspace.
\end{propo}
\begin{proof}
We will show that
	\begin{equation*}
		\overline{ca}=ba\mbox{,}
	\end{equation*}
where the closure is taken with respect to the weak star topology on $ba$.

By Theorem \ref{bipo}, we know
	\begin{equation*}
		\overline{ca}=\left(^\perp  ca \right)^\perp.
	\end{equation*}
It suffices to show that  $^\perp  ca=\left\lbrace 0\right\rbrace $; in that case, taking perpendicular again, we obtain the whole space $ba$.

For each $i$ in $\mathbb{N}$, $\delta_i$ belongs to $ca$, where $\delta_i$ is the Dirac measure \index{Dirac measure} concentrated at $i$:
	\begin{equation*}
		\delta_i(A)=
		\left\{ 
			\begin{array}{lcc}
				1\mbox{,} & \mbox{if} & i\in A\mbox{,} \\
				0\mbox{,} &  \mbox{if} & i\notin A.
			\end{array}
		\right.
	\end{equation*} 
For each $a$ in $\ell_\infty$, 
	\begin{equation} \label{lololo}
		\int_{\mathbb{N}}a\ \mbox{d}\delta_i=\sum_{j\in\mathbb{N}}a_j\delta_i\left(\left\lbrace j\right\rbrace  \right) =a_i.
	\end{equation} 
On the other hand, if $a$ belongs to $^\perp ca $ then $a\perp \delta_i $ for every $i$ in $\mathbb{N}$.
Combining this fact with \eqref{lololo} we get $a_i=0$ for every $i$ and therefore, $a=0$.
\end{proof}
\begin{propo}\label{gershwin}
Let
	\begin{equation} \label{george321}
 		0\rightarrow A \overset{f}{\rightarrow} B \overset{g}{\rightarrow} C \rightarrow 0 
	\end{equation}
be a short exact sequence of Fr\'echet spaces.
Given a continuous linear section $T:C\rightarrow B$ of $g:B\rightarrow C$, so that \eqref{george321} splits by Lemma \ref{ramiro}, one can define a continuous linear section $R:A'\rightarrow B'$ of $f':B'\rightarrow A'$ in the dual short exact sequence of TVS provided with the weak star topology:
	\begin{equation} \label{harrison321}
 		0\rightarrow C' \overset{g'}{\rightarrow} B' \overset{f'}{\rightarrow} A' \rightarrow 0\mbox{,}
	\end{equation}
$($and which splits too by Lemma \ref{Mullova}$)$.
Reciprocally, given  a continuous linear section $R:A'\rightarrow B'$ of $f':B'\rightarrow A'$ in \eqref {harrison321} one can define a continuous linear retraction $I:B\rightarrow A$ of $f:A\rightarrow B$ in \eqref{george321} or a continuous linear section $T:C\rightarrow B$ of $g:B\rightarrow C$ in \eqref{george321}.
\end{propo}
\begin{proof}
First, suppose we are given a continuous linear section $T:C\rightarrow B$ of $g:B\rightarrow C$ in \eqref{george321}. 
By Proposition \ref{ramiro} there is an isomorphism $\Omega:A\oplus C\rightarrow B$
such that the following diagram is commutative
	\begin{equation*}
		\xymatrix{A
		\ar@{=}[d]_{} \ar[r]^{f}& B\ar[r]^{g}& C\ar@{=}[d]_{}\\	A\ar@{^{(}->}[r]^{i {\hspace*{0.5 cm}}}& A\oplus C\ar[u]^{\Omega} \ar@{->>}[r]^{\pi {\hspace*{-0.2 cm}}}& C}
	\end{equation*}
where $i:A\hookrightarrow A\oplus C$ is the canonical inclusion and $\pi: A\oplus C \twoheadrightarrow C $ is the canonical projection. 
Next define $I:B\rightarrow A$ by 
	\begin{equation*}
		I=\pi_A \circ \Omega^{-1}\mbox{,}
	\end{equation*}
where $\pi_A:A\oplus C\rightarrow A$ is the projection. 
Then, $I$ is a continuous linear retraction of $f:A\rightarrow B$ in \eqref{george321}.

Now define $R:A'\rightarrow B'$ as $R=I'$ (\textit{i.e.}: For every $\varphi$ in $A'$, $R(\varphi)$ in $B'$ is defined on each $b$ in $B$ by: $\braket{R(\varphi),b}=\braket{\varphi,I(b)}$).
$R$ is continuous by Proposition \ref{siqueiros}.  

It remains to show that $f'\circ R=Id_{A'}$.
Let $\varphi$ be in $A'$. 
Then for every $a$ in $A$ we have:
	\begin{equation*}
		\begin{split}
			\braket{\left(f'\circ R\right)\left(\varphi\right),a}&= \braket{f'\left(R\left(\varphi\right)\right),a}=\braket{R\left(\varphi\right)\circ f,a}=\braket{R\left(\varphi\right),f(a)}\\
			&=\braket{\varphi,I(f(a))}=\braket{\varphi,a}\mbox{,}
		\end{split}
	\end{equation*}  
from which we conclude $f'\circ R=Id_{A'}$.

Now suppose we are given a continuous linear section $R:A'\rightarrow B'$ of $f':B'\rightarrow A'$ in \eqref {harrison321}.
From the proof of \textit{(ii)} implies \textit{(i)} of Lemma \ref{Mullova} we have that for every $b$ in $B$ there exist unique elements $c$ in $C$ and $a$ in $A$ such that $b=f(a)+\left( g|_{\tilde{C}}\right)^{-1}(c)$. 
This, together with the fact that  $\left( g|_{\tilde{C}}\right)^{-1}$ is an isomorphism by the proof of Lemma \ref{Mullova} and $f$ is continuous by assumption, imply that
	\begin{equation*}
		\Psi=f+\left( g|_{\tilde{C}}\right)^{-1}:A\oplus C\rightarrow B
	\end{equation*}
is a linear continuous  bijection; and by the Open Mapping Theorem, an isomorphism. 

Next, if $\Pi_A: A\oplus C\rightarrow A$ is the canonical projection, define $I:B\rightarrow A$ by $I=\Pi_A\circ \Psi^{-1}$.
which is obviously linear, continuous and satisfies $I\circ f=Id_A$  .
Thus, $I:B\rightarrow A$ is a continuous linear retraction of $f:A\rightarrow B$ in \eqref{george321}.

On the other hand, if we define $T:C\rightarrow B$ by $T=\left( g|_{\tilde{C}} \right)^{-1}$, it 
is obviously linear, continuous and $T$ satisfies $g\circ T=Id_C$. Thus, $T:C\rightarrow B$ is a continuous linear section  of $g:B\rightarrow C$ in \eqref{george321}.
\end{proof}

\chapter{Distributions on open subsets of the Euclidean space}\label{asdf}
In this chapter we first  discuss the spaces of continuously differentiable functions and spaces of test functions.
A thorough description of their topology is made in order to define the notion of a distribution, and to give an  equivalent characterization.
The subject of distributions is further developed in Chapter \ref{macri} in a more general context.
\section{Spaces of continuously differentiable functions and spaces of test functions}

\begin{defi} Let $\Omega$ be a nonempty open subset of some topological space. 
Let $\varphi:\Omega\rightarrow\mathbb{C}$ be a function. 
The \emph{support}\index{support of a function} of $\varphi$ is the set defined by
	\begin{equation*}
		\operatorname{Supp}\left( \varphi\right) =\overline{\left\lbrace x\in \Omega: \varphi(x)\neq 0\right\rbrace }^\Omega. 
	\end{equation*}
Equivalently, a point $x$ is not in $\operatorname{Supp}\left( \varphi\right) $ if there exists an open neighborhood $V$ of $x$, contained in $\Omega$ such that $\varphi(V)=\left\lbrace 0 \right\rbrace$.   
\end{defi}
\begin{defi}
\label{funcinf1}
For each nonempty open subset $\Omega$ of $\mathbb{R}^d$ we denote by $\mathscr{E}\left( \Omega\right) $ $($resp. $\mathscr{E}^m\left( \Omega\right) $$)$ the algebra of  infinitely differentiable, or $\mathcal{C}^\infty$,  $($resp. $m$ times continuously differentiable, or $\mathcal{C}^m$$)$ complex valued  functions on $\Omega$. 

Let $X$ be a closed subset of $\mathbb{R}^d$. We denote by $\mathcal{I}^\infty \left( X,\mathbb{R}^d\right) $ $($resp. $\mathcal{I}^m\left( X,\mathbb{R}^d\right) $ for $m\in\mathbb{N}$$)$ the space of $\mathcal{C}^\infty$ $($resp. $\mathcal{C}^m$$)$ functions which vanish on $X$ together with all of their derivatives $($resp. all of their derivatives of order less than or equal to $m$$)$.
Note that $\mathcal{I}^\infty \left( X,\mathbb{R}^d\right) $ $($resp. $\mathcal{I}^m\left( X,\mathbb{R}^d\right) $$)$ is an ideal of $\mathscr{E}\left( \Omega\right) $ $($resp. $\mathscr{E}^{m}\left( \Omega\right) $$)$.

We define the ideal $\mathcal{I}\left( X,\mathbb{R}^d\right) $ of $\mathscr{E}\left( \mathbb{R}^d\right) $ by
	\begin{equation*}
		\mathcal{I}\left( X,\mathbb{R}^d\right) =\left\lbrace\varphi\in\mathscr{E}\left( \mathbb{R}^d\right) : \operatorname{Supp}\left( \varphi\right) \cap X=\emptyset \right\rbrace\subseteq \mathscr{E}\left( \mathbb{R}^d\right) .
	\end{equation*}
Note the following inclusions: 
	\begin{equation*}
		\mathcal{I}\left( X,\mathbb{R}^d\right) \subseteq \mathcal{I}^\infty\left( X,\mathbb{R}^d\right) \subseteq \dots \subseteq \mathcal{I}^{m+1}\left( X,\mathbb{R}^d\right) \subseteq \mathcal{I}^m\left( X,\mathbb{R}^d\right)\dots
	\end{equation*}
If $K$ is a compact set in $\mathbb{R}^d$, then $\mathscr{D}_K$ denotes the space of all $\varphi$ in  $\mathscr{E}\left( \mathbb{R}^d\right) $ whose support lies in $K$.
If $K$ is a compact subset of $\Omega$ then $\mathscr{D}_K$ may be identified with a subspace of $\mathscr{E}\left( \Omega\right) $, and we will denote it by $\mathscr{D}_{K}\left( \Omega\right) $.

A function  $\varphi:\Omega\rightarrow\mathbb{C}$ is called a \emph{test function}\index{test function} if it is infinitely differentiable and has compact support. 
The set of test functions over $\Omega$ will be denoted by $\mathscr{D}\left( \Omega\right) $. 
Observe that $\mathscr{D}\left( \Omega\right) $ is the union of the subspaces $\mathscr{D}_K\left( \Omega\right) $ as $K$ runs through all compact subsets of $\Omega$.
\end{defi}
\label{nnnnn}
Introducing certain families of seminorms we now proceed to topologize  the spaces given in Definition \ref{funcinf1} above, starting with $\mathscr{E}\left( \Omega\right) $ and $\mathscr{D}_K$.
We first introduce the notion of fundamental sequence of compact sets.
\begin{defi}\label{fundseq}
Let $\Omega$ be an open subset of a topological space.
A collection of compact sets
	\begin{equation*}
		\left\lbrace K_l\right\rbrace_{l\in\mathbb{N}} 
	\end{equation*}
is called a \emph{fundamental sequence}\index{fundamental sequence} if it satisfies the following properties:
	\begin{itemize}
		\item[FS 1.] $K_1\neq \emptyset$;
		\item[FS 2.] $K_l\subseteq\left(  K_{l+1}\right) ^\circ$, for every $l$ in $\mathbb{N}$; and
		\item[FS 3.] $\Omega=\bigcup\limits_{l\in\mathbb{N}}K_l$.
	\end{itemize}
\end{defi}
\begin{defi}
\label{kungfu} Let $\Omega$ be a nonempty open subset of $\mathbb{R}^d$ and $K$ a compact subset of $\Omega$.
For each integer $k\geq 0$ $(\mbox{resp. }0 \leq k\leq m)$ and each function $f$ in $\mathscr{E}\left( \Omega\right) $ $(\mbox{resp. }f$ in $\mathscr{E}^{m}(\Omega))$ we define the seminorm
	\begin{equation}
		\label{seminorms1}
		\parallel f \parallel_k^K=\sup\limits_{\substack{x \in K \\ |\nu| \leq k}}|\partial^{\nu}f(x)|.
	\end{equation}  
\end{defi}
\begin{defi}\label{defiseminorms}
\label{lala} Let $\Omega$ be a nonempty open subset of $\mathbb{R}^d$ and $\{K_l\}_{l\in\mathbb{N}}$ a fundamental sequence of compact sets covering $\Omega$.
For each $k\geq 0$ $(\mbox{resp. }0 \leq k\leq m)$, each $l>0$ and each function $f$ in $\mathscr{E}\left( \Omega\right) $ $(\mbox{resp. }f$ in $\mathscr{E}^{m}(\Omega))$ we define the seminorm
	\begin{equation} \label{seminorms2} 
		\parallel f \parallel_k^l=\parallel f \parallel_k^{K_l}.
	\end{equation}
\end{defi}
\begin{rem}
The topology obtained in Theorem \ref{degas} from the collection of seminorms 
	\begin{equation*}
		\left\lbrace\parallel\cdot\parallel^l_k\right\rbrace_{k,l}
	\end{equation*}
in definition \ref{lala} coincides with the one given by the collection
	\begin{equation*}
		\left\lbrace\parallel\cdot\parallel^K_{k}\right\rbrace_{k,K}
	\end{equation*}
of definition \ref{kungfu}.
This is a consequence of Definition \ref{fundseq}.
\end{rem}

Some simple but most useful results about this topology are given in the following lemma.
\begin{lem}[See \cite{dieu}, 17.1.3, 17.1.4 and 17.1.5]\label{mafi}
Consider $\mathscr{E}\left(\Omega\right)$ with the topology obtained from the collection of seminorms \eqref{seminorms1} $($or equivalently from \eqref{seminorms2}$)$ as described in Theorem \ref{degas}.  
	\begin{itemize}
		\item[(i)] For each multi-index $\nu$ the linear mapping $f\mapsto \partial^\nu f$ of $\mathscr{E}\left( \Omega\right) $ into $\mathscr{E}\left( \Omega\right) $ is continuous.
		\item[(ii)] For each function $\psi$ in $\mathscr{E}\left( \Omega\right) $ $(\mbox{resp. }\psi$ in $\mathscr{E}^{m}(\Omega))$ the linear mapping $\varphi\mapsto \psi \varphi$ of $\mathscr{E}\left( \Omega\right) $ $($resp. $\mathscr{E}^{m}(\Omega)$$)$ into itself is continuous.
		Moreover,  for every $\psi$ in $\mathscr{E}\left( \Omega\right) $ $($resp. $\psi$ in $\mathscr{E}^{m}(\Omega)$$)$, each $k$ in $\mathbb{N}_0$ $(\mbox{resp. }0 \leq k\leq m)$ and each compact set $K$ there exists a constant $M>0$  such that $\parallel \psi\varphi \parallel_k^K \leq M \parallel\varphi \parallel_k^K$, for every $\varphi$ in $\mathscr{E}\left( \Omega\right) $ $($resp. $\varphi$ in $\mathscr{E}^{m}(\Omega)$$)$. 
		\item[(iii)] Let $\varphi$ be a mapping of class $\mathcal{C}^\infty$ $($resp. $\mathcal{C}^m$$)$ of an open set $V$ into an open set $U$, both contained in $\mathbb{R}^d$.
		Then the linear mapping $f\mapsto f\circ\varphi$ of $\mathscr{E}\left( U\right) $ into $\mathscr{E}\left( V\right) $  $($resp. $\mathscr{E}^m\left( U\right) $ into $\mathscr{E}^m\left( V\right) $$)$  is continuous.
	\end{itemize}
\end{lem}
\begin{proof}
For \textit{(i)} and \textit{(ii)}, let $\{K_l\}_{l\in\mathbb{N}}$ be a fundamental sequence of compact subsets of $\Omega$ used to define seminorms $\left\|\cdot \right\|^l_k $ on $\mathscr{E}\left( \Omega\right) $  $($resp. $\mathscr{E}^m\left( \Omega\right) $$)$.
The proof of \textit{(i)} is quite straightforward. We only note that
\begin{equation*}
	\parallel \partial^\nu f \parallel^l_k \leq \parallel  f \parallel^l_{k+\left| \nu \right| }.
\end{equation*}

Now we prove \textit{(ii)}. 
For every $|\nu|\leq k$ and every $x$ in $K$,
	\begin{equation*}
		\begin{split}
			\left|\partial^{\nu}\left(\psi\varphi\right) (x)\right|&\leq\sum\limits_{|i|\leq|\nu|}\binom{\nu}{i}\left|\partial^{i}\varphi(x)\cdot \partial^{\nu-i}\psi(x) \right|
			\leq M'\sum\limits_{|i|\leq|\nu|}\binom{\nu}{i} \left| \partial^{i}\varphi(x) \right|\\
			&\leq M'\left[ \sum\limits_{|i|\leq|\nu|}\binom{\nu}{i} \right] \parallel\varphi \parallel_k^K
			\leq M \parallel\varphi \parallel_k^K\mbox{,}
		\end{split}
	\end{equation*}
where $M'$ and $M$ are defined by 
	\begin{equation*}
		M'=\max\limits_{|\alpha|\leq k}\sup\limits_{x\in K}\left|\partial^{\alpha}\psi(x)\right| \ \mbox{ and }\ M=M' \max\limits_{|\nu|\leq k}\sum\limits_{|i|\leq|\nu|}\binom{\nu}{i}\mbox{,}
	\end{equation*}
and are independent of $\varphi$ and $\nu$. 
Consequently, taking supremum over all the elements $x$ in $K$ and multi-indices $|\nu|\leq k$ on the left hand side of the above inequality we arrive at 
	\begin{equation*}
		\parallel \psi\varphi \parallel_k^K \leq M\parallel\varphi \parallel_k^K.
	\end{equation*} 
Now we prove \textit{(iii)}.
Let $\{K_l\}_{l\in\mathbb{N}}$ be a fundamental sequence of compact subsets of $U$ used to define seminorms $^U\!\!\left\|\cdot \right\|^l_k $ on $\mathscr{E}\left( U\right) $  $($resp. $\mathscr{E}^m\left( U\right) $$)$.
Put
	\begin{equation*}
		\varphi=\left(\varphi_1,\dots,\varphi_n \right)\mbox{,}
	\end{equation*}
where $\varphi_i$ is a scalar valued function.
Let $\{K_l'\}_{l\in\mathbb{N}}$ be a fundamental sequence of compact subsets of $V$ used to define seminorms $^V\!\!\left\|\cdot \right\|^l_k $ on $\mathscr{E}\left( V\right) $  $($resp. $\mathscr{E}^m\left( V\right) $$)$.
For each pair of integers $k$, $l$, let $a_{kl}$ be the gratest of the least upper bounds of the functions 
	\begin{equation*}
		\sup\left\lbrace 1,\left| \partial^{\nu}\varphi_i\right| \right\rbrace 
	\end{equation*} 
on $K'_l$,  for $\left| \nu\right| \leq k$ and $1\leq i \leq n$.
Finally, let $q$ be an integer such that $\varphi\left( K'_l\right) $ is contained in $K_q$.
By repeated application of the formula of partial derivatives of a composite function we obtain for each $f$ in $\mathscr{E}\left( U\right) $  $($resp. $\mathscr{E}^m\left( U\right) $$)$
	\begin{equation*}
		^U\!\!\left\| f \circ \varphi \right\|^l_k \leq c_{kl}a_{kl}^{k}\mbox{ }^U\!\!\left\| f \right\|^l_q.
	\end{equation*}
where $c_{kl}$ is a constant independent of $f$ and $\varphi$.
\end{proof}
\begin{propo}[See \cite{rudin}, Ch. 1, Example 1.46]\label{pinkfloyd}
The topology defined on $\mathscr{E}\left( \Omega\right) $ by the seminorms \eqref{seminorms2}, turns $\mathscr{E}\left( \Omega\right) $ into a Fr\'echet space with the Heine-Borel property, such that $\mathscr{D}_K\left( \Omega\right) $ is a closed subspace of $\mathscr{E}\left( \Omega\right) $ whenever the compact set $K$ is contained in $\Omega$, and therefore, $\mathscr{D}_K\left( \Omega\right) $ is a Fr\'echet space by itself with the Heine-Borel property.
\end{propo}
\begin{proof}
By Theorem \ref{degas}  the seminorms given by \eqref {seminorms2} define a locally convex topology on $\mathscr{E}\left( \Omega\right) $.
By Theorem \ref{liszt} this topology is metrizable, with an invariant metric whose  balls centered at zero are convex and balanced.
A choice of such a metric is given in Remark \ref{pinocho}.

Next, for every $x$ in  $\Omega$ we define the evaluation map \index{evaluation map}
	\begin{align*}
		\operatorname{e}_x: \mathscr{E}\left( \Omega\right) & \rightarrow  \mathbb{C}\\ \varphi &\mapsto  \varphi(x)\mbox{,}
	\end{align*}
which is a continuous functional for this topology.
To see this, just note that if $B_{1/n}(0)$  denotes the open ball of $\mathbb{R}$ centered at zero with radius $1/n$, then
	\begin{equation*}
		\operatorname{e}_x^{-1}\left(B_{1/n}(0) \right)=\left\lbrace \varphi: \left|\braket{\varphi,x} \right|<\frac{1}{n}\right\rbrace = \left\lbrace \varphi: \left\| \varphi\right\|^{\left\lbrace x\right\rbrace}_0<\frac{1}{n}  \right\rbrace  =V\left( \left\| \cdot\right\|^{\left\lbrace x\right\rbrace}_0,n\right)\mbox{,}
	\end{equation*}
where $V\left( \left\| \cdot\right\|^{\left\lbrace x\right\rbrace}_0,n\right)$, as defined in Theorem \ref{degas}, is an open set of $\mathscr{E}\left( \Omega\right) $.

As $\mathscr{D}_K\left( \Omega\right)$ can be written as
	\begin{equation*}
		\mathscr{D}_K\left( \Omega\right)=\bigcap_{x\in \Omega\setminus K\\
		} \operatorname{Ker}\left( \operatorname{e}_x\right) \mbox{,}
	\end{equation*}
it follows that $\mathscr{D}_K\left( \Omega\right)$ is closed in  $\mathscr{E}\left( \Omega\right) $ whenever $K$ is contained in $\Omega$.

The seminorms in \eqref{seminorms2} satisfy $\parallel\cdot\parallel_k^l\leq \parallel\cdot\parallel_{k'}^{l'}$ whenever $(l,k)\leq (l',k')$ and therefore, a local base at zero is given by the sets  
	\begin{equation*}
		V(N)=\left\lbrace \varphi   \in \mathscr{E}\left( \Omega\right)  :\ \parallel\varphi\parallel_N^N< \frac{1}{N} \right\rbrace\mbox{,} 
	\end{equation*}
for $N$ in $\mathbb{N}$.
Having stated this, we can proceed to prove that $\mathscr{E}\left( \Omega\right) $ is complete.
If $\left( \varphi_i \right) _{i\in\mathbb{N}} $ is a Cauchy sequence in $\mathscr{E}\left( \Omega\right) $ and if $N$  in  $\mathbb{N}$ is fixed, then $\varphi_i-\varphi_j$ belongs to $V(N)$ if $i$ and $j$ are sufficiently large. 
Thus $\left|\partial^\nu\varphi_i-\partial^\nu\varphi_j \right|< 1\slash N$ on $K_N$ if $|\nu|\leq  N$.
It follows that each $\partial^\nu\varphi_i$ converges uniformly on compact subsets of $\Omega$ to a function $g_\nu$.
In particular $\varphi_i(x)\rightarrow g_0(x)$.
It is now evident that $g_0$ belongs to $\mathscr{E}\left( \Omega\right) $, that $g_\nu=\partial^\nu g_0$  and that $\varphi_i\rightarrow g_0$ for the topology of  $\mathscr{E}\left( \Omega\right) $.

Thus $\mathscr{E}\left( \Omega\right) $ is a Fr\'echet space. 
The same is true for each of its closed subsets $\mathscr{D}_K\left( \Omega\right)$.

Suppose next that a subset  $B$ of  $\mathscr{E}\left( \Omega\right) $ is closed and bounded.
By Theorem \ref{degas}, the boundedness of $B$ is equivalent to the existence of numbers $M_{lk}<\infty$  such that $\parallel\phi\parallel_k^l\leq M_{lk}$ for every $\phi$ in $B$, and for every pair $(l,k)$ in $\mathbb{N}\times\mathbb{N}_0$.
The inequalities 
	\begin{equation*}
		\left| \partial^\nu \phi(x)\right|\leq M_{lk}\mbox{,} 
	\end{equation*}
valid for $x$ in $K_l$ when $|\nu|\leq k$, imply the equicontinuity of $\left\lbrace \partial^\nu \phi: \phi\in B\right\rbrace $ on $K_{l-1}$, if $|\nu|\leq k-1$.
It now follows from Ascoli's Theorem and Cantor's diagonal process that every sequence in $B$ contains a subsequence $\left(  \phi_i\right) _{i\in\mathbb{N}} $ for which  $\left(  \partial^\nu \phi_i\right) _{i\in\mathbb{N}} $ converges uniformly on compact subsets of $\Omega$, for each multi-index $\nu$.
Hence, $\left(  \phi_i\right)_{i\in\mathbb{N}}$ converges with respect to the topology of $\mathscr{E}\left( \Omega\right) $.
This proves that $B$ is compact.
Hence, $\mathscr{E}\left( \Omega\right) $ has the Heine-Borel property.
The same conclusion holds for $\mathscr{D}_K\left( \Omega\right) $ whenever $K$ has nonempty interior (otherwise, $\mathscr{D}_K\left( \Omega\right) ={0}$), because $\operatorname{dim}\left( \mathscr{D}_K(\Omega)\right) =\infty$ in that case. 
\end{proof}
\begin{rem}\label{toposuppcomp}
The topology of $\mathscr{D}_K\left( \Omega\right) $, as topological  subspace of $\mathscr{E}\left( \Omega\right) $, can be described more thoroughly by constructing a family of seminorms compatible with it.
For this purpose, consider first the space $\mathscr{D}\left( \Omega\right) $ which we can write as
	\begin{equation*}
		\mathscr{D}\left( \Omega\right) =\bigcup_{\substack{K\subseteq \Omega\\ \mbox{\small compact}}}\mathscr{D}_K\left( \Omega\right) .
	\end{equation*}
Next we define a family of norms $\left\lbrace \parallel \cdot \parallel_k\right\rbrace_{k\in\mathbb{N}} $ on  $\mathscr{D}\left( \Omega\right) $ given by
	\begin{equation}\label{norms}
		\parallel \varphi \parallel_k=\sup\limits_{\substack{x\in \Omega\\ |\nu|\leq k}}  \left| \partial^\nu \varphi(x)\right|.  
	\end{equation}
Then, the restrictions of these norms to  any fixed subset  $\mathscr{D}_K\left( \Omega\right) $ of $\mathscr{D}\left( \Omega\right) $ induce the same topology on $\mathscr{D}_K\left( \Omega\right) $ as do the seminorms of equation \eqref{seminorms2},  which define the topology on $\mathscr{E}\left( \Omega\right) $.
To see this, note that to each compact $K$ corresponds an integer $l_0$ such that $K$ is contained in $K_l$ for every $l\geq l_0$. 
For these $l$, $\parallel \varphi \parallel_k=\parallel \varphi \parallel_k^l$ if $\varphi$ belongs to $\mathscr{D}_K\left( \Omega\right) $.
Since 
	\begin{equation*}
		\parallel \varphi \parallel_k^l \leq \parallel \varphi \parallel_{k}^{l+1}\mbox{,}
	\end{equation*}
the topology induced by the sequence of seminorms
	\begin{equation*}
		\big\lbrace\left\|\cdot \right\|_k^l \big\rbrace_{k,l}
	\end{equation*}
is unchanged if we let $l$ start from $l_0$ rather than 1.
Therefore, the topology on $\mathscr{D}_K\left( \Omega\right) $ defined by the seminorms \eqref{seminorms2} coincides with the one defined by the restrictions of the norms \eqref{norms} to $\mathscr{D}_K\left( \Omega\right) $.
In practice, either set of seminorms can be used.

The topology of $\mathscr{D}_K\left( \Omega\right) $ will be denoted by $\tau_K$.
A local base for  $\tau_K$ is formed by the sets 
	\begin{equation*}
		V_N=\left\lbrace\varphi \in \mathscr{D}_K\left( \Omega\right) :\ \parallel \varphi \parallel_{N}<\frac{1}{N} \right\rbrace\mbox{,}
	\end{equation*}
for $N$ in $\mathbb{N}$.
\end{rem}
The norms in equation \eqref{norms} could be used to define a locally convex metrizable topology  on $\mathscr{D}\left( \Omega\right) $ as described in Theroem \ref{degas} and Remark \ref{pinocho}.
But this topology would not be complete.

Instead, for the space $\mathscr{D}\left( \Omega\right) $ we will consider the final topology in the category of LCS with respect to the family 
	\begin{equation}\label{fam}
		\mathscr{F}=  \left\lbrace i_K: K\subseteq \Omega\mbox{,}\  \mbox{\small compact}\right\rbrace \mbox{,}
	\end{equation}
where $i_K$ denotes the inclusion map:
	\begin{equation*}
		i_K:\mathscr{D}_K\left( \Omega\right) \hookrightarrow\mathscr{D}\left( \Omega\right) .
	\end{equation*}

The following theorem exhibits a topology on $\mathscr{D}\left( \Omega\right) $ such that all the functions of the family $\mathscr{F}$ are continuous. 
Later on,  it will be shown that this topology is the finest with respect to this property (\textit{i.e.} it is indeed the final topology with respect to the family  $\mathscr{F}$).
\begin{thm}[See \cite{rudin}, Ch. 6, Thm. 6.4]\label{rutkauskas}
Let $\beta$ be the collection of all convex balanced sets $W$ contained in $\mathscr{D}\left( \Omega\right) $ such that $\mathscr{D}_K\left( \Omega\right) \cap W$ belongs to $\tau_K$ (see Remark \ref{toposuppcomp}) for every compact subset $K$ of $\Omega$.
Let $\tau$ be the collection of all unions of sets of the form $\varphi + W$, with $\varphi$ in $\mathscr{D}\left( \Omega\right) $ and $W$ in $\beta$.
Then,
	\begin{itemize}
		\item[(i)]  $\tau$ is a topology in $\mathscr{D}\left( \Omega\right) $, and $\beta$ is a local base for $\tau$; 
		\item[(ii)] $\tau$ makes $\mathscr{D}\left( \Omega\right) $ into a LCS.
	\end{itemize}
\end{thm}
\begin{proof}
Suppose $V_1$ and $V_2$ belong to $\tau$, and take  $\phi$  in  $V_1\cap V_2$.
To prove \textit{(i)} it is enough to show that
	\begin{equation}\label{pop}
		\phi +W\subseteq V_1 \cap V_2
	\end{equation}  
for some $W $ in  $\beta$.

The definition of $\tau$ shows that there exist $\phi_i$ in   $\mathscr{D}\left( \Omega\right) $ and $W_i$  in  $\beta$ such that
	\begin{equation*}
		\phi \in\phi_i+W_i\subseteq V_i\qquad (i=1\mbox{,}\ 2).
	\end{equation*}
Choose $K$ so that $\mathscr{D}_K\left( \Omega\right) $ contains $\phi_1$, $\phi_2$ and $\phi$.	
Since $\mathscr{D}_K\left( \Omega\right) \cap W_i$ is open in $\mathscr{D}_K\left( \Omega\right) $ and $\phi-\phi_i\in W_i$, we have
	\begin{equation*}
		\phi-\phi_i\in (1-\delta_i)W_i\mbox{,}
	\end{equation*}
for some $0<\delta_i<1$.
The convexity of $W_i$ implies therefore that
	\begin{equation*}
		\phi-\phi_i+\delta_i W_i \subseteq (1-\delta_i)W_i +\delta_i W_i=W_i\mbox{,}
	\end{equation*}
so that
	\begin{equation*}
		\phi+\delta_i W_i\subseteq \phi_i +W_i\subseteq V_i \qquad (i=1\mbox{,}\ 2).
	\end{equation*}
Hence, \eqref{pop} holds with $W=\delta_1W_1\cap \delta_2W_2$.

Now we prove \textit{(ii)}.
Suppose  that $\phi_1$ and $\phi_2$ are distinct elements of $\mathscr{D}\left( \Omega\right) $, and put
	\begin{equation*}
		W=\left\lbrace\phi\in \mathscr{D}\left( \Omega\right) : \parallel \phi \parallel_0<\parallel \phi_1 -\phi_2 \parallel_0  \right\rbrace\mbox{,}
	\end{equation*}
where $\parallel\cdot\parallel_0$ was defined in \eqref{norms}. 
Then, $W$ belongs to $\beta$ and $\phi_1 $ is not in $\phi_2 +W$.
It follows that the singleton $\left\lbrace \phi_1\right\rbrace $ is a closed set relative to $\tau$. 

Addition is $\tau$-continuous, since the convexity of every $W$ in $\beta$  implies that 
	\begin{equation*}
		\left( \Psi_1+\frac{1}{2} W\right) +\left( \Psi_2+\frac{1}{2} W\right) =(\Psi_1+\Psi_2)+W
	\end{equation*}
for any $\Psi_1$, $\Psi_2$ in $\mathscr{D}\left( \Omega\right) $.

To deal with scalar multiplication, pick a scalar $\alpha_0$ and $\phi_0$ in $\mathscr{D}\left( \Omega\right) $.
Then 
	\begin{equation*}
		\alpha\phi-\alpha_0\phi_0=\alpha(\phi-\phi_0)+(\alpha-\alpha_0)\phi_0.
	\end{equation*} 
If $W$ belongs to $\beta$, there exists $\delta>0$ such that 
	\begin{equation*}
		\delta\phi_0\in \frac{1}{2}W.
	\end{equation*}
Then, as $\frac{1}{2}W$ is balanced, whenever $|\alpha-\alpha_0|<\delta$, we have
	\begin{equation*}
		(\alpha-\alpha_0)\phi_0=\frac{\alpha-\alpha_0}{\delta}\ \delta\phi_0\subseteq \frac{1}{2}W.
	\end{equation*}  
Choose $c$ so that $2c(|\alpha_0|+\delta)=1$. 
Then, $|2\alpha c|<1$ for every $\alpha$ such that $|\alpha-\alpha_0|<\delta$.
Then, whenever  $\phi-\phi_0$ is in $cW$ and 
since $\frac{1}{2}W$ is balanced, it follows that
	\begin{equation*}
		\alpha(\phi-\phi_0)\in \alpha cW\subseteq 2\alpha c\ \frac{1}{2} W\subseteq  \frac{1}{2} W.
	\end{equation*}
Then, as $W$ is convex,
	\begin{equation*}
		\alpha \phi-\alpha_0\phi_0\in W\mbox{,}
	\end{equation*}
whenever $|\alpha-\alpha_0|<\delta$ and $\phi-\phi_0$ belongs to $cW$.
This shows that scalar multiplication is continuous.
\end{proof}
\begin{propo}\label{joesgarage}
The topology $\tau$ on $\mathscr{D}\left( \Omega\right) $ defined in the previous theorem, is the final topology in the category of LCS with respect to the family of functions $\mathscr{F}$ defined in \eqref{fam}. 
\end{propo}
\begin{proof}
Let $\tau'$ be a topology that turns $\mathscr{D}\left( \Omega\right) $ into a LCS and such that the functions in the family \eqref{fam} are continuous. 
We seek to show that $\tau'\subseteq \tau$.
Take a set  $V$ in $\tau'$. 
As $\left( \mathscr{D}\left( \Omega\right) ,\tau'\right)$ is a LCS there is a zero neighborhood basis $\beta'$ whose elements are convex balanced sets.
Then $V$ can be written as:
	\begin{equation*}
		V=\bigcup_{\phi\in V}\left( \phi+W_\phi\right)\mbox{,} 
	\end{equation*}
where $W_\phi$ is some set in $\beta'$ such that $\phi+W_\phi$ is contained in $V$.
By hypothesis,
	\begin{equation*}
		i_{K}: \left( \mathscr{D}_K\left( \Omega\right) ,\tau_K\right) \hookrightarrow \left( \mathscr{D}\left( \Omega\right) ,\tau'\right)
	\end{equation*}
is continuous for every compact subset $K$ of $\Omega$. 
Then, 
	\begin{equation*}
		\mathscr{D}_K\left( \Omega\right) \cap W_\phi \in \tau_K\mbox{,} \qquad \forall K\subseteq\Omega \quad \mbox{compact subset.}
	\end{equation*}
Thus, for every $\phi$  in  $V$, $W_\phi$ belongs to   $\beta$ and therefore, $V$ belongs to $\tau$. 
\end{proof}
\begin{thm}[See \cite{rudin}, Ch. 6,  Thm. 6.5]\label{brahms} 
Let $\tau$ denote the topology of $\mathscr{D}\left(\Omega\right)$ (described in Theorem \ref{rutkauskas}) and $\tau_K$ that of $\mathscr{D}_K\left(\Omega\right)$ (see Remark \ref{toposuppcomp}).
	\begin{itemize}
		\item[(i)]The topology $\tau_K$ coincides with the subspace topology that $\mathscr{D}_K\left( \Omega\right) $ inherits from $\mathscr{D}\left( \Omega\right) $.
		\item[(ii)] If $B$ is a bounded subset of $\mathscr{D}\left( \Omega\right) $, then it is contained in  $\mathscr{D}_K\left( \Omega\right) $ for some compact subset  $K$ of $\Omega$, and there are numbers $M_N<\infty$ such that every $\phi$ in $B$ satisfies the inequalities
			\begin{equation*}
				\left\| \phi\right\|_N \leq M_N \qquad (N=0\mbox{, } 1\mbox{, } 2\mbox{, }\dots).
			\end{equation*} 
		\item[(iii)] $\mathscr{D}\left( \Omega\right) $ has the Heine-Borel property.
		\item[(iv)]If $\left( \varphi_i\right) _{i\in\mathbb{N}}$ is a Cauchy sequence in $\mathscr{D}\left( \Omega\right) $ then it is contained in $\mathscr{D}_K\left( \Omega\right) $, for some compact subset   $K$ of $\Omega$ and 
			\begin{equation*}
				\lim\limits_{i,j\rightarrow\infty} \parallel\varphi_i-\varphi_j\parallel_k=0\mbox{,}
			\end{equation*}
		for every $k$ in $\mathbb{N}$.
		\item[(v)] If $\varphi_i\rightarrow 0$ with respect to the topology  of $\mathscr{D}\left( \Omega\right) $, then there is a compact subset $K$ of $\Omega$ which contains the support of every $\varphi_i$, and $\partial^\nu\varphi_i\rightarrow 0$ uniformly as $i\rightarrow \infty$ for every multi-index $\nu$. 
		In other words, there is a compact subset $K$ of $\Omega$ which contains the support of every $\varphi_i$ and $\varphi_i\rightarrow 0$ for the topology  of $\mathscr{D}_K\left( \Omega\right) $.  
		\item[(vi)] $\left( \mathscr{D}\left( \Omega\right) ,\tau\right) $ is sequentially complete.
	\end{itemize}
\end{thm}
\begin{rem}
In view of \textit{(i)}, the necessary conditions expressed by   \textit{(ii)},  \textit{(iv)} and  \textit{(v)} are also sufficient.
\end{rem}
\begin{proof}[Proof of Theorem \ref{brahms}]
To prove  \textit{(i)} suppose first that $V$ belongs to $\tau$. 
Pick $\phi$ in $\mathscr{D}_K\left( \Omega\right) \cap V$.
By theorem \ref{rutkauskas}, $\phi + W$ is contained in $V$ for some $W$ in $\beta$, where we recall that  $\beta$ is the collection of all convex balanced subsets $W$ of $\mathscr{D}\left( \Omega\right) $ such that $\mathscr{D}_K\left( \Omega\right) \cap W$ belongs to $\tau_K$ for every compact subset  $K$ of $\Omega$.
Hence,
	\begin{equation*}
		\phi+\left( \mathscr{D}_K\left( \Omega\right) \cap W\right)\subseteq \mathscr{D}_K\left( \Omega\right) \cap V.
	\end{equation*}
Since $\mathscr{D}_K\left( \Omega\right) \cap W$ is open in $\mathscr{D}_K\left( \Omega\right) $, we have proved that 
	\begin{equation*}
		\mathscr{D}_K\left( \Omega\right) \cap V\in\tau_K\mbox{,}\quad \mbox{ if }\ V\in\tau \ \mbox{ and } \ K\subseteq\Omega.
	\end{equation*}
Now suppose $U$ belongs to $\tau_K$.
We have to show that $U=\mathscr{D}_K\left( \Omega\right) \cap V$ for some $V$ in $\tau$.
The definition of $\tau_K$ implies that to every $\phi$  in $U$ corresponds numbers  $N$ in $\mathbb{N}$ and $\delta>0$ such that
	\begin{equation*}
		\left\lbrace\psi\in\mathscr{D}_K\left( \Omega\right) :\left\| \psi-\phi\right\|_N<\delta  \right\rbrace \subseteq U.
	\end{equation*}
Put 
	\begin{equation*}
		W_{\phi}=\left\lbrace\psi\in\mathscr{D}\left( \Omega\right) :\left\| \psi\right\|_N<\delta  \right\rbrace.
	\end{equation*}
Then $W_\phi$ belongs to $\beta$, and
	\begin{equation*}
		\mathscr{D}_K\left( \Omega\right) \cap\left(\phi + W_\phi \right)=\phi+\left(\mathscr{D}_K\left( \Omega\right) \cap W_\phi \right)=\left\lbrace\psi\in\mathscr{D}_K\left( \Omega\right) :\left\| \psi-\phi\right\|_N<\delta  \right\rbrace \subseteq U.  
	\end{equation*}
If $V$ is taken to be the union of the sets $\phi+W_\phi$, where $\phi$ runs through $U$, then $V$ has the desired property.

To prove  \textit{(ii)}, consider a subset $B$ of  $\mathscr{D}\left( \Omega\right) $ which lies in no $\mathscr{D}_K\left( \Omega\right) $.
Then there are functions $\phi_m$ in $B$ and there are distinct points $x_m$ in $\Omega$, without limit point in $\Omega$, such that $\phi_m(x_m)\neq 0$  ($m=1\mbox{, } 2\mbox{, } 3\mbox{, }\dots$).

Let $W$ be  the set of all $\phi$ in $\mathscr{D}\left( \Omega\right) $ that  satisfy
	\begin{equation*}
		\left| \phi(x_m)\right|<m^{-1}\left|\phi_m(x_m) \right|\qquad  (m=1\mbox{, } 2\mbox{, } 3\mbox{, } \dots).
	\end{equation*}
Since each $K$ contains only finitely many $x_m$, it is easy to see that $\mathscr{D}_K\left( \Omega\right) \cap W$ belongs to $\tau_K$. Indeed, suppose for the sake of simplicity, that $x_1\mbox{, }x_2\mbox{, }\dots\mbox{, } x_{m_0}$ are in  $K$ only.
Then, if $B_{m}$ denotes the ball of $\mathbb{R} $ centered at zero of radius $m^{-1}\left|\phi_m(x_m) \right|$,
	\begin{equation*}
		\mathscr{D}_K\left( \Omega\right) \cap W=\bigcap_{m=1}^{m_0} \left(|\cdot|\circ \operatorname{e}_{x_m}\right) ^{-1} \left(B_{m}\right)\mbox{,}
	\end{equation*} 
where $|\cdot|: \mathbb{C}\rightarrow \mathbb{R}_{\geq 0}$ is the operation of taking the modulus of a complex number; and we recall that 
	\begin{align*}
		\operatorname{e}_x:  \mathscr{D}_K\left( \Omega\right)  & \rightarrow  \mathbb{C}   \\
		\varphi& \mapsto \varphi(x)\nonumber
	\end{align*}
is the evaluation map, which is a continuous functional as $\mathscr{D}_K\left( \Omega\right) $ is a subspace of  $\mathscr{E}\left(\Omega\right) $ whenever $K$ is a compact subset of $\Omega$.

Therefore, we have that $\mathscr{D}_K\left( \Omega\right) \cap W$ is open in $\mathscr{D}_K\left( \Omega\right) $.
As $ W$  is also convex and balanced we conclude that $W$ belongs to $\beta$.
Since $\phi_m$ is not in $ mW$, no multiple of $W$ contains $B$.
This shows that $B$ is not bounded.

It follows that every bounded subset $B$ of $\mathscr{D}\left( \Omega\right) $ lies in some $\mathscr{D}_K(\Omega)$.
By \textit{(i)}, $B$ is then a bounded subset of $\mathscr{D}_K(\Omega)$.
Then, by Theorem \ref{degas}, every seminorm $\left\|\cdot \right\|_N$  with $N$ in $\mathbb{N}$ (see \eqref{norms}) is bounded  on $B$.
This means that there exist numbers $0<M_N<\infty$ such that $\left\|\phi \right\|_N\leq M_N$ for $N=1\mbox{, } 2\mbox{, } \dots$ and for every $\phi$ in $B$.
This completes the proof of  \textit{(ii)}.

Statement  \textit{(iii)} follows from  \textit{(ii)}, since $\mathscr{D}_K(\Omega)$ has the Heine-Borel property (see Proposition  \ref{pinkfloyd}).

Since Cauchy sequences are bounded,  \textit{(ii)} implies that every Cauchy sequence $\left(  \phi_i \right)_{i\in\mathbb{N}} $ in $\mathscr{D}(\Omega)$ lies in some $\mathscr{D}_K(\Omega)$. 
By  \textit{(i)}, $\left(  \phi_i \right)_{i\in\mathbb{N}} $ is also a Cauchy sequence relative to $\tau_K$.
This proves  \textit{(iv)}.

Statement  \textit{(v)} is just a restatement of  \textit{(iv)}.

Finally,  \textit{(vi)} follows from  \textit{(i)},  \textit{(iv)} and the completeness of $\mathscr{D}_K(\Omega)$ 
(recall that $\mathscr{D}_K(\Omega)$ is a Fr\'echet space by Proposition \ref{pinkfloyd}).
\end{proof}
By  \textit{(v)} of Theorem \ref{brahms} the notion of convergence of a  sequence of test functions with respect to the topology of $\mathscr{D}(\Omega)$  (described in Theorem \ref{rutkauskas}) can be given in an  equivalent way as in the following definition.
\begin{defi}
A sequence of functions $\left( \varphi_i\right)_{i\in\mathbb{N}}$  in $\mathscr{D}(\Omega)$ \emph{converges} to $\varphi$ in $\mathscr{D}(\Omega)$ if the following two conditions are satisfied:
	\begin{itemize}
		\item[1.] There exists a compact subset of $\Omega$ such that the supports of all $\varphi_i$ are contained in it. 
		\item[2.] For any n-tuple $k$ of nonnegative integers the functions $\partial^k\varphi_i$ converge uniformly to $\partial^k\varphi$ $($as $i\rightarrow\infty)$.
	\end{itemize}
\end{defi}
We turn to give a characterization of open sets in $\left( \mathscr{D}(\Omega),\tau\right)$ and also of contiuous linear mappings of $ \mathscr{D}(\Omega)$ into a LCS $E$.
\begin{propo}\label{lennon}
Let $\tau$ denote the topology of $\mathscr{D}\left(\Omega\right)$ (described in Theorem \ref{rutkauskas}) and $\tau_K$ that of $\mathscr{D}_K\left(\Omega\right)$ (see Remark \ref{toposuppcomp}).
	\begin{itemize}
		\item[(i)] A set $W$ is open in $\mathscr{D}(\Omega)$ if and only if  $W\cap \mathscr{D}_K(\Omega)$ is open in $\mathscr{D}_K(\Omega)$ for every $K$ compact subset of $\Omega$.
		\item[(ii)] Let $E$ be any LCS.
		A linear mapping 
			\begin{equation*}
			f:\left( \mathscr{D}(\Omega),\tau\right)\rightarrow E 
			\end{equation*}
		is continuous if and only if 
			\begin{equation*}
				f\circ i_K:\left( \mathscr{D}_K(\Omega),\tau_K\right)\rightarrow E
			\end{equation*}
		is continuous for every $K$ compact subset of $\Omega$.
	\end{itemize}
\end{propo}
\begin{proof}
To prove \textit{(i)}, if a subset $W$ of $\mathscr{D}(\Omega)$ satisfies that $W\cap \mathscr{D}_K(\Omega)$ is open in $\mathscr{D}_K(\Omega)$  for every  $K$ contained in  $\Omega$, then, as $\tau_K$ coincides with the subspace topology (see \textit{(i)} of Theorem \ref{brahms}), we have that there exists $V_K$  in  $\tau$ such that $W\cap \mathscr{D}_K(\Omega)=V_K\cap \mathscr{D}_K(\Omega)$ for every $K$ contained in $\Omega$. 
Thus,
	\begin{equation*}
		W=\bigcup_{K\subseteq \Omega}\left( W\cap\mathscr{D}_K(\Omega)\right) =\bigcup_{K\subseteq \Omega}\left( V_K\cap\mathscr{D}_K(\Omega)\right) =\bigcup_{K\subseteq \Omega}V_K\mbox{,}
	\end{equation*}
and therefore $W$ belongs to $\tau$.
The other implication is just \textit{(i)} of Theorem \ref{brahms}.

To prove \textit{(ii)}, first suppose that $f\circ i_K$ is continuous for every $K$ contained in $\Omega$. 
We must show that $f$ is continuous.
Let $W$ be any  convex balanced zero neighborhood of $E$.
As $f$ is linear $f^{-1}(W)$ is also convex and balanced; and $0$ belongs to $f^{-1}(W)$.
Moreover, for every $K$  compact subset of $\Omega$, 
	\begin{equation*}
		f^{-1}(W)\cap\mathscr{D}_K(\Omega)=\left( f\circ i_K \right)^{-1}(W)\in\tau_K
	\end{equation*}
by the fact that $f\circ i_K$ is continuous for every $K$.
Thus, by definition, $f^{-1}(W)$ belongs to $\beta$ and therefore to $\tau$.
This proves that $f$ is continuous.

The other implication of \textit{(ii)} is trivial by definition of $\tau$, which makes continuous the inclusions $i_K$ for every $K$ contained in $\Omega$.
\end{proof}
It has to be noted that, in spite of the fact that $\mathscr{D}(\Omega)$ is contained in $\mathscr{E}\left( \Omega\right) $, the topology just described does not coincide with the subspace topology.
In addition,  the space $\mathscr{D}\left( \Omega\right) $ with this topology is not a Fr\'echet space because it is not metrizable (see \cite{rudin}, Rmk. 6.9). 
\section{Distributions}\label{distrdef}
\begin{defi}\label{defid1}
The space of \emph{distributions}\index{distribution} is defined to be the continuous dual space of $\mathscr{D}(\Omega)$.
The set of distributions forms a vector space which will be denoted by  $\mathscr{D}(\Omega)'$.
\end{defi}
The following theorem gives equivalent characterizations of distributions.
\begin{thm}[See \cite{rudin}, Thm.  6.6]\label{gustavklimt} 
Let $E$ be a LCS and suppose $\Lambda: \mathscr{D}(\Omega)\rightarrow E$ is any linear mapping. 
Then, the following properties are equivalent: 
	\begin{itemize}
		\item[(i)]$\Lambda$ is continuous. 
		\item[(ii)] $\Lambda$ is bounded.
		\item[(iii)] If $\varphi_i\rightarrow 0$ for the topology  of $\mathscr{D}(\Omega)$, then $\Lambda\varphi_i\rightarrow 0$ in $E$. 
		\item[(iv)]The restrictions of $\Lambda$ to every $\mathscr{D}_K(\Omega)\subseteq \mathscr{D}(\Omega)$ are continuous.
	\end{itemize}
\end{thm}
\begin{proof}
The fact that \textit{(i)} implies \textit{(ii)} is contained in Theorem \ref{mielitas}.

Assume $\Lambda$ is bounded and $\varphi_i\rightarrow 0$ in $\mathscr{D}(\Omega)$.
By Theorem   \ref{brahms}, $\varphi_i\rightarrow 0$ in some  $\mathscr{D}_K(\Omega)$, and the restriction of $\Lambda$ to  $\mathscr{D}_K(\Omega)$ is bounded.
As $\mathscr{D}_K(\Omega)$ is metrizable, Theorem \ref{mielitas} applied to 
$\Lambda:\mathscr{D}_K(\Omega)\rightarrow E$ shows that $\Lambda\varphi_i\rightarrow 0$ in $E$.
Thus \textit{(ii)} implies \textit{(iii)}.

Assume \textit{(iii)} holds.
We are going to prove \textit{(iv)}, namely that the restriction of $\Lambda$ to every $\mathscr{D}_K(\Omega)\subseteq \mathscr{D}(\Omega)$ is continuous.
Since $\mathscr{D}_K(\Omega)$ is metrizable,
it suffices to show that $\Lambda$ is sequentially continuous on every $\mathscr{D}_K(\Omega)$.
Let $\left(  \varphi_i\right)_{i\in\mathbb{N}}$ be a sequence contained in $\mathscr{D}_K(\Omega)$ such that  $\varphi_i\rightarrow 0$ in  $\mathscr{D}_K(\Omega)$.
By \textit{(i)} of Theorem \ref{brahms}, 
$\varphi_i\rightarrow 0$ in  $\mathscr{D}(\Omega)$.
Since \textit{(iii)} holds, we have that  $\Lambda\varphi_i\rightarrow 0$ in $E$.
Thus, \textit{(iii)} implies \textit{(iv)} . 

To prove that \textit{(iv)} implies \textit{(i)}, let $U$ be a convex balanced neighborhood of zero in $E$, and put $V=\Lambda^{-1}(U)$.
Then $V$ is convex and balanced.
By \textit{(iv)},   $\mathscr{D}_K(\Omega)\cap V$ is open in $\mathscr{D}_K(\Omega)$ for every $K$ compact subset of $\Omega$.
By definition of $\tau$, the topology of $\mathscr{D}(\Omega)$, $V$ belongs to $\tau$.
\end{proof}
By the previous theorem, if we wish to test whether a functional on $\mathscr{D}_K(\Omega)$ is   continuous, it sufficies to check that it is sequentially continuous.
Therefore, we can write the following equivalent definition of a distribution:
\begin{defi}[Equivalent definition of a distribution]
A \emph{distribution}\index{distribution} is a linear function $t:\mathscr{D}(\Omega)\rightarrow \mathbb{C}$, which satisfies that $t(\varphi_i)\rightarrow t(\varphi)$ in $\mathbb{C}$ whenever  $\varphi_i\rightarrow\varphi$ with respect to the topology of $\mathscr{D}(\Omega)$.
\end{defi} 
Another equivalent definition of distribution\index{distribution}, which follows from the description of $\tau_K$ by means of the norms $\left\| \cdot\right\| _N$ defined in \eqref{norms}, and the equivalence between \textit{(i)} and \textit{(iv)} in Theorem  \ref{gustavklimt} is given in the following proposition.
\begin{propo}[Equivalent definition of a distribution. See \cite{rudin}, Thm. 6.8]\label{joannnn}
If $t$ is a linear functional on $\mathscr{D}(\Omega)$, the following two conditions are equivalent:
	\begin{itemize}
		\item[(i)]$t$ belongs to $\mathscr{D}(\Omega)'$.
		\item[(ii)] To every compact subset $K$ of  $\Omega$ corresponds a nonnegative integer $N$ and a constant $C<\infty$ such that the inequality 
			\begin{equation}\label{obiz}
				\left|\braket{t,\varphi} \right|\leq C  \left\|\varphi\right\|_N 
			\end{equation} 
		holds for every $\varphi$ in $\mathscr{D}_K(\Omega)$.
	\end{itemize} 
\end{propo} 
\begin{rem}
Equation \eqref{obiz} can be replaced by an analogous expression using the seminorms \eqref{seminorms1} or \eqref{seminorms2}, by Remark \ref{toposuppcomp}.
\end{rem}
\chapter{Jet spaces}\markboth{\small\bfseries CHAPTER 3. JET SPACES }{\small\bfseries CHAPTER 3. JET SPACES}
\label{joann3}
The aim of this chapter, is to present Whitney's Extension Theorem, which will be very useful in chapter \ref{cucuu}.
We begin by  giving a  brief introduction to jet spaces.
In particular, we define the subspace of  differentiable functions in the sense of Whitney, wich can be turned into a Banach space by considering a suitable norm.
Throughout this chapter $K$ will denote a compact subset of $\mathbb{R}^d$.

\begin{defi}[See \cite{malgrange}, Ch. 1]\label{stravinsky}
A \emph{jet of order} $m$\index{jet of order $m$} is a family 
	\begin{equation*}
		F=\left(f^k \right)_{|k|\leq m}
	\end{equation*}
of continuous functions $f^k:K\rightarrow\mathbb{C}$, indexed by a  multi-index $k$ with $|k|\leq m$. 

Let $J^m(K)$ denote the space of all jets of order $m$ provided with the natural structure of a vector space on $\mathbb{C}$. 
We define a norm on $J^m(K)$ by
	\begin{equation*}
		\left|F \right|_m^K=\sup\limits_{\substack{x\in K\\ |k|\leq m}}\left|f^k(x) \right|  .
	\end{equation*}
We will write $\left[ F(x)\right]^k=f^k(x)$ for every $x$ in $K$, and $F$ in $J^m(K)$.
If $k=0$ we will just write $F(x)=f^0(x)$.
\end{defi}
\begin{rem}
For $|k|\leq m$ we define the linear map:
	\begin{align*}
		D^k:  J^m(K)& \rightarrow  J^{m-|k|}(K) \\
        F & \mapsto  D^kF:=\left( f^{k+l}\right)_{|l|\leq m-|k|}\mbox{,}\nonumber
	\end{align*}
and for any $g$ in $\mathscr{E}^m\left( \mathbb{R}^d\right) $, $J^m(g)$ denotes the jet
	\begin{equation*}
		J^m(g)=\left(\frac{\partial^{|k|}g}{\partial x^k} \right)_{|k|\leq m}
	\end{equation*}
in $J^m(K)$ where each $\partial^{|k|}_xg$ is understood to be restricted to $K$.
   
The definitions above give the following commutative diagram
	\begin{equation*}
		\xymatrix{
		\mathscr{E}^m\left( \mathbb{R}^d\right)  \ar@{->}[d]_{\partial^k} \ar@{->}[r]^{J^m}& J^m(K)\ar[d]^{D^k}\\
		\mathscr{E}^{m-|k|}\left( \mathbb{R}^d\right)  \ar@{->}[r]_{J^{m-|k|}}& J^{m-|k|}(K) }
	\end{equation*}
\end{rem}
\begin{defi}\label{picasso}
For $x$ in $\mathbb{R}^d$, $a$ in $K$, and $F$ in $J^m(K)$, we define the \emph{Taylor polynomial} $($\emph{of order} m$)$ \emph{of} $F$ \emph{at the point } a \index{Taylor polynomial of a jet}as  
	\begin{equation*}
		T^m_a F(x)=\sum_{|k|\leq m}\frac{(x-a)^k}{k!}f^k(a)\mbox{;}
	\end{equation*}
and noting that $T^m_a F$ belongs to $\mathscr{E}\left( \mathbb{R}^d\right) $ as a function of the variable $x$, we define
	\begin{equation*}
		R^m_a F=F-J^m(T^m_aF).
	\end{equation*}
\end{defi}
\begin{defi}\label{frida}
We define the space $ \mathscr{E}^m\left( K\right) $ of \emph{differentiable functions of order} m \emph{in the sense of Whitney}\index{differentiable function in the sense of Whitney}, as the space of all jets $F$ in $J^m(K)$ such that 
	\begin{equation*}
		\left( R^m_x F\right)(y) =o\left( \left|x-y \right|^{m-|k|} \right) \mbox{ as } \left|x-y \right|\rightarrow 0\mbox{,}
	\end{equation*}
for all elements $x,y$ in $K$, and every  $|k|\leq m$.
We define the following norm on $\mathscr{E}^m\left( K\right) $
	\begin{equation*}
		\left\| F \right\|_m^K = \left|F \right|_m^K+\sup\limits_{\substack{x\mbox{,}y\in K\\x\neq y \\ |k|\leq m}} \frac{\left| \left( R_x^mF\right)^k(y) \right| }{\left| x-y\right|^{m-k} }.
	\end{equation*}
\end{defi}
\begin{propo}\label{leningrad}
$\left( \mathscr{E}^m\left( K\right) , \parallel\cdot\parallel^K_m\right) $ is a Banach space.
\end{propo}
\begin{proof}
We are only going to prove the completeness of the space $\mathscr{E}^m\left( K\right) $ under the norm $ \parallel\cdot\parallel^K_m$. 
Let $\left(F_i \right)_{i\in\mathbb{N}} $ be a Cauchy sequence in $\mathscr{E}^m\left( K\right) $. 
Then, given $\varepsilon>0$ there exists $i_1=i_1(\varepsilon)$ in $\mathbb{N}$ such that 
	\begin{equation*}
		i\mbox{, }j\geq i_1(\varepsilon)\implies \parallel F_i-F_j\parallel^K_m<\varepsilon\mbox{,}
	\end{equation*}
which in turn gives the following implications:
	\begin{subequations}
		\begin{align}
			i\mbox{, }j\geq i_1(\varepsilon)\implies \left| F_i-F_j\right|<\varepsilon \label{pototo}\mbox{,}\\
			i\mbox{, }j\geq i_1(\varepsilon) \implies \sup\limits_{\substack{x\mbox{,}y\in K\\x\neq y \\ |k|\leq m
			}} \frac{\left| \left[  R_x^m\left(F_i-F_j \right) \right] ^k(y) \right| }{\left| x-y\right|^{m-k}}<\varepsilon. \label{pototin}
		\end{align}
	\end{subequations}
The implication \eqref{pototo} implies that for every $k\leq m$,
	\begin{equation*}
		i\mbox{, }j\geq i_1(\varepsilon) \implies \sup\limits_{x\in K}\left|f_i^k(x)-f_j^k(x) \right|<\varepsilon .
	\end{equation*}
In other words, for every $k\leq m$,  $\left( f_i^k\right)_{i\in\mathbb{N}}$ is a Cauchy sequence in the space $\mathscr{E}^0(K)$ of continuous functions  on $K$ with the topology of uniform convergence, which is a Banach space.
Consequently, for every $k\leq m$ there exists a continuous function $f^k:K\rightarrow \mathbb{C}$ such that
	\begin{equation}\label{rachmaninov} 		
		\sup\limits_{x\in K}\left| f_i^k(x)-f^k(x)\right| \underset{i\rightarrow \infty}{\longrightarrow} 0.
	\end{equation}
Let $F=\left(f^k \right)_{|k|\leq m} $.
We affirm  that
	\begin{itemize}
		\item[(i)] $\left\|F_i-F \right\|^K_m\underset{i\rightarrow \infty}{\longrightarrow}0 $, and
		\item[(ii)] $F$ belongs to $\mathscr{E}^m\left( K\right) $.
	\end{itemize}
To prove the first assertion, notice that equation  \eqref{rachmaninov} implies 
	\begin{equation*}
		\left|F_i-F \right|^K_m =\sup\limits_{\substack{x\in K\\k\leq m}}\left| f_i^k(x)-f^k(x)\right| \underset{i\rightarrow \infty}{\longrightarrow} 0.
	\end{equation*}
Then, given $\varepsilon>0$ there exists $i_2=i_2(\varepsilon)$ in $\mathbb{N}$ such that 
	\begin{equation*}
		i\geq i_2(\varepsilon)\implies \left|F_i-F \right|^K_m<\frac{\varepsilon}{2}.
	\end{equation*}
On the other hand, for all elements $x\mbox{, }y$ in $K$, every $k\leq m$ and every $i\geq i_1\left( \frac{\varepsilon}{2}\right) $ we have
	\begin{equation} \label{sarma}
		\begin{split}
			\frac{\left| \left[  R_x^m\left(F_i-F \right) \right] ^k(y) \right| }{\left| y-x\right|^{m-k}}&=\left|\frac{\left[ f_i^k(y)-f^k(y)\right] }{\left( y-x\right)^{m-k}}-\!\sum_{k\leq |l| \leq m }\!\frac{\left(y-x \right)^{l-m}}{\!\!\!\!\!\!\!(l-k)!}\! \left[f^l_i(x)-f^l(x) \right] \! \right| \\
			&=\lim\limits_{j\rightarrow \infty}\left|\frac{\left[ f_i^k(y)-f_j^k(y)\right] }{\left( y-x\right)^{m-k}}-\!\sum_{k\leq |l| \leq m }\! \frac{\left(y-x \right)^{l-m}}{\!\!\!\!\!\!\!(l-k)!}\! \left[f^l_i(x)-f^l_j(x) \right] \! \right|\\
			&= \lim\limits_{j\rightarrow \infty}\frac{\left| \left[  R_x^m\left(F_i-F_j \right) \right] ^k(y) \right| }{\left| y-x\right|^{m-k}}\leq \frac{\varepsilon}{2}\mbox{,}
		\end{split}
	\end{equation}
where we have used equation \eqref{pototin} in the last inequality.
Taking supremum over all elements  $x\mbox{, }y$ in $K$, with  $x\neq y$ and all nonnegative integers $k\leq m$, on the left hand side of the previous inequality, we arrive at  
	\begin{equation}
		\sup\limits_{\substack{x\mbox{,}y\in K\\x\neq y\\|k|\leq m}}\frac{\left| \left[R_x^m\left(F_i-F\right)\right]^k(y)\right|}{\left|x-y\right|^{m-k}}\leq\frac{\varepsilon}{2}\mbox{,}
	\end{equation}
for $i\geq i_1(\frac{\varepsilon}{2})$.

Thus, taking $i_0=i_0(\varepsilon)=\max\left\lbrace i_1\left(\frac{\varepsilon}{2} \right), i_2(\varepsilon) \right\rbrace $ we have that $\left\|F_i-F \right\|^K_m< \varepsilon$ if $i\geq i_0(\varepsilon)$.
  
It remains to show that $F$ belongs to $\mathscr{E}^m\left( K\right) $, by proving the fact
	\begin{equation*}
		R_x^mF^k(y)=o\left( \left| y-x\right|^{m-k}\right).
	\end{equation*}
For this purpose, take $\varepsilon>0$. 
Let $i_0=i_0(\varepsilon)$ be as above such that the bound of equation \eqref{sarma} is valid for every $i\geq i_0$.
In particular, we then have 
	\begin{equation*}
		\left| \left[ R_x^m\left(F-F_{i_0}\right)\right] ^k(y)\right| \leq \frac{\varepsilon}{2}\left| y-x\right|^{m-k}.
	\end{equation*}
Finally, as $F_{i_0}$ itself belongs to the space $\mathscr{E}^m\left( K\right) $, there exists $\delta>0$ such that
	\begin{equation*}
		\left| y-x\right|<\delta\implies \left| \left( R_x^mF_{i_0}\right) ^k(y)\right| <\frac{\varepsilon}{2}\left| y-x\right|^{m-k}.
	\end{equation*}
Then, for all elements $x\mbox{, }y$ in $K$ such that  $x\neq y$ and $\left|y-x \right|<\delta $ we have 
	\begin{equation*}
		\left| \left(R_x^mF\right)^k(y)\right| \leq\left| \left[ R_x^m\left(F-F_{i_0}\right)\right] ^k(y)\right| +\left| \left(R_x^mF_{i_0}\right)^k(y)\right| \leq  \varepsilon\left| y-x\right|^{m-k}.
	\end{equation*}
Thus, $R_x^mF^k(y)=o\left( \left| y-x\right|^{m-k}\right)$.
\end{proof}
The following result is of fundamental importance and it will be useful in Chapter \ref{cucuu}. 
\begin{thm}[Whitney Extension Theorem. See \cite{bierstone}, Thm. 2.3] \label{whitney}\index{Whitney Extension Theorem}
For each positive integer $m$ there exists a continuous linear mapping 
	\begin{equation*}
		R: \mathscr{E}^m\left( K\right) \rightarrow \mathscr{E}^m\left( \mathbb{R}^d\right) 
	\end{equation*}
such that for every $F$ in $\mathscr{E}^m\left( K\right) $, $J^m\circ R(F)=F$ and $R(F)|_{\mathbb{R}^d\setminus K}$ is of class $\mathcal{C}^\infty$. 
\end{thm}
  
\chapter{Distributions on open subsets of manifolds}\label{macri}

In Chapter \ref{asdf} we discussed the spaces of continuously differentiable functions and spaces of test functions, defined on some open subset of $\mathbb{R}^d$.
Later on, we introduced the space of distributions as the continuous dual of the test function space. 
In the present chapter we are going to generalize these ideas, by considering a manifold $\mathcal{M}$ instead of the Euclidean space $\mathbb{R}^d$.

All through this dissertation, $\mathcal{M}$ will denote a $d$-dimensional smooth, paracompact, oriented manifold and $X$ a closed subset of $\mathcal{M}$.
Also, $d$ will denote the distance function induced by some choice of smooth Riemannian metric $g$ on $\mathcal{M}$. 
\section{Spaces of continuously differentiable functions and spaces of test functions}

\begin{defi}\label{defmani}
For each nonempty open subset $\Omega$ of a $d$-dimensional manifold  $\mathcal{M}$ we denote by $\mathscr{E}\left( \Omega\right) $ $($resp. $\mathscr{E}^m\left( \Omega\right) $$)$ the algebra of  infinitely differentiable, or $\mathcal{C}^\infty$,  $($resp. $m$ times continuously differentiable, or $\mathcal{C}^m$$)$ complex valued  functions on $\Omega$. 

If $K$ is a compact set in $\mathcal{M}$, then $\mathscr{D}_K$ denotes the space of all $\varphi$ in  $\mathscr{E}\left( \mathcal{M}\right) $ whose support lies in $K$.
If $K$ is a compact subset of $\Omega$ then $\mathscr{D}_K$ may be identified with a subspace of $\mathscr{E}\left( \Omega\right) $, and we will denote it by $\mathscr{D}_{K}\left( \Omega\right) $.

A function  $\varphi:\Omega\rightarrow\mathbb{C}$ is called a \emph{test function}\index{test function} if it is infinitely differentiable and has compact support. 
The space of test functions over $\Omega$ will be denoted by $\mathscr{D}\left( \Omega\right) $. 
Observe that $\mathscr{D}\left( \Omega\right) $ is the union of the subspaces $\mathscr{D}_K\left( \Omega\right) $ as $K$ runs through all compact subsets of $\Omega$.
\end{defi}
The space $\mathscr{E}\left( \Omega\right) $ (resp. $\mathscr{E}^{m}\left( \Omega\right) $) can be endowed with the structure of a Hausdorff locally convex topological space, defined by a sequence of seminorms and having the following property:
\\
\\
\begin{changemargin}{0.75cm}{0.75cm} 
\textit{A sequence $\left( \varphi_i\right)_{i\in\mathbb{N}} $ in $\mathscr{E}\left( \Omega\right) $ $($resp. $\mathscr{E}^{m}(\Omega)$$)$ converges to zero if and only if for each chart $\left(V, \psi\right) $, each compact subset $K$ of $\psi(V)$
and each multi-index $\nu$ $($resp. such that $|\nu|\leq m$$)$, the sequence }
	\begin{equation*}
	\Big(  \partial^\nu\left( \varphi_i\circ \psi^{-1}\right)|_K\Big) _{i\in\mathbb{N}}
	\end{equation*}  
\textit{converges uniformly to zero.} 
\end{changemargin}
To establish the existence of such a topology we proceed as follows.
\begin{defi}[See \cite{dieu}, Ch. XVII, \S 3]
\label{hendrix}
Let $\mathcal{M}$ be a $d$-dimensional manifold and take $\Omega$ a nonempty open subset of $\mathcal{M}$.
Consider an almost denumerable family of charts  $(V_\alpha,\psi_\alpha)$ of $\Omega$ such that the $V_\alpha$ form a locally finite open covering of $\Omega$.
Let 
	\begin{equation*}
		\left\lbrace K_{l,\alpha}\right\rbrace _{l\in\mathbb{N}}
	\end{equation*}
be a fundamental sequence of compact subsets of $\psi_\alpha(V_\alpha)$, and let $p'_{k,l,\alpha}$ be the seminorm $\parallel \cdot \parallel_k^l$ on $\mathscr{E}\left( \psi_\alpha\left( V_\alpha\right) \right) $, as defined in \eqref{seminorms2}. 
For each pair of integers $k\geq 0$ and $l>0$, and each function $\varphi$  in $\mathscr{E}\left( \Omega\right)$ $($resp. $\mathscr{E}^{m}(\Omega)$$)$ we define the seminorm:
	\begin{equation}\label{castellano}
		p_{k,l,\alpha}(\varphi)= p'_{k,l,\alpha}\left( \varphi\circ\psi_{\alpha}^{-1}\right) .
	\end{equation}
\end{defi}
It is clear that 
	\begin{equation*}
		\mathscr{P}=\left\lbrace p_{k,l,\alpha} \right\rbrace_{k,l,\alpha}
	\end{equation*}
is a separating family of seminorms.
To prove that the required condition is satisfied, we have to show that if $\left(  \varphi_i\right)_{i\in\mathbb{N}}$ tends to zero with respect to the topology defined by these seminorms (see Theorem \ref{degas}), then for each chart $\left(V, \psi\right) $ the sequence  
	\begin{equation*}
		\Big( \partial^\nu\left( \varphi_i\circ \psi^{-1}\right)|_K\Big) _{i\in\mathbb{N}} 
	\end{equation*}  
converges uniformly to zero for every compact subset $K$ of $\psi(V)$.

The compact set $\psi^{-1}(K)$ meets only a finite number of open sets $V_\alpha$, say $V_{\alpha_1}\dots V_{\alpha_q}$.
We can define a fundamental sequence of
	\begin{equation*}
		\bigcup_{h=1}^q V_{\alpha_h}
	\end{equation*}
 by
	\begin{equation*}
		K_l=\bigcup_{h=1}^q \psi_{\alpha_h}^{-1}\left( K_{l,\alpha_h}\right) \mbox{,}
	\end{equation*}
where we recall that 
	\begin{equation*}
		\left\lbrace K_{l,\alpha_h}\right\rbrace_{l\in\mathbb{N}}
	\end{equation*}
is a fundamental sequence of $\psi_{\alpha_h}(V_{\alpha_h})$.
It follows that there exists an integer $l$  such that $K_l$ contains $\psi^{-1}(K)$.
If $w_{i}^h$ is the restriction of $ \varphi_i\circ \psi^{-1}$ to $\psi(V\cap V_{\alpha_h})$, then it is enough to show that the restrictions of the $\partial^{\nu}w_{i}^h$ to 
	\begin{equation}\label{fela}
		K\cap \psi\left( \psi_{\alpha_h}^{-1}\left( K_{l,\alpha_h}\right) \right)
	\end{equation}
converge uniformly to zero.
 
Let 
	\begin{equation*}
		\Psi_h:=\psi_{\alpha_h}\circ \psi^{-1}:\psi\left(V\cap V_{\alpha_h} \right)\rightarrow \psi_{\alpha_h}\left(V\cap V_{\alpha_h} \right) 
	\end{equation*}
be the transition diffeomorphism. 
If we put
	\begin{equation*}
		\varphi_{i}^h:=\varphi_i\circ \psi_{\alpha_h}^{-1}
	\end{equation*}
then we have 
	\begin{equation*}
		w_{i}^h(t)=\varphi_{i}^h\left(\Psi_h(t) \right).  
	\end{equation*}
By \textit{(iii)} of Lemma \ref{mafi}, $\left(  w_{i}^h\right)_{i\in\mathbb{N}}$ converges to zero because each sequence $\left(\varphi_{i}^h\right)_{i\in\mathbb{N}}$ does so in the space 
	\begin{equation*}
		\mathscr{E}\left( \psi_{\alpha_h}\left( V_{\alpha_h}\right) \right)\quad  (\mbox{resp. } \mathscr{E}^m\left( \psi_{\alpha_h}\left( V_{\alpha_h}\right) \right) ).
	\end{equation*} 
Thus, the restrictions of the partial derivatives $\partial^{\nu}w_{i}^h$ to the sets \eqref{fela} converge uniformly to zero on every compact set $K$ contained in  $V$.
\begin{rem}
$\mathscr{D}_K(\Omega)$ $($resp. $\mathscr{D}_K^{m}(\Omega))$  is clearly a closed subspace of $\mathscr{E}\left( \Omega\right) $ $($resp. $\mathscr{E}^{m}(\Omega))$.
\end{rem}

\begin{rem}[On the construction of a  compatible set of seminorms for $\mathscr{E}\left( \Omega'\right) $ if $\Omega'\subseteq \Omega$]\label{sandia}
Following the notation of definition \ref{hendrix}, whenever we have an inclusion $\Omega'\subseteq \Omega$  of open subsets of $\mathcal{M}$, we consider an almost denumerable family of charts  $(V_\alpha,\psi_\alpha)$ of $\Omega$ such that the $V_\alpha$ form a locally finite open covering of $\Omega$.
For $\Omega'$, consider the family of charts given by restriction to $\Omega'$:
	\begin{equation*}
		\left( V_\alpha',\left.\psi_\alpha\right|_{V_\alpha'}\right)\mbox{,}
	\end{equation*}
where $V_\alpha':=V_\alpha\cap\Omega'$.
The latter family of charts forms a locally finite open covering of $\Omega'$.
In addition, if
	\begin{equation*}
		\left\lbrace K_{l,\alpha} \right\rbrace _{l\in\mathbb{N}}
	\end{equation*}
is a fundamental sequence of $\psi_\alpha\left( V_\alpha\right) $, choose a fundamental sequence of $\psi_\alpha\left( V'_\alpha\right) $
	\begin{equation*}
		\left\lbrace K'_{l,\alpha} \right\rbrace _{l\in\mathbb{N}}
	\end{equation*}
such that $K'_{l,\alpha}$ is contained in $K_{l,\alpha}$.
Then, calling $p_{k,l,\alpha}$ and $p^{\Omega'}_{k,l,\alpha}$ the seminorms defined in \eqref{castellano} for $\mathscr{E}\left( \Omega\right) $ and  $\mathscr{E}\left( \Omega'\right) $,  respectively, we have the following immediate consequence:
	\begin{equation*}
		p^{\Omega'}_{k,l,\alpha}(\varphi)\leq 
		p_{k,l,\alpha}(\varphi)
	\end{equation*}
for every $\varphi$ in $\mathscr{E}\left( \Omega\right) $.
\end{rem}
\begin{propo}\label{naranja}
For each function $\psi$ in $\mathscr{E}\left( \Omega\right) $ $(\mbox{resp. }\psi$ in   $\mathscr{E}^{m}(\Omega))$ the linear mapping $\varphi\mapsto \psi \varphi$ of $\mathscr{E}\left( \Omega\right) $ $($resp. $\mathscr{E}^{m}(\Omega)$$)$ into itself is continuous.
Moreover, consider an almost denumerable family of charts  $(V_\alpha,\psi_\alpha)$ of $\Omega$ such that the $V_\alpha$ form a locally finite open covering of $\Omega$. 
If $\varphi,\psi$ belong to $\mathscr{E}\left( \Omega\right) $ $($resp. $\varphi,\psi$ belong to $\mathscr{E}^{m}(\Omega)$$)$, for each $k$ in $\mathbb{N}_0$ $(\mbox{resp. }0 \leq k\leq m)$, each $l$ in $\mathbb{N}$  and each index $\alpha$, there exists a constant $M>0$  independent of $\varphi$ such that $	p_{k,l,\alpha}( \psi\varphi ) \leq M 	p_{k,l,\alpha}(\varphi )$. 
\end{propo}
\begin{proof}
The result follows from
	\begin{equation*}
		\begin{split}
			p_{k,l,\alpha}( \psi\varphi ) &= p'_{k,l,\alpha}\left(( \psi\varphi ) \circ\psi_{\alpha}^{-1}\right)\\
			&= p'_{k,l,\alpha}\left( \left( \psi  \circ\psi_{\alpha}^{-1}\right) \left( \varphi\circ\psi_{\alpha}^{-1}\right)  \right)\mbox{,}
		\end{split}
	\end{equation*}
and the application of \textit{(ii)} of Lemma \ref{mafi}, which ensures there is a constant $M$ which does not depend on $\varphi$ such that
	\begin{equation*}
		p'_{k,l,\alpha}\left( \left(  \psi  \circ\psi_{\alpha}^{-1}\right) \left( \varphi \circ\psi_{\alpha}^{-1}\right)  \right)\leq M p'_{k,l,\alpha}  \left(  \varphi  \circ\psi_{\alpha}^{-1}\right) =M p_{k,l,\alpha}  \left( \varphi\right).
	\end{equation*} 
The statement of the remark  is thus proved.
\end{proof}
\section{Distributions}

\begin{defi}
\label{mariacallas} 
The space of \emph{distributions}\index{distribution} is defined to be the continuous dual space of $\mathscr{D}(\Omega)$.
Equivalently, a \emph{distribution}\index{distribution} is by definition a linear form $t$ on $\mathscr{D}(\Omega)$ whose restriction to each Fr\'echet space $\mathscr{D}_K(\Omega)$ is continuous. 
 The set of distributions forms a vector space which will be denoted by  $\mathscr{D}(\Omega)'$.

In order to verify that a linear function $t$ on $\mathscr{D}(\Omega)$ is a distribution, it must be shown that for each sequence $\left(  \varphi_i\right)_{i\in\mathbb{N}}$ contained in an arbitrary  $\mathscr{D}_K(\Omega)$ and such that converges to zero in $\mathscr{E}\left( \Omega\right) $, the sequence $\left(  \braket{t,\varphi_i}\right)_{i\in\mathbb{N}}$ tends to zero in $\mathbb{C}$. 
\end{defi}
\begin{rem}\label{hormi}
There are alternative definitions of distributions. The reader is referred to Chapter 6 of \cite{hormander}.
\end{rem}
\begin{rem}[Product of a distribution and an infinitely differentiable function]
If $t$ belongs to $\mathscr{D}(\Omega)'$ and $\varphi$ to $\mathscr{E}(\Omega')$, where $\Omega'\subseteq \Omega$ is an inclusion of open sets of $\mathcal{M}$, their product
	\begin{equation}\label{ghjdf}
		\left.t\right|_{\Omega'}\varphi 
	\end{equation} 
belongs to $\mathscr{D}(\Omega')'$ and is given by
	\begin{equation*}
		\left\langle \left.t\right|_{\Omega'}\varphi, \psi \right\rangle=\left\langle \left.t\right|_{\Omega'},\varphi \psi \right\rangle\mbox{,}
	\end{equation*}
for every $\psi$ in $\mathscr{D}(\Omega')$.
To simplify the notation, we will write $t\varphi$ instead of \eqref{ghjdf}.
\end{rem}
\begin{propo}[Equivalent definition of a distribution. See \cite{dieu}, 17.3.1.1] \label{rachmaninov2}
A distribution\index{distribution} is a linear function $t$ defined on $\mathscr{D}(\Omega)$ such that for each compact subset $K$ of $\Omega$ there exist integers $k$ and $l$, and a finite number of indices $\alpha_1\mbox{,}\dots\mbox{, }\alpha_q $ together with a constant $C\geq 0$, such that
	\begin{equation}\label{sesamo}
		\left|\braket{t,\varphi} \right|\leq C\sup_{1\leq h \leq q}  p_{k,l,\alpha_h}(\varphi) 
	\end{equation} 
for every $ \varphi$ in $\mathscr{D}_K(\Omega)$.
\end{propo}
\begin{propo}\label{joplin}
Let $t$ be in $\mathscr{D}(\Omega)'$ and let $O$ be an open subset of $\Omega$ such that $\braket{t,\varphi}=0$ for every $\varphi$ in $\mathscr{D}(\Omega)$ with  $\operatorname{Supp}\left( \varphi\right)$ contained in  $ O$. 
Then we say that $t$ \emph{vanishes on $O$}\index{v@distribution vanishing on an open set}. 
Let $\mathfrak{O}$ be the union of such open sets $O$, which is again open.
Then t vanishes on $\mathfrak{O}$.
\end{propo}
\begin{proof}
Take $\varphi$ in $\mathscr{D}(\Omega)$ such that $\operatorname{Supp}\left( \varphi\right)$ is contained in $\mathfrak{O}$. 
As $\operatorname{Supp}\left( \varphi\right) $ is compact there exist a finite number of open sets $O_1\mbox{,} \dots \mbox{, }O_n$ on which $t$ vanishes and such that
	\begin{equation*}
		K:=\operatorname{Supp}\left( \varphi\right) \subseteq\bigcup_{i=1}^n O_i.
	\end{equation*}
By a partition of unity argument (see \cite{vandijk}, Ch. 2,  Lemma 2.5), there is a set of functions $\left\lbrace\varphi_i\right\rbrace _{i=1}^n$ contained in $\mathscr{D}(\Omega)$ such that
	\begin{equation*}
		\operatorname{Supp}\left( \varphi_i\right) \subseteq O_i\mbox{,}\quad \varphi_i\geq 0\mbox{,}\quad \sum_{i=1}^{n}\varphi_i\leq 1\mbox{,}\ \ \mbox{ and }\quad \sum_{i=1}^{n}\varphi_i = 1 \mbox{ on }K .
	\end{equation*}
Then,
	\begin{equation*}
		\braket{t,\varphi}=\left\langle t,\sum_{i=1}^{n}\varphi_i\varphi\right\rangle =\sum_{i=1}^n\braket{t,\varphi_i\varphi}=0\mbox{,}
	\end{equation*}
where in the last equality we have used the fact that $\varphi_i\varphi$ belongs to $\mathscr{D}\left( O_i\right) $, on which $t$ vanishes.  
\end{proof}
Thus, $\mathfrak{O}$ is the largest open set on which $t$ vanishes. 
\begin{defi}\label{distrcomp}
The \emph{support} of a distribution \index{support of a distribution} $t$ in $\mathscr{D}\left( \Omega\right)' $, denoted by $\operatorname{Supp}\left( t\right) $, is the complement of the largest open set in $\Omega$ on which $t$ vanishes. 
Equivalently, a point $x$ is not in $\operatorname{Supp}\left( t\right) $ if there exists an open neighborhood $V$ of $x$ such that $\braket{t,\varphi}=0$  for every $\varphi$ in $\mathscr{D}\left( V\right) $.

For a fixed compact set $K$, we will denote by $\mathscr{D}'_K(\Omega)$ the space of distributions $t$ in $\mathscr{D}(\Omega)'$ whose support is contained in $K$.
\end{defi}
\begin{rem}
The set $\operatorname{Supp}\left( t\right)$ is clearly seen to be  closed in $\Omega$ with the subspace topology.
\end{rem}
\begin{thm}[See \cite{hormander}, Thm. 2.2.1.]\label{glueunic}
If $t$ belongs to $\mathscr{D}(\Omega)'$ and every point in $\Omega$ has a neighborhood to which the restriction of $t$  is zero, then $t=0$.
\end{thm}
\begin{proof}
If $\varphi$ belongs to $\mathscr{D}(\Omega)$ we can  find for every $x$ in $\operatorname{Supp}\left( \varphi\right) $ an open neighborhood $U$ contained in $\Omega$ such that the restriction of $t$ to $U$ is zero.
As $\operatorname{Supp}\left( \varphi\right) $ is compact we can choose  a finite number of such open sets $U_i$ ($i=1\mbox{,}\dots\mbox{, }n$) which cover $\operatorname{Supp}\left( \varphi\right) $ .

Let $\left\lbrace\phi_i \right\rbrace_{i=i}^n $ be a partition of unity subordinated to the cover $\left\lbrace U_i \right\rbrace_{i=1}^n$.
Then we can write 
	\begin{equation*}
		\varphi=\sum_{i=1}^n \varphi_i\mbox{,} \qquad\mbox{with } \varphi_i:=\varphi\phi_i\in\mathscr{D}(\Omega)\ \forall i=1\mbox{,}\dots\mbox{, }n.
	\end{equation*}
Thus, $\braket{t,\varphi_i}=0$, which implies $\braket{t,\varphi}=\sum_{i=1}^n\braket{t,\varphi_i}=0$.
\end{proof}
\begin{rem}[See \cite{langlang93}, Ch. 11, \S 2, the  discussion immediately after Corollary 2.4]\label{hatcher}
Every compactly supported distribution $t$ in $\mathscr{D}(\Omega)'$ has a continuous extension to the space $\mathscr{E}(\Omega)$.
This is done by taking a function $\chi$ in $\mathscr{D}(\Omega)$ which is equal to 1 in an open neighborhood of $\operatorname{Supp}(t)$, and  defining the extension $t_{ext}$ by 
	\begin{equation}\label{matemate}
		\left\langle t_{ext}, \varphi \right\rangle:=\left\langle t, \chi\varphi \right\rangle 
	\end{equation}
for each $\varphi$ in $\mathscr{E}(\Omega)$.
It is immediate to see that \eqref{matemate} is independent of the choice of  $\chi$ subject to the condition that it is equal to 1 in an open neighborhood of $\operatorname{Supp}(t)$, and that it defines a continuous extension of $t$ to $\mathscr{E}(\Omega)$, namely 
	\begin{equation*}
		\left\langle t, \chi\psi \right\rangle=\left\langle t, \psi \right\rangle
	\end{equation*}
if  $\psi$ belongs to $\mathscr{D}(\Omega)$.
We will generally omit the suffix `\textit{ext}' when we consider this canonical extension, as its use will be clear from context.
\end{rem}
\begin{thm}[See \cite{hormander}, Thm. 2.2.4.]\label{gluedistrib}
Let $\left\lbrace \Omega_i\right\rbrace_{i\in I} $ be an arbitrary family of open sets in $\mathcal{M}$ and set $\Omega=\cup \Omega_i$.
If $t_i$ belongs to $\mathscr{D}\left( \Omega_i\right)' $ and $t_i=t_j$ on $\Omega_i\cap \Omega_j$ for all elements $i,j$ in $I$,
then there exists one and only one $t$ in $\mathscr{D}(\Omega)'$ such that $t_i$ is the restriction of $t$ to $\Omega_i$ for every $i$ in  $I$.
\end{thm}
\begin{proof}
The uniqueness is precisely Theorem \ref{glueunic} so we will prove the existence. 

Let $\left\lbrace \phi_i \right\rbrace_{i\in I} $ be a partition of unity associated to the cover $\left\lbrace \Omega_i\right\rbrace_{i\in I} $.
Then, for any function $\varphi$ in $\mathscr{D}(\Omega)$ we can write
	\begin{equation}\label{jamon}
		\varphi=\sum_{i\in I} \varphi_i\mbox{,}
	\end{equation}
with
	\begin{equation*}
		\varphi_i:=\varphi\phi_i\mbox{,}
	\end{equation*}
that belongs to  $\mathscr{D}(\Omega)$ for every  $i$ in  $I$.
The sum in \eqref{jamon} is finite because $\operatorname{Supp}\left( \varphi\right) $ is contained in the union of a finite number of sets of the cover $\left\lbrace \Omega_i \right\rbrace_{i\in I} $. 
Therefore, we have  that every function $\varphi$ in $\mathscr{D}(\Omega)$ can be written as a finite sum of the type \eqref{jamon}.
 
Next, if $t$ is a distribution with the required properties   
we must have that
	\begin{equation}\label{siria}
		\braket{t,\varphi}=\sum_{i\in I}\big\langle t_i,\left.\varphi_i\right|_{\Omega_i}\big\rangle
	\end{equation}
if $\varphi$ is expressed as a sum like \eqref{jamon}.
We shall prove that 
	\begin{equation*}
			\sum_{i\in I}\big\langle t_i,\left.\varphi_i\right|_{\Omega_i}\big\rangle
	\end{equation*}
is independent of the way in which  $\varphi$ is written. 
This will follow if we show that $\sum_{i\in I}\varphi_i=0$ implies
	\begin{equation*}
		\sum_{i\in I}\big\langle t_i,\left.\varphi_i\right|_{\Omega_i}\big\rangle=0.
	\end{equation*}
Set 
	\begin{equation*}
		K=\bigcup_{i\in I} \operatorname{Supp}\left( \varphi_i\right)\mbox{,}
	\end{equation*}
which is a compact subset of $\Omega$; and choose functions $\psi_j$ in $\mathscr{D}(\Omega)$ such that $\operatorname{Supp}\left( \psi_j\right)$ is contained in $\Omega_j$, and $\sum_{j\in I}\psi_j=1$ on  $K$, where we stress the fact that the last sum is finite. 
Then, we have 
	\begin{equation*}
		\operatorname{Supp}\left( \psi_j\varphi_i\right)\subseteq \Omega_j\cap \Omega_i 
	\end{equation*} 
and
	\begin{equation*}
		\big\langle t_i,\left.\left( \psi_j \varphi_i\right) \right| _{\Omega_i}\big\rangle=\big\langle t_j,\left.\left( \psi_j \varphi_i\right) \right| _{\Omega_j}\big\rangle.
	\end{equation*}  
Hence,
	\begin{equation*}
		\begin{split}
			\sum_{i\in I}\big\langle t_i,\left.\varphi_i\right|_{\Omega_i}\big\rangle&=\sum_{i\in I}\sum_{j\in I}\big\langle t_i,\left.\left( \psi_j \varphi_i\right) \right| _{\Omega_i}\big\rangle=\sum_{i\in I}\sum_{j\in I} \big\langle t_j,\left.\left( \psi_j \varphi_i\right) \right| _{\Omega_j}\big\rangle\\
			&=\sum_{j\in I} \left\langle t_j, \left. \left( \psi_j \sum_{i\in I}\varphi_i\right) \right| _{\Omega_j}\right\rangle=0.
		\end{split}
	\end{equation*}

Having proved that \eqref{siria} defines a linear form on $\mathscr{D}(\Omega)$ we must show that it has the continuity properties required of a distribution.

Choose a compact subset $K$ of $\Omega$, and, as before, choose functions $\psi_j$ in $\mathscr{D}(\Omega)$ such that $\operatorname{Supp}\left( \psi_j\right)$ is contained in $\Omega_j$, and $\sum_{j\in I}\psi_j=1$ on $K$, where we note that the last sum is finite.
If $\varphi$ belongs to  $\mathscr{D}_K(\Omega)$, we have that $\varphi \psi_j$ is in $\mathscr{D}(\Omega)$ whose support lies in $\Omega_j$ and $\varphi=\sum_{j\in I}\varphi\psi_j$.
So \eqref{siria} gives 
	\begin{equation}\label{mouse}
		\braket{t,\varphi}=\sum_{j\in I}\big\langle t_j,\left.\left( \varphi\psi_j\right) \right|_{\Omega_j}\big\rangle.
	\end{equation}

We seek to estimate 
	\begin{equation*}
		\left| \big\langle t_j,\left. \left( \varphi\psi_j\right) \right|_{\Omega_j}\big\rangle\right|
	\end{equation*}
for each $j$.
Since each $t_j$ is a continuous linear functional on $\mathscr{D}\left( \Omega_j\right) $, by Proposition \ref{rachmaninov2}, given the compact set
	\begin{equation*}
		K_j:=\operatorname{Supp}\left( \psi_j\right) \cap K\subseteq \Omega_j\mbox{,}
	\end{equation*} 
there exist integers $k_j$ and $l_j$ and a finite number of indices $\alpha_{1,j}\mbox{,}\dots\mbox{, }\alpha_{q_j,j} $ together with a constant $C_j\geq 0$ such that for every $\psi$ in $\mathscr{D}_{K_j}\left( \Omega_j\right) $
	\begin{equation}\label{joel}
		\left|\braket{t_j,\psi} \right|\leq C_j\sup_{1\leq h \leq q_j}  p^{\Omega_j}_{k_j,l_j,\alpha_{h,j}}(\psi) \leq C_j\sup_{1\leq h \leq q_j}  p_{k_j,l_j,\alpha_{h,j}}(\psi) 
	\end{equation} 
where we have followed the notations and facts  of Remark \ref{sandia} for the seminorms.

Therefore, each of the nonzero terms of the sum \eqref{mouse}, which are finitely many, can be estimated as:
	\begin{equation*}
		\left| \big\langle t_j,\left. \left( \varphi\psi_j\right) \right|_{\Omega_j}\big\rangle\right| \leq C_j\sup_{1\leq h \leq q_j}  p_{k_j,l_j,\alpha_{h,j}}(\varphi\psi_j).
	\end{equation*}
In addition, by Proposition \ref{naranja}, each of the seminorms evaluated on the product $\varphi\psi_j$ can be estimated as follows:
	\begin{equation*}
		p_{k_j,l_j,\alpha_{h,j}}\left( \varphi\psi_j\right) \leq C'_j\  p_{k_j,l_j,\alpha_{h,j}}(\varphi)\mbox{,}\qquad \forall 1\leq h \leq q_j\mbox{,}
	\end{equation*}
where $C'_j$ does not depend on $\varphi$.
Therefore,
	\begin{equation*}
		\left| \braket{t,\varphi}\right| \leq  \sum_{j\in I}C_jC'_j\sup_{1\leq h \leq q_j}  p_{k_j,l_j,\alpha_{h,j}}(\varphi) \leq M \sum_{j\in I}\sup_{1\leq h \leq q_j}  p_{k_j,l_j,\alpha_{h,j}}(\varphi)\mbox{,}
	\end{equation*}
for every $\varphi$ in  $\mathscr{D}_K(\Omega)$, and where $M=\max\left\lbrace C_jC'_j: j\in I\right\rbrace $.
Finally,
	\begin{equation*}
		M \sum_{j\in I}\sup_{1\leq h \leq q_j}  p_{k_j,l_j,\alpha_{h,j}}(\varphi)\leq M \sum_{j\in I}\sup_{\substack{1\leq h \leq q_j\\ j\in I}}  p_{k_j,l_j,\alpha_{h,j}}(\varphi)=M'\sup_{\substack{1\leq h \leq q_j\\ j\in I}}  p_{k_j,l_j,\alpha_{h,j}}(\varphi)\mbox{,}
	\end{equation*}
where $M'=M \sum_{j\in I}1$.
Thus,
	\begin{equation*}
		\left| \braket{t,\varphi}\right| \leq M'\sup_{\substack{1\leq h \leq q_j\\ j\in I}}  p_{k_j,l_j,\alpha_{h,j}}(\varphi)\mbox{,}
	\end{equation*}
which has the form of  estimate \eqref{sesamo}.
This completes the proof.
\end{proof}
\section{Distributions having moderate growth along a closed subset of a manifold}

\begin{defi}[\textit{cf.} with the ambiguous definition in \cite{vietdang}]
\label{grulla} Let $\Omega$ be an open subset of $\mathcal{M}$ and take a family of charts $\left( V_\alpha,\psi_\alpha\right) $  covering $\Omega$, as in definition \ref{hendrix}. 
A distribution $t$ in $\mathscr{D}\left( \Omega\setminus X\right)' $ has \index{moderate growth}\emph{moderate growth along $X$} if for every compact subset $K$ of $\Omega$  there is a finite collection of  seminorms $p_{k,l,\alpha_1}\mbox{, } p_{k,l,\alpha_2}\mbox{,} \dots\mbox{, }p_{k,l,\alpha_q}$, and a pair of constants $C$ and $s$ in $\mathbb{R}_{\geq 0}$ such that
	\begin{equation}\label{modgro1}
		\left| \braket{t,\varphi}\right| \leq C\left[ 1+d\left( \operatorname{Supp}\left( \varphi\right) ,X\right)^{-s}\right] \sup_{1\leq h \leq q}  p_{k,l,\alpha_h}(\varphi) 
	\end{equation}
for every $\varphi$  in $\mathscr{D}_K\left( \Omega\setminus X\right) $. 
\end{defi}
We denote by $\mathcal{T}_{\mathcal{M}\setminus X}\left( \Omega\right) $\label{virchulina} the set of distributions in $\mathscr{D}\left( \Omega\setminus X\right)'$ with moderate growth along $X$. 
Note that the previous definition of $\mathcal{T}_{\mathcal{M}\setminus X}\left( \Omega\right) $ is independent of the choice of the distance $d$, since two Riemannian metrics $g_1$ and $g_2$ on $\mathcal{M}$  are \emph{locally equivalent}\index{locally equivalent metrics}. \textit{i.e.} given any point of $\mathcal{M}$ there exists a neighborhood $V$ and positive constants $\lambda_1$ and $\lambda_2$ such that $g_1\leq \lambda_2g_2$ and $g_2\leq \lambda_1g_1$ on $V$. 
\begin{rem}\label{casopart}
If\ $t$ belongs to $\mathscr{D}\left( \mathcal{M}\right)' $ we can choose $s=0$ in \eqref{modgro1} and have the same estimate without the divergent factor
	\begin{equation*}
		1+d\left(\operatorname{Supp}\left( \varphi\right) ,X\right)^{-s}\mbox{,}
	\end{equation*}
by Proposition \ref{rachmaninov2}.
\end{rem} 

The product of a distribution with moderate growth and a $\mathcal{C}^\infty$ function gives a distribution with moderate growth.
This result is mentioned at the beginning of \S 1.1 of \cite{vietdang}, but the author gives no demonstration.
We thus state this result in the following remark and  give a brief  proof of it.

\begin{propo}\label{viir} 
For every $\varphi$ in $\mathscr{E}\left( \Omega\right) $, whenever  $t$ in $\mathscr{D}\left( \Omega\setminus X\right)' $ has moderate growth, the same property holds for $t\varphi$ in $\mathscr{D}\left( \Omega\setminus X\right)' $. 
\end{propo}
 
\begin{proof} Take an arbitrary compact subset $K$ of $\Omega$, and  $\psi$ in $\mathscr{D}_K(\Omega\setminus X)$.
Then,
	\begin{equation*}
		\left| \braket{\varphi t, \psi}\right| =\left| \braket{t,\varphi\psi}\right| \leq C[1+d\left( \operatorname{Supp}\left( \varphi\psi\right) ,X\right) ^{-s}]\: \sup_{1\leq h \leq q}  p_{k,l,\alpha_h}(\varphi\psi)\mbox{,}
	\end{equation*}
where we have used the fact that
	\begin{equation*}
		\operatorname{Supp}\left( \varphi\psi\right) \subseteq \operatorname{Supp}\left(  \psi\right) \subseteq K\mbox{,}
	\end{equation*}
so that $\varphi\psi$ belongs to $\mathscr{D}_K(\Omega\setminus X)$;  and the fact that $t$ has moderate growth along $X$.
Therefore, the existence of a pair  $(C,s)$ in $\mathbb{R}_{\geq 0}^2$, and that of a finite collection of seminorms $p_{k,l,\alpha_1}\mbox{, } p_{k,l,\alpha_2}\mbox{,} \dots\mbox{, }p_{k,l,\alpha_q}$ such that the previous inequality holds, is justified.  
 
Moreover, as $\operatorname{Supp} \left( \varphi\psi\right)  \subseteq \operatorname{Supp}\left(  \psi\right) $ we have 
	\begin{equation*}
		\begin{split}
			d\left( \operatorname{Supp}\left( \varphi\psi\right) ,X\right) &\geq d\left( \operatorname{Supp}\left( \psi\right) ,X\right) \\ \implies d\left( \operatorname{Supp}\left( \varphi\psi\right) ,X\right) ^{-s}&\leq d\left( \operatorname{Supp}\left( \psi\right) ,X\right) ^{-s}.
		\end{split}	
	\end{equation*}
Therefore,
	\begin{equation*}
		\left| \braket{\varphi t, \psi}\right| \leq C[1+d\left( \operatorname{Supp}\left( \psi\right) ,X\right) ^{-s}] \sup_{1\leq h \leq q}  p_{k,l,\alpha_h}(\varphi\psi).
	\end{equation*}
Finally, we have to estimate the last factor of the above inequality. 
By Proposition \ref{naranja}, for each index $h$ there is a  constant $C_h\geq 0$ which does not depend on $\psi$ such that
	\begin{equation*}
		p_{k,l,\alpha_h}(\varphi\psi)\leq C_h p_{k,l,\alpha_h}(\psi).
	\end{equation*}
Then, if $C'=\max\left\lbrace C_h: 1\leq h \leq q \right\rbrace $, we get
	\begin{equation*}
		\sup_{1\leq h \leq q}  p_{k,l,\alpha_h}(\varphi\psi)\leq C' \sup_{1\leq h \leq q}  p_{k,l,\alpha_h}(\psi).
	\end{equation*}
Thus, 
	\begin{equation*}
		\left| \braket{\varphi t, \psi}\right| \leq CC'[1+d\left( \operatorname{Supp}\left( \psi\right) ,X\right) ^{-s}] \sup_{1\leq h \leq q}  p_{k,l,\alpha_h}(\psi)\mbox{,}
	\end{equation*}
which completes the proof.
\end{proof}
 
\chapter{The extension of distributions}\label{cucuu}

In this chapter we will focus on the problem of extension of distributions originally defined on the complement of a closed set in a  manifold. 
Recall that $\mathcal{M}$ denotes a $d$-dimensional smooth, paracompact, oriented manifold; $X$ a closed subset of $\mathcal{M}$; and $d$ is the distance function induced by some choice of smooth Riemannian metric $g$ on $\mathcal{M}$.
 
A great number of the results mentioned in this chapter appear in \cite{vietdang}, where they are ambiguosly stated or proved in an incomplete and confusing manner. 
This is the case of Theorems \ref{elteo} and \ref{janis}; Lemmas \ref{dumbo}, \ref{technical} and \ref{zappa}; and Proposition \ref{mendelsshon}.
We thus give an exhaustive and clear demonstration of them, filling the gaps left by the author.

\section{Main extension theorem}

We shall prove the following main theorem which gives equivalent conditions for a distribution $t$  in $\mathscr{D}(\mathcal{M}\setminus X)'$ to be continuously extendible to the whole manifold $\mathcal{M}$. 
This theorem appears in \cite{vietdang} as Thm 0.1.
The proof we shall provide is a well-organized and simplified one of that given there.
\begin{thm}\label{elteo}
Let  $t$ be a distribution in $\mathscr{D}\left( \mathcal{M}\setminus X\right)'$.
The three following claims are equivalent:
	\begin{itemize}
		\item [(i)] $t$ has moderate growth along $X$.
		\item [(ii)] $t$  is extendible to $\mathscr{D}\left( \mathcal{M}\right)$.
		\item [(iii)] There is a family of functions
			\begin{equation*}
				(\beta_\lambda)_{\lambda\in(0,1]}\subseteq\mathscr{E}\left( \mathcal{M}\right)
			\end{equation*}
		such that
			\begin{enumerate}
				\item $\beta_\lambda=0\ $ in a neighborhood of $X$,
				\item $\lim\limits_{\lambda\to 0}\beta_\lambda(x)=1\mbox{,}\ $ for every $x$ in $\mathcal{M}\setminus X$, 
			\end{enumerate}
		and a family 
			\begin{equation*}
				(c_\lambda)_{\lambda\in(0,1]}
			\end{equation*} 
		of distributions on $\mathscr{D}\left( \mathcal{M}\right)$ supported on $X$, such that the following limit
			\begin{equation*}
				\lim\limits_{\lambda\to 0}t\beta_\lambda -c_\lambda
			\end{equation*}
		exists and defines a continuous  extension to $\mathscr{D}\left( \mathcal{M}\right)$ of $t$
		\footnote{In principle, the product $t\beta_\lambda$ belongs to the space $\mathscr{D}\left(\mathcal{M}\setminus X  \right)'$. 
		However, for every $\varphi$ in $\mathscr{D}\left(\mathcal{M}\right) $ the action 
			\begin{equation*}
				\left\langle t\beta_\lambda,\varphi\right\rangle=\left\langle t,\beta_\lambda\varphi\right\rangle
			\end{equation*}
		is well-defined, since  $\beta_\lambda=0$ on a neighborhood of $X$.
		Thus,  $t\beta_\lambda$ can be thought to belong to the space $\mathscr{D}\left( \mathcal{M}\right)' $.}.
	\end{itemize}
\end{thm}
\begin{proof}
In a trivial way, \textit{(iii)} implies \textit{(ii)}, being a particular case of the latter.
That \textit{(ii)} implies \textit{(i)} is quite straightforward, for if $\bar{t}$  is a continuous  extension  to $\mathscr{D}\left( \mathcal{M}\right)$ of $t$, by Proposition \ref{rachmaninov2}, it satisfies the estimate \ref{modgro1} in the particular case of $s=0$.
Thus, by definition,  $\bar{t}$ and therefore $t$ satisfies \textit{(i)}.
The proof of \textit{(i)} implies \textit{(iii)} is the matter of the following two sections.
In \S\ref{simple} we reduce the problem to the Euclidean case with a partition of unity argument (see Propositions \ref{locglo} and \ref{lociso}), and in \S\ref{proofEC} we prove the Euclidean case (see Theorem \ref{elteo2}).
\end{proof}
\section{Reduction to the Euclidean case }\label{simple}

Suppose \textit{(i)} of Theorem \ref{elteo} holds, namely $t$ in $\mathscr{D}(\mathcal{M}\setminus X)'$ has moderate growth along $X$.
 
The first step is to localize the problem by a partition of unity argument. 
Choose a locally finite cover of $\mathcal{M}$ by relatively compact open charts $\left( V_\alpha,\psi_\alpha\right) $ and a subordinate partition of unity $\left\lbrace \varphi_\alpha\right\rbrace _\alpha$ such that $\sum_{\alpha} \varphi_\alpha=1$. 
Set
	\begin{equation*}
		K_\alpha:=\operatorname{Supp}\left(  \varphi_\alpha\right) \subseteq V_\alpha.
	\end{equation*}
Then, each 
	\begin{equation}\label{talpha}
		t_\alpha:=t\varphi_\alpha|_{V_\alpha\setminus X}
	\end{equation}
is compactly supported as it vanishes outside $K_\alpha$, and belongs to
	\begin{equation*}
		\mathscr{D}\left( 
		V_\alpha\setminus\left( X\cap K_\alpha\right) \right)'.
	\end{equation*}
By Remark \ref{hatcher}, each $t_\alpha$ can be thought as an element of the space
	\begin{equation*}
		\mathscr{E}\left(V_\alpha\setminus\left( X\cap K_\alpha\right)  \right)'.
	\end{equation*}
In addition, every $t\varphi_\alpha$  has moderate growth along $X\cap K_\alpha$ by Proposition \ref{viir}. 

We shall call \textit{(i)-(iii)-global} the implication \textit{(i)} implies \textit{(iii)} in Theorem \ref{elteo}.
Now suppose we could prove the following local version  of \textit{(i)-(iii)-global}, which we shall call \textit{(i)-(iii)-local}.
\begin{changemargin}{0.75cm}{0.75cm} 
\textit{\textbf{\textit{(i)-(iii)-local}:} For every $\alpha$ let $t_\alpha$ be the compactly supported distribution
	\begin{equation*}
		t_\alpha=t\varphi_\alpha|_{V_\alpha\setminus X}
	\end{equation*}
in the space
	\begin{equation*}
		\mathscr{D}\left(V_\alpha\setminus\left( X\cap K_\alpha\right)  \right)',
	\end{equation*}
or else in
	\begin{equation*}
		\mathscr{E}\left(V_\alpha\setminus\left( X\cap K_\alpha\right)  \right)'
	\end{equation*}
by Remark \ref{hatcher}.
Then,  for every $\alpha$
	\begin{enumerate}[itemindent=2 cm]
		\item[\textbf{(i)-local}] $t_\alpha$ has moderate growth along $X\cap K_\alpha$, implies	
		\item[\textbf{(iii)-local}] there is a family of functions
			\begin{equation*}
				\left( \beta^\alpha_\lambda\right) _{\lambda\in(0,1]}\subseteq\mathscr{E}\left( V_\alpha\right) 
			\end{equation*} 
		such that
			\begin{enumerate}
				\item $\beta^\alpha_\lambda=0\ $ on a neighborhood of $X\cap K_\alpha$,
				\item $\lim\limits_{\lambda\to 0}\beta^\alpha_\lambda(x)=1\mbox{,}\ $  for every $x$ in $V_\alpha\setminus(X\cap  K_\alpha)$,
			\end{enumerate}
		and a family 
			\begin{equation}\label{fliac}
				\left( c^\alpha_\lambda\right)_{\lambda\in(0,1]}
			\end{equation}
		of distributions  on 	$\mathscr{D}\left( V_\alpha\right)$ supported on $X\cap K_\alpha$, viewed in the space $\mathscr{E}\left( V_\alpha\right)'$ by Remark \ref{hatcher}, such that the following limit
			\begin{equation}\label{strangedays}
				\lim\limits_{\lambda\to 0}t_\alpha\beta^\alpha_\lambda -c^\alpha_\lambda
			\end{equation}
		exists and  defines a continuous  extension  to $\mathscr{E}\left( V_\alpha\right)$ of $t_\alpha$
			\footnote{\label{guagw}In principle, the product $t_\alpha\beta^\alpha_\lambda$ belongs to the space $\mathscr{E}\left(V_\alpha\setminus\left( X\cap K_\alpha\right)  \right)'$. 
			However, for every $\varphi$ in $\mathscr{E}\left( V_\alpha\right) $ the action 
				\begin{equation*}
					\left\langle t_\alpha\beta^\alpha_\lambda,\varphi\right\rangle=\left\langle t_\alpha,\beta^\alpha_\lambda\varphi\right\rangle
				\end{equation*}
			is well-defined, since  $\beta^\alpha_\lambda=0$ on a neighborhood of $X\cap K_\alpha$.
			Thus,  $t_\alpha\beta^\alpha_\lambda$ can be thought to belong to the space $\mathscr{E}\left( V_\alpha\right)'$.}.	
	\end{enumerate}}
\end{changemargin}

Setting
	\footnote{\label{explicacion}A notation abuse has been used here: in principle, the product $\left.\varphi_\alpha\right|_{V_\alpha}\beta_{\lambda}^\alpha$ is only defined on the domain of the functions $\left.\varphi_\alpha\right|_{V_\alpha}$ and $\beta_{\lambda}^\alpha$, namely $V_\alpha$.
 	However, we assume that $\left.\varphi_\alpha\right|_{V_\alpha}\beta_{\lambda}^\alpha$ extends by zero outside $V_\alpha$ and therefore,  $\left.\varphi_\alpha\right|_{V_\alpha}\beta_{\lambda}^\alpha$ may be thought to belong to the space $\mathscr{E}\left( \mathcal{M}\right) $.
 	Thus, the sum $\sum\left.\varphi_\alpha\right|_{V_\alpha}\beta_{\lambda}^\alpha$ makes sense and belongs to $\mathscr{E}\left( \mathcal{M}\right) $ .
 	The second sum in \eqref{lluvia} is also a notation abuse, as the distributions $c_{\lambda}^\alpha$ act on different spaces, namely $\mathscr{D}\left( V_\alpha\right)$, or $\mathscr{E}\left( V_\alpha\right)$ by Remark \ref{hatcher}. 
 	However, as  $c^\alpha_\lambda$ has compact support in $X\cap K_\alpha$ we can extend it by zero outside $V_\alpha$.
 	Thus, $c^\alpha_\lambda$ can be thought to belong to the space $\mathscr{E}\left( \mathcal{M}\right)'$, and given $\varphi$ in $\mathscr{D}\left(\mathcal{M}\right)$, we understand the sum defining $c_\lambda$ to act as
 		\begin{equation*}
 			\left\langle c_\lambda,\varphi\right\rangle =\sum_{\alpha}\left\langle c^\alpha_\lambda, \varphi \right\rangle. 
 		\end{equation*}
 		}
	\begin{equation}\label{lluvia}
		\beta_\lambda=\sum_{\alpha}\left.\varphi_\alpha\right|_{V_\alpha}  \beta^\alpha_\lambda
		\quad  \mbox{ and }  \quad c_\lambda=\sum_{\alpha}c^\alpha_\lambda
	\end{equation}
we find that $t\beta_\lambda-c_\lambda$      converges to some extension of $t$ in $\mathscr{D}\left( \mathcal{M}\right)' $ when $\lambda \rightarrow 0$.
Thus, \textit{(iii)} of Theorem \ref{elteo}, holds.

Therefore, if  \textit{(i)} of Theorem \ref{elteo} is  assumed to be true and the implication \textit{(i)-(iii)-local} holds, then \textit{(iii)} of Theorem \ref{elteo} is also true.
In other words we have the following result:
\begin{propo}\label{locglo}
\textit{(i)-(iii)}-local implies \textit{(i)-(iii)}-global.
\end{propo}

Now suppose the condition \textit{(i)-local} of implication \textit{(i)-(iii)-local} holds for every $\alpha$, namely each $t_{\alpha}$  has moderate growth along $X\cap K_\alpha$.

The idea is to use local charts to transfer the situation on each set $V_\alpha$ to an open subset of the Euclidean space and work on $\mathbb{R}^d$. 
On every set $V_\alpha$,  let 	
	\begin{equation*}
		\psi_\alpha:V_\alpha\rightarrow V\subseteq\mathbb{R}^d
	\end{equation*}
denote the corresponding chart where $V=\psi_\alpha\left( V_\alpha\right) $ is open. 
Then the pushforward 	
	\begin{equation}
		\tilde{t}_\alpha:=\psi_{\alpha*}\left( \left(\left. t\varphi_\alpha\right) \right|_{V_\alpha}\right)
	\end{equation}
is compactly supported as it vanishes outside $\psi_{\alpha}(K_\alpha)$, and belongs to
	\begin{equation*}
		\mathscr{D}\left( V\setminus \psi_{\alpha}\left( X\cap K_\alpha\right) \right)'.
	\end{equation*} 
By Remark \ref{hatcher}, each $\tilde{t}_\alpha$ can be seen as an element of the space 
	\begin{equation*}
		\mathscr{E}\left( V\setminus \psi_{\alpha}\left( X\cap K_\alpha\right) \right)'.
	\end{equation*} 
Actually, the compact set 	
	\begin{equation*}
		\psi_{\alpha}(X\cap K_\alpha)
	\end{equation*}
is contained in $V$, since $X\cap K_\alpha$ is contained in  $V_\alpha$ and $\psi_\alpha$ is a diffeomorphism. 
Therefore, the distribution $\tilde{t}_\alpha$ is an element of
	\begin{equation}\label{teblanconegro}
		\mathscr{E}\left( \mathbb{R}^d\setminus \psi_{\alpha}\left( X\cap K_\alpha\right) \right)'
	\end{equation}
and has compact support contained in  $\psi_{\alpha}(K_\alpha)$.

Suppose we could prove the the following implication, which we shall call \textit{(i)-(iii)-iso}.

\begin{changemargin}{0.75cm}{0.75cm} 
\textit{\textbf{\textit{(i)-(iii)-iso}:} 
For every $\alpha$ let $\tilde{t}_\alpha$ be the compactly supported distribution
	\begin{equation*}
			\tilde{t}_\alpha=\psi_{\alpha*}\left( \left(\left. t\varphi_\alpha\right) \right|_{V_\alpha}\right)
	\end{equation*}
in the space
	\begin{equation*}
			\mathscr{D}\left( V\setminus \psi_{\alpha}\left( X\cap K_\alpha\right) \right)'\mbox{,}
	\end{equation*}
or else in
	\begin{equation*}
			\mathscr{E}\left( \mathbb{R} ^d\setminus \psi_{\alpha}\left( X\cap K_\alpha\right) \right)'
	\end{equation*}
by Remark \ref{hatcher} and the discussion above.
Then,  for every $\alpha$
	\begin{enumerate}[itemindent=2 cm]
		\item[\textbf{(i)-iso}] $\tilde{t}_\alpha$ has moderate growth along the set $\psi_{\alpha}\left( X\cap K_\alpha\right)$, implies
		\item[\textbf{(iii)-iso}] there is a family of functions
			\begin{equation*}
				\left( \tilde{\beta}^\alpha_\lambda\right) _{\lambda\in(0,1]}\subseteq\mathscr{E}\left(\mathbb{R}^d\right) 
			\end{equation*} 
		such that
			\begin{enumerate}
				\item $\tilde{\beta}^\alpha_\lambda=0\ $ on a neighborhood of $ \psi_{\alpha}(X\cap K_\alpha)$,
				\item $\lim\limits_{\lambda\to 0}\tilde{\beta}^\alpha_\lambda(x)=1\mbox{,}\ $ for every $x$ in $\mathbb{R}^d\setminus \psi_{\alpha}(X\cap K_\alpha)$, 
			\end{enumerate}
		and a family 
			\begin{equation*}
				\left( \tilde{c}^\alpha_\lambda\right) _{\lambda\in(0,1]}
			\end{equation*}
		of distributions  on 	$\mathscr{D}\left( \mathbb{R}^d\right)$ supported on $\psi_{\alpha}(X\cap K_\alpha)$, viewed in the space $\mathscr{E}\left( \mathbb{R}^d\right)'$ by Remark \ref{hatcher}, such that the following limit	
			\begin{equation}\label{strangedays2}
				\lim\limits_{\lambda\to 0}\tilde{t}_\alpha\tilde{\beta}^\alpha_\lambda -\tilde{c}^\alpha_\lambda
			\end{equation}
		exists and defines a continuous extension to $\mathscr{E}\left( \mathbb{R}^d\right) $ of $\tilde{t}_\alpha$
			\footnote{In principle, the product $\tilde{t}_\alpha\tilde{\beta}^\alpha_\lambda$ belongs to the space $\mathscr{E}\left( \mathbb{R}^d\setminus \psi_{\alpha}\left( X\cap K_\alpha\right) \right)'$. 
			However, for every $\varphi$ in $\mathscr{E}\left( \mathbb{R}^d \right)' $ the action 
				\begin{equation*}
					\left\langle \tilde{t}_\alpha\tilde{\beta}^\alpha_\lambda,\varphi\right\rangle=\left\langle \tilde{t}_\alpha,\tilde{\beta}^\alpha_\lambda\varphi\right\rangle
				\end{equation*}
			is well-defined since  $\tilde{\beta}^\alpha_\lambda=0\ $ on a neighborhood of $ \psi_{\alpha}(X\cap K_\alpha)$.
			Thus,  $\tilde{t}_\alpha\tilde{\beta}^\alpha_\lambda$ can be thought to belong to the space $\mathscr{E}\left(\mathbb{R}^d\right)' $.}.
	\end{enumerate}}
\end{changemargin}

Define the family of functions 
	\begin{equation*}
		\beta^\alpha_\lambda:=\psi_\alpha^\ast\left( \tilde{\beta}^\alpha_\lambda\right) =\tilde{\beta}^\alpha_\lambda \circ \psi_\alpha\mbox{,}
	\end{equation*}
contained in  $\mathscr{E}\left( V_\alpha\right)$.
Define also the family of distributions 
	\begin{equation*}
		c^\alpha_\lambda:=\psi_\alpha^\ast\left( \tilde{c}^\alpha_\lambda \right)
	\end{equation*}
on $\mathscr{D}\left( V_\alpha\right)$ supported on $X\cap K_\alpha$, viewed in the space $\mathscr{E}\left( V_\alpha\right)'$ by Remark \ref{hatcher}, and where
	\begin{equation*}
		\big\langle \psi_\alpha^\ast\left( \tilde{c}^\alpha_\lambda \right), \phi\big\rangle =\big\langle\tilde{c}^\alpha_\lambda ,\phi\circ\psi_\alpha^{-1}\big\rangle\mbox{,}
	\end{equation*}
for every $\phi$ in $\mathscr{E}\left( V_\alpha\right)$.
We then obtain that
	\begin{equation}
		\lim\limits_{\lambda\to 0}t_\alpha\beta^\alpha_\lambda -c^\alpha_\lambda
	\end{equation}
exists and defines a continuous extension in $\mathscr{E}\left( V_\alpha\right)$ of $t_\alpha$ (see \eqref{talpha} and footnote \ref{guagw}).
Thus, \textit{(iii)-local} holds.

Therefore, if the condition   \textit{(i)-local} of implication \textit{(i)-(iii)-local}  is  assumed to be true for every $\alpha$ and the implication \textit{(i)-(iii)-iso} holds, then \textit{(iii)-local} of of implication \textit{(i)-(iii)-local} is also true for every $\alpha$.
In other words we have the following result:
\begin{propo}\label{lociso}
\textit{(i)-(iii)-iso} implies \textit{(i)-(iii)-local}. 
\end{propo}

\section{Proof of the Euclidean case}\label{proofEC}

The results of Propositions \ref{locglo} and \ref{lociso} show that to prove  \textit{(i)} implies  \textit{(iii)} in Theorem \ref{elteo},  it suffices to prove the implication \textit{(i)-(iii)-iso}, \textit{i.e.} we only  have to show that \textit{(i)-iso} implies \textit{(iii)-iso} for every $\alpha$.
Observe that for every value of $\alpha$ the result to be proved is basically the one described in the following theorem, whose proof will be therefore sufficient.
\begin{thm}\label{elteo2}
Let $X$ be a compact subset of the Euclidean space  $\mathbb{R}^d$.
Let $t$ be a compactly supported distribution defined on $\mathscr{D}\left( \mathbb{R}^d\setminus X\right)$, and viewed in the space $\mathscr{E}\left( \mathbb{R}^d\setminus X\right)'$ by Remark \ref{hatcher},
with moderate growth along $X$.
Set 
	\begin{equation*}
		\mathcal{K}:=\operatorname{Supp}\left( t\right).
	\end{equation*}
Then, there is a family of functions
	\begin{equation*}
		(\beta_\lambda)_{\lambda\in(0,1]} \subseteq  \mathscr{E}\left( \mathbb{R}^d\right) 
	\end{equation*}
such that
	\begin{enumerate}
		\item $\beta_\lambda=0\ $ on a neighborhood of $X$,
		\item $\lim\limits_{\lambda\to 0}\beta_\lambda(x)=1\mbox{,}\ $ for every $x$ in $\mathbb{R}^d\setminus X$, 
	\end{enumerate}
and a family 
	\begin{equation*}
		(c_\lambda)_{\lambda\in(0,1]}
	\end{equation*}
of distributions  on 	$\mathscr{D}\left( \mathbb{R}^d\right)$ supported on $X$, viewed in the space $\mathscr{E}\left( \mathbb{R}^d\right)'$ by Remark \ref{hatcher}, such that the following limit
	\begin{equation*}
		\lim\limits_{\lambda\to 0}t\beta_\lambda -c_\lambda
	\end{equation*}
exists and defines a continuous extension to
$\mathscr{E}\left( \mathbb{R}^d\right)$ of $t$
\footnote{\label{guagw32}In principle, the product $t\beta_\lambda$ belongs to the space $\mathscr{E}\left(\mathbb{R}^d\setminus X \right)'$. 
However, for every $\varphi$ in $\mathscr{E}\left(\mathbb{R}^d\right) $ the action 
	\begin{equation*}
		\left\langle t\beta_\lambda,\varphi\right\rangle=\left\langle t,\beta_\lambda\varphi\right\rangle
	\end{equation*}
is well-defined, since  $\beta_\lambda=0$ on a neighborhood of $X$.
Thus,  $t\beta_\lambda$ can be thought to belong to the space $\mathscr{E}\left( \mathbb{R}^d\right)'$.}.
\end{thm} 

Before giving the proof of Theorem \ref{elteo2}, we will first discuss some intermediate results which will simplify the task, namely Lemmas \ref{dumbo}, \ref{technical} and \ref{zappa}; Proposition \ref{mendelsshon}; and Theorem \ref{janis}.

In the following lemma, for which we recall the definition of the ideal $\mathcal{I}\left( X,\mathbb{R}^d\right) $ of $\mathscr{E}\left( \mathbb{R}^d\right)$ (see Definition \ref{funcinf1}):
	\begin{equation*}
		\mathcal{I}\left( X,\mathbb{R}^d\right) =\left\lbrace\varphi\in\mathscr{E}\left( \mathbb{R}^d\right) : \operatorname{Supp}\left( \varphi\right) \cap X=\emptyset \right\rbrace\mbox{,}
	\end{equation*}
we show that the  moderate growth condition for a compactly supported distribution can be reformulated in terms of its natural growth when evaluated on elements of $\mathcal{I}\left( X,\mathbb{R}^d\right)$, instead of compactly supported functions.
This result is stated in equation (6) of \cite{vietdang}.
However, the author gives no proof of it,  so we provide it. 
\begin{lem}
\label{dumbo}
Let $t$ be a compactly supported distribution defined on $\mathscr{D}\left( \mathbb{R}^d\setminus X\right)$, and viewed in the space $\mathscr{E}\left( \mathbb{R}^d\setminus X\right)'$ by Remark \ref{hatcher}.
Let  $\{K_l\}_{l\in\mathbb{N}}$ be a fundamental sequence of compact sets covering $\mathbb{R}^d$ used to construct a family of seminorms as in Definition \ref{defiseminorms}. 
Then $t$ has moderate growth along $X$ if and only if there exist a pair $(C,s)$ in $\mathbb{R}_{\geq 0}^2$ and a seminorm $\parallel \cdot \parallel_k^l$  such that 
	\begin{equation} \label{modgro2}
		\left|\braket{t,\varphi} \right|\leq C \left[1+d\left(\operatorname{Supp} \left( \varphi\right) ,X \right)^{-s} \right]\parallel \varphi \parallel_k^l\mbox{,}
	\end{equation}
for every $\varphi$ in $\mathcal{I}\left( X,\mathbb{R}^d\right)$.
\end{lem}
\begin{proof}
We have to show that the following statements are equivalent:
	\begin{enumerate}
		\item[(a)]For any compact subset  $K$ of $\mathbb{R}^d$  there is a pair of constants $(C,s)$ in $\mathbb{R}^2_{\geq 0}$ and a seminorm $ \parallel \cdot \parallel_k^l$ such that
		\begin{equation*}
			\left| \braket{t,\varphi}\right| \leq C\left[ 1+d\left( \operatorname{Supp} \left( \varphi\right) ,X\right) ^{-s}\right] \parallel \varphi \parallel_k^l\mbox{,}
		\end{equation*}
	for every $\varphi$ in $\mathscr{D}_K\left( \mathbb{R}^d\setminus X\right) $.
	\item[(b)] There exist a pair of constants $ (C,s)$ in $\mathbb{R}_{\geq 0}^2$  and a seminorm $\parallel \cdot \parallel_k^l$ such that 
		\begin{equation*}
 			\left|\braket{t,\varphi} \right|\leq C \left[1+d\left(\operatorname{Supp} \left(  \varphi\right) ,X \right)^{-s} \right]\parallel \varphi \parallel_k^l\mbox{,}
		\end{equation*}
	for every $\varphi$ in $\mathcal{I}\left( X,\mathbb{R}^d\right) $.
	\end{enumerate}
To prove that \textit{(b)} implies \textit{(a)}, take a compact subset  $K$ of $\mathbb{R}^d$.
As
	\begin{equation*}
		\mathscr{D}_K\left( \mathbb{R}^d\setminus X\right) \subseteq \mathcal{I}\left( X,\mathbb{R}^d\right)\mbox{,}
	\end{equation*}
we have, by item \textit{(b)}, that there is a pair of constants $(C,s)$ in $\mathbb{R}^2_{\geq 0}$ and a seminorm $\parallel \cdot \parallel_k^l $ such that
	\begin{equation*}
		\left| \braket{t,\varphi}\right| \leq C\left[ 1+d\left( \operatorname{Supp} \left( \varphi\right) ,X\right) ^{-s}\right]\parallel \varphi \parallel_k^l\mbox{,}
	\end{equation*}
for every $\varphi$ in  $\mathscr{D}_K\left( \mathbb{R}^d\setminus X\right) $.

To prove that \textit{(a)} implies \textit{(b)}, set
	\begin{equation*}
		\mathcal{K}:=\operatorname{Supp}\left( t\right).
	\end{equation*}
Let $\varphi$ be in $\mathcal{I}\left( X,\mathbb{R}^d\right) $ and take  $K$ a compact subset such that $\mathcal{K}$ is contained in $K^{\circ}$. 
For this $K$,  \textit{(a)} implies there is a seminorm $\parallel \cdot \parallel_k^l$ and a pair of constants $(C,s)$ in $\mathbb{R}^2_{\geq 0}$ such that  
	\begin{equation*}
		\left| \braket{t,\psi}\right| \leq C\left[ 1+d\left( \operatorname{Supp} \left( \psi\right) ,X\right) ^{-s}\right] \parallel \psi \parallel_k^l\mbox{,}
	\end{equation*}
for every $\psi$ in $\mathscr{D}_K\left( \mathbb{R}^d\setminus X\right) $.
Take a bump function $\chi$ in  $\mathscr{D}_K\left( \mathbb{R}^d\right) $ such that $\chi=1$ on $\mathcal{K}$. 
Then, the above inequality applies to $\chi \varphi$, which belongs to $\mathscr{D}_K\left( \mathbb{R}^d\setminus X\right) $, and we get
	\begin{equation*}
		\left| \braket{t,\varphi}\right| =\left| \braket{t,\chi \varphi}\right| \leq C\left[ 1+d\left( \operatorname{Supp} \left( \chi\varphi\right) ,X\right) ^{-s}\right]\parallel \chi \varphi \parallel_k^l.
	\end{equation*}
Moreover, as 		
	\begin{equation*}
		\operatorname{Supp}\left(  \chi\varphi\right) = \operatorname{Supp}\left(  \chi \right) \cap \operatorname{Supp}\left(  \varphi\right) \subseteq \operatorname{Supp}\left( \varphi\right)\mbox{,}
	\end{equation*}
we have
	\begin{equation*}
		d\left( \operatorname{Supp}\left( \chi\varphi\right) ,X\right) \geq d\left( \operatorname{Supp}\left( \varphi\right) ,X\right)\mbox{,}
	\end{equation*} 
and therefore,
	\begin{equation*}
		d\left( \operatorname{Supp}\left( \chi\varphi\right) ,X\right) ^{-s}\leq d\left( \operatorname{Supp}\left( \varphi\right) ,X\right) ^{-s}.
	\end{equation*}
Then,
	\begin{equation*}
		\left| \braket{t,\varphi}\right| \leq C\left[ 1+d\left( \operatorname{Supp} \left( \varphi\right) ,X\right) ^{-s}\right]\parallel \chi \varphi \parallel_k^l.
	\end{equation*}
Finally, by item \textit{(ii)} of Lemma \ref{mafi}, there exists a positive constant $M$, independent of $\varphi$, such that
	\begin{equation*}
		\parallel \chi \varphi \parallel_k^l\leq M \parallel  \varphi \parallel_k^l.
	\end{equation*}
Thus,
	\begin{equation*}
		\left| \braket{t,\varphi}\right| \leq CM\left[ 1+d\left( \operatorname{Supp} \left( \varphi\right) ,X\right) ^{-s}\right]\parallel \varphi \parallel_k^l.
	\end{equation*}
The lemma is thus proved.
\end{proof}
Now we turn to show in Proposition \ref{mendelsshon} a very strong  result that will enable us to define a continuous extension of a compactly supported distribution with moderate growth, namely  estimate \eqref{modgro2} is satisfied with $s=0$.
To achieve this we will need the following lemma.

We recall we denote by $\mathcal{I}^\infty \left( X,\mathbb{R}^d\right) $ (resp. $\mathcal{I}^m\left( X,\mathbb{R}^d\right) $) the space of $\mathcal{C}^\infty$ (resp. $\mathcal{C}^m$) functions which vanish on $X$ together with all of their derivatives (resp. all of their derivatives of order less than or equal to $m$). 
In addition we observe that we have the inclusions (see Definition \ref{funcinf1}): 
	\begin{equation*}
		\mathcal{I}\left( X,\mathbb{R}^d\right) \subseteq \mathcal{I}^\infty\left( X,\mathbb{R}^d\right) \subseteq \dots \subseteq \mathcal{I}^{m+1}\left( X,\mathbb{R}^d\right) \subseteq \mathcal{I}^m\left( X,\mathbb{R}^d\right)  \dots
	\end{equation*}

\begin{lem}[\textit{cf.} \cite{vietdang}, Lemma 1.1]\label{technical}
For every pair of positive integers $n$ and $m$, consider the closed ideal $\mathcal{I}^{m+n}\left( X,\mathbb{R}^d\right) $. 
Let $K$  be a compact subset of $\mathbb{R}^d$. 
Then, there is a family of functions $\chi_\lambda$ in $\mathscr{E}\left( \mathbb{R}^d\right) $, parametrized by $\lambda$ in $(0,1]$, which satisfies
	\begin{itemize}
		\item[(i)]$\chi_\lambda(x)=1$ if $d(x,X)\leq\frac{\lambda}{8}$,
		\item[(ii)]$\chi_\lambda(x)=0$ if $d(x,X)\geq\lambda$, and
		\item[(iii)] there is a constant $\tilde{C}\geq 0$ such that, for every $\lambda$  in $(0,1]$ and for every $\varphi$ in $\mathcal{I}^{m+n}\left( X,\mathbb{R}^d\right)$,
	 		\begin{equation*}
				\parallel \chi_\lambda \varphi \parallel_{m}^{K}\leq \tilde{C}\lambda^n \parallel \varphi \parallel_{m+n}^{K\cap \left\lbrace x: d(x,X)\leq \lambda\right\rbrace }\mbox{,}
			\end{equation*}
	\end{itemize}
and the constant $\tilde{C}$ depends neither on $\varphi$ nor $\lambda$ (for the description of the seminorms $\parallel \cdot \parallel_{m}^{K}$ see Definition \ref{kungfu}).
\end{lem}
\begin{proof}
Choose $\phi$ in $\mathscr{E}\left( \mathbb{R}^d\right) $ such that
	\begin{enumerate}
		\item $\phi\geq 0$,
		\item $\int\limits_{\mathbb{R}^d}\phi(x) \,\mathrm{d}x=1$, and
		\item $\phi(x)=0$ if $|x|\geq 3/8$.
	\end{enumerate} 
Set 
	\begin{equation*}
		\phi_\lambda:=\lambda^{-d}\phi(\lambda^{-1}\cdot) 
	\end{equation*}
and let $\alpha_\lambda$ be the characteristic function of the set 
	\begin{equation*}
		\left\lbrace x: d\left( x,X\right)\leq \frac{\lambda}{2} \right\rbrace.
	\end{equation*} 
Then, the convolution product $\chi_\lambda\equiv\phi_\lambda \ast \alpha_\lambda$ satisfies
	\begin{equation}\label{vir}
		\chi_\lambda(x)\equiv\phi_\lambda \ast \alpha_\lambda(x) = \left\{
			\begin{array}{lr}
				1\mbox{,} & \mbox{if } d\left(x,X \right)\leq \frac{\lambda}{8}\mbox{,} \\0\mbox{,} & \mbox{if } d\left(x,X \right)\geq \lambda. 
			\end{array}
		\right.
	\end{equation}
We show this last statement.  
By definition, we have
	\begin{equation}\label{fosforo}
		\begin{split}
			\phi_\lambda   \ast\alpha_\lambda(x)&= \int\limits_{\mathbb{R}^d}      \phi_\lambda (x-y)\alpha_\lambda (y)\, \mathrm{d}y = \int\limits_{\left\lbrace y: d(y,X)\leq \frac{\lambda}{2}\right\rbrace }\phi_\lambda (x-y)\, \mathrm{d}y \\&=    \int\limits_{\left\lbrace y: d(y,X)\leq \frac{\lambda}{2}\right\rbrace } \lambda^{-d}\phi\left( \lambda^{-1}\left( x-y\right) \right)\mathrm{d}y =\int\limits_{\Omega_{x,\lambda}}              \phi\left( z\right) \mathrm{d}z\mbox{,}
		\end{split}
	\end{equation}
where 
	\begin{equation*}
		\Omega_{x,\lambda}=\left\lbrace z:\,\exists\,y \mbox{ such that }\; z=(x-y)\lambda^{-1}\;\mbox{ and }\;\; d(y,X)\leq \frac{\lambda}{2}\right\rbrace. 
	\end{equation*}
Suppose $d\left(x,X \right)\leq \lambda/8$.
We want to show that the last integral in \eqref{fosforo} is equal to one.
For this, it suffices to  prove that if $|z|\leq  3/8$, then $z$ belongs to $\Omega_{x,\lambda}$; if this was the case,  using the third condition satisfied by $\phi$ we would have
	\begin{equation*}
		1\geq\int\limits_{\Omega_{x,\lambda}}              \phi\left( z\right) \mathrm{d}z \geq	\int\limits_{\left\lbrace z:  \left| z\right| \leq 3/8\right\rbrace }\phi(z) \,\mathrm{d}z=\int\limits_{\mathbb{R}^d}\phi(z) \,\mathrm{d}z=1\mbox{,}
	\end{equation*}
which implies that the last integral in \eqref{fosforo} is equal to one.
So, suppose $d\left(x,X \right)\leq \lambda/8$ and  $|z|\leq  3/8$.
Let $y$ be such that $z=(x-y)\lambda^{-1}$ (such an element $y$ always exists given $z$ and $x$). 
We seek to show $d(y,X)\leq \lambda/2$.
For every $\tilde{x}$ in $X$ we have
	\begin{equation*}
		d(y,X)\leq |y-\tilde{x}|\leq|y-x|+|x-\tilde{x}|\leq \frac{3}{8}\lambda+|x-\tilde{x}|.
	\end{equation*}
Then, taking infimum over all the elements $\tilde{x}$ in $X$ on the right hand side of the above inequality we get
	\begin{equation*}
		d(y,X)\leq \frac{3}{8}\lambda+ d(x,X)\leq \frac{3}{8}\lambda+ \frac{\lambda}{8}=\frac{\lambda}{2}.
	\end{equation*}
Therefore, if $d\left(x,X \right)\leq \lambda/8$, we have  $\phi_\lambda \ast \alpha_\lambda(x)=1.$

Now suppose that $d\left(x,X \right)\geq \lambda$. 
Take $z=(x-y)\lambda^{-1}$ in $\Omega_{x,\lambda}$  where  $d(y,X)\leq \lambda/2$. 
Then, for every $\tilde{x}$ in $X$
	\begin{equation*}
		|z|\geq \frac{|x-\tilde{x}|}{\lambda}-\frac{|y-\tilde{x}|}{\lambda}\geq  \frac{d(x,X)}{\lambda}-\frac{|y-\tilde{x}|}{\lambda} \geq  1-\frac{|y-\tilde{x}|}{\lambda}.
	\end{equation*}
Taking infimum over all the elements $ \tilde{x}$ in $X$ in the above inequality we get
	\begin{equation*}
		|z|\geq  1-\frac{d(y,X)}{\lambda} \geq 1-\frac{1}{2} \geq\frac{1}{2} > \frac{3}{8}.
	\end{equation*}
Therefore, if $d\left(x,X \right)\geq \lambda$, $\Omega_{x,\lambda}$ is contained in $\left\lbrace z: \left| z\right| > 3/8 \right\rbrace $.
Then, we have
	\begin{equation*}
		\phi_\lambda \ast \alpha_\lambda(x)=\int\limits_{\Omega_{x,\lambda}} \phi\left( z\right) \;\mathrm{d}z\leq \int\limits_{\left\lbrace z:  \left| z\right| > 3/8\right\rbrace} \phi\left( z\right) \;\mathrm{d}z=0.
	\end{equation*}
In order to prove the third item of the lemma, let $\varphi$ be in  $\mathcal{I}^{m+n}\left( X,\mathbb{R}^d\right) $. 
By Leibniz's rule,
	\begin{equation*}
		\partial_x^\nu\left( \chi_\lambda \varphi\right)(x) =\sum_{|i|\leq |\nu|} \binom{\nu}{i}\left( \partial_x^i\chi_\lambda \partial_x^{\nu-i}\varphi\right) (x)\mbox{,}
	\end{equation*}
for each $|\nu|\leq m$. 
Then, it suffices to estimate each term $\left( \partial^i_x\chi_\lambda \partial^{\nu-i}_x\varphi\right) (x)$ (with $|i|\leq |\nu|\leq m$) of the above sum.  
For every multi-index $i$, there is some constant $C_i$ such that for every $x$ in $\mathbb{R}^d\setminus X$,
	\begin{equation}
		\quad|\partial^i_x\chi_\lambda(x)|\leq \frac{C_i}{\lambda^{|i|}}.
	\end{equation}
To prove this, set 
	\begin{equation*}
		B:=\left\lbrace y:d(y,X)\leq \frac{\lambda}{2}\right\rbrace\mbox{,}
	\end{equation*}
and note 
	\begin{equation*}
		\begin{split}
			\left| \partial^i_x\chi_\lambda(x)\right| &= \left|  \partial^i_x\left( \;  \int\limits_{\mathbb{R}^d} \phi_\lambda\left( x-y\right)\alpha_\lambda(y) \;\mathrm{d}y\right) \right| = \left|  \partial^i_x\left(  \,\int\limits_{B } \lambda^{-d}\phi\left( \lambda^{-1}\left( x-y\right) \right) \;\mathrm{d}y\right) \right| \\ &= \left|\,\int\limits_{B} \lambda^{-d}\lambda^{-|i| }\left( \partial^i\phi\right) \left( \lambda^{-1}\left( x-y \right) \right) \;\mathrm{d}y \right|= \left|\,\int\limits_{\Omega_{x,\lambda}} \lambda^{-|i|} \partial^i\phi\left( z\right) \;\mathrm{d}z \right| \\&\leq \lambda^{-|i|}\int\limits_{\Omega_{x,\lambda}}\left| \partial^i\phi\left( z\right) \right| \mathrm{d}z \leq\lambda^{-|i|}\underbrace{ \int\limits_{\mathbb{R}^d}\left| \partial^i\phi\left( z\right) \right| \mathrm{d}z}_{C_i} .
		\end{split}  
	\end{equation*}

To estimate the other factor, namely  $\partial^{\nu-i}\varphi(x)$, we can assume $x$ to be in the support of $\partial^i\chi_\lambda \partial^{\nu-i}\varphi$ (otherwise the whole term is zero).
As
	\begin{equation*}
		\operatorname{Supp}\left(\partial^i\chi_\lambda \partial^{\nu-i}\varphi \right)\subseteq \operatorname{Supp}\left(\partial^i\chi_\lambda  \right)\subseteq \left\lbrace u: d\left(u,X \right)\leq \lambda \right\rbrace\mbox{,} 
	\end{equation*}
where result \eqref{vir} has been used, we have $d\left(x,X \right)\leq \lambda$.
Fix $y$ in $X$ such that $d(x,X)=|x-y|$.
As $\varphi$ belongs to $\mathcal{I}^{m+n}\left( X,\mathbb{R}^d\right) $, $\partial^{\nu-i}\varphi $ vanishes at $y$ at order $i+n$.
Therefore,
	\begin{equation*}
		\partial^{\nu-i}_x \varphi(x)=\sum\limits_{|\beta|=|i|+n}(x-y)^{\beta}R_\beta (x)\mbox{,}
	\end{equation*}
where the right hand side is just the integral remainder in Taylor's expansion of $\partial^{\nu-i} \varphi$ around $y$. 
At this point, note that by all our assumptions we have
	\begin{equation*}
		\left| x-y\right| ^{\beta}=d(x,X)^{\beta}\leq \lambda^{\beta}.
	\end{equation*}
Hence,
 \begin{equation*}
 		\left|\partial^{\nu-i}\varphi(x)\right|\leq \sum\limits_{|\beta|=|i|+n}\left| (x-y)^{\beta}R_\beta (x) \right|\leq  \lambda  ^{|i|+n}\sum\limits_{|\beta|=|i|+n}\left| R_\beta (x) \right| .
 	\end{equation*}
The estimations of both factors can be put together to obtain
	\begin{equation*}
		\left|\partial^i\chi_\lambda \partial^{\nu-i}\varphi(x)\right|\leq\frac{C_i}{\lambda^{|i|}}\lambda  ^{|i|+n}\sum\limits_{|\beta|=|i|+n}\left| R_\beta (x) \right|\leq C_i \lambda^{n} \sum\limits_{|\beta|=|i|+n}\left|R_\beta (x) \right|.
	\end{equation*}
It is easy to see that $R_\beta$ only depends on the jets of $\varphi$ of order less than or equal to $m+n$, since
	\begin{equation*}
		\sum\limits_{|\beta|=|i|+n}\left|R_\beta (x) \right|\leq \sum\limits_{|\beta|=|i|+n}\quad\sum\limits_{|\nu|=|\beta|}\frac{|\beta|}{\nu!}\parallel\varphi\parallel_{\nu}^{K\cap \left\lbrace u:\, d(u,X)\leq \lambda \right\rbrace} \leq A_i\parallel\varphi\parallel_{m+n}^{K\cap \left\lbrace u:\, d(u,X)\leq \lambda \right\rbrace}\mbox{,}
	\end{equation*}
for some constant $A_i$ which does not depend on $\varphi$.
Then,
	\begin{equation*}
		\left|\partial^i\chi_\lambda \partial^{\nu-i}\varphi(x)\right|\leq		  C_i \lambda^{n} A_i\parallel\varphi\parallel_{m+n}^{K\cap \left\lbrace u:\, d(u,X)\leq \lambda \right\rbrace}.
	\end{equation*}
Finally, we have
	\begin{equation*}
		\begin{split}
			\left| \partial^\nu\left( \chi_\lambda \varphi\right)(x)\right| & \leq \sum_{|i|\leq |\nu|} \binom{\nu}{i}\left| \partial^i\chi_\lambda \partial^{\nu-i}\varphi(x)\right| \\ &\leq  \left[ \sum_{|i|\leq |\nu|} \binom{\nu}{i} C_i A_i\right] \lambda^n \parallel\varphi\parallel_{m+n}^{K\cap \left\lbrace u:\, d(u,X)\leq \lambda \right\rbrace}\\ &\leq  \tilde{C} \lambda^n \parallel\varphi\parallel_{m+n}^{K\cap \left\lbrace u:\, d(u,X)\leq \lambda \right\rbrace}\mbox{,}
		\end{split}
	\end{equation*}
where $\tilde{C}$ is given by
	\begin{equation*}
		\tilde{C}=\max\limits_{\nu\leq m}\left\lbrace \sum_{|i|\leq |\nu|} \binom{\nu}{i} C_i A_i\right\rbrace.
	\end{equation*}
Taking supremum over all the elements $x$ in $K$ and $\nu\leq m$ on the left hand side of the previous inequality, we obtain the required result.
\end{proof}

Using the previous lemma we can now prove the following result.

\begin{propo}[\textit{cf.} \cite{vietdang}, Thm. 1.1]\label{mendelsshon} 
Let $t$ be a compactly supported distribution defined on $\mathscr{D}\left( \mathbb{R}^d\setminus X\right)$, viewed in the space $\mathscr{E}\left( \mathbb{R}^d\setminus X\right)'$ by Remark \ref{hatcher}, and with moderate growth along $X$.
Let  $\{K_l\}_{l\in\mathbb{N}}$ be a fundamental sequence of compact sets covering $\mathbb{R}^d$ used to construct a family of seminorms as in Definition \ref{defiseminorms}. 
Then, there exist a constant $C$ in $\mathbb{R}_{\geq 0}$ and a seminorm  $\parallel \cdot \parallel_{k}^{l}$ such that 
	\begin{equation}\label{elsa}
		\left|\braket{t,\varphi} \right|\leq C \parallel \varphi \parallel_{k}^{l}\mbox{,}
	\end{equation}
for every $\varphi$ in $\mathcal{I}\left( X,\mathbb{R}^d\right)$ (see Definition \ref{funcinf1}).
\end{propo}
\begin{proof}
Suppose $t$ has moderate growth along $X$.
By Lemma \ref{dumbo}, there exist positive numbers $C$ and $s$, and a seminorm $\parallel \cdot \parallel_k^{l}$ such that
	\begin{equation}\label{democrat}
		\left|\braket{t,\varphi} \right|\leq C \left[1+d\left(\operatorname{Supp} \left(  \varphi\right) ,X \right)^{-s} \right]\parallel \varphi \parallel_k^l\mbox{,}
	\end{equation}
for every $\varphi$ in  $\mathcal{I}\left( X,\mathbb{R}^d\right)$.

If $s=0$ in \eqref{democrat}, there is nothing to prove.
Therefore, we shall treat the case where $s>0$. 
The idea is to absorb the divergence by a dyadic decomposition in the following way.
We begin by writing
	\begin{equation}\label{republican}
		\braket{t,\varphi} =\braket{t,\left( \chi_{1}-\chi_{1}+1\right) \varphi}          =\braket{t,\chi_{1}\varphi}+\braket{t,\left( 1-\chi_{1}\right) \varphi}.
	\end{equation}
We are going to estimate each of the terms in \eqref{republican}.
We can easily estimate the term  $\braket{t,\left(1-\chi_{1}\right)\varphi}$. 
Set
	\begin{equation*}
		\mathcal{K}:=\operatorname{Supp}\left( t\right).
	\end{equation*}
Let $\chi$ in $\mathscr{D}\left( \mathbb{R}^d\right) $ be  such that $\chi=1$ on $\mathcal{K}$ and $\chi=0$ outside $V$, where $V$ is a bounded  open set that contains $\mathcal{K}$.
Then, $\chi\left(1-\chi_1 \right)\varphi$ belongs to $\mathscr{D}\left( \mathbb{R}^d\setminus X\right) $ and, by continuity of $t$, there is some constant $C>0$ and a seminorm $\parallel \cdot \parallel^{l_1}_{k_1}$ such that
	\begin{equation*}
		\left| \braket{t,\left(1-\chi_1 \right)\varphi}\right| = \left| \braket{t,\chi\left(1-\chi_1 \right)\varphi}\right| \leq C\parallel \chi\left(1-\chi_1 \right)\varphi \parallel^{l_1}_{k_1}
\end{equation*}
for every $\varphi$ in $\mathscr{E}\left( \mathbb{R}^d\right)$. 
On the other hand, from Lemma \ref{mafi} there exists some  positive constant $M$ which does not depend on  $\varphi$ and such that 
	\begin{equation*}
		\parallel \chi\left(1-\chi_1 \right)\varphi \parallel^{l_1}_{k_1}\leq M \parallel \varphi \parallel^{l_1}_{k_1}.
	\end{equation*} 
Taking $C_1=CM$ we conclude that 
\begin{equation*}
\left| \braket{t,\left(1-\chi_{1}\right)\varphi}\right| \leq C_1\parallel \varphi \parallel^{l_1}_{k_1}
\end{equation*}
for some seminorm $\parallel \cdot \parallel^{l_1}_{k_1}$ and some constant $C_1>0$ which does not depend on $\varphi$. 

To estimate the remaining term in \eqref{republican}, for every $\varphi$ in $\mathcal{I}\left( X,\mathbb{R}^d\right) $ there exists a number $N$ in $\mathbb{N}$ such that  $\chi_{2^{-N}}\varphi=0$.
Then, we can write
	\begin{equation*}
		\begin{split}
			\braket{t,\chi_{1}\varphi}&=\braket{t,\left( \chi_{1}-\chi_{2^{-1}}+\chi_{2^{-1}}-\dots-\chi_{2^{-N+1}}+\chi_{2^{-N+1}}-\chi_{2^{-N}}\right) \varphi}\\ &=\sum_{j=0}^{N-1} \braket{t,\left(\chi_{2^{-j}}-\chi_{2^{-j-1}}\right) \varphi}. 
		\end{split}
	\end{equation*}
Thus,
	\begin{equation*}
		\begin{split}
			\left| \braket{t,\chi_{1}\varphi}\right| &\leq \sum_{j=0}^{N-1}\left| \braket{t, \left(\chi_{2^{-j}}-\chi_{2^{-j-1}}\right) \varphi}\right|  \\
			&\leq C_2 \sum_{j=0}^{N-1}\left[1+d \left( \operatorname{Supp}\left(\left( \chi_{2^{-j}}-\chi_{2^{-j-1}}\right) \varphi\right) ,X \right) ^{-s}\right] \parallel \left( \chi_{2^{-j}}-\chi_{2^{-j-1}}\right)\varphi\parallel _{k_2}^{l_2}  
		\end{split}
	\end{equation*}
where we have used condition \eqref{democrat}.
Now choose $n$ in $\mathbb{N}$ such that  $n-s>0$. 
Applying Lemma \ref{technical} with $K=K_{l_2}$ we have
	\begin{equation*}
		\begin{split}
			\left| \braket{t,\chi_{1}\varphi}\right|&\leq C_2 \sum_{j=0}^{N-1}\left(1+2^{s(j+4)}\right)\left(2^{-jn}+ 2^{-(j+1)n}\right)\tilde{C}_2 \parallel \varphi\parallel _{k_2+n}^{l_2}  \\
			& \leq C'_2 \parallel \varphi \parallel _{k_2+n}^{l_2}\mbox{,}
		\end{split}
	\end{equation*}
for $C'_2$ defined by
	\begin{equation*}
		C'_2= \tilde{C}_2C_2\left( 1+2^{-n}\right) \sum\limits_{j=0}^{N-1} 2^{-jn}\left(1+2^{s(j+4)}\right)<+\infty\mbox{,}
	\end{equation*}
which is independent of $\varphi$ and $N$.
To complete the proof we finally write
	\begin{equation*}
		\left| \braket{t,\varphi}\right|  \leq\left| \braket{t,\chi_{1}\varphi}\right| +\left| \braket{t,\left( 1-\chi_{1}\right) \varphi}\right|  \leq C'_2 \parallel \varphi \parallel _{k_2+n}^{l_2}+C_1\parallel \varphi \parallel^{l_1}_{k_1}\leq  C'\parallel \varphi \parallel _{k'}^{l'}\mbox{,}
	\end{equation*}
where $l'=\max\left\lbrace l_1,l_2\right\rbrace$,  $\ k'=\max\left\lbrace k_1,k_2+n \right\rbrace$ and $\ C'=C'_2+C_1$.
\end{proof}
We are going to prove in Lemma \ref{zappa}  that given a compactly supported distribution $t$  on $\mathscr{D}\left( \mathbb{R}^d\setminus X\right)$ with moderate growth along $X$, and a number $m$ in $\mathbb{N}\cup \left\lbrace\infty \right\rbrace $, there is a unique continuous extension of $t$ to the space $\mathcal{I}^m\left( X,\mathbb{R}^d\right)$.
Part of this result appears in Lemma 1.2  of \cite{vietdang}.
However, the proof given by the author is unsatisfactory because he omits  several steps.
For instance, the author states that the mentioned extension  is unique but does not show this fact in his proof.
We give in turn a more detailed demonstration.
In addition, we extend the scope of Lemma 1.2 of \cite{vietdang} to include the case $m=\infty$.

To this end, we first give the definition of Dirac sequences and in Theorem \ref{aproxdirac} we state a  useful result related to them.  
\begin{defi}[See \cite{langlang}, Ch. 11, \S 1]\label{dirac}
A sequence $\left(   \phi_i \right)  _{i\in\mathbb{N}}$ of complex valued functions defined on all $\mathbb{R}^d$ is called a \emph{Dirac sequence}\index{Dirac sequence} if it satisfies the following properties:
	\begin{itemize}
		\item[DIR 1.]$\phi_i\geq 0$ for every $i$ in $\mathbb{N}$ and for every $x$ in $\mathbb{R}^d$.
		\item[DIR 2.]Each $\phi_i$ is continuous, and 
			\begin{equation*}
				\int\limits_{\mathbb{R}^d}\phi_i(t)\mbox{d}t=1. 
			\end{equation*}
		\item[DIR 3.]Given $\varepsilon>0$ and $\delta>0$, there exists $i_0$ such that if $i\geq i_0$ then 
			\begin{equation*}
				\int\limits_{\mathbb{R}^d\setminus B_{\delta}(0)}\phi_i(t)\mbox{d}t <\varepsilon.
			\end{equation*}
	\end{itemize}
\end{defi}
\begin{thm}[See \cite{langlang}, Ch. 11, Thm. 1.1]\label{aproxdirac}
Let $f:\mathbb{R}^d\rightarrow \mathbb{R}$ be a piecewise continuous function, and assume that $f$ is bounded. 
Let $\left( \phi_i \right)_{i\in\mathbb{N}}$ be a Dirac sequence and for each $i$ in $\mathbb{N}$ define $f_i = \phi_i * f$. 
Let $K$ be a compact subset of $\mathbb{R}^d$ on which $f$ is continuous. 
Then the sequence $\left(    f_i \right)_{i\in\mathbb{N}} $ converges to $f$ uniformly on $K$.
\end{thm}

\begin{lem}[\textit{cf.}\cite{vietdang}, Lemma 1.2] \label{zappa}
Let $m$ be a positive integer  or eventually infinity.
Let $t$ be a compactly supported distribution defined on $\mathscr{D}\left( \mathbb{R}^d\setminus X\right)$, viewed in the space $\mathscr{E}\left( \mathbb{R}^d\setminus X\right)'$ by Remark \ref{hatcher}, and with moderate growth along $X$.
Then, there is a unique continuous extension $t_m$ to the space $\mathcal{I}^m\left( X,\mathbb{R}^d\right)$.
If $m$ belongs to $\mathbb{N}$, such an extension is given on every $\varphi$ in $\mathcal{I}^m\left( X,\mathbb{R}^d\right)$ by the  iterated limit
	\begin{equation}\label{milka}
		 \braket{t_m,\varphi}=\lim\limits_{\lambda\rightarrow 0}\lim \limits_{i\rightarrow \infty} \braket{t,  \left(1-\chi_\lambda \right)\phi_i\ast \varphi }\mbox{,}
	\end{equation}
where $\left( \chi_\lambda\right)_{\lambda\in\left( 0,1\right] }$ is the family of cut-off functions defined in Lemma \ref{technical}  and $\left(  \phi_i \right)_{i\in\mathbb{N}} $ is a Dirac sequence of functions defined on $\mathbb{R}^d$.

Moreover, if $\varphi$ belongs to $\mathcal{I}^m\left( X,\mathbb{R}^d\right) \cap \mathscr{E}\left( \mathbb{R}^d\right)  $, the continuous extension $t_m$ is given by 
	\begin{equation}\label{milka2}
		\braket{t_m,\varphi}=\lim\limits_{\lambda\rightarrow 0}\braket{t,  \left(1-\chi_\lambda \right) \varphi }.
	\end{equation}
This means that the evaluation of  $t_m$ on such a function $\varphi$ can be obtained without the need of a Dirac sequence.
\end{lem}
\begin{proof}
Consider first that $m$ belongs to $\mathbb{N}$. 
To show that \eqref{milka} is well-defined we begin by proving the following convergences with respect to the topology of $\mathscr{E}^m\left( \mathbb{R}^d\right) $:
	\begin{itemize}
		\item[(i)] Let $\left(  \phi_i \right)_{i\in\mathbb{N}} $ be a Dirac sequence; then we have
			\begin{equation*}
				\lim\limits_{i \rightarrow \infty}\left( 1-\chi_\lambda\right)\phi_i\ast \varphi= \left( 1-\chi_\lambda\right) \varphi
			\end{equation*}
		for a  fixed value of $\lambda$ in $(0,1]$ and a fixed function $\varphi$ in $\mathscr{E}^m\left( \mathbb{R}^d\right) $.
		\item[(ii)]$\lim\limits_{\lambda\rightarrow 0}\left( 1-\chi_\lambda\right)\varphi=\varphi$ for a fixed function $\varphi$ in $\mathcal{I}^m\left( X,\mathbb{R}^d\right) $.
		\item[(iii)] $\lim\limits_{\lambda\rightarrow 0}\lim \limits_{i\rightarrow \infty}\left(1-\chi_\lambda \right)\phi_i\ast \varphi =\varphi$ for a fixed function $\varphi$ in $\mathcal{I}^m\left( X,\mathbb{R}^d\right) $.
	\end{itemize}
The proof of these statements is as follows.
Let $K$ be a compact subset of $\mathbb{R}^d$ and $k$ a nonnegative integer. 
Then, for every $\varphi$ in $\mathscr{E}^m\left( \mathbb{R}^d\right) $, each multi-index  $\alpha$ such that $|\alpha|\leq k$, and every $x$ in $K$ we have
	\begin{equation*}
		\begin{split}
			\left| \partial^{\alpha}_x\left[ \left( 1-\chi_\lambda\right)  \phi_i\ast\varphi\right.\right.&-\left.\left.\left( 1-\chi_\lambda\right)  \varphi\right]\right(x)| = \left| \partial^{\alpha}_x\left[ \left( 1-\chi_\lambda\right)\left(   \phi_i\ast\varphi- \varphi\right) \right](x)\right|\\ & \leq\sum_{|j|\leq |\alpha|}\binom{\alpha}{j}\left|\left[  \partial_x^{\alpha-j}\left( 1-\chi_\lambda\right)\partial_x^{j}\left(   \phi_i\ast\varphi- \varphi\right)\right] (x)\right|\\ &\leq \sum_{|j|\leq |\alpha|}\binom{\alpha}{j} \parallel 1-\chi_\lambda\parallel _k^K \left| \partial_x^{j}\left(   \phi_i\ast\varphi- \varphi\right)(x)\right|\\ &= \parallel 1-\chi_\lambda\parallel _k^K \sum_{|j|\leq |\alpha|}\binom{\alpha}{j}\left|\left(  \phi_i\ast\partial^{j}\varphi- \partial^{j}\varphi\right) (x)\right|.
		\end{split}
	\end{equation*}
From Theorem \ref{aproxdirac} it immediately follows that 
	\begin{equation*}
		\phi_i\ast\partial^{j}\varphi- \partial^{j}\varphi \underset{i\rightarrow \infty}{\longrightarrow} 0
	\end{equation*}
uniformly on $K$. 
The previous discussion implies that 
	\begin{equation*}
		\partial^{\alpha}_x\left[ \left( 1-\chi_\lambda\right)  \phi_i\ast\varphi-\left( 1-\chi_\lambda\right)  \varphi\right] \underset{i\rightarrow \infty}{\longrightarrow}0
	\end{equation*}
uniformly on $K$  for every $|\alpha|\leq k$, from which it follows that
	\begin{equation}\label{blanco}
		\left( 1-\chi_\lambda\right)\varphi = \lim_{i \rightarrow \infty}\left( 1-\chi_\lambda\right)\phi_i\ast\varphi
	\end{equation}
for the topology of $\mathscr{E}^m\left( \mathbb{R}^d\right) $, for every $\varphi$ in $\mathscr{E}^m\left( \mathbb{R}^d\right) $.
In other words, given $\eta >0$, there exists $i_0=i_0(\eta,\lambda,\varphi)$ such that, for every  $i\geq i_0$,
	\begin{equation*}
		\parallel \left( 1-\chi_\lambda\right)  \phi_i\ast\varphi-\left( 1-\chi_\lambda\right)  \varphi\parallel_k^K<\eta.
	\end{equation*}

For the second limit in the variable $\lambda$ we proceed as follows. 
By item \textit{(iii)} of Lemma \ref{technical} , we have that 
	\begin{equation*}
		 \parallel \chi_\lambda \varphi \parallel ^{K}_k\leq \parallel \chi_\lambda \varphi \parallel ^{K}_m \leq \tilde{C}\parallel \varphi\parallel_m^{K\cap \left\lbrace d(x,X)\leq \lambda\right\rbrace }\underset{\lambda\rightarrow 0}{\longrightarrow} 0\mbox{,}
	\end{equation*}
for every $\varphi$ in $\mathcal{I}^m\left( X,\mathbb{R}^d\right) $ and every nonnegative integer $k$, with $k\leq m$.
Therefore, 
	\begin{equation}\label{negro}
		\varphi = \lim_{\lambda\rightarrow 0}\left( 1-\chi_\lambda\right)\varphi
	\end{equation}
for the topology of $\mathscr{E}^m\left( \mathbb{R}^d\right) $, for every  $\varphi$ in $\mathcal{I}^m\left( X,\mathbb{R}^d\right) $.
In other words, given $\eta>0$ there exists $\lambda_0=\lambda_0(\eta,\varphi)$ such that
	\begin{equation*}
		\left\| \left( 1-\chi_\lambda\right)\varphi-\varphi\right\|^K_k =\left\|\chi_\lambda\varphi\right\|^K_k<\eta\mbox{,}
	\end{equation*}
for every $\lambda\leq \lambda_0$.

We turn now to prove the third limit.
Let $\varphi$ be in $\mathcal{I}^m\left( X,\mathbb{R}^d\right) $ and $\eta>0$.
Then if $\lambda<\lambda_0\left( \frac{\eta}{2},\varphi\right) $ and $i>i_0\left(\frac{\eta}{2},\varphi,\lambda\right) $ we obtain
	\begin{equation*}
		\begin{split}
			\left\| \varphi-\left(1-\chi_{\lambda}\right) \phi_i\ast\varphi\right\| _k^K& \leq\left\| \varphi-\left(1-\chi_{\lambda} \right) \varphi\right\| _k^K+\left\| \left(1-\chi_{\lambda} \right) \varphi-\left(1-\chi_{\lambda} \right) \phi_i\ast\varphi\right\| _k^K \\ &\leq  \frac{\eta}{2}+\frac{\eta}{2}=  \eta.
		\end{split}
	\end{equation*}
We have proven
	\begin{equation}
		\varphi=\lim\limits_{\lambda\rightarrow 0}\lim \limits_{i\rightarrow \infty}\left(1-\chi_\lambda \right)\phi_i\ast \varphi
	\end{equation}
with respect to the topology of $\mathscr{E}^m\left( \mathbb{R}^d\right) $, for every $\varphi$ in $\mathcal{I}^m\left( X,\mathbb{R}^d\right) $.
Observe that this result implies $\mathcal{I}\left( X,\mathbb{R}^d\right) $ is dense in $\mathcal{I}^m\left( X,\mathbb{R}^d\right) $. 

The statement of the lemma is  now staightforward.
We shall prove that \eqref{milka} uniquely defines a continuous extension $t_m$ in  $\mathcal{I}^m\left( X,\mathbb{R}^d\right) '$.
We begin by showing the double limit in \eqref{milka} exists.
Take $\left( \lambda_j\right)_{j\in\mathbb{N}} $ such that
	\begin{equation*}
		\lim\limits_{j\rightarrow \infty} \lambda_j =0.
	\end{equation*}
On $ \mathbb{N}\times\mathbb{N}$ we define the partial ordering: $(i,j)\leq (i',j')$ if and only if $ i\leq i'\ $ and $\  j\leq j'$, which turns $\mathbb{N}\times\mathbb{N}$ into a directed set.
As the field $\mathbb{C}$ is complete it suffices to see that
	\begin{equation*}
		\bigg(  \left\langle t,\left(1-\chi_{\lambda_j} \right)\phi_i\ast\varphi\right\rangle \bigg)_{(i,j)\in \mathbb{N}\times\mathbb{N}}
	\end{equation*}
is a Cauchy net.
Take $\eta>0$. If we choose 
	\begin{equation*}
		i\mbox{, }i'> \max\bigg\lbrace i_0\left(\frac{\eta}{4C},\varphi,\lambda_j\right),i_0\left(\frac{\eta}{4C},\varphi,\lambda_{j'}\right) \bigg\rbrace\ \mbox{ and } \  j\mbox{, }j'>j_0
	\end{equation*}
where $j_0$ is such that $\lambda_{j}<\lambda_0\left( \frac{\eta}{4C},\varphi\right)\ $ for every $\ j>j_0$, then
	\begin{equation*}
		\begin{split}
			\left| \left\langle t,\Big(1-\chi_{\lambda_j} \Big)\phi_i\ \ast\right.\right.& \left.\Big.\varphi\Big\rangle -\left\langle t,\left(1-\chi_{\lambda_{j'}}\right)\phi_{i'}\ast\varphi\right\rangle  \right|\\ & = \left| \left\langle t,\Big(1-\chi_{\lambda_j} \Big)\phi_i\ast\varphi-\left(1-\chi_{\lambda_{j'}} \right)\phi_i'\ast\varphi\right\rangle \right|\\ &\leq C \left\|  \Big(1-\chi_{\lambda_j} \Big)\phi_{i}\ast\varphi-\left(1-\chi_{\lambda_{j'}} \right)\phi_{i'}\ast\varphi\right\| _k^K \\ &\leq C\bigg(\left\|  \Big(1-\chi_{\lambda_j} \Big)\phi_{i}\ast\varphi-\varphi\right\| _k^K  +\left\|  \varphi -\left(1-\chi_{\lambda_{j'}} \right)\phi_{i'}\ast\varphi\right\|_k^K \bigg) \\ &\leq C\bigg(\frac{\eta}{2C}+ \frac{\eta}{2C} \bigg)= \eta\mbox{,}\nonumber
		\end{split}
	\end{equation*}
where in the first inequality we have used estimate \eqref{elsa} of Proposition \ref{mendelsshon}.
Set
	\begin{equation*}
		w:=\lim\limits_{\lambda\rightarrow 0}\lim \limits_{i\rightarrow \infty}\left\langle t,\left(1-\chi_\lambda \right)\phi_i\ast \varphi\right\rangle .
	\end{equation*}
Suppose now that $\left( \psi_j\right)_{j\in\mathbb{N}} $ is another sequence of functions in $\mathcal{I}\left( X,\mathbb{R}^d\right) $ that converges with respect to the topology of $\mathscr{E}^m\left( \mathbb{R}^d\right) $  to $\varphi$ in $\mathcal{I}^m\left( X,\mathbb{R}^d\right) $, and set
	\begin{equation*}
		\tilde{w}:=\lim\limits_{j\rightarrow\infty}\braket{t,\psi_j}.
	\end{equation*}
Then,
	\begin{equation*}
		\left| w-\tilde{w}\right|\leq \left|w-\braket{t,\left(1-\chi_\lambda \right)\phi_i\ast \varphi} \right| +\left|\braket{t,\left(1-\chi_\lambda \right)\phi_i\ast \varphi}- \braket{t,\psi_j}\right|+\left| \braket{t,\psi_j}-\tilde{w}\right|   .
	\end{equation*}
It is clear that the first and last terms on the right hand side can be made arbitrarily small by choosing $i$ and $j$ sufficiently large, and $\lambda$ sufficiently small. 
For the term in the middle, note that as $t$ satisfies \eqref{elsa}, there exist some constant $C$ and a seminorm $\parallel\cdot \parallel_k^K$ such that 
	\begin{equation*}
		\begin{split}
			\left|\braket{t,\left(1-\chi_\lambda \right)\phi_i\ast \varphi}- \braket{t,\psi_j}\right|&=\left|\braket{t,\left(1-\chi_\lambda \right)\phi_i\ast \varphi-\psi_j}\right|\\ &\leq C \parallel \left(1-\chi_\lambda \right)\phi_i\ast \varphi-\psi_j \parallel_k^K.
		\end{split}
	\end{equation*}
Furthermore,
	\begin{equation*}
		\parallel \left(1-\chi_\lambda \right)\phi_i\ast \varphi-\psi_j \parallel_k^K\leq \parallel \left(1-\chi_\lambda \right)\phi_i\ast \varphi-\varphi \parallel_k^K+\parallel \varphi-\psi_j \parallel_k^K
	\end{equation*}
and each of the terms on the right hand side can be made arbitrarily small by choosing $i$ and $j$ sufficiently large, and $\lambda$ sufficiently small.
The previous argument shows that $\left| w-\tilde{w}\right|\leq \varepsilon\ $ for every $\ \varepsilon>0$. 
Then $w=\tilde{w}$.

We are now allowed to define for every $\varphi$ in $\mathcal{I}^m\left(X,\mathbb{R}^d\right) $
	\begin{equation*}
		\braket{t_m,\varphi}:=\lim\limits_{j\rightarrow\infty}\braket{t,\psi_j}\mbox{,}
	\end{equation*}
where $\left( \psi_j\right)_{j\in\mathbb{N}} $ is any sequence of functions in $\mathcal{I}\left( X,\mathbb{R}^d\right) $ such that 
	\begin{equation*}
		\psi_j\underset{j\rightarrow\infty}{\longrightarrow}\varphi\
	\end{equation*}
with respect to the topology of $\mathscr{E}^m\left( \mathbb{R}^d\right) $.
By the previous argument, one can always choose
	\begin{equation*}
		\braket{t_m,\varphi}=\lim\limits_{\lambda\rightarrow 0}\lim\limits_{i\rightarrow \infty}\braket{t,\left(1-\chi_\lambda\right)\phi_i\ast \varphi}\mbox{,}
	\end{equation*} 
for every $\varphi$ in $\mathcal{I}^m\left( X,\mathbb{R}^d\right)$.

If $\varphi$ belongs to $\mathcal{I}\left( X,\mathbb{R}^d\right) $, we can take $\psi_j=\varphi\ $ for every $j$ in $\mathbb{N}$ and
	\begin{equation*}
		\braket{t_m,\varphi}=\lim\limits_{j\rightarrow\infty}\braket{t,\varphi}=\braket{t,\varphi}\mbox{,}
	\end{equation*}
so that $t_m$ is a continuous extension of  $t$. 

Suppose $\tilde{t}_m$ is another continuous extension of $t$. 
Let $\varphi$ be in $\mathcal{I}^m\left( X,\mathbb{R}^d\right) $ and take $\left( \psi_j\right)_{j\in\mathbb{N}} $ any sequence of functions in $\mathcal{I}\left( X,\mathbb{R}^d\right) $ such that 
	\begin{equation*}
		\psi_j\underset{j\rightarrow\infty}{\longrightarrow}\varphi
	\end{equation*}
with respect to the topology of $\mathscr{E}^m\left( \mathbb{R}^d\right) $. 
Then,	
	\begin{equation*}
		\braket{\tilde{t}_m,\varphi}=\braket{\tilde{t}_m,\lim\limits_{j\rightarrow\infty}\psi_j}=\lim\limits_{j\rightarrow\infty}\braket{\tilde{t}_m,\psi_j}=\lim\limits_{j\rightarrow\infty}\braket{t,\psi_j}=\braket{t_m,\varphi}.
	\end{equation*}
Thus, $t_m$ is unique.

Finally, $t_m$ is linear as a consequence of $t$ being linear itself on $\mathcal{I}\left( X,\mathbb{R}^d\right) $. 
In effect, let $\varphi$ and $\tilde{\varphi}$ be two functions in $\mathcal{I}^m\left( X,\mathbb{R}^d\right) $ and let  $\alpha$ be in $\mathbb{C}$. 
Take any pair of sequences of functions $\left( \psi_j\right)_{j\in\mathbb{N}} $ and $( \tilde{\psi}_j)_{j\in\mathbb{N}} $ in $\mathcal{I}\left( X,\mathbb{R}^d\right) $ such that
	\begin{equation*}
		\psi_j\underset{j\rightarrow\infty}{\longrightarrow}\varphi \ \mbox{ and } \ \tilde{\psi}_j\underset{j\rightarrow\infty}{\longrightarrow}\tilde{\varphi}\mbox{,}
	\end{equation*} 
with respect to the topology of $\mathscr{E}^m\left( \mathbb{R}^d\right) $.
Then, 
	\begin{equation*}
		\braket{t_m,\varphi+\alpha\tilde{\varphi}}=\lim\limits_{j\rightarrow\infty}\braket{t,\psi_j+\alpha\tilde{\psi}_j}=\lim\limits_{j\rightarrow\infty}\braket{t,\psi_j}+\alpha\lim\limits_{j\rightarrow\infty}\braket{t,\tilde{\psi}_j}=\braket{t_m,\varphi}+\alpha\braket{t_m,\tilde{\varphi}}.
	\end{equation*}
%

Finally, if $\varphi$ is an element of  $\mathcal{I}^\infty\left( X,\mathbb{R}^d\right) $, we can define $t_\infty$ on $\varphi$ by the right hand side of \eqref{milka} because  $\mathcal{I}^\infty\left( X,\mathbb{R}^d\right) \subset\mathcal{I}^m\left( X,\mathbb{R}^d\right) $ and the proof given above obviously holds in this restricted case.

Now, item \textit{(i)} at the beginning of this proof  holds for $\varphi$ in $\mathscr{E}\left( \mathbb{R}^d\right) $.
Then, the sequence 
	\begin{equation*}
		\big( \left( 1-\chi_\lambda\right)\phi_i\ast \varphi\big)_{i\in\mathbb{N}} 
	\end{equation*}
converges to $ \left( 1-\chi_\lambda\right) \varphi$ in the space $\mathscr{E}\left( \mathbb{R}^d\right) $.
By the fact that $t$ is sequentially continuous, we may write
	\begin{equation*}
		\lim \limits_{i\rightarrow \infty} \braket{t,  \left(1-\chi_\lambda \right)\phi_i\ast \varphi }=\braket{t,  \left(1-\chi_\lambda \right) \varphi }\mbox{,}
	\end{equation*}
so that in the case that $\varphi$ is an element of $\mathscr{E}\left( \mathbb{R}^d\right) $, the right hand side of equation  \eqref{milka} reduces to the right hand side of equation \eqref{milka2}, and no mollifiers are needed to define an extension of $t$.
\end{proof}

Choosing a fixed positive integer $m$, Lemma \ref{zappa} thus gives a continuous extension to $\mathcal{I}^m\left(X,\mathbb{R}^d \right)$ of a distribution $t$ in $\mathscr{D}\left( \mathbb{R}^d\setminus X\right)'$, whenever $t$ is compactly supported and has moderate growth along $X$.

Briefly, we repeat the way in which this extension is constructed. 
Let $\left( \chi_\lambda\right)_{\lambda\in\left( 0,1\right] }$ be the family of cut-off functions defined in Lemma \ref{technical}  and let $\left(  \phi_i \right)_{i\in\mathbb{N}} $ be a Dirac sequence of functions defined on $\mathbb{R}^d$.
Setting $\beta_{\lambda}=1-\chi_\lambda$ we have that $t_m$, defined on a function $\varphi$ in $\mathcal{I}^m\left( X,\mathbb{R}^d\right) $ by 
	\begin{equation*}
		\left\langle t_m,\varphi\right\rangle =\lim\limits_{\lambda\rightarrow 0}\lim \limits_{i\rightarrow \infty} \left\langle t,  \beta_\lambda\left(  \phi_i\ast \varphi\right) \right\rangle \mbox{,}
	\end{equation*}
is a continuous extension to $\mathcal{I}^m\left( X,\mathbb{R}^d\right)$ of $t$.
In other words, such a distribution $t_m$ makes the following diagram commutative
	\begin{equation*}
		\xymatrix{\mathcal{I}\left( X,\mathbb{R}^d\right)  \ar[d]_t \ar@{^{(}->}[r]^i& \mathcal{I}^m\left( X,\mathbb{R}^d\right) \ar@{-->}[dl]^{\quad \exists ! \ t_m}\\\mathbb{C}&}
	\end{equation*}
Observe that by construction (see Lemma \ref{technical}), the family
	\begin{equation*}
		(\beta_\lambda)_{\lambda\in(0,1]}\subseteq\mathscr{E}\left( \mathbb{R}^d\right) 
	\end{equation*}
already satisfies the requirements in Theorem \ref{elteo2}.
 
To prove Theorem \ref{elteo2} we need to find a further continuous extension of  $t_m$ (and therefore of $t$)  to the space $\mathscr{E}\left( \mathbb{R}^d\right)$,   satisfying the description in the statement of the mentioned theorem. 

To construct the new extension $\bar{t}$ we will need  Theorem \ref{janis} below.
This result appears as  Theorem 1.2 of  \cite{vietdang}, and unfortunately it is stated in a rather unclear fashion.
Nor does the author give a definition of splitting of a short exact sequence of LCS (which we have given, see Definition \ref{galletitasdelimon}), neither does he specify the topology of the dual spaces that appear in his theorem, which is by no means a minor omission (see Remark \ref{seqspa23} and Example \ref{seqspa}).
We give a complete and organised proof of this result.
\begin{thm}\label{janis}
There is a bijection between:
	\begin{itemize}
		\item [1.] The collection  $\mathfrak{D}$ of closed subspaces $D$ of $\mathscr{E}^m\left( \mathbb{R}^d\right) $ such that
			\begin{equation*}
				\mathscr{E}^m\left( \mathbb{R}^d\right) =\mathcal{I}^m\left( X,\mathbb{R}^d\right) \oplus D.
			\end{equation*} 
		$D$ is called a  \emph{renormalization scheme}\index{renormalization scheme}.
		\item [2.] The space $\mathscr{S}$ of continuous linear sections
			\begin{equation*}
				T:\mathscr{E}^m\left( X\right) \rightarrow \mathscr{E}^m\left( \mathbb{R}^d\right) 
			\end{equation*}
 		of  
			\begin{equation*}
				J^m:\mathscr{E}^m\left( \mathbb{R}^d\right) \rightarrow \mathscr{E}^m\left( X\right) 
			\end{equation*}
		in the short exact sequence
			\begin{equation}\label{eroica} 			
				0\rightarrow \mathcal{I}^m\left( X,\mathbb{R}^d\right) \overset{i}{ \rightarrow} \mathscr{E}^m\left( \mathbb{R}^d\right) \overset{J^m}{\rightarrow} \mathscr{E}^m\left( X\right) \rightarrow 0.
			\end{equation}
		\item [3.] The space $\mathscr{R}$ of \emph{renormalization sections}\index{renormalization section}, \textit{i.e.} the space of continuous linear sections
			\begin{equation*}
				\mathcal{R}:\mathcal{I}^m\left( X,\mathbb{R}^d\right) '\rightarrow\mathscr{E}^m\left( \mathbb{R}^d\right)' 
			\end{equation*}
 		of  
			\begin{equation*}
				i':\mathscr{E}^m\left( \mathbb{R}^d\right)' \rightarrow\mathcal{I}^m\left( X,\mathbb{R}^d\right) '
			\end{equation*}
		in the dual short exact sequence
			\begin{equation} \label{pastoral}
				0\rightarrow \mathscr{E}^{m}\left(X\right)'  \overset{{J^m}'}{\rightarrow} \mathscr{E}^m\left( \mathbb{R}^d\right)' \overset{i'}{\rightarrow} \mathcal{I}^m\left( X,\mathbb{R}^d\right) '\rightarrow 0\mbox{,}
			\end{equation}
		where the dual spaces are considered with the weak star topology.
	\end{itemize}
In addition, the space $\mathfrak{D}$ is not empty.
Therefore the bijection between the above items gives a one to one correspondence between nonempty sets.
\end{thm}
\begin{proof}
The exactness of \eqref{eroica} and the existence of linear continuous sections of \eqref{eroica} is a consequence of the Whitney Extension Theorem (see Theorem \ref{whitney} and \cite{bierstone}, Thm. 2.6 for the exactness).
Therefore, the last assertion of Theorem \ref{janis}, namely $\mathfrak{D}$ is not empty, follows when we have proved the bijection between the spaces listed in it.

Since  \eqref{eroica} is a continuous exact sequence of Fréchet spaces, the dual sequence
	\begin{equation} 
		0\rightarrow \mathscr{E}^{m}\left(X\right)' \overset{{J^m}'}{\rightarrow} \mathscr{E}^m\left( \mathbb{R}^d\right)' \overset{i'}{\rightarrow} \mathcal{I}^m\left( X,\mathbb{R}^d\right) '\rightarrow 0
	\end{equation}
is exact (see \cite{meisevogt}, Proposition 26.4). 

We begin by proving there is a bijection between the space $\mathfrak{D}$ of renormalization schemes  and the space $\mathscr{R}$ of renormalization sections.
We define a mapping $\mathcal{F}:\mathscr{R}\rightarrow \mathfrak{D}$ such that on each renormalization section 
	\begin{equation*}
		\mathcal{R}:\mathcal{I}^m\left( X,\mathbb{R}^d\right) '\rightarrow\mathscr{E}^m\left( \mathbb{R}^d\right)' 
	\end{equation*} 
is given by the formula:
	\begin{equation} \label{Ff}
		\mathcal{F}(\mathcal{R})=\left\lbrace \varphi \in \mathscr{E}^m\left( \mathbb{R}^d\right) : \mu(\varphi)=0\ \forall \mu\in\mathcal{R}\left( \mathcal{I}^m\left( X,\mathbb{R}^d\right) '\right)  \right\rbrace .
	\end{equation}
By the proof of \textit{(ii)} implies \textit{(i)} of Lemma \ref{Mullova}, $\mathcal{F}(\mathcal{R})$ is a closed subspace such that   
	\begin{equation*}
		\mathscr{E}^m\left( \mathbb{R}^d\right) =\mathcal{I}^m\left( X,\mathbb{R}^d\right) \oplus\mathcal{F}\left( \mathcal{R}\right) .
	\end{equation*}
In other words, $\mathcal{F}(\mathcal{R})$ is a renormalization scheme.

Next, we define the mapping  $\mathcal{G}:\mathfrak{D}\rightarrow \mathscr{R}$ such that on each renormalization scheme $D$ is defined by 
	\begin{equation*}
		\mathcal{G}\left( D\right) =\Pi_1',
	\end{equation*}
where 
	\begin{equation*}
		\Pi_1:\mathcal{I}^m\left( X,\mathbb{R}^d\right) \oplus D\rightarrow \mathcal{I}^m\left( X,\mathbb{R}^d\right)
	\end{equation*}
is the canonical projection onto $\mathcal{I}^m\left( X,\mathbb{R}^d\right) $. 
Note that $\Pi_1$ is continuous because $\operatorname{Ker}\left( \Pi_1\right) =D$ is closed by assumption.

We assert that $\mathcal{G}\circ \mathcal{F}=\operatorname{Id}_{\mathscr{R}}$ and that $\mathcal{F}\circ \mathcal{G}=\operatorname{Id}_{\mathfrak{D}}$.

To prove the first equality, let 
	\begin{equation*}
		\mathcal{R}:\mathcal{I}^m\left( X,\mathbb{R}^d\right) '\rightarrow\mathscr{E}^m\left( \mathbb{R}^d\right)'
	\end{equation*}
be a renormalization section.
Then $\mathcal{F}\left( \mathcal{R}\right) $ is given by \eqref{Ff}. 
Applying $\mathcal{G}$ we must show that
	\begin{equation*}
		\mathcal{G}\circ\mathcal{F}(\mathcal{R})=\Pi_1'=\mathcal{R}\mbox{,}
	\end{equation*}
where, in this particular case,
	\begin{equation*}
		\Pi_1:\mathcal{I}^m\left( X,\mathbb{R}^d\right) \oplus\mathcal{F}\left( \mathcal{R}\right) \rightarrow \mathcal{I}^m\left( X,\mathbb{R}^d\right) 
	\end{equation*}
is the canonical projection onto  $\mathcal{I}^m\left( X,\mathbb{R}^d\right) $.

Define  
	\begin{equation*}
		\Pi_2:\mathcal{I}^m\left( X,\mathbb{R}^d\right) \oplus \mathcal{F}\left( \mathcal{R}\right) \rightarrow \mathcal{F}\left( \mathcal{R}\right)\mbox{,}
	\end{equation*}
the canonical  projection onto $\mathcal{F}\left( \mathcal{R}\right) $, which satisfies
	\begin{equation*}
		\Pi_2=\operatorname{Id}_{\mathscr{E}^m\left( \mathbb{R}^d\right) }-\Pi_1.
	\end{equation*}
Notice that 
	\begin{equation*}
		\operatorname{Ker}\left( \Pi_2\right) =\mathcal{I}^m\left( X,\mathbb{R}^d\right)
	\end{equation*}
is closed. 
Thus, $\Pi_2$ is continuous.
Then, for every $\tau$ in $\mathcal{I}^m\left( X,\mathbb{R}^d\right) '$ and $\psi$  in $\mathscr{E}^m\left( \mathbb{R}^d\right) $ we have
	\begin{equation*}
		\begin{split}
			\left\langle \Pi_1'(\tau),\psi\right\rangle & =\left\langle \tau,\Pi_1(\psi)\right\rangle =\left\langle i'\circ\mathcal{R}(\tau),\Pi_1(\psi)\right\rangle =\left\langle \mathcal{R}\left( \tau\right) \circ i,\Pi_1(\psi)\right\rangle \\ &=\left\langle \mathcal{R}(\tau),\Pi_1(\psi)\right\rangle =\left\langle \mathcal{R}(\tau),\psi-\Pi_2(\psi)\right\rangle =\left\langle \mathcal{R}(\tau),\psi\right\rangle\mbox{,}
		\end{split}
	\end{equation*} 
where in the last equality we have used the fact that  $\Pi_2(\psi)$ belongs to the space $\mathcal{F}(\mathcal{R})$ given by \eqref{Ff}, so that $\mathcal{R}(\tau)$ applied to it is zero.
Thus, $\mathcal{G}\circ \mathcal{F}=\operatorname{Id}_{\mathscr{R}}$.

To prove the other equality, namely $\mathcal{F}\circ \mathcal{G}=\operatorname{Id}_{\mathfrak{D}}$, let $D$ be a renormalization scheme. 
Then, by definition  $D$ is a closed subspace of $\mathscr{E}^m\left( \mathbb{R}^d\right) $ and it satisfies 
	\begin{equation*}
		\mathscr{E}^m\left( \mathbb{R}^d\right) =\mathcal{I}^m\left( X,\mathbb{R}^d\right) \oplus D.
	\end{equation*}
By definition of $\mathcal{G}$,
	\begin{equation*}
		\mathcal{G}\left( D\right) =\Pi_1'\mbox{,}
	\end{equation*}
where 
	\begin{equation*}
		\Pi_1:\mathcal{I}^m\left( X,\mathbb{R}^d\right) \oplus D\rightarrow \mathcal{I}^m\left( X,\mathbb{R}^d\right)
	\end{equation*}
is now  the canonical projection onto  $\mathcal{I}^m\left( X,\mathbb{R}^d\right) $. 
Applying $\mathcal{F}$ we get
	\begin{equation}\label{cjto}
		\mathcal{F}\circ \mathcal{G}\left( D\right) =\left\lbrace \varphi \in \mathscr{E}^m\left( \mathbb{R}^d\right) : \mu\left( \varphi\right) =0\ \forall \mu\in\Pi_1'\left( \mathcal{I}^m\left( X,\mathbb{R}^d\right) '\right)  \right\rbrace. 
	\end{equation}
We then have to prove the following equality
	\begin{equation}\label{blanquita}
		D=\left\lbrace \varphi \in \mathscr{E}^m\left( \mathbb{R}^d\right) : \mu(\varphi)=0\ \forall \mu\in\Pi_1'\left( \mathcal{I}^m\left( X,\mathbb{R}^d\right) '\right)  \right\rbrace. 
	\end{equation}
If $\varphi$ belongs to $D$, $\Pi_1(\varphi)=0$.
Then, for every $\mu=\Pi_1'\left( \phi\right) $, where $\phi$ belongs to $\mathcal{I}^m\left( X,\mathbb{R}^d\right) '$, we have
	\begin{equation*}
		\mu(\varphi)=\braket{\Pi_1'(\phi),\varphi}=\braket{\phi,\Pi_1(\varphi)}=0\mbox{,}
	\end{equation*}
where in the last equality we have used the fact that $\Pi_1(\varphi)=0$. Thus, $\varphi$ belongs to \eqref{cjto}.

On the other hand, if $\varphi$ in $\mathscr{E}^m\left( \mathbb{R}^d\right)  $ is such that $\mu(\varphi)=0$ for every $\mu$  belonging to 
	\begin{equation*}
		\Pi_1'\left( \mathcal{I}^m\left( X,\mathbb{R}^d\right) '\right)\mbox{,}
	\end{equation*}
we have to show that $\Pi_1(\varphi)=0$. 
But this is equivalent to showing that the evaluation maps,  $\operatorname{e}_x:\mathscr{E}^m\left( \mathbb{R}^d\right) \rightarrow \mathbb{C}$,  satisfy  
	\begin{equation*}
		\operatorname{e}_x\left( \Pi_1(\varphi)\right) =\braket{\Pi_1(\varphi),x}=0
	\end{equation*}
for every $x$ in $\mathbb{R}^d$, which in turn happens if and only if
	\begin{equation}\label{truee}
		\Pi_1'\left( \operatorname{e}_x\right) (\varphi)=0\mbox{,}
	\end{equation}
for every $x$ in $\mathbb{R}^d$. 
As $\operatorname{e}_x$ belongs to $\mathscr{E}^m\left( \mathbb{R}^d\right)' $ for every $x$ in $\mathbb{R}^d$, \eqref{truee} holds, and therefore, the equality \eqref{blanquita} is proved.  

The bijection between the space $\mathscr{R}$ of renormalization sections and the space $\mathscr{I}$ of continuous linear sections of
	\begin{equation*}
		J^m:\mathscr{E}^m\left( \mathbb{R}^d\right) \rightarrow \mathscr{E}^m\left( X\right)
	\end{equation*}
in \eqref{eroica} is defined in the same way as the argument described in Proposition  \ref{gershwin}. 
Namely, we define a mapping $\mathcal{J}:\mathscr{R}\rightarrow \mathscr{I}$ such that on each renormalization section
	\begin{equation*}
		\mathcal{R}:\mathcal{I}^m\left( X,\mathbb{R}^d\right) '\rightarrow\mathscr{E}^m\left( \mathbb{R}^d\right)' 
	\end{equation*} 
is given by the formula:
	\begin{equation*}
		\mathcal{J}\left( \mathcal{R}\right) =\left( J^m|_{\mathcal{F}\left( \mathcal{R}\right) }\right)^{-1}
	\end{equation*}
where we recall that 
	\begin{equation} \label{F}
		\mathcal{F}\left( \mathcal{R}\right) =\left\lbrace \varphi \in \mathscr{E}^m\left( \mathbb{R}^d\right) : \mu\left( \varphi\right) =0\ \forall \mu\in\mathcal{R}\left( \mathcal{I}^m\left( X,\mathbb{R}^d\right) '\right)  \right\rbrace. 
	\end{equation}
By Proposition \ref{gershwin}, $\mathcal{J}\left( \mathcal{R}\right) $ is continuous and satisfies 
	\begin{equation*}
		J^m\circ\mathcal{J}(\mathcal{R})=\operatorname{Id}_{\mathscr{E}^m\left( X\right) }.
	\end{equation*}

Next, define a mapping $\mathcal{K}:\mathscr{I}\rightarrow \mathscr{R}$ such that on each continuous linear section 
	\begin{equation*}
		\mathcal{T}:\mathscr{E}^m\left( X\right) \rightarrow \mathscr{E}^m\left( \mathbb{R}^d\right)
	\end{equation*}
of 
	\begin{equation*}
		J^m:\mathscr{E}^m\left( \mathbb{R}^d\right) \rightarrow \mathscr{E}^m\left( X\right)
	\end{equation*}
in \eqref{eroica} is given by
	\begin{equation*}
		\mathcal{K}(\mathcal{T})=\left( \pi_1\circ \Omega^{-1}_{\mathcal{T}}\right)'\mbox{,}
	\end{equation*}
where 
	\begin{equation*}
		\Omega_{\mathcal{T}}:  \mathcal{I}^m\left( X,\mathbb{R}^d\right)\oplus \mathscr{E}^m\left( X\right) \rightarrow  \mathscr{E}^m\left(\mathbb{R}^d\right)
	\end{equation*}
is an isomorphism that makes the following diagram 
	\begin{equation*}
		\xymatrix{\mathcal{I}^m\left( X,\mathbb{R}^d\right)
		\ar@{=}[d]_{} \ar[r]^{i}& \mathscr{E}^m\left(\mathbb{R}^d\right)\ar[r]^{J^m}&   \mathscr{E}^m\left( X\right)\ar@{=}[d]_{}\\	\mathcal{I}^m\left( X,\mathbb{R}^d\right)\ar@{^{(}->}[r]^{ {\hspace*{0.8 cm}}}& \mathcal{I}^m\left( X,\mathbb{R}^d\right)\oplus \mathscr{E}^m\left( X\right)\ar[u]^{\Omega_{\mathcal{T}}} \ar@{->>}[r]^{\pi_2 {\hspace*{-1 cm}}}&  \mathscr{E}^m\left( X\right)}
	\end{equation*} 
commutative and 
	\begin{equation*}
		\pi_1:  \mathcal{I}^m\left( X,\mathbb{R}^d\right)\oplus \mathscr{E}^m\left( X\right) \rightarrow  \mathcal{I}^m\left( X,\mathbb{R}^d\right)
	\end{equation*}
is the canonical projection.
By Proposition \ref{gershwin}, $\mathcal{K}(\mathcal{T})$ is a continuous linear section of  
	\begin{equation*}
		i':\mathscr{E}^m\left( \mathbb{R}^d\right)' \rightarrow\mathcal{I}^m\left( X,\mathbb{R}^d\right) '
	\end{equation*}
in the dual exact sequence \eqref{pastoral}.

We have to prove  $\mathcal{K}\circ \mathcal{J}=\operatorname{Id}_{\mathscr{R}}$ and  $\mathcal{J}\circ \mathcal{K}=\operatorname{Id}_{\mathscr{I}}$.
We begin with the first equality.
Let $\mathcal{R}
$ be in $\mathscr{R}$ and take $t$ in $\mathcal{I}^m\left( X,\mathbb{R}^d\right) '$ and $\varphi$ in $\mathscr{E}^m\left( \mathbb{R}^d\right) $.
Then, 
	\begin{equation*}
		\begin{split}
			\braket{\mathcal{K}\circ\mathcal{J}(\mathcal{R})(t),\varphi}&=\left\langle\left( \pi_1\circ \Omega^{-1}_{\mathcal{J}(\mathcal{R})}\right) '(t)\ ,\varphi\right\rangle =\left\langle t\ ,\left( \pi_1\circ \Omega^{-1}_{\mathcal{J}(\mathcal{R})}\right) (\varphi)\right\rangle \\ &=\left\langle \left( i'\circ\mathcal{R }\right)(t)\ ,\left( \pi_1\circ \Omega^{-1}_{\mathcal{J}(\mathcal{R})}\right) (\varphi)\right\rangle =\left\langle \mathcal{R }(t)\ ,\left( \pi_1\circ \Omega^{-1}_{\mathcal{J}(\mathcal{R})}\right) (\varphi)\right\rangle \\			&=\left\langle \mathcal{R }(t)\ ,\varphi-\left( \mathcal{J}\left( \mathcal{R}\right) \circ J^{m}\right) (\varphi)\right\rangle 			=\left\langle \mathcal{R }(t)\ ,\varphi\right\rangle 
		\end{split}
	\end{equation*} 
where in the last equality we have used the fact that 
	\begin{equation*}
		\left( \mathcal{J}\left( \mathcal{R}\right) \circ J^{m}\right) (\varphi)\in\mathcal{F}(\mathcal{R})
	\end{equation*}
and therefore
	\begin{equation*}
		\left\langle \mathcal{R }(t)\ ,  \left( \mathcal{J}\left( \mathcal{R}\right) \circ J^{m}\right) (\varphi)\right\rangle=0.
	\end{equation*}
Thus, we conclude $\mathcal{K}\circ \mathcal{J}=\operatorname{Id}_{\mathscr{R}}$.

To prove the other equality, namely $\mathcal{J}\circ \mathcal{K}=\operatorname{Id}_{\mathscr{I}}$, let $\mathcal{T}$ be in $\mathscr{I}$ and take $F$ in $\mathscr{E}^m\left( X\right) $.
Then
	\begin{equation*}
		\begin{split}
			\left[ \mathcal{J}\circ\mathcal{K}(\mathcal{T})\right] (F) &=\left( J^m|_{\mathcal{F}(\mathcal{K}(\mathcal{T}) )}\right)^{-1}(F)=\operatorname{Id}_{\mathscr{E}^m\left( \mathbb{R}^d\right)}\circ\left( J^m|_{\mathcal{F}(\mathcal{K}(\mathcal{T}) )}\right)^{-1}(F)\\& =\left[ \pi_1\circ\Omega_{\mathcal{T}}^{-1}+\mathcal{T}\circ J^m \right] \left( \mathbb{R}^d\right)\circ\left( J^m|_{\mathcal{F}(\mathcal{K}(\mathcal{T}) )}\right)^{-1}(F) \\ &=\left[ \mathcal{T}\circ J^m\circ\left( J^m|_{\mathcal{F}(\mathcal{K}(\mathcal{T}) )}\right)^{-1}\right] (F) =\mathcal{T}(F).
		\end{split}	
	\end{equation*}
Thus, $\mathcal{J}\circ\mathcal{K}=\operatorname{Id}_{\mathscr{I}}.$
\end{proof}
Now we can give the proof of Theorem \ref{elteo2}.
\begin{proof}[Proof of Theorem \ref{elteo2}]
Let $\left( \chi_\lambda\right)_{\lambda\in\left( 0,1\right] }$ be the family of cut-off functions defined in Lemma \ref{technical}  and let $\left(  \phi_i \right)_{i\in\mathbb{N}} $ be a Dirac sequence of functions defined on $\mathbb{R}^d$.
Set $\beta_{\lambda}:=1-\chi_\lambda$ and choose some fixed positive integer $m$. 

By hypothesis, $t$ belongs to  $\mathscr{D}\left( \mathbb{R}^d\setminus X\right)'$, is compactly supported and has moderate growth along $X$.
Then, from Lemma \ref{zappa} we have that $t_m$, defined on a function $\varphi$ in $\mathcal{I}^m\left( X,\mathbb{R}^d\right) $ by 
	\begin{equation*}
		\left\langle t_m,\varphi\right\rangle =\lim\limits_{\lambda\rightarrow 0}\lim \limits_{i\rightarrow \infty} \left\langle t,  \beta_\lambda\left(  \phi_i\ast \varphi\right) \right\rangle\mbox{,}
	\end{equation*}
is a continuous extension to  $\mathcal{I}^m\left( X,\mathbb{R}^d\right)$ of $t$.

Observe that by construction (see Lemma \ref{technical}), the family
	\begin{equation*}
		(\beta_\lambda)_{\lambda\in(0,1]}\subseteq\mathscr{E}\left( \mathbb{R}^d\right) 
	\end{equation*}
already satisfies the requirements in the statement of Theorem \ref{elteo2}.

Let $D$ be a fixed renormalization scheme (see Theorem \ref{janis} where we define renormalization schemes and prove their existence) and let  
	\begin{equation*}
		I^m_{D}:\mathscr{E}^m\left( \mathbb{R}^d\right) \rightarrow \mathcal{I}^m\left( X,\mathbb{R}^d\right) \quad\mbox{ and } \quad P^m_D=\operatorname{Id}_{\mathscr{E}\left( \mathbb{R}^d\right)}-I^m_D:\mathscr{E}^m\left( \mathbb{R}^d\right) \rightarrow D\mbox{,}
	\end{equation*} 
be the canonical projections relative to $D$, \textit{i.e.} the space $\mathscr{E}^m\left( \mathbb{R}^d\right)$ is written as
	\begin{equation*}
		\mathscr{E}^m\left( \mathbb{R}^d\right)=\mathcal{I}^m\left( X,\mathbb{R}^d\right)\oplus D.
	\end{equation*}
Then, the desired extension $\bar{t}$ of $t$ can be defined on a given $\varphi$ in $\mathscr{E}\left( \mathbb{R}^d\right) $ as follows: first applying the projection $I^m_{D}$ to obtain an element in  $\mathcal{I}^m\left( X,\mathbb{R}^d\right) $, and then applying the already found extension $t_m$.
Explicitly,  we have that for every $\varphi$ in $\mathscr{E}\left( \mathbb{R}^d\right)$,
	\begin{equation}\label{amoabeethoven}
		\left\langle \bar{t},\varphi\right\rangle =\lim\limits_{\lambda\rightarrow 0}\lim \limits_{i\rightarrow \infty}\underset{finite\ part}{\left\langle t,\beta_\lambda \left( \phi_i\ast I^m_{D}\varphi\right) \right\rangle } = \lim\limits_{\lambda\rightarrow 0}\lim \limits_{i\rightarrow \infty} \left\langle t,\beta_\lambda\varphi\right\rangle - \underset{singular\ part}{\left\langle t,\beta_\lambda \left( \phi_i\ast P^m_{D}\varphi\right)\right\rangle } 
	\end{equation}
is a well-defined extension of $t$, called the  \emph{renormalization}\index{renormalization of a distribution} of the distribution $t$.
From equation 
\eqref{amoabeethoven} we see that the definition of the required distributions $c_{\lambda}$ supported on $X$ is given by 
	\begin{equation*}
		\left\langle c_{\lambda},\varphi\right\rangle =\lim \limits_{i\rightarrow \infty}  \big\langle t,\beta_\lambda \left( \phi_i\ast P_D^m\varphi\right) \big\rangle\mbox{,}
	\end{equation*}
for every $\varphi$  in $\mathscr{E}\left( \mathbb{R}^d\right) $.
Thus, the proof of the theorem is complete.
\end{proof}
\begin{rem}
It has to be noted that the construction of the extension $\bar{t}$ depends on the fixed positive integer $m$ chosen for the extension $t_m$ and in the renormalization scheme $D$ selected thereafter.
\end{rem}
\begin{rem}[See comments immediately after Proposition 1.2 in \cite{vietdang}]
Observe that $\left\langle t,\beta_\lambda\varphi\right\rangle$ diverges as $\lambda \rightarrow 0$ in \eqref{amoabeethoven} whenever $\varphi$ does not belong to the space $\mathcal{I}^m\left( X,\mathbb{R}^d\right)$.
However, these divergences are local in the sense they can be subtracted by the counterterm $\left\langle t,\beta_\lambda \left( \phi_i\ast P^m_{D}\varphi\right)\right\rangle$ as $\lambda \rightarrow 0$.
\end{rem}
\chapter{Renormalization of Feynman amplitudes in Euclidean quantum field theories}\label{joann6}

This chapter is devoted to the problem of extending the so called \emph{Feynman amplitudes} (to be defined later) in quantum field theories whose underlying spacetime has a Riemannian structure.
Recall that $\mathcal{M}$ denotes a $d$-dimensional smooth, paracompact, oriented manifold; and 
$X$ a closed subset of $\mathcal{M}$.
We denote by $d$ the distance function induced by some choice of smooth Riemannian metric $g$ on $\mathcal{M}$. 

Many of the results appearing in this chapter are stated in \cite{vietdang} in an ambiguos fashion or with confusing proofs, 
namely Theorems \ref{maggotbrain} and \ref{cocainedecisions}; Lemma \ref{silvia}; and Proposition \ref{G}.
We thus give an exhaustive and clear explanation of them.

In the first section, the problem of extending a product of a function and a distribution is treated, in an almost general way.
This means that we are going to require the function to satisfy certain features, namely to be tempered along the subset over which the product is intended to be extended.

In the second section we concentrate on the particular case of a continuous extension of Feynman amplitudes, using the general results obtained before.
\section{Renormalized products}\label{secrenprod}

We begin by introducing a class of functions whose features are relevant for the purpose of extending products, namely the class of tempered functions along $X$.

The notion of tempered function will be a local one, so it will be defined by means of a partition of unity argument.
We will use the same notations of \S \ref{simple}: choose a locally finite cover of $\mathcal{M}$ by relatively compact open charts $(V_\alpha,\psi_\alpha)$, with  
	\begin{equation*}
		\psi_\alpha:V_\alpha\rightarrow V\subseteq\mathbb{R}^d
	\end{equation*}
where $V=\psi_\alpha(V_\alpha)$ is open.
Let $\left\lbrace \varphi_\alpha\right\rbrace _\alpha$ be a  subordinate partition of unity such that $\sum_{\alpha} \varphi_\alpha=1$, and denote
	\begin{equation*}
		K_\alpha=\operatorname{Supp}\left(  \varphi_\alpha\right) \subseteq V_\alpha.
	\end{equation*}
We first give the notion of a tempered function defined on  the Euclidean space $ \mathbb{R}^d$.
\begin{defi}\index{tempered function}
Let $\Omega$ be an open subset of $ \mathbb{R}^d$  and $X$ a closed set contained in $\Omega$.
Let $d$ denote the distance function in the Euclidean space.
A function $f$ in $\mathscr{E}\left( \Omega\setminus X\right) $ is \emph{tempered along } X if it satisfies the following estimate: for every $k$ in $\mathbb{N}$ and every compact subset $ K$ of $\Omega$ there exists a pair of positive constants $C$ and $s$ such that
	\begin{equation}\label{temper}
		\sup\limits_{|\nu|\leq k}\left|\partial^\nu f(x) \right| \leq C \left[1+ d(x,X) ^{-s}\right]   
	\end{equation}
for every $x$ in $K\setminus X$. 
\end{defi}
We now generalize the preceding definition to cover functions defined on a manifold $\mathcal{M}$.
\begin{defi}\label{tempfunc}\index{tempered function}
A function $f$ in $\mathscr{E}\left( \mathcal{M}\setminus X\right) $ is \emph{tempered along }X if in any local chart $\psi_\alpha:V_\alpha\rightarrow V$, the pushforward
	\begin{equation}\label{villavicencio}
		\psi_{\alpha\ast}\left( \varphi_\alpha f\right):V\setminus \psi_\alpha\left( X\cap V_\alpha\right) \rightarrow \mathbb{C}\mbox{,}
	\end{equation}
satisfies the following estimate: for every $k$ in $\mathbb{N}$ and  every $K$ compact subset of  $V$, there exists a pair of positive constants $C$  and $s$ such that 
	\begin{equation}\label{funkadelic}
		\sup\limits_{|\nu|\leq k}\left|\partial^\nu\left(  \psi_{\alpha\ast}\left( \varphi_\alpha f\right)\right) (x) \right| \leq C \left[1+ d(x,\psi_\alpha\left( X\cap V_\alpha\right)) ^{-s}\right]\mbox{,}
	\end{equation}
for every $x$ in $K\setminus \psi_\alpha\left( X\cap V_\alpha\right)$. 
Here, $d$ denotes the distance function in the Euclidean space.
   
In other words, $f$ in $\mathscr{E}\left( \mathcal{M}\setminus X\right) $ is tempered along X if for every $\alpha$,  the pushforward $\psi_{\alpha\ast}\left( \varphi_\alpha f\right)$ in $\mathscr{E}\left( V\setminus \psi_\alpha\left( X\cap V_\alpha\right) \right) $ is tempered along $\psi_\alpha\left( X\cap V_\alpha\right)$.

The class of tempered functions along a closed subset $X$ of a manifold  $\mathcal{M}$  forms  an algebra by Leibniz's rule, and will be denoted by $\mathcal{T}(X,\mathcal{M})$.
\end{defi}
Now, we establish in Theorem \ref{maggotbrain} a result about renormalization of a distribution multiplied by a tempered function, for which we will need the following proposition.
\begin{propo}[\textit{cf.} \cite{vietdang}, Proposition 3.1]\label{G}
Let $t$ be a compactly supported distribution in $\mathscr{D}\left( \mathbb{R}^d\setminus X \right)'$, and let $f$ be  a function in $\mathscr{E}\left( \mathbb{R}^d\setminus X \right) $, which satisfy the following estimates:
	\begin{enumerate}
		\item There exist a pair of positive constants $C_1$ and $s_1$, and a seminorm $\parallel \cdot \parallel_k^l$  such that 
			\begin{equation*}
				\left|\braket{t,\varphi} \right| \leq C_1 \left[1+d\left(\operatorname{Supp} \left(  \varphi\right) ,X \right)^{-s_1} \right]\parallel \varphi \parallel_k^l\mbox{,}
			\end{equation*}
		for every $\varphi$ in $\mathcal{I}\left( X,\mathbb{R}^d\right)$.
		\item For every $k$ in $\mathbb{N}$ and for every compact subset $K$ of $\mathbb{R}^d$, there is a pair of positive constants $C_2$ and $s_2$ such that 		\begin{equation*}
			 \sup\limits_{|\nu|\leq k}\left|\partial^\nu f(x) \right| \leq C_2 \left[1+ d(x,X) ^{-s_2}\right]\mbox{,}
		\end{equation*}
	for every $x$ in $K\setminus X$.
	\end{enumerate}
Then, $ft$ satisfies the following estimate: there exists a positive constant $C$ such that 
	\begin{equation*}
		 \left|\braket{ft,\varphi} \right|\leq C \left[1+d\left(\operatorname{Supp} \left(  \varphi\right) ,X \right)^{-\left( s_1+s_2\right) } \right]\parallel \varphi \parallel_k^l\mbox{,}
	\end{equation*}
for every $\varphi$ in $\mathcal{I}\left( X,\mathbb{R}^d\right) $.
\end{propo}
\begin{proof}
The claim follows from the following estimate.
For every $\varphi$ in $\mathcal{I}\left( X,\mathbb{R}^d\right) $,
	\begin{equation*}
		\begin{split}
			\left| \braket{ft,\varphi}\right|&\!\leq C_1  \left[1+d\left(\operatorname{Supp} \left( \varphi\right) ,X \right)^{-s_1} \right] \left\| f\varphi \right\|_k^l\\
			&\leq \!C_1C_2 2^{kn} \!\left[1+d\left(\operatorname{Supp} \left( \varphi\right) ,X \right)^{-s_1} \right]\!\!\left[1+ d(x,X) ^{-s_2}\right] \!\left\| \varphi \right\|_k^l \\
			&\leq\!\underbrace{ 4C_1C_2 2^{kn}}_{C}\left[1+d\left(\operatorname{Supp} \left(  \varphi\right) ,X \right)^{-\left( s_1+s_2\right) } \right]\left\| \varphi \right\|_k^l\mbox{,}
		\end{split}
	\end{equation*}
where we have used the Leibniz rule for the second inequality.
\end{proof}
\begin{thm}[\textit{cf.} \cite{vietdang}, Thm. 3.1]\label{maggotbrain}
For every function $f$ in $\mathcal{T}\left( X,\mathcal{M}\right) $ and every distribution $t$ in $\mathscr{D}(\mathcal{M})'$, there exists a distribution $\mathcal{R}(ft)$ in  $\mathscr{D}(\mathcal{M})'$ which coincides with the regular product $ft$ outside $X$.
\end{thm}
Before giving the proof of Theorem  \ref{maggotbrain} we give the following interesting consequence.
\begin{cor}\label{maggotbrain2}
A function $f$ in $\mathcal{T}\left( X,\mathcal{M}\right) $ can always be considered as a distribution in $\mathscr{D}\left( \mathcal{M}\setminus X\right)' $.
In that case, $f$ has a continuous extension to $\mathscr{D}\left( \mathcal{M}\right)$ by considering $t=1$ in Theorem \ref{maggotbrain}.
\end{cor}
\begin{proof}[Proof of Theorem \ref{maggotbrain}]

By a partition of unity argument, the proof of the theorem may be reduced to the case where $X$ is a closed subset of $\mathcal{M}=\mathbb{R}^d$,  $f$ belongs to $\mathcal{T}\left( X,\mathbb{R}^d\right) $, and $t$ is a compactly supported distribution in $\mathscr{D}\left( \mathbb{R}^d\right)'$.
This is done in the following way.
Choose a locally finite cover of $\mathcal{M}$ by relatively compact open charts $(V_\alpha,\psi_\alpha)$, with  
	\begin{equation*}
		\psi_\alpha:V_\alpha\rightarrow V\subseteq\mathbb{R}^d
	\end{equation*}
where $V=\psi_\alpha(V_\alpha)$ is open.
Let $\left\lbrace \varphi_\alpha\right\rbrace _\alpha$ be a  subordinate partition of unity such that $\sum_{\alpha} \varphi_\alpha=1$, and denote
	\begin{equation*}
		K_\alpha=\operatorname{Supp}\left(  \varphi_\alpha\right) \subseteq V_\alpha.
	\end{equation*}
For each $\alpha$ consider $t_\alpha=t\varphi_\alpha$, which belongs to $\mathscr{D}\left( \mathcal{M}\right)' $ and is supported in  $K_\alpha$. 
Set $f_\alpha=\left.f\right|_{V_\alpha}$,  which is  tempered along $X\cap V_\alpha$.
Then, $f_\alpha t_\alpha$  belongs to  $\mathscr{D}\left( V_\alpha\setminus X\right)' $ and has compact support contained in ${K_\alpha}$.
For each $\alpha$ it suffices to find a continuous extension $\mathcal{R}(f_\alpha t_\alpha)$ to the space  $\mathscr{D}\left( V_\alpha\right)$,  supported on $K_\alpha$, which coincides with $f_\alpha t_\alpha$ on  $V_\alpha\setminus X$, because we could then define
	\begin{equation*}
		\mathcal{R}(ft)=\sum_{\alpha} \mathcal{R}(f_\alpha t_\alpha)\mbox{,}
	\end{equation*}
which is a locally finite sum that satisfies the conditions of the theorem
	\footnote{A notation abuse has been used here: in principle, the product $\mathcal{R}(f_\alpha t_\alpha)$ is only defined on $\mathscr{D}\left( V_\alpha\right) $.
 	However, we assume that $\mathcal{R}(f_\alpha t_\alpha)$ extends by zero outside $V_\alpha$ because it has compact support contained in ${K_\alpha}$, and therefore,  it may be thought to belong to the space $\mathscr{D}\left( \mathcal{M}\right)' $.
 	Thus, the sum $\sum_{\alpha} \mathcal{R}(f_\alpha t_\alpha)$ makes sense and belongs to $\mathscr{D}\left( \mathcal{M}\right)' $ .}. 

Now, in order to find such an extension $\mathcal{R}(f_\alpha t_\alpha)$ for each $\alpha$, we reduce the problem to the Euclidean case.
We define the pushforward
	\begin{equation}
		\tilde{f}_\alpha:=\psi_{\alpha\ast}(f_\alpha)=f_\alpha\circ \psi_{\alpha}^{-1}.  
	\end{equation}
The function $\tilde{f}_\alpha$ is tempered along $\psi_\alpha \left(X\cap V_\alpha \right)$ which is contaied in $V$.
Also, define  
	\begin{equation}
		\tilde{t}_\alpha:=\psi_{\alpha\ast}(t_\alpha)\mbox{,}
	\end{equation}
that belongs to  $\mathscr{D}\left( \mathbb{R}^d\right)'$ and has compact support contained in $\psi_\alpha(K_\alpha)$. 
If we could find a continuous extension $\mathcal{R}( \tilde{f_\alpha}\tilde{t_\alpha}) $ to $\mathscr{D}\left( \mathbb{R}^d\right)$  of $\tilde{f_\alpha}\tilde{t_\alpha}$, supported in $\psi_\alpha(K_\alpha)$, and  such that it coincides with $\tilde{f_\alpha}\tilde{t_\alpha}$ outside $\psi_\alpha\left(X\cap V_\alpha \right)$, then we could set 
	\begin{equation*}
		\mathcal{R}(f_\alpha t_\alpha)=\psi_\alpha^{\ast}(\mathcal{R}(\tilde{f_\alpha}\tilde{t_\alpha}))\mbox{,}
	\end{equation*}
that belongs to $\mathscr{D}(\mathcal{M})'$, is supported on $K_\alpha$, coincides with $f_\alpha t_\alpha$ on  $V_\alpha\setminus X$, and the statement of the theorem follows.

Thus, we will prove Theorem \ref{maggotbrain} in the case where $X$ is a closed subset of $\mathcal{M}=\mathbb{R}^d$,  $f$ belongs to $\mathcal{T}\left( X,\mathbb{R}^d\right) $ and $t$ is a compactly supported distribution in $\mathscr{D}\left( \mathbb{R}^d\right)' $. 
By Theorem \ref{elteo}, distributions with moderate growth are extendible.
Therefore, it suffices to show that $ft$ has moderate growth along $X$, but this is an immediate consequence of Lemma \ref{dumbo} and Proposition \ref{G}. 
\end{proof}
\begin{cor}
$\mathcal{T}_{\mathcal{M}\setminus X}(\mathcal{M})$ is a $\mathcal{T}(X,\mathcal{M})$-module.
\end{cor}
In what follows we give an example of renormalization of products where Theorem \ref{maggotbrain} is applied.
\begin{example}[Particular case of renormalization of products. See \cite{vietdang}, comments after Thm. 4.4]
Let us consider a function $\varphi$ in $\mathscr{E}\left( \mathbb{R}^d\right) $.
Define $X=\left\lbrace \varphi=0 \right\rbrace $ and suppose $\varphi\mathscr{E}\left( \mathbb{R}^d\right) $ is a closed ideal of $\mathscr{E}\left( \mathbb{R}^d\right) $.
Then, a result of Malgrange (see \cite{malgrange}, inequality 2.1, p.88), yields that $\varphi$ satisfies the Lojasiewicz inequality: for every $K$ compact subset of $\mathbb{R}^d$, there exists a pair of positive constants $C$ and $s$ such that 
	\begin{equation*}
	\left| \varphi(x)\right|\geq C d(x,X)^s\mbox{,}
	\end{equation*}  
for every $x$ in $K$.
The Leibniz rule tells us that $f:\mathbb{R}^d\setminus X\rightarrow \mathbb{C}$ given by $f(x)=\left( \varphi(x)\right)^{-1}$ must be tempered along $X$. 
Then, by Theorem \ref{maggotbrain} $ft$ has an  extension $\mathcal{R}(ft)$ which  coincides with the product outside the set $X$.
\end{example}

The previous example is extended to manifolds in Theorem \ref{cocainedecisions} below.
Such a result could be proved by a  partition of unity argument, being the previous example the proof of the local case.
However, we present here an alternative demonstration of this statement due to Malgrange (see \cite{malgrange}, Thm. 2.1, p.100).

\begin{thm}[\textit{cf.} \cite{vietdang}, Thm. 4.5.]\label{cocainedecisions}
Let $\varphi$ belong to $\mathscr{E}\left( \mathcal{M}\right) $, $X=\left\lbrace \varphi=0\right\rbrace $ and suppose $\varphi\mathscr{E}\left( \mathcal{M}\right) $ is a closed ideal of $\mathscr{E}\left( \mathcal{M}\right) $. 
Let $f:\mathcal{M}\setminus X\rightarrow \mathbb{C}$ be given by $f(x)=\left( \varphi\left( x\right) \right)^{-1}$.
Then, for every $t$ in $\mathscr{D}\left( \mathcal{M}\right)'$, there exists $s$ in $\mathscr{D}\left( \mathcal{M}\right)'$ such that  $\varphi s=t$.
In particular, $s=ft$ outside $X$.  
\end{thm}
\begin{proof}
It suffices to prove that the linear map
	\begin{align*}
		m_\varphi:\mathscr{E}\left( \mathcal{M}\right)' &\rightarrow \mathscr{E}\left( \mathcal{M}\right)' \\s&\mapsto\varphi s
	\end{align*}
is onto if $\varphi\mathscr{E}\left( \mathcal{M}\right) $ is closed in $\mathscr{E}\left( \mathcal{M}\right) $.
For this purspose we will show that $\operatorname{Im}\left( m_\varphi\right) $ is closed and dense in $\mathscr{E}\left( \mathcal{M}\right)' $. 

To prove the first assertion, consider the map
	\begin{align*}
		M_\varphi:\mathscr{E}\left( \mathcal{M}\right) &\rightarrow \mathscr{E}\left( \mathcal{M}\right) \\ \psi&\mapsto \varphi \psi.
	\end{align*}
Observe that $M_\varphi'=m_\varphi$.
As $\varphi\mathscr{E}\left( \mathcal{M}\right) =\operatorname{Im}\left( M_\varphi\right) $ is closed in $\mathscr{E}\left( \mathcal{M}\right) $ and $\mathscr{E}\left( \mathcal{M}\right) $ is Fr\'echet,  $\operatorname{Im}\left( m_\varphi\right) $ is closed in $\mathscr{E}\left( \mathcal{M}\right)' $ (see \cite{meisevogt}, Thm. 26.3). 

To prove the second assertion, observe that the fact that  $\varphi\mathscr{E}\left( \mathcal{M}\right) =\operatorname{Im}\left( M_\varphi\right) $ is closed in $\mathscr{E}\left( \mathcal{M}\right) $ implies $\varphi\mathscr{E}\left( \mathcal{M}\right) =\operatorname{Im}\left( M_\varphi\right) $ is Fr\'echet.
Then, the continuous linear map 
	\begin{equation}\label{correstr}
		\left. M_{\varphi}\right\vert^{\varphi\mathscr{E}\left( \mathcal{M}\right) }:\mathscr{E}\left( \mathcal{M}\right) \rightarrow \varphi\mathscr{E}\left( \mathcal{M}\right) 
	\end{equation}
is open, by the Open Mapping Theorem.
Then, for  every continuous seminorm $\left\| \cdot\right\|_k^l $  on $\mathscr{E}\left( \mathcal{M}\right) $, there exists a seminorm $\left\| \cdot\right\|_{k'}^{l'} $ such that
	\begin{equation*}
		\left\| \psi\right\|_k^l \leq \left\| \varphi\psi\right\|_{k'}^{l'}.
	\end{equation*}
To show this last statement is true, consider the open set 
	\begin{equation*}
		U=\left\lbrace \psi\in\mathscr{E}\left(M \right) : \left\|\psi \right\|_k^l<1 \right\rbrace. 
	\end{equation*}
As \eqref{correstr} is linear and  open, the image of $U$ under this map contains the zero function and is open.
Therefore, there is some positive number $s$ and a seminorm $\left\| \cdot\right\|_{k'}^{l'}$ such that
	\begin{equation*}
		0\in V:=\left\lbrace \psi\in\varphi\mathscr{E}\left(M \right) : \left\|\psi \right\|_{k'}^{l'}<s \right\rbrace \subseteq \overline{V} \subseteq \left. M_{\varphi}\right\vert^{\varphi\mathscr{E}\left( \mathcal{M}\right) }\left( U\right). 
	\end{equation*}
In particular, 
	\begin{equation*}
		V\cap M_\varphi|^{\varphi\mathscr{E}\left( \mathcal{M}\right) }\left( \partial U\right)=\emptyset.
	\end{equation*}
Therefore, for every $\psi\neq 0$, 
	\begin{equation*}
		\left\| \left. M_{\varphi}\right\vert^{\varphi\mathscr{E}\left( \mathcal{M}\right) }\left( \frac{\psi}{\left\|\psi \right\|_{k}^{l} }\right) \right\|_{k'}^{l'}\geq s\mbox{,}
	\end{equation*}
which can be reexpressed, using the definition of the map  \eqref{correstr}, as
	\begin{equation*}
		\left\| \varphi \psi \right\|_{k'}^{l'}\geq s\left\|\psi \right\|_{k}^{l}\mbox{,} 
	\end{equation*}
for every $\psi\neq 0$.
Since the previous inequality  holds also  for $\psi=0$, we have 
	\begin{equation*}
		\varphi\psi=0  \implies \psi=0.
	\end{equation*}
Then,
	\begin{equation*}
		\begin{split}
			^\perp \left(\operatorname{Im}\left( m_\varphi\right)  \right) &=\left\lbrace \psi \in \mathscr{E}\left( \mathcal{M}\right) : \braket{\varphi s,\psi}=0\mbox{, } \ \forall s\in\mathscr{E}\left( \mathcal{M}\right)' \right\rbrace \\ &= \left\lbrace \psi \in \mathscr{E}\left( \mathcal{M}\right) : \varphi\psi=0 \right\rbrace=\left\lbrace 0\right\rbrace\mbox{,}
		\end{split}  
	\end{equation*}
from which we conclude $\operatorname{Im}\left( m_\varphi\right) $ is dense.
\end{proof}
\section{Extension of Feynman amplitudes}\label{expliFA}

We denote by $\Delta_g$ the Laplace-Beltrami operator corresponding to the metric $g$ and consider the Green function $G$ in $\mathscr{D}\left( \mathcal{M}^2\right)' $ of the operator $\Delta_g+m^2$, where $m$ belongs to  $\mathbb{R}_{\geq 0}$.
$G$ is the Schwartz kernel of the operator inverse of $\Delta_g+m^2$ which always exists when $\mathcal{M}$ is compact and $-m^2$ does not belong to $\operatorname{Spec}\left( \Delta_g\right) $.
In the noncompact case, the general existence and uniqueness result for the Green function usually depends on the global properties of $\Delta_g$ and $\left( \mathcal{M},g \right) $.

However, if $G$ exists, then there is a fundamental result about asymptotics of $G$ near the diagonal 
	\begin{equation}\label{ñlkj}
		D_2=\left\lbrace (x,y)\in\mathcal{M}^2: x=y \right\rbrace\subseteq\mathcal{M}^2.
	\end{equation}
\begin{lem}[\textit{cf.} \cite{vietdang}, Lemma 4.1]\label{silvia}
Let $\left( \mathcal{M},g \right) $ be a smooth Riemannian manifold and $\Delta_g$ the corresponding Laplace-Beltrami operator.
If $G$ in $\mathscr{D}\left( \mathcal{M}^2\right)' $ is the fundamental solution of $\Delta_g+m^2$, then $G$ is tempered along the diagonal $D_2$, given by \eqref{ñlkj}. 
\end{lem}
\begin{proof}
Temperedness is a local property therefore it suffices to prove the Lemma for some compact domain 
	\begin{equation*}
		K\times K\subseteq \mathbb{R}^d\times\mathbb{R}^d
	\end{equation*}
and in the case that $g$ is a Riemannian metric on $\mathbb{R}^d$.
This means that given a positive integer $k$ and a  compact subset $K\times K$ of $\mathbb{R}^d\times\mathbb{R}^d$, we seek to show there exist two positive numbers $C$ and $s$ such that the inequality
	\begin{equation} \label{rabinovitch}
		\left| \partial_x^\alpha \partial_y^\beta G(x,y)\right| \leq C\left[1+d\left( \left( x,y \right) , D_2 \right)  ^{-s}\right]  
	\end{equation}
holds for every $\left( x,y\right) $ in $K\times K\setminus D_2$, and every pair of multi-indices $\alpha$ and $\beta$ such that $|\alpha|+|\beta|\leq k$.

The differential operator $\Delta_g+m^2$ is elliptic with smooth coefficients, and $G$ is an fundamental solution of $\Delta_g+m^2$.
In other words, it is a parametrix of $\Delta_g+m^2$ constructed by means of an elliptic pseudodifferential operator with polyhomogeneous symbol (see \cite{shimakura}, Thm. 2.7).

Let $D=\operatorname{diam}\left( K\right) $ be the diameter of $K$, and set $z=y-x$ and 
	\begin{equation*}
		\mathcal{E}(x,z):=G(x,x+z).
	\end{equation*}
We will treat first the case in which $0<|z|\leq 1$. By \cite{shimakura} Thm. 3.3, there exist  two sequences of functions  
	\begin{equation*}
		\big( \mathcal{A}_q\left( x,\zeta\right) \big)_q \mbox{ and }\big( \mathcal{B}_q\left( x,\zeta\right)\big)_q \mbox{,}
	\end{equation*}
smooth on $\mathbb{R}^{d}$ with respect to $x$ and real analytic on $\mathbb{S}^{d-1}$ with respect to $\zeta$, such that  $\mathcal{E}$ satisfies the following estimate: for every positive integer $N$, multi-indices $\alpha$, $\beta$, and every compact subset $K$ of $\mathbb{R}^{d}$ there exists a positive number $c=c\left( N,\alpha,\beta,K\right) $ such that
	\begin{equation*}
		\left| \partial_x^\alpha \partial_z^\beta\left[\mathscr{E}\left( x,z\right) -\sum_{q=0}^{N}\left| z\right|^{2+q-d} \left\lbrace \mathcal{A}_q\left(x,\frac{z}{|z|} \right)\operatorname{log}|z| +\mathcal{B}_q\left(x,\frac{z}{|z|} \right)\right\rbrace  \right] \right| \leq c|z|^{\lambda}
	\end{equation*}
for every pair  $(x,z)$ in $K\times \mathbb{R}^d$, where $z$ is such that $0<|z|\leq 1$, and where 
	\begin{equation*}
		\lambda=\min \left\lbrace 0, 2+N-|\beta|-d\right\rbrace.
	\end{equation*}
Now, for arbitrarily chosen multi-indices $\alpha$, $\beta$, there always is some positive integer $N$ such that $ 2+N-|\beta|-d >0$ which tells us that  $\lambda=0$ in the estimate above.
Then, for arbitrarily  chosen multi-indices $\alpha$, $\beta$ and every compact subset $K$ of $\mathbb{R}^{d}$ there exists a sufficiently large 
positive integer $N$ and a positive number $c=c(N,\alpha,\beta,K)$  such that 
	\begin{equation}\label{libe}
		\left| \partial_x^\alpha \partial_z^\beta\mathscr{E}\left( x,z\right) \right| \leq c+\left|\partial_x^\alpha \partial_z^\beta \sum_{q=0}^{N}\left| z\right|^{2+q-d} \left\lbrace   \mathcal{A}_q\left(x,\frac{z}{|z|} \right)\operatorname{log}|z| +\mathcal{B}_q\left(x,\frac{z}{|z|}\right)\right\rbrace  \right|
	\end{equation}
for every pair $(x,z)$ in $K\times \mathbb{R}^d$, where $0<|z|\leq 1$.

Observe that the second term on the right hand side of equation \eqref{libe} is bounded by
	\begin{equation}\label{mantis}
		\begin{split}
 			\sum_{q=0}^{N}&\left| \partial_z^\beta\left|  z\right|^{2+q-d}\right|  \left| \partial_x^\alpha  \mathcal{A}_q\left(x,\frac{z}{|z|} \right)\operatorname{log}|z| +\partial_x^\alpha \mathcal{B}_q\left(x,\frac{z}{|z|}\right) \right|+\dots\\  &+\sum_{q=0}^{N} \left|  z\right|^{2+q-d} \left| \partial_z^\beta\left[  \partial_x^\alpha  \mathcal{A}_q\left(x,\frac{z}{|z|} \right)\operatorname{log}|z| \right]\right|+ \sum_{q=0}^{N} \left|  z\right|^{2+q-d} \left| \partial_z^\beta\partial_x^\alpha \mathcal{B}_q\left(x,\frac{z}{|z|}\right) \right|.
		\end{split}
	\end{equation}
For the sake of simplicity, in what remains of the proof the word \emph{const.} will denote some constant which is independent of $z$, and whose value is irrelevant.

By the Leibniz rule we have  
	\begin{equation} \label{cotapot}
		\left| \partial_z^\beta\left|  z\right|^{2+q-d}\right| \leq const. \left|  z\right|^{2+q-d-|\beta|}\leq const.  \frac{1}{\left|  z\right|^{d+|\beta|}}
	\end{equation}
and
	\begin{equation} \label{cotalog}
		\left| \partial_z^\beta\operatorname{log}\left|  z\right|\right| \leq const.\frac{1}{ \left|  z\right|^{|\beta|+1}}.
	\end{equation}
Moreover, taking into account that for every $\varepsilon>0$  
	\begin{equation*}
		\lim\limits_{u\rightarrow 0^+}u^\varepsilon \operatorname{log}u=0\mbox{,}
	\end{equation*}
we have
	\begin{equation*}
		\operatorname{log}|z|\leq \frac{const.}{|z|} . 
	\end{equation*}
	
In addition, as  $\mathcal{A}_q\left(x,\zeta  \right) $ and $\mathcal{B}_q\left(x,\zeta  \right)$ are smooth on $\mathbb{R}^{d}$ with respect to $x$, and real analytic on $\mathbb{S}^{d-1}$ with respect to $\zeta$, we have that their derivatives up to order $\alpha$ with respect to $x$, and up to order  $\beta$ with respect to $z$ are bounded by a factor of the type
	\begin{equation}\label{beinmyvideo}
		\left| \partial_z^\nu\partial_x^\rho \mathcal{A}_q\left(x,\frac{z}{|z|}\right) \right|\mbox{,}\ \left| \partial_z^\nu\partial_x^\rho \mathcal{B}_q\left(x,\frac{z}{|z|}\right) \right|\leq const.\frac{1}{|z|^{|\nu|}} \leq const. \frac{1}{|z|^{|\beta|}} 
	\end{equation}
for every pair $(x,z)$ in $K\times \mathbb{R}^d$, where  $0<|z|\leq 1$ and $\nu\leq\beta$, $\rho\leq\alpha$.

We can now estimate the first sum in \eqref{mantis} by
	\begin{equation*}
		\begin{split}
			\sum_{q=0}^{N}&\left| \partial_z^\beta\left|  z\right|^{2+q-d}\right|  \left| \partial_x^\alpha  \mathcal{A}_q\left(x,\frac{z}{|z|} \right)\operatorname{log}|z| +\partial_x^\alpha \mathcal{B}_q\left(x,\frac{z}{|z|}\right) \right|\\ & \leq  const.\sum_{q=0}^{N} \frac{1}{\left|  z\right|^{d+|\beta|}}\left( \frac{1}{|z|} + 1\right) \leq  const. \frac{1}{\left|  z\right|^{d+|\beta|+1}} \leq const.  \frac{1}{|z|^{d+2|\beta|+1}}.
		\end{split}
	\end{equation*}

The second sum in  \eqref{mantis} can be treated in the same fashion.
First, by Leibniz's rule  we have
	\begin{equation*}
		\begin{split}
			\left| \partial_z^\beta\left[  \partial_x^\alpha  \mathcal{A}_q\left(x,\frac{z}{|z|} \right)\operatorname{log}|z| \right]\right|& \leq \sum_{\gamma \leq \beta} \binom{\beta}{\gamma}\left| \partial_z^\gamma  \partial_x^\alpha  \mathcal{A}_q\left(x,\frac{z}{|z|} \right)\partial_z^{\beta-\gamma} \operatorname{log}|z|\right|\\ &\leq  \sum_{\gamma \leq \beta} const. \frac{1}{|z|^{|\beta|}}\frac{1}{\left|  z\right|^{|\beta|+1}}\leq const.\frac{1}{\left|  z\right|^{2|\beta|+1}}.
		\end{split} 
	\end{equation*} 
Therefore,
	\begin{equation*}
		\begin{split}
			\sum_{q=0}^{N} \left|  z\right|^{2+q-d} \left| \partial_z^\beta\left[  \partial_x^\alpha  \mathcal{A}_q\left(x,\frac{z}{|z|} \right)\operatorname{log}|z| \right]\right|&\leq const. \sum_{q=0}^{N}\frac{1}{\left|  z\right|^{d}} \frac{1}{|z|^{2|\beta|+1}} \\ &\leq const.  \frac{1}{|z|^{d+2|\beta|+1}}.
		\end{split}
	\end{equation*}
Finally, we write the estimate for the last sum in \eqref{mantis} as
	\begin{equation*}
		\begin{split}
			\sum_{q=0}^{N} \left|  z\right|^{2+q-d} \left| \partial_z^\beta\partial_x^\alpha \mathcal{B}_q\left(x,\frac{z}{|z|}\right) \right|&\leq const. \sum_{q=0}^{N}\frac{1}{\left|  z\right|^{d}} \frac{1}{|z|^{|\beta| }} \leq const.  \frac{1}{|z|^{d+|\beta|}}\\ &\leq const.  \frac{1}{|z|^{d+2|\beta|+1}}.
		\end{split}
	\end{equation*}
Therefore, the expression in \eqref{mantis} is bounded by a multiple of  $|z|^{-\left( d+2|\beta|+1\right) }$.
Using this in the inequality in \eqref{libe} we obtain that
	\begin{equation}\label{gregory}
		\left| \partial_x^\alpha \partial_z^\beta\mathscr{E}\left( x,z\right) \right| \leq c+  const.  \frac{1}{|z|^{d+2|\beta|+1}}\leq const. \left( 1+\frac{1}{|z|^{d+2|\beta|+1}}\right) 
	\end{equation}
for every pair $(x,z)$ in $K\times \mathbb{R}^d$, where $0<|z|\leq 1$.

The next step is to show that the previous inequality holds for $1\leq|z|\leq D$. 
Since the set 
	\begin{equation*}
		K\times \left\lbrace 1\leq|z|\leq D \right\rbrace\subseteq \mathbb{R}^d\times \mathbb{R}^d
	\end{equation*} 
is compact, there exists a positive constant $M$ such that  
	\begin{equation*}
		\left| \partial_x^\alpha \partial_z^\beta\mathscr{E}\left( x,z\right) \right| \leq M
	\end{equation*}
for every pair of multi-indices $\alpha$ and $\beta$ such that $|\alpha|+|\beta|\leq k$.
We also have that 
	\begin{equation*}
		\begin{split}
			M &=M\frac{D^{d+|\beta|+1}}{D^{d+|\beta|+1}}\leq M\frac{D^{d+|\beta|+1}}{|z|^{d+|\beta|+1}}\\ & \leq MD^{d+|\beta|+1}\left( 1+\frac{1}{|z|^{d+|\beta|+1}}\right) = const. \left( 1+\frac{1}{|z|^{d+|\beta|+1}}\right).
		\end{split}
	\end{equation*}
Therefore, \eqref{gregory} also holds for every pair  $(x,z)$ in $K\times \mathbb{R}^d$ with $1\leq|z|\leq D$.
Considering that 
	\begin{equation*}
		|z|\geq d\left(\left( x,x+z\right) , D_2 \right)=d\left(\left( x,y\right) , D_2\right)\mbox{,}
	\end{equation*} 
and that
	\begin{equation*}
		\partial_x^\alpha \partial_z^\beta\mathscr{E}\left( x,z\right) = \partial_x^\alpha \partial_y^\beta G(x,y)\mbox{,}
	\end{equation*}
we conclude that
	\begin{equation}\label{gregory2}
		\left| \partial_x^\alpha \partial_y^\beta G(x,y)\right|\leq const. \left[  1+d\left(\left( x,y\right) , D_2 \right)^{-(d+|\beta|+1)}\right] 
	\end{equation}
for every pair $(x,y)$ in  $K\times K\setminus D_2$.
Therefore, $G(x,y)$ is tempered along $D_2$. 
\end{proof}
In Theorem \ref{dinsky} we present a result  that involves the extension of the so called Feynman amplitudes.
This requires the introduction of the following  definitions.
\begin{defi}\label{defidiag}
For every finite subset $I$ of $\mathbb{N}$ and open subset $U$ of $\mathcal{M}$, we define the \emph{configuration space of}  \emph{$|I|$ particles in  $U$ labelled by the set $I$}, as \index{configuration space}
	\begin{equation*}
		U^I=\left\lbrace\left( x_i \right)_{i\in I}: x_i \in U\mbox{, }\forall i\in I    \right\rbrace.
	\end{equation*}

We will distinguish two types of diagonals in $U^I$.
The \emph{big diagonal}\index{big diagonal} is given by
	\begin{equation*}
		D_I=\left\lbrace \left( x_i \right)_{i\in I}:\exists i\mbox{,}j\in I\mbox{, }  i\neq j\mbox{, } x_i=x_j \right\rbrace\mbox{,}
	\end{equation*}
and represents configurations where at least two particles in $U^I$  collide. 
Whenever $J$ is a subset of $I$, we set
	\begin{equation*}
		d_{I,J}=\left\lbrace \left( x_i \right)_{i\in I}:\forall i\mbox{, }j\in J\mbox{, } x_i=x_j \right\rbrace. 
	\end{equation*}
In particular,   the \emph{small diagonal} $d_{I,I}$ \index{small diagonal}represents configurations where  all the  particles in $U^I$ collapse  over the same element. 

The configuration space $\mathcal{M}^{\left\lbrace 1,\dots ,n \right\rbrace }$ will be denoted by $\mathcal{M}^n$, for simplicity.
The corresponding big diagonal $D_{\left\lbrace 1,\dots ,n \right\rbrace }$   will be denoted by $D_n$.
In addition, we will use the following compact notations:
	\begin{equation}\label{othernotations}
		\begin{split}
			d_{\left\lbrace i,j,\dots,k\right\rbrace}&:=d_{\left\lbrace i,j,\dots ,k \right\rbrace ,\left\lbrace i,j,\dots,k \right\rbrace }\mbox{,}\\d_{n,\left\lbrace i,j,\dots,k\right\rbrace}&:=d_{\left\lbrace 1,\dots ,n \right\rbrace ,\left\lbrace i,j,\dots,k \right\rbrace }\mbox{,}\\d_n&:=d_{\left\lbrace 1,\dots ,n \right\rbrace ,\left\lbrace 1,\dots,n \right\rbrace }.
		\end{split}
	\end{equation}
\end{defi}
\begin{thm}[See \cite{vietdang}, Thm. 4.1]\label{dinsky}
Let $\left( \mathcal{M}, g\right)$ be a smooth Riemannian manifold, $\Delta_g$ the corresponding Laplace-Beltrami operator and $G$ the Green function of $\Delta_g +m^2$.
Then, all \emph{Feynman amplitudes}\index{Feynman amplitude} of the form
	\begin{equation}\label{vainilla}
		\prod_{1\leq i<j\leq n } G^{n_{ij}}\left( x_i,x_j\right) \in \mathscr{E}\left( \mathcal{M}^n\setminus D_n \right)\mbox{,}\quad n_{ij}\in \mathbb{N}\mbox{,} 
	\end{equation}
are tempered along $D_n$, and therefore extendible to $\mathcal{M}^n$. 
\end{thm}
\begin{proof}

As $d_{n,\left\lbrace i,j  \right\rbrace }$ is contained in $D_n$,  we have
	\begin{equation*}
		d\left( \left( x_i,x_j\right) ,d_{\left\lbrace  i,j  \right\rbrace }\right)^{-s}=d\left( \left(x_1,\dots ,x_n \right) ,d_{n,\left\lbrace  i,j  \right\rbrace }\right)^{-s}\leq d\left( \left(x_1,\dots ,x_n \right),D_n \right)^{-s}
	\end{equation*}
for every $s\geq 0$.
The previous inequality, together with the fact that $G\left( x_i,x_j\right)$ is tempered along $d_{\left\lbrace i,j  \right\rbrace }$ by Theorem \ref{silvia} imply that  $G\left( x_i,x_j\right)$ belongs to $\mathcal{T}\left( D_n,\mathcal{M}^n\right) $.
Since $\mathcal{T}\left( D_n,\mathcal{M}^n\right) $ is an algebra, we have
	\begin{equation*}
		\prod_{1\leq i<j\leq n } G^{n_{ij}}\left( x_i,x_j\right) \in \mathcal{T}\left( D_n,\mathcal{M}^n\right)\mbox{,}
	\end{equation*}
and \eqref{vainilla} is therefore extendible to $\mathcal{M}^n$, by Corollary \ref{maggotbrain2}.
\end{proof}

\chapter{Renormalization maps}\label{tengofiebree}

In \S \ref{expliFA}, specifically in \eqref{vainilla}, Feynmann amplitudes were introduced and it has been shown they are extendible (Theorem \ref{dinsky}).
However, in quantum field theory, renormalization is not only involved with the extension of Feynman amplitudes in configuration space.
Renormalization is intended to be applicable over the algebra which they generate along with $\mathcal{C}^\infty$ functions, namely
	\begin{equation}\label{moustache}
		\mathcal{O}(D_I,\Omega):=\left\langle f\prod_{i<j\in I}\left.G^{n_{ij}}\left(x_i,x_j \right)\right|_{\Omega\setminus D_I}: n_{ij}\in\mathbb{N} \ \forall i<j\in I\mbox{,} \ f\in\mathscr{E}\left( \Omega\right)  \right\rangle _{\mathbb{C}}\mbox{,}
	\end{equation}
where $I$ is a finite subset of $\mathbb{N}$, $\Omega$ is an open subset of the product $\mathcal{M}^I$ of $\left| I\right|$ copies of the Riemannian manifold $\mathcal{M}$, and $D_I$ is given in Definition \ref{defidiag}.
Recall that for any open subset $\Omega$ of $\mathcal{M}^I$, we denote by $\mathcal{T}\left(D_I,\Omega \right) $ the algebra of tempered functions along $D_I$.
As the generators of $\mathcal{O}(D_I,\Omega)$ are tempered along $D_I$ (see Theorem   \ref{dinsky}), we have that  $\mathcal{O}(D_I,\Omega)$ is contained in $\mathcal{T}\left(D_I,\Omega \right) $. 

We define a collection 
	\begin{equation}\label{renmaps}
		\mathscr{R}=\left\lbrace \mathcal{R}_{\Omega}^I : \mathcal{O}(D_I,\Omega)\rightarrow \mathscr{D}(\Omega)':I\subseteq \mathbb{N}\mbox{, } |I|<\infty\mbox{,} \ \Omega\subseteq \mathcal{M}^I\ open \right\rbrace 
	\end{equation}
of objects $\mathcal{R}_{\Omega}^I$, called \emph{renormalization maps}\index{renormalization map}. 
Each of these maps is defined to be an extension operator, 
	\begin{equation*}
		\mathcal{R}_{\Omega}^I: \mathcal{O}(D_I,\Omega)\rightarrow \mathscr{D}\left( \Omega\right)'\mbox{,}
	\end{equation*}
and will be used in the extension procedure we intend to define in the present chapter.
This extension procedure should satisfy some consistency conditions in order to be compatible with the fundamental requirement of locality.

Before we begin, we introduce the following notation, for simplicity: if $\Omega=\mathcal{M}^I$,
we will write $\mathcal{R}_{I}$ instead of $\mathcal{R}_{\mathcal{M}^I}^I$.
If $I=\iota_n$, where
	\begin{equation}\label{defiota}
		\iota_n:=\left\lbrace 1,\dots ,n\right\rbrace
	\end{equation}
and $\Omega=\mathcal{M}^n$ we will write $\mathcal{R}_{n}$ instead of $\mathcal{R}_{\mathcal{M}^{\iota_n}}^{\iota_n}$.

In the first section, we will define the axioms the renormalization maps should satisfy in the extension procedure, and in the second section we will show that such a family of renormalization maps does exist.
\section{Axioms for renormalization maps consistent with locality} \label{axioms}

The collection of renormalization maps \eqref{renmaps} will be required to satisfy the  axioms listed below (\textit{cf.} \cite{vietdang}, Def. 4.1):
	\begin{itemize}
	\item[1.] For every finite subset $I$ of $\mathbb{N}$ and every open subset $\Omega$ of $\mathcal{M}^I$, 
		\begin{equation*}
			\mathcal{R}_{\Omega}^I: \mathcal{O}\left( D_I,\Omega\right) \rightarrow \mathscr{D}\left( \Omega\right)' 
		\end{equation*}
	is a linear extension operator.
	\item[2.] For every inclusion of open subsets 
		\begin{equation}\label{inclu}
			\Omega_1\subseteq \Omega_2\subseteq \mathcal{M}^I\mbox{,}
		\end{equation}
	 we require the following diagram
		\begin{equation*}
			\xymatrix{
			\mathcal{O}\left( D_I,\Omega_2\right) \ar@{->}[d]_\rho   \ar@{->}[r]^{\mathcal{R}_{\Omega_2}^{I}}& \mathscr{D}\left( \Omega_2\right)' \ar@{->}[d]\\ \mathcal{O}\left( D_I,\Omega_1\right) \ar@{->}[r]^{\mathcal{R}_{\Omega_1}^{I}} & \mathscr{D}\left( \Omega_1\right)' },
		\end{equation*}
	to be commutative,	where the vertical arrows correspond to the restriction to $\Omega_1$.
	In other words, we require that 
		\begin{equation*}
			\braket{\mathcal{R}_{\Omega_2}^{I}(f),\varphi}=\braket{ \mathcal{R}_{\Omega_1}^{I}(f),\varphi}\mbox{,}
		\end{equation*}  
	for every $f$ in $\mathcal{O}(D_I,\Omega_2)$ and every $\varphi$ in $\mathscr{D}\left( \Omega_1\right) $.
	Note that all the generators in the algebra $\mathcal{O}(D_I,\Omega)$ come from restricting some Feynmann amplitude originally defined on $\mathcal{M}^I\setminus D_I$. 
	Thus, the map $\rho$ is surjective and this axiom implies that, for every inclusion of open sets \eqref{inclu}, the operator $\mathcal{R}^I_{\Omega_1}$ is completely determined by the operator $\mathcal{R}^I_{\Omega_2}$, whenever the latter is defined.
	\item [3.] For every $G$ in $\mathcal{O}(D_I,\Omega)$ and every function $f$ in $\mathscr{E}\left( \Omega\right) $
		\begin{equation*}
			\mathcal{R}^I_{\Omega}\left( f G\right)=f\mathcal{R}^I_{\Omega}\left(  G\right). 
		\end{equation*} 
	\item [4.] \emph{Factorization axiom}\index{factorization axiom}: for every nontrivial partition\index{nontrivial partition}\footnote{\label{defnontrivpart}$I_1\sqcup I_2= I$ stands for $I_1\cup I_2= I$ and $I_1\cap I_2= \emptyset$; "nontrivial" means that $I_1\neq\emptyset\neq I_2$.} $I_1\sqcup I_2= I$
		\begin{equation}\label{todor}
			\left.\mathcal{R}_{I}\left( G_I\right)\right\vert_{C_{\left\lbrace I_1I_2\right\rbrace }}  =\left(  \left.\mathcal{R}_{I_1}\left(  G_{I_1}\right)  \otimes \mathcal{R}_{I_2}\left(  G_{I_2} \right)   \right)  G_{\left\lbrace I_1,I_2\right\rbrace }\right\vert_{C_{\left\lbrace I_1I_2\right\rbrace }}\mbox{,}
		\end{equation} 
	where, if $I=\left\lbrace i_1,\dots,i_n \right\rbrace $,
		\begin{equation}
			\label{loscij}
			C_{\left\lbrace I_1I_2\right\rbrace }=\left\lbrace \left( x_{i_1},\dots ,x_{i_n}\right): \forall (i,j)\in I_1\times I_2, \ x_i\neq x_j \right\rbrace \subseteq \mathcal{M}^I\mbox{,}
		\end{equation}
		\begin{equation}
			\label{losGI1I2}
			G_{\left\lbrace I_1,I_2\right\rbrace }:= \prod_{\substack{ (i,j)\in I_1\times I_2\\i<j}} G^{n_{ij}}\left( x_i,x_j \right)
		\end{equation}   
	and, for every  finite subset $J$ of $\mathbb{N}$, we define
		\begin{equation}
			\label{losGJ}
			G_{J}:=\prod_{i<j\in J} G^{n_{ij}}\left( x_i,x_j\right).
		\end{equation}
	Observe that 	
		\begin{equation*}
			G_J\in    \mathcal{O}(D_J,\mathcal{M}^J)\subseteq \mathcal{T}(D_J,\mathcal{M}^J)\mbox{,}
		\end{equation*}	
	by  Theorem \ref{dinsky}.
	We also note that the right hand side of \eqref{todor} contains a well-defined product of distributions  due to the restriction to the domain $C_{\left\lbrace I_1I_2\right\rbrace }$, where  $G_{\left\lbrace I_1,I_2 \right\rbrace }$ is a smooth  function.
	\end{itemize}
	
\begin{rem}
For the definition of the factorization axiom we chose not to follow \cite{vietdang}. 
Instead, we preferred the clearer version of this axiom appearing in \cite{todorov}.
\end{rem}
\section{Existence of renormalization maps: recursive procedure}

This section will be dedicated to prove recursively the existence of renormalization maps on general Riemannian manifolds.
Recall that for any open subset $\Omega$ of $\mathcal{M}^I$, where $I$ is a finite subset of $\mathbb{N}$, $\mathcal{O}\left( D_I,\Omega\right) $ denotes the algebra \eqref{moustache}.
We are to prove the following important result that appears in \cite{vietdang}.
We present here a correct and complete proof of it.
 
\begin{thm}[\textit{cf.} \cite{vietdang}, Thm. 4.2]\label{exrenmap}
There exists a collection of renormalization maps 
	\begin{equation}\label{genoa}
		\mathscr{R}=\left\lbrace \mathcal{R}_{\Omega}^{I}: \mathcal{O}\left( D_I,\Omega\right) \rightarrow \mathscr{D}\left( \Omega\right)'  :I\subseteq \mathbb{N}\mbox{, } |I|<\infty\mbox{,} \ \Omega\subseteq \mathcal{M}^I\ open \right\rbrace\mbox{,} 
	\end{equation}
that satisfies the axioms described in \S \ref{axioms}.
\end{thm}

In order to prove Theorem \ref{exrenmap} we first observe the following facts.
\begin{rem}\label{eldato}
By the second axiom and because of the fact that every element in  $\mathcal{O}\left( D_I,\Omega\right) $ is the restriction of some element in $\mathcal{O}\left( D_I,\mathcal{M}^I\right) $, we have that $\mathcal{R}_{\Omega}^{I}$ is completely determined by $\mathcal{R}_{I}$. 
Therefore, to establish the existence of the collection of renormalization maps \eqref{genoa}, it suffices to prove the existence of the linear maps $\mathcal{R}_{I}$, where $I$ runs over all finite subsets of $\mathbb{N}$.

Moreover, by the first and third axioms it suffices to establish the existence of the operator $\mathcal{R}_{I}\left(G_{I}\right)$ for a generic Feynman amplitude $G_I$ $($see \eqref{losGJ}$)$, in a way that it satisfies the last axiom.

In addition, for an arbitrary index set $I$ of positive integers consisting of $n$ elements there is a unique monotonic isomorphism
	\begin{equation*}
		\iota_n:=\left\lbrace 1,\dots ,n\right\rbrace\simeq I\mbox{;}
	\end{equation*}
and under this isomorphism we identify 
	\begin{equation*}
		\mathscr{D}\left( \mathcal{M}^I\right)' \simeq \mathscr{D}\left( \mathcal{M}^{n}\right)'  \qquad \mbox{and}\qquad \mathcal{O}\left( D_I,\mathcal{M}^I\right) \simeq \mathcal{O}\left( D_{n},\mathcal{M}^{n}\right).
	\end{equation*}
Using these identifications we lift the map $\mathcal{R}_{n}$ to a linear map
	\begin{equation*}
		\mathcal{R}_{I}:\mathcal{O}\left( D_I,\mathcal{M}^I\right) \rightarrow \mathscr{D}\left( \mathcal{M}^I\right)'.
	\end{equation*}

Therefore, to prove Theorem \ref{exrenmap} it suffices to prove the existence of  $\mathcal{R}_{n}\left(G_{\iota_n}\right)$  satisfying the last axiom (factorization property), for every positive integer $n$. 
\end{rem}
The idea of proof of Theorem \ref{exrenmap} is as follows.
By Remark \ref{eldato}, we need only define the extension $\mathcal{R}_{n}\left(G_{\iota_n}\right)$ for every $n$ in $\mathbb{N}$.
We shall assume recursively that the problem of extension is already solved for every proper subset  $ J$ of $\iota_n$. 
Namely, we suppose that for every such $J$ we are given a distribution $\mathcal{R}_{J}\left(  G_{J}\right)$ in $\mathscr{D}\left( \mathcal{M}^J\right)'$ with the property that for every nontrivial partition $J_1{\sqcup}J_2 = J$ equation \eqref{todor} holds with $I$, $I_1$ and $I_2$ replaced by $J$, $J_1$ and $J_2$, respectively.
It is convenient to set
	\begin{equation}\label{convention}
		\mathcal{R}_{I}\left(  G_{I}\right)=1\mbox{,}\qquad \mbox{ if }|I|\leq 1.
	\end{equation}
Thus, the starting point of the renormalization recursion will be the two point case, $I^2= \left\lbrace 1,2\right\rbrace $. 
The existence of $\mathcal{R}_{2}\left( G_{2}\right) $ is guaranteed by Theorem \ref{dinsky}.
The factorization axiom is satisfied trivially in this case by the extra assumption \eqref{convention}.

However, before giving the proof of the inductive step of Theorem \ref{exrenmap} (which is done in \S \ref{lademospostaa}), we shall  first ascertain some technical results presented in the following subsection.

\subsection{Technical results for the proof of existence of renormalization maps }

Before proving the existence of renormalization maps we will need some results described in the following lemmas.

The first lemma  states that $\mathcal{M}^n\setminus d_n$ (see Definition \ref{defidiag}) can be partitioned as a union of the open sets  $C_{\left\lbrace I_1I_2\right\rbrace }$, on which the renormalization map  $\mathcal{R}_{n}$ can be factorized (see \eqref{todor}).
\begin{lem}[See \cite{todorov}, Lemma 2.2]\label{covlemma}
Let $\mathcal{M}$ be a smooth d-dimensional manifold and, for every nontrivial partition $I_1\sqcup I_2= \iota_n $ (see Definition \ref{defiota} and footnote \ref{defnontrivpart}), let $C_{\left\lbrace I_1I_2\right\rbrace }$ be given by  
	\begin{equation}\label{kenshin}
		C_{\left\lbrace I_1I_2\right\rbrace }=\left\lbrace \left( x_1,\dots ,x_n\right): \forall (i,j)\in I_1\times I_2\mbox{,} \ x_i\neq x_j \right\rbrace \subseteq \mathcal{M}^n. 
	\end{equation}
Then,
	\begin{equation}\label{union}
		\bigcup_{\substack{I_1\mbox{,}I_2\neq \emptyset \\
		I_1\sqcup I_2 =\iota_n  }} C_{\left\lbrace I_1I_2\right\rbrace }=\mathcal{M}^n\setminus d_n.
	\end{equation}
\end{lem}   
\begin{proof}
Let $\left(x_1,\dots,x_n \right)$ be an element of $\mathcal{M}^n\setminus d_n $.
Then, there are at least two different indices $j_1\neq j_2$ in $\iota_n$ with $x_{j_1}\neq x_{j_2}$.
We define $I_1$ to be the set of indices $j$ in $\iota_n$ such that $x_j=x_{j_1}$ and $I_2=\iota_n\setminus I_1$.
Then, the partition $\iota_n=I_1{\sqcup} I_2$ is proper and $\left(x_1,\dots,x_n \right)$ belongs to $C_{\left\lbrace I_1I_2\right\rbrace }$.
The other inclusion is trivial.
\end{proof}
\begin{lem}[See \cite{todorov}, Lemma 2.3]\label{consistency} 
Assume the recursion hypothesis, namely that for every proper subset  $ J$ of $\iota_n$ (see Definition \ref{defiota}) there exists a distribution $\mathcal{R}_{J}\left(  G_{J}\right)$ in $\mathscr{D}\left( \mathcal{M}^J\right)'$ with the property that for every nontrivial partition $J_1{\sqcup}J_2 = J$ (see footnote \ref{defnontrivpart}) equation \eqref{todor} holds with $I$, $I_1$ and $I_2$ replaced by $J$, $J_1$ and $J_2$, respectively.
Then, for every pair of  nontrivial partitions $I_1\sqcup I_2=\iota_n$ and $J_1\sqcup J_2=\iota_n$, we have
	\begin{equation}\label{consistency2}
		\begin{split}
			\left(\mathcal{R}_{I_1}\left(  G_{I_1}\right)  \otimes \mathcal{R}_{I_2}\left(  G_{I_2} \right)\right)     & \left.G_{\left\lbrace I_1,I_2\right\rbrace }\right\vert_{C_{\left\lbrace I_1I_2\right\rbrace }\cap C_{\left\lbrace J_1J_2\right\rbrace }}= \dots \\
			&\dots \left( \left.\mathcal{R}_{J_1}\left(  G_{J_1}\right)  \otimes \mathcal{R}_{J_2}\left(  G_{J_2} \right)   \right)   G_{\left\lbrace J_1,J_2\right\rbrace }\right\vert_{C_{\left\lbrace I_1I_2\right\rbrace }\cap C_{\left\lbrace J_1J_2\right\rbrace }}.
		\end{split}
	\end{equation}
\end{lem}
\begin{proof}
Let us introduce the sets 
	\begin{equation*}
		K_{a,b}:=I_a\cap J_b\mbox{,}
	\end{equation*}
some of which can be empty, and which form a partition of $\iota_n$
	\begin{equation*}
		\iota_n=K_{1,1}\ \sqcup\ K_{1,2}\ \sqcup\ K_{2,1}\ \sqcup\ K_{2,2}.
	\end{equation*}
In addition we have
	\begin{equation*}
		I_a=K_{a,1}\ \sqcup\ K_{a,2}\quad \mbox{and}\quad J_b=K_{1,b}\ \sqcup \ K_{2,b}.
	\end{equation*}
With the previous definitions the following relation holds
	\begin{equation}\label{ayuda1}
		C_{\left\lbrace I_1,I_2 \right\rbrace }\cap C_{\left\lbrace J_1,J_2 \right\rbrace }=C_{\left\lbrace K_{1,1},K_{1,2} \right\rbrace }\cap C_{\left\lbrace K_{2,1},K_{2,2} \right\rbrace }\cap C_{\left\lbrace K_{1,1},K_{2,1} \right\rbrace }\cap C_{\left\lbrace K_{1,2},K_{2,2} \right\rbrace }.
	\end{equation}
Then, in equation \eqref{consistency2} we can think each of the extensions restricted to any of the four sets appearing on the right hand side of \eqref{ayuda1}. 

Next, by the recursively assumed condition \eqref{todor}, it follows  that
	\begin{equation}\label{ayuda2}
		\left.\mathcal{R}_{I_a}\left(G_{I_a}\right)\right\vert_{C_{\left\lbrace K_{a,1}K_{a,2}\right\rbrace }}  =\left(  \left.\mathcal{R}_{K_{a,1}}\left(  G_{K_{a,1}}\right)  \otimes \mathcal{R}_{K_{a,2}}\left(  G_{K_{a,2}} \right)   \right)  G_{\left\lbrace K_{a,1},K_{a,2}\right\rbrace }\right\vert_{C_{\left\lbrace K_{a,1}K_{a,2}\right\rbrace }}\mbox{,}
	\end{equation}
and 
	\begin{equation}\label{ayuda3}
		\left.\mathcal{R}_{J_b}\left(G_{J_b}\right)\right\vert_{C_{\left\lbrace K_{1,b}K_{2,b}\right\rbrace }}  =\left(\mathcal{R}_{K_{1,b}}\left(  G_{K_{1,b}}\right)  \otimes \mathcal{R}_{K_{2,b}}\left(  G_{K_{2,b}} \right)  \right)   \left.G_{\left\lbrace K_{1,b},K_{2,b}\right\rbrace }\right\vert_{C_{\left\lbrace K_{1,b}K_{2,b}\right\rbrace }}\mbox{,}
	\end{equation}
where in addition to the convention \eqref{convention}
we set
	\begin{equation*}
		G_{\left\lbrace I,J\right\rbrace }=1\mbox{,}\quad\mbox{ if } I=\emptyset \mbox{ or } J=\emptyset.
	\end{equation*}
By substituting these equalities in  equation \eqref{consistency2}  we see that it holds if and only if 
	\begin{equation*}
		\begin{split}
			&\left(  \mathcal{R}_{K_{1,1}}\left(  G_{K_{1,1}}\right)\right. 
			\otimes
			\left.\mathcal{R}_{K_{1,2}}\left(  G_{K_{1,2}} \right)\right) 
			\left. G_{\left\lbrace K_{1,1},K_{1,2}\right\rbrace }\right\vert_{C_{\left\lbrace K_{1,1}K_{1,2}\right\rbrace }}
			\otimes
			\dots
			\\
			\dots
			&\otimes 
			\left( \mathcal{R}_{K_{2,1}}\left(  G_{K_{2,1}}\right)  \otimes
			\mathcal{R}_{K_{2,2}}\left(  G_{K_{2,2}} \right)\right) 
			\left. G_{\left\lbrace K_{2,1},K_{2,2}\right\rbrace }\right\vert_{C_{\left\lbrace K_{2,1}K_{2,2}\right\rbrace }} 
			\left.  G_{\left\lbrace I_1,I_2\right\rbrace }\right\vert_{C_{\left\lbrace I_1I_2\right\rbrace }\cap C_{\left\lbrace J_1J_2\right\rbrace }}
			\\
			&  
			=\left( \mathcal{R}_{K_{1,1}}\left(  G_{K_{1,1}}\right) \right.
			\otimes \left.
			\mathcal{R}_{K_{2,1}}\left(  G_{K_{2,1}} \right)\right) 
			\left. G_{\left\lbrace K_{1,1},K_{2,1}\right\rbrace }\right\vert_{C_{\left\lbrace K_{1,1}K_{2,1}\right\rbrace }}
			\otimes
			\dots
			\\
			\dots
			&\otimes 
			\left( \mathcal{R}_{K_{1,2}}\left(  G_{K_{1,2}}\right)
			\otimes
			\mathcal{R}_{K_{2,2}}\left(  G_{K_{2,2}} \right)\right) 
			\left. G_{\left\lbrace K_{1,2},K_{2,2}\right\rbrace }\right\vert_{C_{\left\lbrace K_{1,2}K_{2,2}\right\rbrace }} 
			\left. G_{\left\lbrace J_1,J_2\right\rbrace }\right\vert_{C_{\left\lbrace I_1I_2\right\rbrace }\cap C_{\left\lbrace J_1J_2\right\rbrace }}.
		\end{split}
	\end{equation*}
By inspection, the last equality holds if and only if
	\begin{equation}\label{consistency3}
		\left(  G_{\left\lbrace K_{1,1},K_{1,2}\right\rbrace }
		\otimes
		G_{\left\lbrace K_{2,1},K_{2,2}\right\rbrace }\right) 
		G_{\left\lbrace I_1,I_2\right\rbrace }=
		\left(  G_{\left\lbrace K_{1,1},K_{2,1}\right\rbrace }
		\otimes
		G_{\left\lbrace K_{1,2},K_{2,2}\right\rbrace }\right) 
		G_{\left\lbrace J_1,J_2\right\rbrace }
	\end{equation}
holds on $C_{\left\lbrace I_1I_2\right\rbrace }\cap C_{\left\lbrace J_1J_2\right\rbrace}$.
Therefore, it suffices to prove this last assertion.
To achieve this, it only remains to state that both sides of \eqref{consistency3} are equal to  
	\begin{equation*}
		G_{\left\lbrace\! K_{1,1}K_{1,2}K_{2,1}K_{2,2}\!\right\rbrace }\!=\!G_{\left\lbrace\! K_{1,1},K_{1,2}\!\right\rbrace }G_{\left\lbrace\! K_{1,1},K_{2,1}\!\right\rbrace }G_{\left\lbrace\! K_{1,1},K_{2,2}\!\right\rbrace }G_{\left\lbrace\! K_{1,2},K_{2,1}\!\right\rbrace }
		G_{\left\lbrace\! K_{1,2},K_{2,2}\!\right\rbrace }
		G_{\left\lbrace\! K_{2,1},K_{2,2}\!\right\rbrace }
	\end{equation*}
on $C_{\left\lbrace I_1I_2\right\rbrace }\cap C_{\left\lbrace J_1J_2\right\rbrace}$, because on $C_{\left\lbrace I_1I_2\right\rbrace }\cap C_{\left\lbrace J_1J_2\right\rbrace}$ we have the equalities
	\begin{equation*}
		\begin{split}
			G_{\left\lbrace I_1,I_2 \right\rbrace }&=  G_{\left\lbrace K_{1,1},K_{2,1}\right\rbrace }G_{\left\lbrace K_{1,1},K_{2,2}\right\rbrace }G_{\left\lbrace K_{1,2},K_{2,1}\right\rbrace }
			G_{\left\lbrace K_{1,2},K_{2,2}\right\rbrace }\mbox{,}\\
			G_{\left\lbrace J_1,J_2 \right\rbrace }
			&= G_{\left\lbrace K_{1,1},K_{1,2}\right\rbrace }G_{\left\lbrace K_{1,1},K_{2,2}\right\rbrace }G_{\left\lbrace K_{1,2},K_{2,1}\right\rbrace }
			G_{\left\lbrace K_{2,1},K_{2,2}\right\rbrace }.
		\end{split}
	\end{equation*}
The lemma is thus proved.
\end{proof}
\begin{rem}\label{hesse}
Observe that under the recursion hypothesis,  condition \eqref{consistency2},  together with Lemma \ref{covlemma} and  Theorem \ref{gluedistrib},  imply that a distribution on $\mathcal{M}^n\setminus d_n$ is completely determined by defining it on each open set $C_{\left\lbrace I_1I_2\right\rbrace }$ $($see \eqref{kenshin}$)$ as the distribution given by the expression on the right hand side of \eqref{todor}.

By the second axiom of renormalization maps, the distribution on  $\mathcal{M}^n\setminus d_n$ constructed recursively in  this way will coincide with the restriction to  $\mathcal{M}^n\setminus d_n$ of the required extension $\mathcal{R}_n\left( G_{\iota_n}\right) $ we intend to define in the $n$-th step of the induction.
\end{rem}
The next lemma involves the construction of a tempered partition of unity associated to the cover described in Lemma \ref{covlemma}.
\begin{lem}\label{parlemma}
Let $\mathcal{M}$ be a smooth $d$-dimensional manifold and let
	\begin{equation*}
		\left\lbrace   C_{\left\lbrace I_1I_2\right\rbrace } \right\rbrace _{\left(  I_1I_2\right)  }
	\end{equation*}
be the cover of $\mathcal{M}^n\setminus d_n$  defined in Lemma \ref{covlemma}.
Then, there exists a subordinated partition of unity 
	\begin{equation*}
		\left\lbrace \chi_{I_1I_2} \right\rbrace_{\left( I_1I_2\right) } 
	\end{equation*}
such that every function $\chi_{I_1I_2}$ is tempered along $d_n$.  
\end{lem}
\begin{proof}
We will first construct a partition of unity in some neighborhood $U$ of $d_n$.
Consider the normal bundle  
	\begin{equation*}
		Nd_n=\coprod_{x\in d_n} T_{x}\mathcal{M}^n/T_xd_n 
	\end{equation*}
associated to the canonical embedding
	\begin{align*}
	i: d_n&\hookrightarrow \mathcal{M}^n\\
	 x&\mapsto  x
	\end{align*}
of $d_n$ as a closed subspace of $\mathcal{M}^n$.
Calling $ 0_{x}$ the null vector in 
	\begin{equation*}
		N_xd_n=T_{x}\mathcal{M}^n/T_xd_n
	\end{equation*}
we  denote by  $\zeta$ the mapping 
	\begin{align*}
	\zeta: d_n&\rightarrow  Nd_n\\
	 x&\mapsto  (x,0_{x})\mbox{,}
	\end{align*}
which embeds $d_n$ as a closed submanifold (the zero cross section,   $\zeta_0:=\zeta\left( d_n\right) $) of $Nd_n$.

We can think the normal bundle inside the tangent bundle over $\mathcal{M}^n$, using local coordinates, as follows. 
Let $\left\lbrace U_i\right\rbrace_i $ be some open cover of $\mathcal{M}$. 
Consider the diagonal map
	\begin{align*}
		\Delta: U_i&\rightarrow  \Delta U_i\\
		 x&\mapsto  \underbrace{(x,x,\dots ,x)}_{\mbox{$n$ times}}.
	\end{align*}
We trivialize the normal bundle over $\Delta U_i$ for every $i$ by means of the trivializing maps
	\begin{equation*}
		\Psi_{U_i}:N d_n|_{\Delta U_i}\rightarrow U_i\times \mathbb{R}^{d(n -1)}. 
	\end{equation*}	
In other words, we use local coordinates of the form 
	\begin{equation*}
		(x,h_1,\dots,h_{n-1}) \in U_i\times \mathbb{R}^{d(n-1)}\mbox{,}
	\end{equation*}
where $x$ in  $U_i$  represents the  point $\Delta(x)$ in $d_n$;  and  each $h_j$ belongs to $\mathbb{R}^d$ (for $j=1, \dots , n-1$) so that the vector $(h_1,\dots ,h_{n-1})$ belongs to $\mathbb{R}^{d(n-1)}$, and represents an element in the fiber attached to $\Delta(x)$.
Thus, the injection of $N d_n|_{\Delta U_i}$ in $\left.T\mathcal{M}^n\right|_{\Delta U_i\cap d_n}$ may be represented locally by an exact sequence
	\begin{equation*}
		\begin{array}{rcccl}
		0\rightarrow & U_i\times \mathbb{R}^{d(n-1)}&\overset{T}{\rightarrow} & U_i\times \mathbb{R}^{dn}&\rightarrow 0  \\
		&(x,h_1,\dots,h_{n-1})&\mapsto & (x,w_1,\dots,w_{n})&
	 	\end{array}
	\end{equation*}
where
	\begin{equation}\label{tantui}
		\left\{
		\begin{array}{rcl}
		w_i&=& h_i\mbox{,}\qquad\qquad\qquad\mbox{ if } \ i=1\mbox{,}\dots\mbox{, }n-1\mbox{,}\\
		w_n&=& -\sum_{i=1}^{n-1}h_i\ .
		\end{array}\right.
	\end{equation}
Now, let $\xi$ be a spray on $\mathcal{M}$ (see \cite{langdiff}, Thm. 3.1, for the existence of $\xi$).
For every $\upsilon$ in $T\mathcal{M} $ denote by $\beta_\upsilon$ the unique integral curve of $\xi$ such that $\beta_\upsilon(0)=\upsilon$.
We denote by $\mathfrak{D}$ the open subset of $T\mathcal{M}$ given by 
	\begin{equation*}
		\mathfrak{D}=\left\lbrace\upsilon \in T\mathcal{M}: \beta_\upsilon \mbox{ is defined at least on } [0,1] \right\rbrace. 
	\end{equation*} 
Then, the exponential map associated to the spray $\xi$ is defined by 
	\begin{align}
		\operatorname{exp}^\xi: \mathfrak{D}&\rightarrow  \mathcal{M}\\
		\upsilon &\mapsto  \pi_\mathcal{M}\left( \beta_\upsilon (1)\right) \mbox{,}\nonumber
	\end{align}
where $\pi_\mathcal{M}:T\mathcal{M}\rightarrow \mathcal{M}$ denotes the canonical projection.

The spray $\xi$ in turn induces a spray $\xi^n$ on $\mathcal{M}^n$ in a canonical way.  
For every $\upsilon$ in $T\mathcal{M}^n$, let $\beta^n_{\upsilon}$ denote the integral curve of $\xi^n$ 
such that $\beta^n_{\upsilon}(0)=\upsilon$.
The integral curve $\beta^n_{\upsilon}$ can be expressed in terms of the integral curves of $\xi$ in a simple way: if $\upsilon$ belongs to $T\mathcal{M}^n$ then there are elements $\upsilon_i$ in $T\mathcal{M}$ ($i=1\mbox{,}\dots \mbox{, }n$) with $\pi_i\upsilon=\upsilon_i$, where $\pi_i:T\mathcal{M}^n\rightarrow T\mathcal{M}$ denotes the projection to the $i$-th copy of $T\mathcal{M}$ in $T\mathcal{M}^n$.
In this case we then have
	\begin{equation*}
		\beta_{\upsilon}^n=\left(\beta_{\upsilon_1},\dots,\beta_{\upsilon_n} \right) 
	\end{equation*}
and the exponential map associated to the spray $\xi^n$ is given by
	\begin{align*}
		\operatorname{exp}^{\xi^n}: \mathfrak{D}^n&\rightarrow  \mathcal{M}^n\\
		\upsilon &\mapsto \pi_{\mathcal{M}^n}\left( \beta^n_\upsilon (1)\right) \mbox{,}
	\end{align*}
with
	\begin{equation*}
		\pi_{\mathcal{M}^n}\left( \beta^n_\upsilon (1)\right) =\Big(\pi_\mathcal{M}\left( \beta_{\upsilon_1}(1)\right) ,\dots,\pi_\mathcal{M}\left( \beta_{\upsilon_n}(1)\right)  \Big)=\Big(\operatorname{exp}^\xi\left( \upsilon_1\right) ,\dots,\operatorname{exp}^\xi\left( \upsilon_n\right) \Big)\mbox{,}
	\end{equation*}
where $\pi_{\mathcal{M}^n}:T\mathcal{M}^n\rightarrow \mathcal{M}^n$ denotes the canonical projection.

By the Tubular Neighborhood Theorem (see \cite{langdiff}, Ch. IV, Thm. 5.1) there is an open  neighborhood $Z$ (in the topology of $Nd_n$), contained in $\mathfrak{D}^n$ (viewed inside the tangent bundle $T\mathcal{M}^n$, as described above), of the zero cross section $\zeta_0$.
The set $Z$ is such that the restriction of the exponential map to the normal bundle $Nd_n$, denoted by $\left.\operatorname{exp}^{\xi^n}\right|_N$,  is a diffeomorphism of $Z$ onto an open set $U$ in $\mathcal{M}^n$,
	\begin{equation*}
		\left.\operatorname{exp}^{\xi^n}\right|_N : Z\rightarrow   U\mbox{,}
	\end{equation*}
which makes the following diagram
	\begin{equation*}
		\xymatrix{
		Z \ar@/^/[ddrr]^{\left.\operatorname{exp}^{\xi^n}\right|_N}   && &\\
		& & &\\
		d_n\ar@{->}[uu]^{\zeta}\ar@{->}[rr]_{i} && U\ar@/^/[uull]^{\Gamma}\ar@{^{(}->}[r] & \mathcal{M}^n 
		}
	\end{equation*}
commutative, where we have denoted by $\Gamma$ the inverse mapping of $\left.\operatorname{exp}^{\xi^n}\right|_N$.

Let us denote by $\dot{Z}$ the neighborhood $Z\setminus \zeta_0$. 
Then,
	\begin{equation}\label{aurora}
		\left\lbrace   \Gamma\left(  C_{\left\lbrace I_1I_2\right\rbrace }\right) \right\rbrace _{\left(  I_1I_2\right)  }
	\end{equation}
forms an open cover of $\dot{Z}$.

The next step is to find the relations between the  coordinates describing the points in the sets  \eqref{aurora}.
In other words, we seek to describe the sets 
	\begin{equation}\label{auroraaurora}
			V_{I_1I_2}^i:=\Psi_{U_i}\circ \Gamma\left(C_{\left\lbrace I_1I_2\right\rbrace }\right).
	\end{equation}
In order to do that, we will show that
	\begin{equation}\label{maxi}
	 	V_{I_1I_2}^i=\bigcap_{(i,j)\in I_1\times I_2}\left\lbrace (x,h_1,\dots,h_{n-1}) \in U_i\times \mathbb{R}^{d(n-1)}: \ h_i-h_j\neq 0 \right\rbrace.
	\end{equation}
To prove \eqref{maxi}, first notice that we have the following commutative diagram
	\begin{equation*}
		\xymatrix{
		\Gamma\left(  C_{\left\lbrace I_1I_2\right\rbrace }\right) \cap Z \cap \left.Nd_n\right|_{\Delta U_i}   \ar@{->}[dd]_{\left.\operatorname{exp}^\xi \right|_N}\ar@{->}[rrr]^{\Psi_{U_i}} &  &  &  V_{\left\lbrace I_1I_2\right\rbrace }^i\cap U_i\times \mathbb{R}^{d(n-1)}\ar@{->}[dd]^{T}\\
		&&&
		\\ C_{\left\lbrace I_1I_2\right\rbrace }\cap \left.\operatorname{exp}^{\xi^n}\right|_N\left(Z  \cap \left.Nd_n\right|_{\Delta U_i} \right)  \ar@{->}[rrr]^{}  & & &   T\left( V_{\left\lbrace I_1I_2\right\rbrace }^i\right) \cap U_i\times \mathbb{R}^{dn} 
		}
	\end{equation*}
from which we easily see that given a fixed point $(x_1,\dots ,x_n)$ belonging to the open subset of $U$
	\begin{equation*}
		C_{\left\lbrace I_1I_2\right\rbrace 	}\cap \left.\operatorname{exp}^{\xi^n}\right|_N\left(Z  \cap \left.Nd_n\right|_{\Delta U_i} \right)
	\end{equation*}
there correspond unique elements
	\begin{equation*}
		\upsilon=\left( \upsilon_1,\dots,\upsilon_n\right)\mbox{,}\ \left( x,h_1,\dots,h_{n-1}\right) \mbox{ and }  (x,w_1,\dots,w_{n}) 
	\end{equation*}
belonging to
	\begin{equation*}
		\Gamma\left(  C_{\left\lbrace I_1I_2\right\rbrace }\right)\cap Z \cap \left.Nd_n\right|_{\Delta U_i}\mbox{,}\
		V^i_{\left\lbrace I_1I_2\right\rbrace }\cap U_i\times \mathbb{R}^{d(n-1)}\  \mbox{ and }
		T\left( V^i_{\left\lbrace I_1I_2\right\rbrace }\right) \cap U_i\times \mathbb{R}^{dn}\mbox{,}
	\end{equation*}
respectively, with $\sum w_i=0$,  and such that
	\begin{equation*}
		\xymatrix{
		\upsilon=\left( \upsilon_1,\dots,\upsilon_n\right)   \ar@{|->}[dd]_{\left.\operatorname{exp}^\xi \right|_N}\ar@{|->}[rrr]^{\Psi_{U_i}} & &  &  \left( x,h_1,\dots,h_{n-1}\right) \ar@{|->}[dd]^{T}\\
		&&&
		\\ \left(\operatorname{exp}^\xi\left( \upsilon_1\right) ,\dots,\operatorname{exp}^\xi\left( \upsilon_n\right) \right) = (x_1,\dots ,x_n)  \ar@{|->}[rrr]^{}  & & &   (x,w_1,\dots,w_{n})}
	\end{equation*}
commutes.
Also note that 
	\begin{equation}\label{bucket}
		T\circ \Psi_{U_i}\Bigg( \! \Big(\! \left(\pi_\mathcal{M}(\upsilon_1),0\right) ,\dots,\left(\pi_\mathcal{M}(\upsilon_i),0 \right),\dots,\left(\pi_\mathcal{M}(\upsilon_n),0 \right) \! \Big)\!    \Bigg)\!=\!\left( x,0,\dots,w_i,\dots,0\right)\mbox{,}
	\end{equation}
for $i=1\mbox{,}\dots\mbox{, }n$.
Therefore,
	\begin{equation*}
		x_i=x_j \iff   \operatorname{exp}^\xi\left( \upsilon_i \right)=\operatorname{exp}^\xi\left( \upsilon_j \right) \iff \upsilon_i=\upsilon_j \iff w_i=w_j\iff  h_i=h_j\mbox{,}
	\end{equation*}
where in the last equivalence we have used the fact that $T$ is injective on each fiber $\mathbb{R}^{d(n-1)}$.

From the previous argument, \eqref{maxi} is proved and therefore, the sets \eqref{auroraaurora} are invariant under scalings of the type
	\begin{align*}
		\operatorname{Id}_{U_i}\times M_\lambda :\quad U_i\times \mathbb{R}^{d(n-1)}\
		&\rightarrow U_i\times \mathbb{R}^{d(n-1)}\\
		(x,h_1,\dots,h_{n-1})&\mapsto (x,\lambda h_1,\dots,\lambda h_{n-1})\mbox{,}
	\end{align*}
where $\lambda$ is a nonzero real number; 
and they are also invariant under translations of the type
	\begin{align}\label{actiontrans}
		\operatorname{Id}_{U_i}\times T_u:\quad
		U_i\times \mathbb{R}^{d(n-1)}\ &\rightarrow   U_i\times \mathbb{R}^{d(n-1)}\\
		(x,h_1,\dots,h_{n-1})&\mapsto  (x,u + h_1,\dots,u + h_{n-1})\mbox{,}\nonumber
	\end{align}
where $u$ is any element of the Euclidean space $\mathbb{R}^d$.
Therefore,
	\begin{equation}\label{conical}
		\Big\{ V_{\left\lbrace I_1I_2\right\rbrace }^i\Big\}_{(I_1,I_2)}
	 \end{equation}
is an open conical cover of
$U_i\times\left( \mathbb{R}^{d(n-1)}\setminus L \right)$, 
where 
	\begin{equation*}
		L=\left\lbrace\left( h_1,\dots,h_{n-1}\right)\in\mathbb{R}^{d(n-1)}:h_i=h_j \forall i\mbox{, }j  \right\rbrace
	\end{equation*}
is the small diagonal.

We seek to construct a partition of unity
	\begin{equation}\label{metengoquebañar}
		\Big\{\chi^i_{\left\lbrace I_1I_2\right\rbrace } \Big\} _{(I_1,I_2)}
	\end{equation}
subordinated to this cover, whose elements are tempered along $U_i\times L$.
By translation invariance of the sets \eqref{conical}, it is possible to construct this partition in such a way that its elements are  translation invariant.
Also, as the invariance of the sets \eqref{conical} under translations is independent of the base point $x$ in $U_i$, the functions \eqref{metengoquebañar} will not need to depend on this coordinate.
The way in which this translation symmetry simplifies the task is as follows.
We first consider the quotient space
	\begin{equation*}
		\mathfrak{Q}= \mathbb{R}^{d(n-1)}\biggr/ L
	\end{equation*}
under the action of $L$ by translations \eqref{actiontrans}. 
In particular, we focus on the subspace
	\begin{equation*}
		\mathfrak{Q}_L=\left( \mathbb{R}^{d(n-1)}\setminus L \right)\biggr/ L.
	\end{equation*}
Let
	\begin{equation*}
		p:\mathbb{R}^{d(n-1)} \rightarrow\mathfrak{Q}
	\end{equation*}
be the canonical projection.
We shall use square brakets to denote the class $p\left( \left(h_1,\dots, h_{n-1} \right)\right) $ of a given element $\left(h_1,\dots, h_{n-1} \right)$.
We then have that
	\begin{equation}
		\big[ \left(h_1,\dots, h_{n-1} \right)\big] = \left[ \left(h'_1,\dots, h'_{n-1} \right)\right] \mbox{ in }\mathfrak{Q}  \iff  \left(h_1-h'_1,\dots, h_{n-1}-h'_{n-1} \right)\in L. 
	\end{equation}

Next, consider the isomorphism
	\begin{align}\label{mofo}
		\phi: U_i\times \mathfrak{Q}& \rightarrow U_i\times \mathbb{R}^{d(n-2)}\\
		\left( x,\left[ h_1,\dots,h_{n-1}\right] \right) &\mapsto  \left( x,h_1-h_{n-1},\dots,h_{n-2}-h_{n-1}\right) \nonumber\mbox{,}
	\end{align}
under which the diagonal $L$ is mapped to the origin of the Euclidean space $\mathbb{R}^{d(n-2)}$ and the 
quotient space $\mathfrak{Q}_L$ is therefore  isomorphic to the punctured Euclidean space $\mathbb{R}^{d(n-2)}\setminus \left\lbrace 0 \right\rbrace$.
If we now consider the composition
	\begin{equation*}
		\Phi:U_i\times \mathbb{R}^{d(n-1)}\overset{\operatorname{Id}_{U_i}\times p}{\xrightarrow{\hspace*{1.5cm}}} U_i\times \mathfrak{Q}\overset{\operatorname{Id}_{U_i}\times \phi}{\xrightarrow{\hspace*{1.5cm}}} U_i\times \mathbb{R}^{d(n-2)}\mbox{,}
	\end{equation*}
then, the collection of sets 
	\begin{equation}\label{kit}
		\left\lbrace \Phi\left( V^i_{\left\lbrace I_1I_2\right\rbrace } \right) \right\rbrace _{(I_1,I_2)}
	\end{equation}
is an open covering of the space 
	\begin{equation*}
		\mathcal{M}\times \mathbb{R}^{d(n-2)}\setminus \left\lbrace 0\right\rbrace.
	\end{equation*}
Thus, it suffices to construct a partition of unity
	\begin{equation}\label{kitkat}
		\Big\{ \bar{\chi}^i_{\left\lbrace I_1I_2\right\rbrace }  \Big\}_{(I_1,I_2)}
	\end{equation}
subordinated to this open cover, such that  every function $\bar{\chi}^i_{\left\lbrace I_1I_2\right\rbrace }$ is tempered along $U_i\times\left\lbrace 0 \right\rbrace $ in $U_i\times \mathbb{R}^{d(n-2)}$.
If this was the case, we could take
	\begin{equation}\label{yamismo}
		\chi^i_{\left\lbrace I_1I_2\right\rbrace }:=\Phi^\ast\Big( \bar{\chi}^i_{\left\lbrace I_1I_2\right\rbrace }\Big) 
	\end{equation}
for every pair $(I_1,I_2)$, which is tempered along $U_i\times L$ by the chain rule, and then the lemma follows.

As the cover \eqref{conical} is conical and the fiber is independent of the base point $x$ it is possible to take every $\bar{\chi}^i_{\left\lbrace I_1I_2\right\rbrace }$ as a function only of the direction given by the vector component of each point $\left(x,h_1,\dots,h_{n-2} \right) $; \textit{i.e.} it can be taken of the form
	\begin{equation*}
		\bar{\chi}^i_{\left\lbrace I_1I_2\right\rbrace }(x,h)=f^i_{\left\lbrace I_1I_2\right\rbrace }\left(\frac{h}{|h|} \right)\mbox{,}
	\end{equation*}
where $h=\left(h_1,\dots,h{_{n-2}} \right) $,
	\begin{equation*}
		|h|=\sqrt{\sum_{l=1}^{n-2}h_l^2},\ \mbox{ and }\ f^i_{\left\lbrace I_1I_2\right\rbrace } \in \mathscr{E}\left(\mathbb{R}^{d(n-2)}\setminus \left\lbrace 0 \right\rbrace \right).
	\end{equation*}
Therefore, combining the Faa Di Bruno formula
\footnote{The Faa Di Bruno formula is an identity generalizing the chain rule to higher derivatives.
Here we use the generalised Faa Di Bruno formula which is an extension of the former to multivariate functions. 
The reader is referred to  \url{http://digital.csic.es/bitstream/10261/21265/3/FaadiBruno\%20-\%20copia.pdf}}
and the fact that 
	\begin{equation*}
		\left|\partial_h^k\left( \frac{h}{|h|}\right) \right| \leq C_{k} \left( 1+|h|^{-|k|}\right)\mbox{,}
	\end{equation*}
we get
	\begin{equation*}
		\left|\partial_h^\alpha \bar{\chi}^i_{\left\lbrace I_1I_2\right\rbrace }(h) \right| =\left|\partial_h^\alpha f^i_{\left\lbrace I_1I_2\right\rbrace }\left(\frac{h}{|h|} \right)  \right|\leq C_{\alpha} \left( 1+|h|^{-|\alpha|}\right).
	\end{equation*}
This implies that $\bar{\chi}^i_{\left\lbrace I_1I_2\right\rbrace }$ is tempered along $U_i\times \left\lbrace 0 \right\rbrace$, and so each function \eqref{yamismo} is tempered along $U_i\times L$.
Therefore, 
	\begin{equation*}
		\Psi^\ast_{U_i}\chi^i_{\left\lbrace I_1I_2\right\rbrace }
	\end{equation*}
is tempered along the zero cross section  $\left.\zeta_0\right| _{U_i}$ contained in $\left.Nd_n\right| _{\Delta U_i}$.
 
Let $\left\lbrace \varphi_i\right\rbrace_i $  be a partition of unity subordinated to the cover $\left\lbrace U_i\right\rbrace _i$ of $\mathcal{M}$.
Then,
	\begin{equation*}
		\left\lbrace\sum_i \varphi_i \Psi^\ast_{U_i}\chi^i_{\left\lbrace I_1I_2\right\rbrace } \right\rbrace_{(I_1,I_2)} 
	\end{equation*}
 
is a partition of unity of $Nd_n\setminus \zeta_0$ which is subordinated to the conical cover
	\begin{equation*}
		\left\lbrace \Gamma\left(  C_{\left\lbrace I_1I_2\right\rbrace }\right)  \right\rbrace _{\left( I_1,I_2\right) }.
	\end{equation*}
To go back to the configuration space $\mathcal{M}^n$, choose a neighborhood $U'$ of $d_n$ such that 
	\begin{equation*}
		\overline{U'}\subseteq U.
	\end{equation*}
We then have the inclusions 
	\begin{equation*}
		d_n\subseteq U'\subseteq U.
	\end{equation*}
Let $\left\lbrace\chi_{\alpha},\chi_{\beta} \right\rbrace $ be a  partition of unity subordinated to the cover 
	\begin{equation*}
		\left\lbrace U, \mathcal{M}^n\setminus \overline{U'}\right\rbrace
	\end{equation*}
and choose 
	\begin{equation*}
		\left\lbrace \tilde{\chi}_{\left\lbrace I_1I_2\right\rbrace }\right\rbrace_{(I_1,I_2)}
	\end{equation*}
to be an arbitrary partition of unity subordinated to the cover 
	\begin{equation*}
		\left\lbrace  C_{\left\lbrace I_1I_2\right\rbrace }  \right\rbrace _{\left( I_1,I_2\right) }
	\end{equation*}
of $\mathcal{M}^n\setminus d_n$.
Finally,  set 
	\begin{equation*}
		\chi_{\left\lbrace I_1I_2\right\rbrace }= \chi_{\alpha}\Gamma^\ast\left( \sum_i \varphi_i \Psi^\ast_{U_i}\chi^i_{\left\lbrace I_1I_2\right\rbrace } \right)+\chi_{\beta} \tilde{\chi}_{\left\lbrace I_1I_2\right\rbrace }\mbox{,}
	\end{equation*}
where $I_1\mbox{, }I_2$ run through all proper  subsets of $\iota_n$ such that $I_1\sqcup I_2=\iota_n$.
Then, it follows by construction that every function $\chi_{\left\lbrace I_1I_2\right\rbrace }$ is tempered along $d_n$.
Moreover, the sum over all such subsets $I_1\mbox{, }I_2$ gives one, since
	\begin{equation*}
		\begin{split}
		\sum_{I_1I_2}  \chi_{\left\lbrace I_1,I_2\right\rbrace }&= \sum_{I_1I_2} \chi_{\alpha}\Gamma^\ast\left( \sum_i \varphi_i \Psi^\ast_{U_i}\chi^i_{\left\lbrace I_1I_2\right\rbrace } \right)+\sum_{I_1I_2}\chi_{\beta} \tilde{\chi}_{\left\lbrace I_1I_2\right\rbrace }\\ &=\chi_{\alpha}\Gamma^\ast\left( \sum_i \varphi_i \Psi^\ast_{U_i}\left( \sum_{I_1I_2}\chi^i_{\left\lbrace I_1I_2\right\rbrace } \right) \right)+\chi_{\beta} \sum_{I_1I_2}\tilde{\chi}_{\left\lbrace I_1I_2\right\rbrace }\\
		&=\chi_{\alpha}\Gamma^\ast\left( \sum_i \varphi_i \Psi^\ast_{U_i}\left(1\right) \right)+\chi_{\beta}=\chi_{\alpha}\Gamma^\ast\left( \sum_i \varphi_i \right)+\chi_{\beta}\\
		&=\chi_{\alpha}\Gamma^\ast\left( 1\right)+\chi_{\beta}=\chi_{\alpha}+\chi_{\beta}=1.
		\end{split}
	\end{equation*}
This completes the proof.
\end{proof}

\subsection{Proof of the existence of renormalization maps }\label{lademospostaa}
Now we are in a position to give the proof of Theorem \ref{exrenmap}.
\begin{proof}[Proof of Theorem \ref{exrenmap}]
Recall that by Remark \ref{eldato}, we need only to define the extension $\mathcal{R}_{n}\left(G_{\iota_n}\right)$ for every $n$ in $\mathbb{N}$, where $\iota_n:=\left\lbrace 1,\dots ,n\right\rbrace$.
In addition, this must be done in such a way that the extension satisfies the factorization axiom (see \eqref{todor}).
We proceed by induction on $n$, starting with $n=2$. 
As mentioned earlier,  the existence of  $\mathcal{R}_{2}\left( G_{2}\right)$ 
is guaranteed by Theorem \ref{dinsky}.
The factorization axiom is satisfied trivially in this case by the extra assumption \eqref{convention}.

For the inductive step, we assume recursively that the problem of extension is already solved for every proper subset  $ J$ of $\iota_n$ ($n>2$). 
Namely, we suppose that for every such $J$ we are given a distribution $\mathcal{R}_{J}\left(  G_{J}\right)$ in $\mathscr{D}\left( \mathcal{M}^J\right)' $ with the property that, for every nontrivial partition $J_1\sqcup J_2 = J$, equation \eqref{todor} holds with $I$, $I_1$ and $I_2$ replaced by $J$, $J_1$ and $J_2$, respectively.

By Lemma \ref{covlemma}, the complement $\mathcal{M}^n\setminus d_n$ of the small diagonal $d_n$ in $\mathcal{M}^n$ is covered by open sets of the form
	\begin{equation}\label{cij2}
		C_{\left\lbrace I_1I_2\right\rbrace }=\left\lbrace \left( x_1,\dots ,x_n\right): \forall (i,j)\in I_1\times I_2 \ x_i\neq x_j \right\rbrace \subseteq \mathcal{M}^n\mbox{,}
	\end{equation}
where $I_1\mbox{, }I_2$ run through all proper  subsets of $\iota_n$ such that $I_1\sqcup I_2=\iota_n$.  
In Lemma \ref{parlemma} we constructed a partition of unity  
	\begin{equation*}
		\left\lbrace \chi_{I_1I_2} \right\rbrace_{\left( I_1,I_2\right) }
	\end{equation*}
subordinated to this cover, such that every function $\chi_{I_1I_2}$ is tempered along $d_n$.

The key idea is that the product 
	\begin{equation*}
		\mathcal{R}_{I_1}\left(  G_{I_{1}}\right)  \otimes \mathcal{R}_{I_2}\left(  G_{I_{2}} \right)
	\end{equation*}
is well-defined in $\mathscr{D}\left( \mathcal{M}^n\right)'$ and $G_{\left\lbrace I_1,I_2\right\rbrace }$
is tempered along $\partial C_{\left\lbrace I_1,I_2\right\rbrace }$.
Then,
	\begin{equation*}
		t_{I_1I_2}:=\underbrace{\left( \mathcal{R}_{I_1}\left(  G_{I_1}\right)  \otimes \mathcal{R}_{I_2}\left(  G_{I_2} \right)\right) }_{\in \mathscr{D}\left( \mathcal{M}^n\right)' }     \underbrace{\left( \chi_{I_1I_2}G_{\left\lbrace I_1,I_2\right\rbrace }\right)}_{\in \mathcal{T}\left( \partial C_{\left\lbrace I_1I_2\right\rbrace },\mathcal{M}^n\right)}
	\end{equation*} 
is a product of a tempered function along $\partial C_{\left\lbrace I_1,I_2\right\rbrace }$ and a distribution in $\mathscr{D}\left( \mathcal{M}^n\right)' $. 
Therefore, it has a continuous extension $\bar{t}_{I_1I_2}$ to $\mathscr{D}\left( \mathcal{M}^n\right)$ by Theorem \ref{maggotbrain}, which is supported on  
	\begin{equation*}
		K_{I_1I_2}:=\operatorname{Supp}\left( \chi_{I_1I_2}\right)\mbox{,}
	\end{equation*}
in which case we may write 
	\begin{equation*}
		\bar{t}_{I_1I_2}\in\mathscr{D}'_{K_{I_1I_2}}\left( \mathcal{M}^n\right) .
	\end{equation*}
By construction, $\chi_{I_1I_2}$ vanishes in some neighborhood of $\partial C_{\left\lbrace I_1,I_2\right\rbrace }\setminus d_n$ in $\mathcal{M}^n\setminus d_n $ which implies  that the equality 
	\begin{equation*}
		t_{I_1I_2}=\left( \mathcal{R}_{I_1}\left(  G_{I_1}\right)  \otimes \mathcal{R}_{I_2}\left(  G_{I_2} \right) \right)  \left( \chi_{I_1I_2}G_{\left\lbrace I_1,I_2\right\rbrace }\right)
		=\bar{t}_{I_1I_2} 
	\end{equation*}
holds in $\mathscr{D}\left( \mathcal{M}^n\setminus d_n\right)'$. 
Then, we define $\mathcal{R}_{\iota_n}\left(G_{\iota_n} \right)$ to be the distribution given by
	\begin{equation}\label{pudao1}
		\mathcal{R}_{\iota_n}\left(G_{\iota_n} \right) = \sum_{I_1,I_2} \bar{t}_{I_1I_2}\mbox{,}
	\end{equation}
where $I_1\mbox{, }I_2$ run through all proper  subsets of $\iota_n$ such that $I_1\sqcup I_2=\iota_n$.

Now we verify that the extension \eqref{pudao1} satisfies the factorization axiom \eqref{todor}. 
Fix $I_1$ and $I_2$ two proper  subsets of $\iota_n$ such that $I_1\sqcup I_2=\iota_n$.
Then,
	\begin{equation*}
		\begin{split}
			\left.\mathcal{R}_{\iota_n}\left(G_{\iota_n} \right) \right\vert_{C_{\left\lbrace I_1I_2\right\rbrace }}&= \sum_{J_1,J_2} \left.\bar{t}_{J_1J_2}\right\vert_{C_{\left\lbrace I_1I_2\right\rbrace }}\\ &=\sum_{J_1,J_2}\left. \left( \mathcal{R}_{J_1}\left(  G_{J_1}\right)  \otimes \mathcal{R}_{J_2}\left(  G_{J_2} \right) \right)  \left( \chi_{J_1J_2}G_{\left\lbrace J_1,J_2\right\rbrace }\right) \right\vert_{C_{\left\lbrace I_1I_2\right\rbrace }} \\
			&=\sum_{J_1,J_2}\chi_{J_1J_2}\left. \left[ \left( \mathcal{R}_{J_1}\left(  G_{J_1}\right)  \otimes \mathcal{R}_{J_2}\left(  G_{J_2} \right) \right)   G_{\left\lbrace J_1,J_2\right\rbrace } \right] \right\vert_{C_{\left\lbrace I_1I_2\right\rbrace }}\\
			&=\sum_{J_1,J_2} \chi_{J_1J_2}\left. \left[ \left( \mathcal{R}_{I_1}\left(  G_{I_1}\right)  \otimes \mathcal{R}_{I_2}\left(  G_{I_2} \right) \right)   G_{\left\lbrace I_1,I_2\right\rbrace } \right] \right\vert_{C_{\left\lbrace I_1I_2\right\rbrace }} \\
			&=\left. \left( \mathcal{R}_{I_1}\left(  G_{I_1}\right)  \otimes \mathcal{R}_{I_2}\left(  G_{I_2} \right) \right)   G_{\left\lbrace I_1,I_2\right\rbrace }  \right\vert_{C_{\left\lbrace I_1I_2\right\rbrace }}\mbox{,}
		\end{split}
	\end{equation*}
where $J_1\mbox{, }J_2$ in the sums above run through all proper  subsets of $\iota_n$ such that $J_1\sqcup J_2=\iota_n$. 
We have used the consistency relation \eqref{consistency2} given by Lemma \ref{consistency} and the fact that $\sum \chi_{J_1J_2}=1$.

The theorem is thus proved.
\end{proof}
%

\clearpage
\chapter*{List of symbols and acronyms}
\markboth{LIST OF SYMBOLS AND ACRONYMS}{LIST OF SYMBOLS AND ACRONYMS}
\addcontentsline{toc}{chapter}{List of symbols and acronyms}
\renewcommand*{\arraystretch}{1.37}
\setlength\LTleft{0pt}
\setlength\LTright{0pt}
\begin{longtable}{  p{3cm}@{\hspace{6mm}}  p{9cm}l  }
\textbf{QFT}& quantum field theory& \pageref{introooo}\\
\textbf{pQFT}& perturbative quantum field theory& \pageref{introooo}\\
\textbf{TVS}& topological vector space& \pageref{trutru}\\
\textbf{LCS}& locally convex space& \pageref{triojgh}\\
$\bm{\psi_\ast(\zeta)}$& pushforward of the object $\zeta$ &\pageref{pulpush}\\
${\bm{\psi^\ast(\zeta)}}$& pullback of the object $\zeta$ & \pageref{pulpush}\\
$\bm{T_{h}}$& translation operator& \pageref{trasla}\\
$\bm{M_{\lambda}}$&  multiplication operator& \pageref{multi}\\
$\bm{\mu_V}$& Minkowski functional& \pageref{minki}\\
 \bm{$E/F$}& quotient space of $E$ modulo $F$& \pageref{espcoc1}\\
$\bm{\sigma(E',E)}\mbox{ or }\bm{\sigma^\ast}$& weak star topology& \pageref{wstdef}\\
$\bm{b^\ast}$& strong topology&  \pageref{stdef}\\
$\bm{E'}$& continuous dual space of a topological vector space $E$& \pageref{kutiii}\\
$\bm{u'}$& transpose of the map $u$& \pageref{trdef}\\
$\bm{M^{\perp}}$, $\bm{^{\perp}N}$& orthogonal of a subspace& \pageref{orthdef}\\
$\bm{M^{\circ}}$, $\bm{^{\circ}N}$& polar of a set& \pageref{polarsdef}\\

$\bm{\ell_{\infty}}$& space of bounded complex valued sequences& \pageref{linfc0}\\
$\bm{\ell_{1}}$& space of complex valued sequences of absolutely convergent associated series & \pageref{elle12}, \pageref{elle1}\\
$\bm{c_0}$& space of complex valued sequences whose limit is zero& \pageref{linfc0}\\
$\bm{ba}$& space of finitely additive signed measures& \pageref{badef}\\
$\bm{ca}$& space of countably additive and bounded signed measures& \pageref{cadef}\\

$\bm{\mathscr{E}\left( \Omega\right)}$&  space of infinitely differentiable, or $\bm{\mathcal{C}^\infty}$ complex valued  functions on $\Omega$& \pageref{funcinf1}, \pageref{defmani}\\
$\bm{\mathscr{E}^m\left( \Omega\right)}$ &  space of $m$ times continuously differentiable, or $\mathcal{C}^m$ complex valued  functions on $\Omega$& \pageref{funcinf1}, \pageref{defmani}\\
$\bm{\mathcal{I}^\infty \left( X,\mathbb{R}^d\right)}$&  space of $\bm{\mathcal{C}^\infty}$ functions which vanish on $X$ together with all of their derivatives& \pageref{funcinf1}\\
$\bm{\mathcal{I}^m \left( X,\mathbb{R}^d\right)}$& space of $\mathcal{C}^m$ functions which vanish on $X$ together with all of their derivatives of order less than or equal to $m$& \pageref{funcinf1}\\
$\bm{\mathcal{I}\left( X,\mathbb{R}^d\right)}$& space of $\mathcal{C}^\infty$ functions whose support does not intersect $X$& \pageref{funcinf1}\\
$\bm{\mathscr{D}\left( \Omega\right)}$& space of test functions& \pageref{nnnnn}, \pageref{defmani}\\
$\bm{\mathscr{D}_K}$, $\mathscr{D}_K\left( \Omega\right)$& spaces of test functions whose support lies in $K$& \pageref{nnnnn}, \pageref{defmani}\\
$\bm{\mathscr{D}\left(\Omega\right)'}$& space of distributions& \pageref{distrdef}, \pageref{mariacallas}\\
$\bm{\mathscr{D}'_K\left( \Omega\right)}$& space of distributions with compact support contained in $K$& \pageref{distrcomp}\\

$\bm{J^m(K)}$& space of jets of order $m$& \pageref{stravinsky}\\
$\bm{ \mathscr{E}^m\left( K\right) }$& space of differentiable functions of order $m$ in the sense of Whitney& \pageref{frida}\\

$\bm{\mathcal{T}_{\mathcal{M}\setminus X}(\mathcal{M})}$& space of distributions with moderate growth along a closed set $X$ of a manifold $\mathcal{M}$& \pageref{virchulina}\\
$\bm{\mathcal{T}(X,\mathcal{M})}$& algebra of tempered functions along a closed set $X$ of a manifold $\mathcal{M}$& \pageref{tempfunc}\\
$\bm{D_I}$, $\bm{D_n}$& big diagonal& \pageref{defidiag}\\
$\bm{d_{I,J}}$, $\bm{\ d_{\left\lbrace i,j,\dots,k\right\rbrace}}$, $\bm{\ d_{n,\left\lbrace i,j,\dots,k\right\rbrace}}$, $\bm{\ d_n}$& small diagonal& \pageref{defidiag} \\
$\bm{\mathcal{O}(D_I,\Omega)}$& algebra generated by Feynman amplitudes and $\mathcal{C}^\infty$ functions& \pageref{moustache}
\end{longtable}

\newpage
\addcontentsline{toc}{chapter}{Index}
\printindex


\begin{bibdiv}
\begin{biblist}
\addcontentsline{toc}{chapter}{Bibliography}
\bib{bastiani}{article}{
   author={Bastiani, Andr{\'e}e},
   title={Applications diff\'erentiables et vari\'et\'es diff\'erentiables de dimension infinie},
   language={French},
   journal={J. Analyse Math.},
   volume={13},
   date={1964},
   pages={1--114},
   }
\bib{bierstone}{article}{
   author={Bierstone, Edward},
   title={Differentiable functions},
   journal={Bol. Soc. Brasil. Mat.},
   volume={11},
   number={2},
   publisher={Springer-Verlag, New York},
   date={1980},
   pages={139--189},
}

\bib{brunetti}{article}{
   author={Brunetti, Romeo},
   author={Fredenhagen, Klaus},
   title={Microlocal analysis and interacting quantum field theories:
   renormalization on physical backgrounds},
   journal={Comm. Math. Phys.},
   volume={208},
   date={2000},
   number={3},
   pages={623--661},
}

\bib{conway}{book}{
   author={Conway, John B. },
   title={A Course in Functional Analysis.},
   series={Graduate Texts in Mathematics},
   volume={96},
   edition={2},
   publisher={Springer-Verlag, New York},
   date={1990},
   pages={xvi+399},
}
\bib{vietdang}{article}{
   author={Dang, Viet Nguyen},
   title={Extension of distributions, scalings and renormalization of QFT on Riemannian manifolds},
   date={2014},
   eprint={http://arxiv.org/abs/1411.3670},
}
\bib{diestel}{book}{
   author={Diestel, Joseph},
   title={Sequences and series in Banach spaces},
   series={Graduate Texts in Mathematics},
   volume={92},
   publisher={Springer-Verlag, New York},
   date={1984},
   pages={xii+261},
   isbn={0-387-90859-5},
   review={\MR{737004}},
   doi={10.1007/978-1-4612-5200-9},
}
\bib{dieu}{book}{
   author={Dieudonn{\'e}, J.},
   title={Treatise on analysis. Vol. III},
   note={Translated from the French by I. G. MacDonald;
   Pure and Applied Mathematics, Vol. 10-III},
   publisher={Academic Press, New York-London},
   date={1972},
   pages={xvii+388},
}
\bib{dutsch}{article}{
   author={D{\"u}tsch, Michael}, 
   title={Connection between the renormalization groups of St\"uckelberg-Petermann and Wilson}, 
   journal={Confluentes Math.}, 
   volume={4}, 
   date={2012}, 
   number={1}, 
   pages={1240001, 16}, 
   issn={1793-7442},  
   review={\MR{2921528}}, 
   doi={10.1142/S1793744212400014}
   }
\bib{grothe}{article}{
   author={Grothendieck, Alexander},
   title={Sur quelques points d'alg\`ebre homologique},
   journal={The T\^ohoku Math. J. (2)},
   language={French},
   series={Second Series},
   volume={9},
   date={1957},
   pages={119--221},
}
\bib{hormander}{book}{
    author={H{\"o}rmander, Lars},
    title={The analysis of linear partial differential operators. I},
    series={Grundlehren der Mathematischen Wissenschaften [Fundamental
    Principles of Mathematical Sciences]},
    volume={256},
    edition={2},
    note={Distribution theory and Fourier analysis},
    publisher={Springer-Verlag, Berlin},
    date={1990},
    pages={xii+440},
 }
 \bib{langlang93}{book}{
   author={Lang, Serge},
   title={Real and functional analysis},
   series={Graduate Texts in Mathematics},
   volume={142},
   edition={3},
   publisher={Springer-Verlag, New York},
   date={1993},
   pages={xiv+580},
}
\bib{langlang}{book}{
   author={Lang, Serge},
   title={Undergraduate analysis},
   series={Undergraduate Texts in Mathematics},
   edition={2},
   publisher={Springer-Verlag, New York},
   date={1997},
   pages={xvi+642},
}
 \bib{langdiff}{book}{
    author={Lang, Serge},
    title={Fundamentals of differential geometry},
    series={Graduate Texts in Mathematics},
    volume={191},
    publisher={Springer-Verlag, New York},
    date={1999},
    pages={xviii+535},
 }
\bib{malgrange}{book}{
   author={Malgrange, B.},
   title={Ideals of differentiable functions},
   series={Tata Institute of Fundamental Research Studies in Mathematics,
   No. 3},
   publisher={Tata Institute of Fundamental Research, Bombay; Oxford
   University Press, London},
   date={1967},
}


\bib{meisevogt}{book}{
   author={Meise, Reinhold},
   author={Vogt, Dietmar},
   title={Introduction to functional analysis},
   series={Oxford Graduate Texts in Mathematics},
   volume={2},
   note={Translated from the German by M. S. Ramanujan and revised by the
   authors},
   publisher={The Clarendon Press, Oxford University Press, New York},
   date={1997},
   pages={x+437},
}
\bib{Mi}{article}{
   author={Michal, A. D.},
   title={Differential calculus in linear topological spaces},
   journal={Proc. Nat. Acad. Sci. U. S. A.},
   volume={24},
   date={1938},
   pages={340--342},
}
\bib{Mi2}{article}{
   author={Michal, A. D.},
   title={Differentials of functions with arguments and values in
   topological abelian groups},
   journal={Proc. Nat. Acad. Sci. U. S. A.},
   volume={26},
   date={1940},
   pages={356--359},
}
\bib{nikolov}{article}{
   author={Nikolov, Nikolay M.},
   title={Anomalies in quantum field theory and cohomologies of
      configuration spaces},
   date={2009},
   eprint={http://arxiv.org/abs/0903.0187},
}
\bib{Nik09}{article}{
   author={Nikolov, Nikolay M.},
   title={Talk on anomalies in quantum field theory and cohomologies of
   configuration spaces},
   journal={Bulg. J. Phys.},
   volume={36},
   date={2009},
   number={3},
   pages={199--213},
}
\bib{todorov}{article}{
   author={Nikolov, Nikolay M.},
   author={Stora, Raymond},
   author={Todorov, Ivan},
   title={Renormalization of massless Feynman amplitudes in configuration space},
   journal={Rev. Math. Phys.},
   volume={26},
   date={2014},
   number={4},
   pages={1430002, 65},
}
\bib{quillen}{book}{
   author={Quillen, Daniel},
   title={Higher algebraic $K$-theory. I},
   conference={
       title={Algebraic $K$-theory, I: Higher $K$-theories},
       address={Proc. Conf., Battelle Memoirs Inst., Seattle, Wash.},
       date={1972},
       },
   book={
      publisher={Springer, Berlin},
   },
   date={1973},
   pages={85--147. Lecture Notes in Math., Vol: 341},
   volume={341},
}



\bib{rudin}{book}{
   author={Rudin, Walter},
   title={Functional analysis},
   series={International Series in Pure and Applied Mathematics},
   edition={2},
   publisher={McGraw-Hill, Inc., New York},
   date={1991},
   pages={xviii+424},
}


\bib{shimakura}{book}{
   author={Shimakura, Norio},
   title={Partial differential operators of elliptic type},
   series={Translations of Mathematical Monographs},
   volume={99},
   note={Translated and revised from the 1978 Japanese original by the
   author},
   publisher={American Mathematical Society, Providence, RI},
   date={1992},
   pages={xiv+288},
}
\bib{treves}{book}{
   author={Tr{\`e}ves, Fran{\c{c}}ois},
   title={Topological vector spaces, distributions and kernels},
   note={Unabridged republication of the 1967 original},
   publisher={Dover Publications, Inc., Mineola, NY},
   date={2006},
   pages={xvi+565},
}

\bib{vandijk}{book}{
   author={van Dijk, Gerrit},
   title={Distribution theory},
   series={De Gruyter Graduate Lectures},
   note={Convolution, Fourier transform, and Laplace transform},
   publisher={De Gruyter, Berlin},
   date={2013},
   pages={viii+109},
}
\bib{zeidler}{book}{ 
author={Zeidler, Eberhard}, 
title={Quantum field theory. II. Quantum electrodynamics},
note={A bridge between mathematicians and physicists}, 
publisher={Springer-Verlag, Berlin}, date={2009}, 
pages={xxxviii+1101}, isbn={978-3-540-85376-3}, 
review={\MR{2456465 (2010a:81002)}}
}
\end{biblist}
\end{bibdiv}
\end{document}